\documentclass[letterpaper,11pt]{article}
\usepackage[margin=0.94in]{geometry}

\newcommand{\alex}[1]{{\color{red}\textsf{Alex: {#1}}}}

\usepackage{cmap} 
\usepackage[utf8]{inputenc}
\usepackage[english]{babel}
\usepackage[T1]{fontenc}
\usepackage{amsmath}
\usepackage{amssymb}
\usepackage{amsfonts}
\usepackage{dsfont}
\usepackage{wasysym}
\usepackage{bm}
\usepackage{xcolor}
\definecolor{blueviolet}{rgb}{0.2, 0.2, 0.6}
\definecolor{webgreen}{rgb}{0,.5,0}
\definecolor{webbrown}{rgb}{.6,0,0}
\usepackage{setspace}
\usepackage[pdftex,
	bookmarks=false,
	colorlinks=true, 
	urlcolor=webbrown, 
	linkcolor=blueviolet, 
	citecolor=webgreen,
	pdfstartpage=1,
	pdfstartview={FitH},  
	bookmarksopen=false
	]{hyperref}
\allowdisplaybreaks
\usepackage{braket}
\usepackage[numbers,sort&compress]{natbib}
\usepackage{amsthm}
\usepackage{dsfont}
\usepackage{listings}
\usepackage[capitalize]{cleveref}
\usepackage{appendix}

\usepackage{graphicx}
\usepackage{physics}
\usepackage{mathtools}
\usepackage[normalem]{ulem}
\usepackage{wrapfig}
\usepackage{comment}
\usepackage{thmtools,thm-restate}
\usepackage[export]{adjustbox}
\usepackage{bbm}
\usepackage{authblk}

\numberwithin{equation}{section}

\newtheorem{theorem}{Theorem}
\newtheorem{corollary}{Corollary}

\newtheorem{definition}{Definition}

\newtheorem{lemma}{Lemma}
\newtheorem{proposition}{Proposition}
\newtheorem{fact}{Fact}

\theoremstyle{definition}

\newcommand{\rom}[1]{\mathtt{\uppercase\expandafter{\romannumeral #1\relax}}}

\DeclareMathOperator*{\E}{{\mathbb{E}}}

\DeclareMathOperator{\poly}{poly}

\providecommand{\bij}{\mathsf{bij}}

\providecommand{\spfo}{\mathsf{pfO}}
\providecommand{\pf}{\mathsf{pf}}

\newcommand{\eps}{\varepsilon}
\newcommand{\noEPR}{\Pi^{\text{nE}}}
\newcommand{\oneEPR}{\Pi^{>1\text{E}}}
\newcommand{\noP}{P^{\text{nE}}}
\newcommand{\noI}{I^{\text{nE}}}

\DeclareMathOperator{\Wg}{Wg}

\usepackage{xargs}
\usepackage[colorinlistoftodos,prependcaption,textsize=tiny]{todonotes}
\newcommandx{\lz}[2][1=]{\todo[linecolor=red,backgroundcolor=red!10,bordercolor=red,#1]{LZ: #2}}

\newtheorem*{theorem*}{Theorem}
\newtheorem{conjecture}{Conjecture}

\newtheorem*{task*}{Task}
\newtheorem*{proposition*}{Proposition}

\providecommand{\scC}{\mathsf{cC}}
\providecommand{\scD}{\mathsf{cD}}

\newcommand{\frakD}{\mathfrak{D}}
\newcommand{\init}{\mathsf{init}}

\newcommand{\ee}{\end{equation}}

\newcommand{\calO}{\mathcal{O}}

\providecommand{\Id}{\mathbbm{1}}

\providecommand{\num}{\mathsf{num}}
\providecommand{\sSym}{\mathsf{Sym}}
\providecommand{\ssym}{\mathsf{sym}}

\providecommand{\lrfo}{\mathsf{lrfO}}

\providecommand{\Compress}{\mathsf{Compress}}

\providecommand{\sSym}{\mathsf{Sym}}
\providecommand{\ssym}{\mathsf{sym}}

\providecommand{\TD}{\mathsf{TD}}

\providecommand{\sA}{\mathsf{A}}
\providecommand{\sB}{\mathsf{B}}

\providecommand{\sF}{\mathsf{F}}

\providecommand{\sH}{\mathsf{H}}

\providecommand{\sL}{\mathsf{L}}

\providecommand{\sP}{\mathsf{P}}
\providecommand{\sQ}{\mathsf{Q}}
\providecommand{\sR}{\mathsf{R}}

\providecommand{\sW}{\mathsf{W}}
\providecommand{\sX}{\mathsf{X}}
\providecommand{\sY}{\mathsf{Y}}

\DeclareMathOperator{\Dom}{Dom}

\providecommand{\gsA}{{\textcolor{gray}{\mathsf{A}}}}
\providecommand{\gsA}{{\textcolor{gray}{\mathsf{A}_{\mathsf{abc}}}}}
\providecommand{\gsB}{{\textcolor{gray}{\mathsf{B}}}}
\providecommand{\gsC}{{\textcolor{gray}{\mathsf{C}}}}
\providecommand{\gsD}{{\textcolor{gray}{\mathsf{D}}}}

\providecommand{\gsF}{{\textcolor{gray}{\mathsf{F}}}}

\providecommand{\gsH}{{\textcolor{gray}{\mathsf{H}}}}

\providecommand{\gsL}{{\textcolor{gray}{\mathsf{L}}}}
\providecommand{\gsLabc}{{\textcolor{gray}{\mathsf{L}_{\mathsf{abc}}}}}
\providecommand{\gsLab}{{\textcolor{gray}{\mathsf{L}_{\mathsf{ab}}}}}
\providecommand{\gsLbc}{{\textcolor{gray}{\mathsf{L}_{\mathsf{bc}}}}}
\providecommand{\gsRabc}{{\textcolor{gray}{\mathsf{R}_{\mathsf{abc}}}}}
\providecommand{\gsRab}{{\textcolor{gray}{\mathsf{R}_{\mathsf{ab}}}}}
\providecommand{\gsRbc}{{\textcolor{gray}{\mathsf{R}_{\mathsf{bc}}}}}

\providecommand{\gsP}{{\textcolor{gray}{\mathsf{P}}}}

\providecommand{\gsR}{{\textcolor{gray}{\mathsf{R}}}}

\providecommand{\gsX}{{\textcolor{gray}{\mathsf{X}}}}
\providecommand{\gsY}{{\textcolor{gray}{\mathsf{Y}}}}

\providecommand{\calA}{\mathcal{A}}

\providecommand{\calD}{\mathcal{D}}

\providecommand{\calH}{\mathcal{H}}
\providecommand{\calI}{\mathcal{I}}

\providecommand{\calO}{\mathcal{O}}

\providecommand{\calR}{\mathcal{R}}

\providecommand{\calU}{\mathcal{U}}

\DeclareMathOperator{\dist}{{dist}}
\DeclareMathOperator{\lcdist}{{lcdist}}

\def\:={\,\raisebox{0.85pt}{.}\hspace{-2.78pt}\raisebox{2.85pt}{.}\!\!=\,}
\def\=:{\,=\!\!\raisebox{0.85pt}{.}\hspace{-2.78pt}\raisebox{2.85pt}{.}\,}

\newlength{\fighskip} \fighskip=2pt
\newlength{\figvskip} \figvskip=3pt

\newcommand*{\figbox}[2]{{
  \def\figscale{#1}
  \def\arraystretch{0.8}
  \arraycolsep=0pt
  \begin{array}{c}
    \vbox{\vskip\figscale\figvskip
      \hbox{\hskip\figscale\fighskip
        \includegraphics[scale=\figscale]{#2}}}
  \end{array}}}

\usepackage{authblk}

\begin{document}

\title{Strong random unitaries and fast scrambling}

\author[1,2]{Thomas Schuster}
\author[3]{Fermi Ma}
\author[4]{Alex Lombardi}
\author[5,1]{\\Fernando Brand{\~a}o}
\author[1,2]{Hsin-Yuan Huang}

\affil[1]{California Institute of Technology}
\affil[2]{Google Quantum AI}
\affil[3]{UC Berkeley, New York University}
\affil[4]{Princeton University}
\affil[5]{AWS Center for Quantum Computing}

\date{\today}

\maketitle

\begin{abstract}\normalsize
Understanding how fast physical systems can resemble Haar-random unitaries is a fundamental question in physics. Many experiments of interest in quantum gravity and many-body physics, including the butterfly effect in quantum information scrambling and the Hayden-Preskill thought experiment, involve queries to a random unitary~$U$ alongside its inverse~$U^\dagger$, conjugate~$U^*$, and transpose~$U^T$. However, conventional notions of approximate unitary designs and pseudorandom unitaries (PRUs) fail to capture these experiments. In this work, we introduce and construct \emph{strong unitary designs} and \emph{strong PRUs} that remain robust under all such queries. Our constructions achieve the optimal circuit depth of $\mathcal{O}(\log n)$ for systems of $n$ qubits. 
We further show that strong unitary designs can form in circuit depth $\mathcal{O}(\log^2 n)$ in circuits composed of independent two-qubit Haar-random gates, and that strong PRUs can form in circuit depth $\poly(\log n)$ in circuits with no ancilla qubits.
Our results provide an operational proof of the fast scrambling conjecture from black hole physics: every observable feature of the fastest scrambling quantum systems reproduces Haar-random behavior at logarithmic times.
\end{abstract}

\pagenumbering{arabic} 
\setcounter{page}{1}

\addtocontents{toc}{\protect\setcounter{tocdepth}{0}}


\section{Introduction}

Understanding how fast a quantum system can scramble information is a fundamental question with implications throughout quantum science. 
In quantum computing, quantum information scrambling by random circuits enables efficient device benchmarking~\cite{emerson2005scalable,ambainis2007quantum, knill2008randomized,elben2023randomized}, quantum state tomography~\cite{guta2020fast, huang2020predicting,zhao2021fermionic}, and quantum advantage demonstrations~\cite{arute2019quantum, morvan2023phase, abanin2025constructive}.
In quantum cryptography~\cite{ji2018pseudorandom,ananth2022cryptography,kretschmer2023quantum}, scrambling is characterized by the computational indistinguishability of quantum circuits from Haar-random unitaries, and enables new cryptographically secure protocols.
In fundamental physics, scrambling and the emergence of Haar-random unitary behaviors provide powerful theoretical tools for modeling complex phenomena across diverse areas, from quantum many-body dynamics~\cite{fisher2023random,nahum2017entgrowth,cotler2022fluctuations} to quantum chaos and thermalization~\cite{deutsch1991quantum,srednicki1994chaos,rigol2008thermalization} to quantum gravity and black hole physics~\cite{sekino2008fast,hayden2007black,brown2023quantum,nezami2023quantum,schuster2022many}.

Across all of these contexts, a central question concerns the minimum time required to scramble quantum information. This question is crucial for quantum technologies: the shorter the time required, the more experimentally applicable random unitary protocols become. In physics, this question is captured by the \emph{fast scrambling conjecture}, proposed by Sekino and Susskind~\cite{sekino2008fast}, which states\footnote{The fast scrambling conjecture also posits that optimal scrambling is achieved in black hole systems. Verifying this claim remains experimentally inaccessible with current technology and is beyond the scope of this work.}:

\begin{conjecture}[Fast scrambling conjecture, Sekino and Susskind~\cite{sekino2008fast}]
\label{conjecture:fast-scrambling}
The minimum time to \emph{scramble} information in quantum systems of $n$ qubits under all-to-all connectivity is $\Theta(\log n)$.
\end{conjecture}
\noindent In their original formulation, scrambling was characterized through entanglement growth: for every $n$-qubit pure state $\ket{\psi}$, every sufficiently small subsystem of the time-evolved state $U\ketbra{\psi}U^\dagger$ should achieve near-maximal entanglement entropy in $\mathcal{O}(\log n)$ time. However, the modern physical understanding of quantum information scrambling encompasses numerous signatures beyond entanglement growth, including, in physical studies, the decay of out-of-time-order correlators~\cite{xu2024scrambling}, information recovery protocols like the Hayden-Preskill thought experiment~\cite{hayden2007black}, and the saturation of operator size distributions and operator entanglement entropies to their Haar-random values~\cite{roberts2018operator}. 

In recent years, \emph{indistinguishability from Haar-random} has emerged as a powerful conceptual framework for characterizing scrambling: a system scrambles if its dynamics $U$ become operationally indistinguishable from Haar-random evolution in any physical experiment. Operational indistinguishability is commonly defined through \emph{approximate unitary $k$-designs}~\cite{emerson2003pseudo,emerson2004random,gross2007evenly,dankert2005efficient,dankert2009exact,brandao2016local, haah2024efficient, chen2024incompressibility,laracuente2024approximate,schuster2024random,west2025no,grevink2025will,cui2025unitary} and \emph{pseudorandom unitaries} (PRUs)~\cite{ji2018pseudorandom,metger2024simple, chen2024efficient, ma2024construct}. The former guarantees that $U$ is indistinguishable from a Haar-random unitary within any quantum experiment that queries $U$ up to $k$ times, while the latter concerns any polynomial-time quantum experiment. These general operational frameworks offer a crucial advantage of also capturing any future scrambling diagnostics yet to be discovered. This naturally motivates the following operational formulation of the fast scrambling conjecture:

\begin{conjecture}[Operational fast scrambling conjecture] \label{conjecture:fast-scrambling-PRU}
The minimum circuit depth to form $n$-qubit unitary $k$-designs (for constant $k$) and PRUs under all-to-all connectivity is $\Theta(\log n)$.
\end{conjecture}
\noindent The past decade has seen remarkable progress in understanding the depths needed to form unitary designs and PRUs~\cite{brandao2016local, haah2024efficient, chen2024incompressibility,laracuente2024approximate,schuster2024random,west2025no,grevink2025will,cui2025unitary,ji2018pseudorandom,metger2024simple, chen2024efficient, ma2024construct,foxman2025random}. Unfortunately, standard notions of designs and PRUs possess two critical limitations, which undermine the utility of the fast scrambling conjecture formulated above.

\vspace{0.45em}
\textbf{Limitation 1: Forward-only access.} Standard unitary designs and PRUs only guarantee indistinguishability under forward queries to $U$. However, many scrambling diagnostics require access to the inverse $U^\dagger$, conjugate $U^*$, or transpose $U^T$ operations. For example, out-of-time-order correlators require time-reversal operations using $U^\dagger$ to be efficiently measured~\cite{cotler2023information}, while efficient information recovery in the Hayden-Preskill protocol requires complex conjugation $U^*$~\cite{yoshida2017efficient}. Similarly, recent work in quantum cryptography~\cite{zhandry2024model} suggests that the strongest notion of PRUs should allow access to all of $U$, $U^\dagger$, $U^*$, $U^T$ (as well as their controlled versions), reflecting the fact that a user with knowledge of the gates composing $U$ should be able to implement all these transformations. 

Standard designs and PRUs fail to capture these essential features. In fact, they do not even satisfy an $\Omega(\log n)$ depth lower bound: recent work~\cite{schuster2024polynomial, cui2025unitary} shows that standard unitary designs and PRUs can be constructed in $\Theta(\log \log n)$ depth, exponentially faster than the conjectured minimum scrambling time of $\Theta(\log n)$. This demonstrates that forward-only indistinguishability is insufficient to capture the complete range of scrambling behaviors often desired in physics and cryptography.

\vspace{0.45em}
\textbf{Limitation 2: Unphysical use of ancillary systems.} All existing PRU constructions~\cite{metger2024simple, ma2024construct, schuster2024random} over $n$ qubits require $m = \mathrm{poly}(n)$ ancilla qubits initialized to $\ket{0^m}$ and returned to $\ket{0^m}$ to achieve pseudorandomness on the original $n$ qubits. While this use of ancilla qubits is acceptable for cryptographic applications, it creates a fundamental mismatch when modeling physical quantum dynamics. Physical scrambling processes, whether in black holes or many-body quantum systems, operate on fixed Hilbert spaces and do not involve auxiliary degrees of freedom with fine-tuned initialization and finalization conditions. As a result, standard cryptographic PRUs do not necessarily provide appropriate evidence for physical scrambling processes. 

\vspace{0.45em}
In this work, we address both of these limitations. First, we introduce \emph{strong unitary $k$-designs} and \emph{strong pseudorandom unitaries} (PRUs), which are indistinguishable from Haar-random in any experiment that queries the unitary $U$ or its inverse $U^\dagger$, conjugate $U^*$, or transpose $U^T$.
Second, we initiate the study of \emph{ancilla-free PRUs}, which are PRUs with efficient ancilla-free circuit implementations. 
These definitions motivate a strengthened version of the fast scrambling conjecture:

\begin{conjecture}[Strong fast scrambling conjecture] \label{conjecture:strong-fast-scrambling}
The minimum depth to form $n$-qubit \textbf{strong} unitary designs and PRUs under all-to-all connectivity is $\Theta(\log n)$, achievable \textbf{without ancilla}.
\end{conjecture}

\noindent This conjecture captures the strongest possible operational meaning of fast scrambling: quantum dynamics that remain indistinguishable from Haar-random under any efficient quantum experiment involving any combination of operations, realized using only the physical degrees of freedom. 

Our main results provide compelling evidence for, and an almost full resolution to, the strong fast scrambling conjecture:

\begin{enumerate}
\item \textit{\textbf{Strong unitary designs.}} We provide the first proof of existence for strong unitary designs and establish that the minimum depth to form them is precisely $\Theta(\log n)$. This proves that every property measurable in finitely many queries scrambles in logarithmic time.

\item \textit{\textbf{Strong PRUs.}} We provide the first proof of existence for strong PRUs secure against all operations $U$, $U^\dagger$, $U^*$, $U^T$ and establish that the minimum depth to form them is precisely $\Theta(\log n)$ under a well-established cryptographic assumption: subexponential hardness of learning with errors (LWE)~\cite{regev2009lattices}. This represents a significant advance over prior work~\cite{ma2024construct}, which requires $\mathrm{poly}(n)$ depth and only achieved security against $U$ and $U^\dagger$. Our result proves that every property measurable by polynomial-time experiments scrambles in logarithmic time.

\item \textit{\textbf{Ancilla-free constructions:}} We show that both strong unitary designs and strong PRUs can be implemented without using any ancilla qubits in $\mathrm{poly}(\log n)$ depth. While slightly larger than the conjectured $\Theta(\log n)$, this provides the first ancilla-free PRU construction for any security notion, including standard forward-only PRUs. These constructions are secure under the subexponential hardness\footnote{In order to construct polynomial-depth ancilla-free PRUs, it suffices to assume the polynomial hardness of LWE.} of learning with errors \cite{regev2009lattices}. 

\item \textit{\textbf{Generic emergence.}} We prove that all-to-all random circuits consisting of independent Haar-random two-qubit gates form strong unitary designs in $\mathcal{O}(\log^2 n)$ depth, providing evidence that fast scrambling is a generic phenomenon rather than requiring careful engineering.
\end{enumerate}

\noindent All of our results immediately extend to quantum experiments with access to controlled versions of $U, U^\dagger, U^*, U^T$, following the general reduction laid out in~\cite{sheridan2009approximating,tang2025controlled}.

Our constructions of strong unitary designs and strong PRUs introduce several new random unitary ensembles and techniques for working with strong random unitaries. Our main technical contributions are threefold. \textit{First}, we introduce the Luby-Rackoff-Function-Clifford (LRFC) ensemble and prove that it forms both strong unitary designs and strong PRUs. The LRFC ensemble adapts the Permutation-Function-Clifford (PFC) ensemble of~\cite{metger2024simple}, replacing the random permutation with a Luby-Rackoff construction to achieve exponentially lower circuit depths while maintaining security against all queries to $U$, $U^\dagger$, $U^*$, and $U^T$. \textit{Second}, inspired by~\cite{schuster2024random}, we prove a gluing theorem for strong unitary designs and strong PRUs. This powerful technique allows us to reduce the circuit depth of our strong constructions to the optimal value of $\mathcal{O}(\log n)$, which we prove is the minimum possible for any strong unitary design or strong PRU, and to establish strong unitary designs in depth $\mathcal{O}(\log^2 n)$ using all-to-all random circuits with Haar-random two-qubit gates. \textit{Third}, to construct ancilla-free PRUs, we combine a new classical-to-quantum circuit compilation technique with our gluing theorems to construct standard PRUs in $\mathrm{poly}(\log n)$ depth over 1D geometries and strong PRUs in $\mathrm{poly}(\log n)$ depth over all-to-all geometries, both without using any ancilla qubits.

\paragraph{Organization of this paper.}

\begin{figure}[t]
    \centering
    \includegraphics[width=1.0\textwidth]{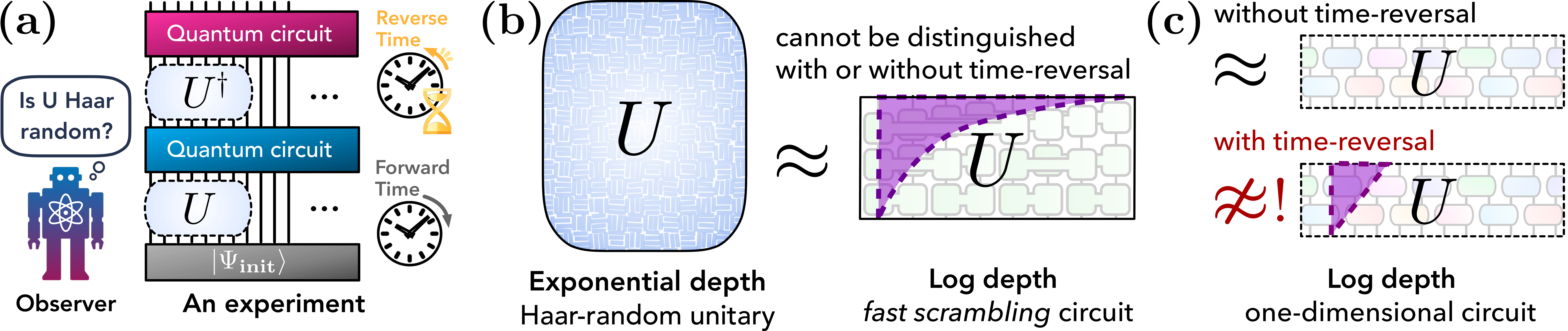}
    \caption{Illustration of our main results. \textbf{(a)} A strong approximate unitary $k$-design is a random unitary ensemble that is indistinguishable from Haar in any quantum experiment that queries $U$ or its inverse (i.e.~time-reversal), conjugate, or transpose $k$ times. A strong pseudorandom unitary (PRU) is similarly indistinguishable in any polynomial-time experiment. \textbf{(b)} We construct strong unitary designs and  PRUs on $n$ qubits in depth $\mathcal{O}(\log n)$. Our constructions use long-range two-qubit gates to scramble quantum information over all $n$ qubits as fast as possible. \textbf{(c)} In comparison, low-depth one-dimensional quantum circuits can only scramble information over local regions. This allows them to form conventional designs~\cite{schuster2024random,laracuente2024approximate} and PRUs~\cite{schuster2024random}, but not strong designs or strong PRUs.}
    \label{fig:1}
\end{figure}

Our manuscript is organized as follows. In Section~\ref{sec: design}, we define strong unitary designs and summarize our main results on their circuit depth. We also illustrate the failure of standard definitions of approximate unitary designs to capture experiments involving the time-reverse $U^\dagger$ and conjugate $U^*$. In Section~\ref{sec: PRU}, we define strong PRUs and summarize our main results on their circuit depth. In Section~\ref{sec: constructions}, we provide detailed descriptions of our constructions and proofs of strong unitary designs and PRUs. These involve three key ingredients as described above: the LRFC random unitary ensemble, a gluing construction for optimizing the circuit depth of strong random unitaries, and an adapted gluing construction for local random circuits. In Section~\ref{sec: scrambling}, we discuss the relation between strong random unitaries and fast quantum information scrambling. We conclude in Section~\ref{sec: discussions} with  discussions and open questions.

\section{Strong approximate unitary designs} \label{sec: design}

Approximate unitary designs seek to mimic Haar-random unitaries in applications that involve a random unitary only a finite number of times. Let us first review their standard definitions and then provide our strong definition and summarize our main results on strong unitary designs.

\paragraph{Background.} An exact unitary $k$-design on $n$ qubits is a random unitary ensemble $\mathcal{E}$ whose $k$-th moment, $\Phi_\mathcal{E}(\cdot) \equiv \E_{U \sim \mathcal{E}}[ U^{\otimes k}(\cdot) U^{\dagger, \otimes k}]$, equals the $k$-th moment of the Haar ensemble on $U(2^n)$: $\Phi_\mathcal{E} = \Phi_H$. While exact designs provide the strongest possible guarantees, efficient realizations beyond $k \leq 3$ are exceedingly rare.

To address this limitation, several notions of \emph{approximate unitary $k$-designs} have been introduced. We begin with the strongest such notion, the so-called \emph{relative error}~\cite{brandao2016local}. A unitary ensemble $\mathcal{E}$ is an approximate unitary $k$-design with relative error $\varepsilon$ if its moment obeys $(1-\varepsilon) \Phi_H \preceq \Phi_\mathcal{E} \preceq (1+\varepsilon)\Phi_H$, where $\mathcal{A} \preceq \mathcal{B}$ indicates that $\mathcal{B}-\mathcal{A}$ is a completely positive map. Physically, the relative error guarantees that a random unitary $U \sim \mathcal{E}$ cannot be distinguished from Haar-random in any quantum experiment that queries it up to $k$ times~\cite{schuster2024random}. It also provides even stronger guarantees on properties that cannot be efficiently measured in any quantum experiment~\cite{cui2025unitary}.

Remarkably, unitary $k$-designs with relative error $\varepsilon$ over $n$ qubits can form in extremely low circuit depths of $\tilde{\mathcal{O}}(k \log n/\varepsilon)$~\cite{laracuente2024approximate,schuster2024random}, growing only logarithmically in the number of qubits $n$. This holds even in one-dimensional systems with small light-cones. The dependence on $n$, $k$, and $\varepsilon$ was further improved exponentially to $\mathcal{O}(k \log k \cdot \log \log (n/\varepsilon))$, to achieve relative error, and $\mathcal{O}(\log k \cdot \log \log (n/\varepsilon))$, to achieve a more physical notion of measurable error, for systems with long-range two-qubit gates~\cite{cui2025unitary}.

The existence of such low-depth unitary designs is counter-intuitive, as they appear to capture many features of Haar-random unitaries~\cite{schuster2024random,laracuente2024approximate} without developing other characteristic features such as large light-cones, high entanglement, decay of out-of-time-order correlations, and good quantum encoding properties. Notably, these latter features are precisely the standard diagnostics of \emph{quantum information scrambling} in many-body quantum physics and quantum gravity~\cite{sekino2008fast,shenker2014black,roberts2015localized,hosur2016chaos,roberts2018operator,swingle2016measuring,nahum2017entgrowth,nahum2018operator,schuster2023operator,xu2024scrambling}.

A resolution to this apparent paradox was provided in~\cite{schuster2024random}: these scrambling-related features cannot be detected efficiently in any quantum experiment that queries only the forward evolution $U$. Consequently, they do not form barriers to realizing low-depth unitary designs. However, many scrambling diagnostics \emph{can} be efficiently detected in quantum experiments that involve the inverse $U^\dagger$, conjugate $U^*$, or transpose $U^T$ of the unitary $U$. These are precisely the experiments traditionally studied in quantum information scrambling~\cite{swingle2016measuring,garttner2017measuring,landsman2019verified,blok2020quantum,sanchez2021emergent,mi2021information,abanin2025constructive}. For example, estimating out-of-time-order correlators to study butterfly effects requires time-reversal operations $U^\dagger$~\cite{swingle2016measuring,garttner2017measuring, garttner2017measuring}, while the decoding protocol for the Hayden-Preskill thought experiment involves complex conjugation $U^*$~\cite{yoshida2017efficient}. This motivates a stronger notion of approximate unitary designs that captures experiments involving not just $U$, but also $U^\dagger$, $U^*$, and $U^T$.

\paragraph{Strong unitary designs.}
We define a \emph{strong $\varepsilon$-approximate unitary $k$-design} as any random unitary ensemble $\mathcal{E}$ that cannot be distinguished from Haar-random in any quantum experiment that makes any $k$ queries to the unitary $U$ or its inverse $U^\dagger$, conjugate $U^*$, or transpose $U^T$. To be precise, if we denote the output of a general quantum experiment as $\ket*{\psi_W^U} = W_{k+1} U^{\circ_k} W_k U^{\circ_{k-1}} \cdots U^{\circ_1} W_1 \ket{0}$, where each $\circ_j \in \{ \cdot, \dagger, T, * \}$ represents forward evolution, inverse, transpose, or conjugate respectively, and $W_j$ are arbitrary quantum operations applied between successive queries, then we demand
\begin{equation}
    \left\lVert \E_{U \sim \mathcal{E}} \Big[ \dyad*{\psi^U_W} \Big] - \E_{U \sim H} \Big[ \dyad*{\psi^U_W} \Big] \right\rVert_1 \leq \varepsilon
\end{equation}
for all choices of $W_j$ and $\circ_j$. This generalizes the notion of adaptive security for pseudorandom unitaries~\cite{ji2018pseudorandom,ma2024construct} and measurable error for unitary designs~\cite{cui2025unitary} to incorporate all variants of the unitary. We refer to Appendix~\ref{app: preliminaries} for further discussion, including strong versions of other design approximation metrics.

With this definition, a fundamental question arises: what circuit depths are required for strong unitary designs to form? A basic light-cone argument, which we formalize later, shows that strong unitary designs cannot form until information can propagate between any pair of qubits in the system. This requires the light-cone of the evolution to encompass all $n$ qubits, demanding depth $\Omega(\log n)$ in general quantum circuits and dynamics. Can this extremely fast speed of scrambling actually be achieved? The \emph{fast scrambling conjecture} from black hole physics posits that all-to-all connected quantum systems can achieve logarithmic scrambling times~\cite{sekino2008fast}. However, existing progress toward proving this conjecture has focused on specific scrambling diagnostics, such as the decay of out-of-time-order correlators or the encoding properties of random unitaries~\cite{maldacena2016remarks,kitaev2015simple,roberts2018operator,brown2012scrambling,brown2013short,brown2015decoupling,lashkari2013towards,cleve2015near,belyansky2020minimal,bentsen2019fast,vikram2024exact}. A fully general operational proof of this conjecture has remained an open question.

Our main result establishes that strong unitary designs can indeed form in optimal circuit depth $\mathcal{O}(\log n)$ in all-to-all-connected architectures. This proves that every property of a random unitary measurable in a constant number of queries scrambles in logarithmic time.

\begin{theorem}[Fast formation of strong unitary designs] \label{thm:strong-design-depth}
    Strong $\varepsilon$-approximate unitary $k$-designs can be realized in the following circuit depths:
    \begin{enumerate}
        \item $d = \mathcal{O}\big(\log n + \log k \cdot \log \log (nk/\varepsilon) \big)$ using all-to-all structured circuits with $\tilde{\mathcal{O}}(nk)$ ancilla qubits.
        \item $d = \mathcal{O}\big( \log n + k \cdot \log \log (nk/\varepsilon) \big)$ using all-to-all structured circuits with $\tilde{\mathcal{O}}(n)$ ancilla qubits.
    \end{enumerate}
    For all-to-all random circuits consisting of independent Haar-random two-qubit gates without ancilla qubits, $d = \mathcal{O}(k \cdot \poly \log k \cdot \log (n/\varepsilon) + \log n \cdot \log (n/\varepsilon))$.
\end{theorem}

\noindent The structured circuits achieve the optimal $\mathcal{O}(\log n)$ scaling in system size when $k$ and $\varepsilon$ are held constant, while the random circuits achieve an $\mathcal{O}(\log^2 n)$ scaling. The upper bounds for structured circuits are nearly optimal across all parameters, as confirmed by our lower bounds:

\begin{proposition} \label{prop: lower bound design}
    \emph{(Depth lower bounds for strong unitary designs)}
    For any $\varepsilon < 1/4$, any circuit ensemble over $n$ qubits that forms a strong $\varepsilon$-approximate unitary $k$-design requires circuit depth $d$:
    \begin{enumerate}
    \item $d = \Omega\big(\log n + \log k \big)$ for any all-to-all circuits with any number of ancilla qubits.
    \item $d = \Omega\big(\log n + k/\log(nk) \big)$ for any all-to-all circuits with at most $\mathcal{O}(n)$ ancilla qubits.
    \end{enumerate}
    In contrast, for any 1D circuits with any number of ancilla qubits, $d = \Omega\big(n + k/\log(nk) \big)$.
\end{proposition}

\noindent The two items confirm near-optimality of our all-to-all constructions, while we also show an exponential separation between all-to-all connectivity and finite-dimensional geometries. We provide detailed constructions and proof techniques in Section~\ref{sec: constructions} and complete proofs in Appendix~\ref{app: designs}. In Appendix~\ref{sec: lower bound}, we also establish a surprising result showing that local random circuits require $\Omega(n)$ depth to realize strong unitary designs with \emph{relative error}, regardless of connectivity.

\begin{figure}[t]
    \centering
    \includegraphics[width=1.0\textwidth]{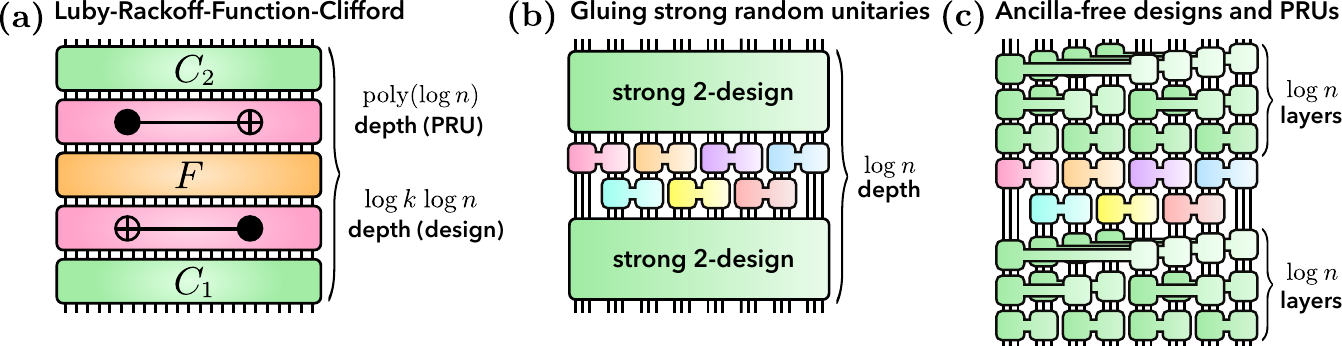}
    \caption{Our constructions of strong unitary $k$-designs and strong PRUs.
    \textbf{(a)} The Luby-Rackoff-Function-Clifford (LRFC) ensemble sandwiches classical shuffle and phase gates (pink and orange) between random Clifford unitaries (green). It forms a strong unitary design and a strong PRU in the stated circuit depths. 
    \textbf{(b)} To further reduce these depths, we consider a glued construction, with two layers of small $2\xi$-qubit random unitaries (various colors)  sandwiched between strong $n$-qubit unitary 2-designs (green). This forms a strong unitary $k$-design when $\xi = \Omega(\log (nk/\varepsilon))$ and a strong PRU when $\xi = \omega(\log n)$. We instantiate each small unitary with the LRFC ensemble.
    \textbf{(c)} To obtain ancilla-free constructions, we replace each $n$-qubit 2-design with a fast scrambling circuit of depth $\log n$ composed of $2\xi$-qubit 2-designs. For ancilla-free strong unitary designs consisting of Haar-random two-qubit gates, each small unitary is drawn from a random circuit on $2\xi$ qubits. For ancilla-free strong PRUs, each small unitary is implemented by reusing neighboring qubits as ancillae.}
    \label{fig:constructions}
\end{figure}

\begin{figure}[t]
    \centering
    \includegraphics[width=1.0\textwidth]{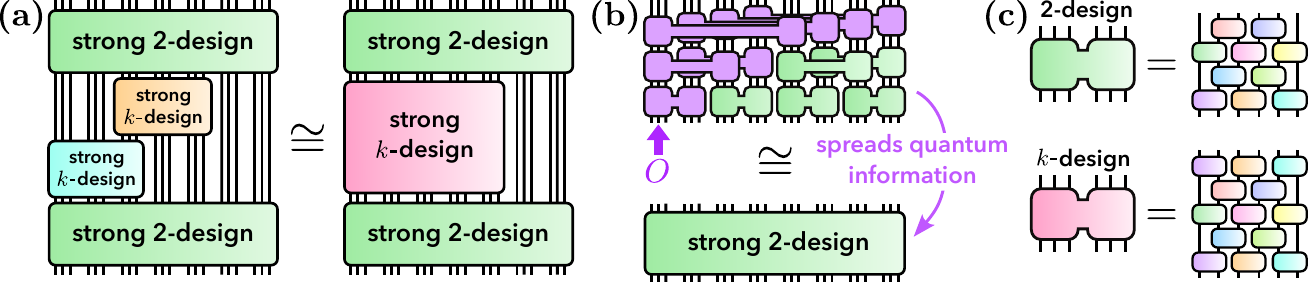}
    \caption{Illustration of key ideas from the proof of Theorem~\ref{thm:strong-design-depth} and Theorem~\ref{thm:strong-PRU-depth}.
    \textbf{(a)} To analyze our glued construction, we prove that two strong unitary $k$-designs ``glue'' together whenever they are sandwiched by larger unitary 2-designs.
    This does not hold in the absence of the larger 2-designs.
    \textbf{(b)} To show that the blocked fast scrambling circuit forms a strong 2-design, we prove that any unitary ensemble that uniformly spreads quantum information (purple) is a strong 2-design. The circuit spreads information over $n$ qubits in $\log n$ layers.
    \textbf{(c)} To replace each $2\xi$-qubit  unitary with a local random circuit, we prove a technical lemma that translates spectral gaps~\cite{brandao2016local,haferkamp2022random,chen2024incompressibility} to strong unitary designs.
    }
    \label{fig:design-proof}
\end{figure}

\section{Strong pseudorandom unitaries} \label{sec: PRU}

Pseudorandom unitaries (PRUs) seek to mimic Haar-random unitaries in efficient quantum experiments. Let us first review their standard definition and then provide our strong definition and summarize our main results on strong PRUs.

\paragraph{Background.}

A random unitary ensemble $\mathcal{E}$ is a PRU if no efficient quantum algorithm can distinguish a random unitary $U \sim \mathcal{E}$ from a Haar-random unitary under polynomially many queries to $U$~\cite{ji2018pseudorandom}. More precisely, $\mathcal{E}$ is a PRU with security against any $t(n)$-time quantum adversary if it cannot be distinguished from Haar-random in any $t(n)$-time quantum experiment, where $t(n)$ is some function of the number of qubits $n$.

The premier example is the Permutation-Function-Clifford (PFC) ensemble~\cite{metger2024simple}, $U = PFC$, formed by multiplying a pseudorandom permutation $P$, a pseudorandom function $F$, and a random Clifford unitary $C$. Under standard cryptographic assumptions, the PFC ensemble achieves security against subexponential-time quantum adversaries~\cite{ma2024construct}. However, it requires circuit depth $\mathrm{poly}(n)$ to implement, even on all-to-all-connected geometries, due to the circuit depth of the pseudorandom permutation~\cite{zhandry2016note, schuster2024random}.

Recent work has achieved exponential improvements in the circuit depths of standard PRUs. Ref.~\cite{schuster2024random} reduces the depth to $\mathrm{poly}(\log n)$ while maintaining security against polynomial-time quantum adversaries. The same work also discusses PRU realizations in even smaller circuit depths $\mathrm{poly}(\log \log n)$ using the LRFC ensemble introduced in our work.\footnote{The LRFC construction was developed by several of the authors of this work at the time of publication of Ref.~\cite{schuster2024random} but the security proof remained unpublished until now.} We will show that the LRFC ensemble forms a strong PRU in circuit depth $\mathrm{poly}(\log n)$, which yields circuit depth $\mathrm{poly}(\log \log n)$ for the two-layer LRFC ensemble considered in Ref.~\cite{schuster2024random}.

Similar to standard unitary designs, existing PRUs exhibit a fundamental limitation: they only guarantee security against experiments that query the forward evolution $U$. However, as discussed for strong unitary designs, many important quantum phenomena require experiments involving the inverse $U^\dagger$, conjugate $U^*$, or transpose $U^T$ to be detected. Prior work partially addressed this by extending the PFC ensemble to achieve security against experiments querying both $U$ and $U^\dagger$, using the construction $U = DPFC$ with an additional random Clifford $D$~\cite{ma2024construct}. However, this approach still omits conjugate and transpose operations, and requires $\mathrm{poly}(n)$ circuit depth to implement. In our results that follow, we establish security under all queries to $U, U^\dagger, U^*, U^T$ and exponentially reduce the circuit depth from $\mathrm{poly}(n)$ to $\mathcal{O}(\log n)$.

\paragraph{Strong pseudorandom unitaries.}
We define a \emph{strong PRU} as any random unitary ensemble $\mathcal{E}$ that cannot be distinguished from Haar-random in any efficient quantum experiment involving any combination of queries to the unitary $U$ or its inverse $U^\dagger$, conjugate $U^*$, or transpose $U^T$. Formally, $\mathcal{E}$ is a strong PRU with $t(n)$-time security if it remains indistinguishable from Haar-random in any $t(n)$-time quantum experiment, regardless of which operations are queried.

Our main result establishes that strong PRUs can form in optimal circuit depth $\mathcal{O}(\log n)$ in all-to-all-connected architectures. This proves that every efficiently observable property of quantum dynamics can achieve pseudorandomness in logarithmic time.

\begin{theorem}[Fast formation of strong PRUs] \label{thm:strong-PRU-depth}
    Under standard cryptographic assumptions, strong PRUs with polynomial-time security can be realized in the following circuit depths:
    \begin{enumerate}
        \item $d = \mathcal{O}(\log n)$ using all-to-all structured circuits with $\tilde{\mathcal{O}}(n)$ ancilla qubits.
        \item $d = \mathrm{poly}(\log n)$ using all-to-all structured circuits with no ancilla qubits.
    \end{enumerate}
\end{theorem}
\noindent We describe our constructions and proof methods in Section~\ref{sec: constructions} and provide a complete proof in Appendix~\ref{app: PRU} and Appendix~\ref{app: no-ancilla}. The first result utilizes our LRFC ensemble [Fig.~\ref{fig:constructions}(a)] combined with our gluing theorem for strong random unitaries [Fig.~\ref{fig:constructions}(b)]. A circuit depth lower bound of $\Omega(\log n)$ can be easily proven by noting that if the depth of $U$ is sublinear in $\log n$, then an experiment that measures $U^\dagger X_1 U \ket{0^n}$ in the all $Z$ basis will result in a bitstring with almost all zeros, whereas a Haar-random unitary $U$ will result in a bitstring with almost equal number of zeros and ones. Hence, the achieved $d = \mathcal{O}(\log n)$ scaling is optimal. The second result combines the LRFC ensemble, our gluing theorem, and a new strategy for compiling random classical functions in quantum circuits and gluing them together by reusing system qubits on other local patches as ancilla qubits. Both results rely only on the subexponential hardness of standard learning with errors (LWE)~\cite{regev2009lattices}. All of our results immediately extend to quantum experiments involving queries to the controlled versions of $U, U^\dagger, U^*, U^T$, following the general reduction in~\cite{sheridan2009approximating,tang2025controlled}.

\section{Our constructions} \label{sec: constructions}

Having summarized our main results, we now introduce our random unitary ensembles.
We proceed in three parts.
In Section~\ref{sec: LRFC}, we introduce the Luby-Rackoff-Function-Clifford (LRFC) ensemble.
We prove that the LRFC ensemble forms a strong unitary design in $\mathcal{O}(\log k \cdot \log n)$ depth and a strong PRU in $\poly(\log n)$ depth.
In Section~\ref{sec: gluing}, we introduce a gluing construction for strong random unitaries, which allows us to further optimize each circuit depth to $\mathcal{O}(\log n + \log k \log\log k)$ and $\mathcal{O}(\log n)$.
In Section~\ref{sec: local random circuit}, we introduce a modified gluing construction inspired by toy-model fast scrambling circuits in black hole physics~\cite{hayden2007black,sekino2008fast}.
We use this to prove that all-to-all-connected local random circuits can form strong unitary $k$-designs in depth $\mathcal{O}(k \cdot \poly \log k \cdot \log n/\varepsilon +  \log n \cdot \log n/\varepsilon)$.

\subsection{The Luby-Rackoff-Function-Clifford (LRFC) ensemble}  \label{sec: LRFC}

The Luby-Rackoff-Function-Clifford (LRFC) ensemble is inspired by the Permutation-Function-Clifford (PFC) ensemble introduced in~\cite{metger2024simple}.
The key difference is that it replaces the random permutation in the PFC ensemble with a pair of random shuffle gates, which use only random functions.
This replacement leads to an exponential improvement in circuit depth, since known constructions of quantum-secure pseudorandom functions~\cite{zhandry2021PRF} are much more efficient than those of pseudorandom permutations~\cite{zhandry2016note}. 
The random shuffle gates are inspired by the Luby-Rackoff block cipher in classical cryptography~\cite{luby1988construct}.
Crucially, we show that the random shuffle gates mimic the action of a random permutation \emph{within the LRFC circuit}. 

We define the LRFC ensemble as follows. We consider the random unitary,
\begin{equation}
    U = D \cdot S_R \cdot F \cdot S_L \cdot C.
\end{equation}
The unitaries $C$ and $D$ are drawn from any strong unitary 2-design on $n$ qubits.
For example, they can each be a random Clifford unitary, which are exact unitary 2-designs.
The unitary $F$ is a random ternary phase gate, $F = \sum_{x \in \{0,1\}^n} \omega^{f(x)} \dyad{x}$, where $\omega \equiv e^{i \frac{2\pi}{3}}$ and $f: \{0,1\}^n \rightarrow \{0,1,2\}$ is a random function.
Finally, the unitaries $S_L$ and $S_R$ are random shuffle gates, which shuffle the bitstring of the left (or right) $n/2$ qubits conditional on the value of the right (or left) $n/2$ qubits.
That is, $S_L \ket{x_L,x_R} = \ket{x_L + h_1(x_R), x_R}$ where $x_L, x_R \in \{0,1\}^{n/2}$ are the left and right $n/2$ bits and $h_1: \{ 0,1\}^{n/2} \rightarrow \{ 0,1\}^{n/2}$ is a random function.
Similarly, $S_R \ket{x_L,x_R} = \ket{x_L, x_R + h_2(x_L)}$.

%
To implement the LRFC ensemble efficiently, we will replace each random function in the phase and shuffle gates with a less-random efficient approximation.
To construct strong unitary $k$-designs, we will replace each random function with an exact $2k$-wise independent function~\cite{wegman1981new}.
This replicates the first $2k$ moments of a random function, which guarantees that any quantum experiment making $k$ queries will proceed identically to if the function was random.
The factor of two accounts for the bra and ket of the quantum state.
Exact $2k$-wise independent functions can be implemented in quantum circuit depth $\mathcal{O}(k \cdot \log n)$ using $\tilde{\mathcal{O}}(n)$ ancilla qubits, or depth $\mathcal{O}(\log k \cdot \log n)$ using $\tilde{\mathcal{O}}(nk)$ ancilla qubits~\cite{cui2025unitary}.
To construct strong PRUs, we will replace each random function with a quantum-secure pseudorandom function (PRF)~\cite{zhandry2021PRF}.
This is indistinguishable from a random function in any bounded time quantum experiment.
Strong PRFs with security against any subexponential-time quantum adversary can be implemented in quantum circuit depth $\poly(\log n)$~\cite{zhandry2021PRF,schuster2024random}.

Our main result is that the LRFC circuit forms a strong unitary $k$-design (when each random function is replaced with a $2k$-wise independent function) and a strong PRU (when each random function is replaced with a quantum-secure pseudorandom function).
\begin{theorem}[The LRFC ensemble is a strong unitary design] \label{thm:LRFC-design}
    Let $f, h_1, h_2$ be $2k$-wise independent functions. 
    Then the LRFC ensemble is a strong $\varepsilon$-approximate unitary $k$-design with  $\varepsilon = \mathcal{O}(k^2 / 2^{n/6})$.
\end{theorem}
\begin{theorem}[The LRFC ensemble is a strong PRU] \label{thm:LRFC-PRU}
    Let $f, h_1, h_2$ be subexponentially\footnote{A cryptographic primitive is defined to be sub-exponentially secure if, for a security parameter $n$, it is secure against attacks running in time $2^{O(n^\epsilon)}$, for some constant $\epsilon > 0$.} quantum-secure pseudorandom functions.
    Then the LRFC ensemble is a strong PRU with subexponential security.
\end{theorem}
\noindent The LRFC ensemble can be compiled in $\mathcal{O}(k \cdot \log n)$ circuit depth using $\tilde{\mathcal{O}}(n)$ ancilla qubits, or $\mathcal{O}(\log k \cdot \log n)$ circuit depth using $\tilde{\mathcal{O}}(nk)$ ancilla qubits, when each function is $2k$-wise independent~\cite{cui2025unitary}.
It can be compiled in $\poly(\log n)$ circuit depth when each function is pseudorandom~\cite{schuster2024random}.

\subsection{Gluing strong random unitaries} \label{sec: gluing}

We can further improve upon the circuit depths of the LRFC ensemble by establishing a fundamental property of strong random unitaries.
Namely, we prove that two strong random unitaries on overlapping subsystems ``glue'' together, whenever they are surrounded by larger unitary 2-designs [Fig.~\ref{fig:design-proof}(a)]. 
Intuitively, the larger 2-designs scramble the input to the strong random unitaries, which guarantees that the overlap of the input state with counter-examples to the strong gluing construction is exceedingly small.
This allows us to reduce the circuit depth of strong unitary $k$-designs and strong PRUs to match the circuit depth of unitary 2-designs, i.e.~$\mathcal{O}(\log n)$~\cite{cleve2015near}.

We consider the random unitary ensemble depicted in Fig.~\ref{fig:constructions}(b).
Inspired by~\cite{schuster2024random}, we partition the $n$ qubits into $n/\xi$ patches of $\xi$ qubits each, arranged in a 1D line.
We then form a two-layer  circuit composed of small strong random unitaries, where the small unitaries act on two neighboring patches each and are arranged in a brickwork fashion between the two layers.
Finally, we ``scramble'' the two-layer circuit by appending it with an $n$-qubit strong unitary 2-design on either side.

Our main result is that the scrambled two-layer  ensemble forms a strong unitary $k$-design (when each small random unitary is drawn from a strong unitary $k$-design) and a strong PRU (when each small random unitary is drawn from a strong PRU).
\begin{theorem}[The scrambled two-layer ensemble is a strong unitary design] \label{thm:two-layer-design}
    Let each small random unitary be a strong $\frac{\varepsilon}{n}$-approximate unitary $k$-design on $2\xi$ qubits.
    Then the scrambled two-layer ensemble is a strong $\varepsilon$-approximate unitary $k$-design when $\xi \geq \frac{16}{3}\log_2(nk^2/\varepsilon) + \mathcal{O}(1)$. 
\end{theorem}
\begin{theorem}[The scrambled two-layer ensemble is a strong PRU] \label{thm:two-layer-PRU}
    Let each small random unitary be a strong PRU with $\poly n$-time security  on $2\xi$ qubits.
    Then the scrambled two-layer ensemble is a strong PRU with $\poly n$-time security when $\xi = \omega(\log n)$. 
\end{theorem}
\noindent The theorems immediately yield strong unitary designs and strong PRUs in the circuit depths  in Theorems~\ref{thm:strong-design-depth} and~\ref{thm:strong-PRU-depth}.
This follows by adding the circuit depth $\mathcal{O}(\log n)$ of an exact unitary 2-design~\cite{cleve2015near} to the circuit depths of the LRFC ensemble on $2\xi$ qubits\footnote{For strong unitary $k$-designs, we set $\xi = \mathcal{O}(\log nk/\varepsilon)$ (Theorem~\ref{thm:two-layer-design}) which yields a circuit $\mathcal{O}(\log k \log \log nk/\varepsilon)$ for the $2\xi$-qubit LRFC  unitary $k$-design. We then add this to the 2-design circuit depth and note that $\mathcal{O}(\log n + \log k \log \log nk/\varepsilon) = \mathcal{O}(\log n + \log k \log \log k/\varepsilon)$. For strong PRUs, we set $\xi = \omega(\log n)$ (Theorem~\ref{thm:two-layer-PRU}), which yields circuit depth $\poly(\log \log n)$ for the $2\xi$-qubit LRFC PRU. This is  sub-leading to the  2-design circuit depth  $\mathcal{O}(\log n)$.}.

As mentioned above, we analyze the scrambled two-layer ensemble by proving that one can glue small strong random unitaries together one brick at a time~\cite{schuster2024random}.
Without the larger unitary 2-designs, this gluing does not hold.
Applying the following lemma $n/\xi$ times yields Theorems~\ref{thm:two-layer-design} and~\ref{thm:two-layer-PRU}.

\begin{lemma}[Gluing strong random unitaries] \label{lemma: strong gluing}
    Let $\mathsf a, \mathsf b, \mathsf c$ be three subsystems of size at least $\xi$.
    Consider the unitary ensemble $U_1 = D_{\mathsf{abc}} U_{\mathsf{bc}} U_{\mathsf{ab}}   C_{\mathsf{abc}}$, where $C_{\mathsf{abc}}, D_{\mathsf{abc}}$ are strong $\varepsilon_2$-approximate unitary 2-designs and $U_{\mathsf{ab}},U_{\mathsf{bc}}$ are strong $\varepsilon_{\mathsf{ab}}$- and $\varepsilon_{\mathsf{bc}}$-approximate unitary $k$-designs on their respective subsystems.
    Then $U_1$ forms a strong $\varepsilon$-approximate unitary $k$-design with measurable error  $\varepsilon = \varepsilon_{\mathsf{ab}} + \varepsilon_{\mathsf{bc}} + \mathcal{O}(k^2/2^{(3/16)\xi}) + \mathcal{O}(k^{5/8}\varepsilon_2^{1/8})$.
    %
\end{lemma}
\noindent The strong gluing lemma also allows us to straightforwardly extend Theorems~\ref{thm:two-layer-design} and~\ref{thm:two-layer-PRU} to allow approximate strong unitary 2-designs instead of exact unitary 2-designs.
This extension will be helpful for our construction of ancilla-free designs and PRUs. 

We prove Lemma~\ref{lemma: strong gluing} in Appendix~\ref{sec: proof strong gluing} using the path-recording framework from~\cite{ma2024construct}.
We find that the path-recording framework enables the most effective analyses of strong random unitaries.
This contrasts with standard approximate unitary designs and PRUs, where other succinct approaches based on the permutation group are possible~\cite{schuster2024random,cui2025unitary}.

\subsection{Ancilla-free pseudorandom unitaries}
In this section, we give the first constructions of ancilla-free PRUs and strong PRUs. Our main crytographic building block will be pseudorandom functions \cite{goldreich1986construct} computable in the complexity class ``logspace-uniform $\mathsf{TC}^1$''. A function is computable in logspace-uniform $\mathsf{TC}^1$ if (1) it is computable by a family of $\mathcal{O}(\log n)$-depth circuits with large fan-in threshold gates and (2) this family of circuits is output by a logspace Turing machine on the input $1^n$. Crucially, it is known that such PRFs exist under the LWE assumption \cite{banerjee2012pseudorandom}, and that this construction is post-quantum secure \cite{zhandry2021PRF}. 

To construct ancilla-free PRUs, we prove (under standard cryptographic assumptions) the existence of an intermediate object: an ``ancilla-independent'' strong PRU. We say that an $(n+a)$-qubit circuit implements an $n$-qubit unitary $U$ in an ancilla-independent way if the circuit act as $U \otimes \mathbbm{1}_{2^a}$. Note that the standard notion of a PRU only requires that the ancilla is undisturbed when it is initialized to the all $0$ state; in comparison, an ancilla-independent implementation requires that the ancilla is undisturbed no matter how it is instantiated. 

\begin{theorem}\label{thm:ancilla-independent-PRU-intro}
    Assuming that there exist polynomially secure (respectively, sub-exponentially secure) post-quantum PRFs computable in logspace-uniform $\mathsf{TC}^1$, there exist polynomially secure (respectively, sub-exponentially secure) ancilla-independent strong PRUs. 
\end{theorem}

We show this by instantiating the LRFC ensemble with ancilla-independent implementations of the underlying pseudorandom functions. Towards this goal, our main technical result on ancilla-free computation is as follows. 

\begin{theorem}\label{thm:ancilla-free-implementation-intro}
    Let $f: \{0,1\}^n \rightarrow \mathbb Z_q^m$ be any logspace-uniform $\mathsf{TC}^1$-computable function, where $q = O(1)$. Then, there is a $\mathsf{poly}(n, m)$-size reversible circuit implementing the permutation
    \begin{equation} (x,y,a) \mapsto (x, y + f(x) \!\!\!\! \mod q, a), 
    \end{equation}
    where $a$ denotes an arbitrary setting of the ancilla register. 
\end{theorem}

To prove \cref{thm:ancilla-free-implementation-intro}, we leverage and build upon recent work on catalytic quantum computation~\cite{TQC:BFMSSST25-catalytic}. Based on this work, it is known that logspace-uniform $\mathsf{TC}^1$ functions can be implemented in a \emph{somewhat} ancilla-independent way. Namely, they require $\poly(n)$ ancilla qubits but only $O(\log n)$ clean ancilla qubits: if the clean ancillae are all initialized to $\ket{0}$, the circuit properly computes the function and acts as identity on the remaining ancillae. In~\cref{app: no-ancilla}, we show how to remove the need for these these last $\mathcal{O}(\log n)$ clean ancillae by exploiting a number of reversible circuit identities. 

\paragraph{From ancilla-independent PRUs to ancilla-free PRUs.} Finally, to compile ancilla-independent strong PRUs into ancilla-free strong PRUs, we instantiate the scrambled two-layer ensemble with ancilla-independent strong PRUs on $n^{\epsilon}$ qubits each, where $\epsilon$ is a constant greater than zero.
Crucially, the strong PRU blocks can \emph{reuse} registers that serve as the ancilla registers of other PRU blocks, due to their ancilla-independence. This gives ancilla-free strong PRUs on $n$ qubits of $\poly(n)$ depth. To compress the depth down to $\poly(\log n)$, we instantiate the two-layer ensemble again, with
\begin{itemize}
    \item $\omega(\log^2 n)$-depth ancilla-free instantiations of the unitary $2$-designs; see Lemma~\ref{lemma: strong 2-designs} of the following section and set $\varepsilon^{-1} = \omega(\poly n)$, and
    \item ancilla-free PRUs on $\poly(\log n)$ qubits each, which have depth $\poly(\log n)$. Note that for these to be secure against all $\poly(n)$-time adversaries, we need to assume that the underlying cryptography is sub-exponentially secure.
\end{itemize}
This yields $\poly(\log n)$-depth strong ancilla-free PRUs.
\begin{theorem}\label{thm:ancilla-free-PRU-intro}
    Assuming post-quantum PRFs computable in $\mathsf{TC}^1$, there exist ancilla-free strong PRUs. Moreover, assuming sub-exponentially secure post-quantum PRFs computable in $\mathsf{TC}^1$, there exist ancilla-free strong PRUs computable in depth $\poly(\log n)$ with all-to-all circuits. 
\end{theorem}

Since logspace-uniform $\mathsf{TC}^1$-computable PRFs are known under the LWE assumption \cite{banerjee2012pseudorandom}, we obtain instantiations of our results under LWE.

\begin{corollary}\label{cor:LWE-ancilla-free-PRUs-intro}
    Assuming the post-quantum hardness of LWE there exist ancilla-free PRUs. Assuming the sub-exponential post-quantum hardness of LWE, there exist ancilla-free strong PRUs computable in depth $\poly(\log n)$ with all-to-all circuits. 
\end{corollary} 

\subsection{Strong unitary designs from local random circuits} \label{sec: local random circuit}

We now present our final construction of strong random unitaries, in all-to-all-interacting local random circuits with a specific architecture.
This demonstrates that the fast formation of strong random unitaries is possible even in highly unstructured ensembles, complementing our highly structured implementations thus far.
This generality is not trivial; for example, standard unitary designs can be implemented in circuit depth $\mathcal{O}(\log \log n)$ in highly structured unitary ensembles, but require depth $\Omega(\log n)$ in local random circuits on any architecture~\cite{schuster2024random,fefferman2024anti,cui2025unitary}.
%
%

Our random circuit construction builds upon the scrambled two-layer ensemble introduced in the previous section [Fig.~\ref{fig:constructions}(c)].
We modify the ensemble in two ways to enable a random circuit realization.
First, we replace both of the $n$-qubit unitary 2-designs with blocked fast scrambling circuits, composed of small unitary 2-designs acting on $2\xi$ qubits each.
%
%
The small unitary 2-designs are arranged such that the $i$-th patch of qubits is coupled to the $(i+2^{d-1})$-patch of qubits at the $d$-th circuit layer.
This guarantees that the light-cone of every qubit doubles at every circuit layer, so that every light-cone encompasses all $n$ qubits after $\log_2(n/\xi)$ layers.
Second, we replace all of the small random unitaries in the ensemble with 1D local random circuits~\cite{brandao2016local}.
We specify the depths of these circuits below.
The circuit is highly unstructured, in the sense that each two-qubit gate is drawn independently at random from the Haar measure on $U(4)$.

To prove that the first modification is valid, we show that the blocked fast scrambling circuit is a strong unitary 2-design.
\begin{lemma}[The blocked fast scrambling circuit is a strong unitary 2-design] \label{lemma: strong 2-designs}
    Let each small random unitary be a strong $\frac{\varepsilon}{n}$-approximate unitary 2-design with depth $d$.
    Then the blocked fast scrambling circuit forms a strong $\varepsilon$-approximate unitary 2-design with depth $d \log_2(n/\xi)$ for any $\xi \geq \log_2( 5 n / \varepsilon)$. %
\end{lemma}
\noindent We then prove that each small strong approximate unitary design can be replaced with a 1D local random circuit.
\begin{lemma}[1D random circuits are strong unitary $k$-designs] \label{lemma: strong design 1D RUC}
    1D random circuits on $n$ qubits form strong $\varepsilon$-approximate unitary $k$-designs in depth $d = \Omega \big( \! \log(k)^7 ( nk + \log(1/\varepsilon) ) \big)$ for any $\varepsilon \geq 2k^2/2^n$.
\end{lemma}
\noindent This requires circuit depth $\mathcal{O}(\poly \log k  (\xi k + \log n/\varepsilon))$  for the two layers of small unitary $k$-designs, and  depth $\mathcal{O}(\xi+\log n/\varepsilon)$ per layer for the $\log_2(n/\xi)$ layers of small unitary 2-designs.
Setting $\xi = \mathcal{O}(\log nk/\varepsilon)$ as in Theorem~\ref{thm:two-layer-design} yields a total circuit depth of $\mathcal{O}(k \cdot \poly \log k \cdot \log n/\varepsilon + \log n \cdot \log n k/\varepsilon)$. 

The proofs of Lemma~\ref{lemma: strong 2-designs} and Lemma~\ref{lemma: strong design 1D RUC} are contained in Appendix~\ref{sec: proof blocked fast scrambling} and~\ref{sec: strong linear random}, respectively.
We prove Lemma~\ref{lemma: strong 2-designs} by establishing a formal connection between operator spreading and strong unitary 2-designs.
Operator spreading refers to the growth in support of initially local operators under a quantum circuit or time dynamics; this signifies the scrambling of local information into non-local correlations~\cite{roberts2015localized,roberts2018operator,xu2024scrambling}.
In random circuits, this growth is probabilistic and is characterized by a probability distribution over the set of Pauli strings~\cite{nahum2018operator}.
We show that an ensemble forms a strong unitary 2-design if and only if its operator spreading distributions are close to those of a Haar-random unitary. 
To establish that the blocked fast scrambling circuit satisfies this condition, we prove that with high probability every operator is randomly supported within its light-cone at every layer.

The statement of Lemma~\ref{lemma: strong design 1D RUC} for strong unitary designs closely parallels an analogous statement for standard unitary designs~\cite{chen2024incompressibility}.
These statements are derived by translating lower bounds on the spectral gap of random circuits to upper bounds on their design depth~\cite{brandao2016local}.
This translation involves several steps: One first uses the spectral gap to bound the so-called additive error of a unitary design, and then translates the additive error to a bound on the relative error.
For standard unitary designs, this translation can either be proven using the Schur-Weyl duality~\cite{brandao2016local} or from elementary properties of the permutation operators~\cite{schuster2024random}.
In Appendix~\ref{app: preliminaries}, we provide an analogous additive-to-relative-error translation result for strong unitary designs.
Our analysis yields numerous insights into the structure of the mixed Haar twirl, which may be useful in other contexts.
For example, we formalize the following statement: Within any quantum experiment, a Haar-random unitary $U$ and its inverse $U^\dagger$ either cancel one another, or behave identically to two independent random unitaries $U$ and $V$.

\section{Fast scrambling} \label{sec: scrambling}

\begin{figure}[t]
    \centering
    \includegraphics[width=1.0\textwidth]{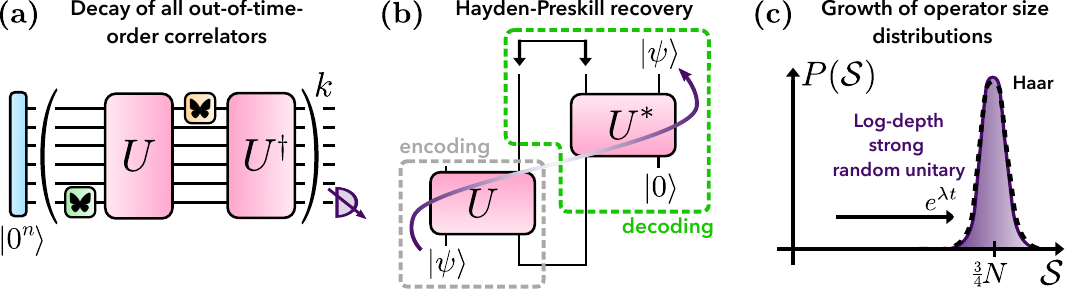}
    \caption{By allowing queries to $U^\dagger$ and $U^*$, strong random unitaries capture hallmark features of quantum information scrambling. 
    \textbf{(a)} In a strong unitary $k$-design, every $k$-point time-ordered and out-of-time-order correlation function decays to near zero with high probability.
    \textbf{(b)} Strong random unitaries are good encoders of quantum information, as in the Hayden-Preskill thought experiment~\cite{hayden2007black}. This follows because the decoding protocol uses the conjugate random unitary~\cite{yoshida2017efficient}.
    \textbf{(c)} In a strong unitary 4-design or strong pseudorandom unitary, the operator size distribution approaches its Haar-random value to within small total variational distance.
    }
    \label{fig:scrambling}
\end{figure}

We now elaborate on the connections between our results and fast quantum information scrambling.
Quantum information scrambling is a broad field focused on understanding the spreading of quantum information in dynamical many-body quantum systems. 
As discussed in detail in our introduction, the \emph{fast scrambling} conjecture~\cite{sekino2008fast} posits that (i) all-to-all-connected quantum systems can scramble information in time growing only logarithmically in the system size $n$, and (ii) this is the fastest that any quantum system can scramble information.

Our construction of strong unitary $k$-designs and strong pseudorandom unitaries in $\mathcal{O}(\log n)$ circuit depth provides the strongest confirmation to date of the fast scrambling conjecture.
These results prove that every observable feature of unitary quantum dynamics can mimic Haar-random behavior in logarithmic depth.
This significantly expands upon existing works, which, as aforementioned, focus on either a small number of specific signatures or solely unitary 2-design properties~\cite{maldacena2016remarks,kitaev2015simple,roberts2018operator,brown2012scrambling,brown2013short,brown2015decoupling,lashkari2013towards,cleve2015near,belyansky2020minimal,bentsen2019fast,vikram2024exact}.
Meanwhile, our lower bound, which is simple to derive, confirms that a logarithmic depth is optimal for any unitary to appear Haar-random in experiments that involve the inverse or conjugate.


We illustrate the connection between strong unitary designs and strong PRUs and quantum information scrambling in more detail through several examples (Fig.~\ref{fig:scrambling}).
We consider four hallmark diagnostics of quantum information scrambling: (a) the decay of all local time-ordered and out-of-time-order correlation functions~\cite{xu2024scrambling}, (b) the Hayden-Preskill thought experiment~\cite{hayden2007black}, (c) the saturation of operator size distributions to their Haar-random profile~\cite{roberts2018operator}, and (d) the growth of the entanglement and operator entanglement entropies.
The implications of our results for each behavior follow  straightforwardly from our definitions of strong unitary $k$-designs and strong PRUs.
We discuss this example-by-example below.

\paragraph{Decay of all out-of-time-order correlation functions.} 
Perhaps the simplest diagnostic of scrambling is the decay of all local correlation functions to zero~\cite{hosur2016chaos}.
This includes both conventional \emph{time-ordered} correlations functions (TOCs), as well as \emph{out-of-time-order} correlations functions (OTOCs)~\cite{xu2024scrambling}.
The former can be measured with forward time-evolution under $U$.
Crucially however, the latter can only be efficiently measured by alternating forward and backward time-evolution under $U$ and $U^\dagger$ [Fig.~\ref{fig:scrambling}(a)].
In Appendix~\ref{app: scrambling}, we provide a short proof that all local $k$-point correlation functions are near zero with high probability in any strong unitary $2k$-design.
A similar statement holds for strong PRUs. 
This complements existing observations that OTOCs can decay in logarithmic time in the Sachdev-Ye-Kitaev model~\cite{kitaev2015simple,maldacena2016remarks} and all-to-all-connected random circuits~\cite{schuster2022many,schuster2023operator}.
Our results capture both standard four-point OTOCs as well as higher-point OTOCs, which have risen in interest in recent work~\cite{abanin2025constructive,pappalardi2022eigenstate,fava2025designs,dowling2025free,vardhan2025free}.

\paragraph{Hayden-Preskill thought experiment.}
A key inspiration for early studies of quantum information scrambling came from the black hole information paradox~\cite{hawking1976breakdown}.
In this context, Hayden and Preskill proposed that, if one models the dynamics of a black hole by a random unitary, then one could use the Hawking radiation collected from a black hole to recover a quantum state that fell into the black hole at an earlier time~\cite{hayden2007black}.
An explicit decoding protocol was later provided by Yoshida and Kitaev~\cite{yoshida2017efficient}.
Crucially, the decoding protocol utilizes the conjugate random unitary $U^*$ [Fig.~\ref{fig:scrambling}(b)].
Thus, any strong unitary 2-design or strong pseudorandom unitary will encode information precisely as well as a Haar-random unitary in the Hayden-Preskill thought experiment\footnote{We note that the original work by Hayden and Preskill~\cite{hayden2007black} utilized an early definition of approximate unitary 2-designs~\cite{dankert2009exact} that in fact bears some resemblance with our strong definition. This definition fell out of use in later works which focused on unitary $k$-designs for general $k$.}.
Intriguingly, this result does not seem to extend to all applications of Haar-random unitaries as quantum codes; in particular, it applies only when the decoding can be performed \emph{efficiently} using query access to $U$, $U^\dagger$, $U^*$, $U^T$.

\paragraph{Growth of operator size distributions.} An even more fine-grained diagnostic of quantum information scrambling is the \emph{operator size distribution}~\cite{roberts2018operator}.
The size distribution characterizes the support of a time-evolved operator when expanded in the Pauli basis, $O(t) \equiv UOU^\dagger = \sum_P c_P(t) P$.
Specifically, $P(\mathcal{S}) \equiv \sum_{|P|=\mathcal{S}} |c_P(t)|^2$, where the sum is over all Pauli operators of weight (i.e.~\emph{size}) $\mathcal{S}$.
Operator size distributions and are central to applications of information scrambling in quantum gravity~\cite{roberts2018operator,brown2023quantum,nezami2023quantum,schuster2022many}, quantum sensing~\cite{kobrin2024universal}, and understanding the impact of noise on quantum systems~\cite{schuster2023operator,sanchez2021emergent,schuster2024polynomial}.
In Appendix~\ref{app: scrambling}, we prove that the operator size distribution of any strong $\varepsilon$-approximate unitary 4-design is $n^2\varepsilon$-close to its Haar-random value in total variational distance.
Setting $\varepsilon = 1/\omega(\poly n)$ yields operator size distributions that are super-polynomially close to their Haar values in circuit depth $\mathcal{O}(\log n)$.

\paragraph{Entanglement and operator entanglement entropy.} Finally, we consider the entanglement and operator entanglement entropies.
In principle, neither of these quantities can be efficiently measured in any quantum experiment.
Therefore their values in a Haar-random state need not be replicated by approximate strong unitary designs or strong PRUs.
Nonetheless, in Appendix~\ref{app: scrambling}, we show that an especially precise form of strong unitary designs, characterized by a small \emph{relative error}, can efficiently capture the entanglement and operator entanglement entropies.
In Appendix~\ref{app: LRFC}, we show that such designs can be formed in $\mathcal{O}(\log n)$ circuit depth.
From this, we find that the Renyi-2 entanglement entropy of any subsystem of any time-evolved state $\ket{\psi(t)} \equiv U \ket{\psi}$ reaches its Haar-random value in $\mathcal{O}(\log n)$ depth.
In addition, the Renyi-2 operator entanglement entropy of any subsystem of any time-evolved operator $O(t) \equiv UOU^\dagger$ reaches its Haar-random value in $\mathcal{O}(\log n)$ depth.
Formally, these bounds are achieved by reformulating each entropy as the expectation value of a positive-valued operator after $U$ and $U^*$ are applied to a fictitious larger system.

\section{Discussions} \label{sec: discussions}

Our results provide the strongest constructions of approximate unitary designs and PRUs to date. Moreover, the circuit depths of our constructions achieve the optimal $\Theta(\log n)$ scaling, as predicted by the fast scrambling conjecture. These results establish a rigorous operational foundation for quantum information scrambling and represent the most comprehensive confirmation of fast scrambling to date, capturing all efficiently observable quantum experiments. Our work leaves open several interesting questions.

We have motivated strong unitary designs and strong PRUs from applications to quantum information scrambling and black hole physics. What other applications of strong random unitaries might exist? For example, can strong unitary designs help us understand the classical hardness of recent quantum advantage proposals involving time-reversal dynamics~\cite{mi2021information,abanin2025constructive}? Or perhaps the sensitivity of such experiments to experimental noise~\cite{schuster2023operator}? More broadly, can strong random unitaries assist in device benchmarking and other quantum learning tasks? 

Our constructions provide the first examples of PRUs secure against queries to all of $U$, $U^\dagger$, $U^*$, and $U^T$. Can the existence of these strong PRUs and the ability to generate them in logarithmic depth enable new quantum cryptographic applications~\cite{zhandry2024model}? Our ancilla-free constructions address a fundamental limitation of previous PRU constructions, which relied on auxiliary systems that are unrealistic in physical settings. This advance opens a deeper question: can we build upon ancilla-free strong random unitaries to demonstrate that dynamics naturally arising in condensed matter, quantum chaos, and high-energy physics form strong random unitaries?

More broadly, giving ancilla-free constructions of \emph{any} quantum cryptographic object---such as commitments, encryption, uncloneable cryptography, etc.---can constitute a stronger form of evidence (compared to a pure existential result) that natural physical phenomena may possess these cryptographic properties. This motivates understanding whether, and under what hardness assumptions, other quantum cryptographic primitives have efficient ancilla-free instantiations. 

Finally, our new techniques for analyzing strong random unitaries raise several questions. Can the circuit depth of our local random circuit construction of strong unitary designs be further improved, from $\mathcal{O}(\log^2 n)$ to $\mathcal{O}(\log n)$? In particular, our treatment of the blocked fast scrambling circuit is extremely coarse and likely leaves room for an improved analysis. Along similar lines, can one further remove the structure from our random circuit designs, and prove the fast formation of strong unitary designs in the standard model of all-to-all random circuits~\cite{brown2012scrambling,brown2013short,brown2015decoupling}? More broadly, can our refined techniques for analyzing the mixed-unitary twirl~\cite{grinko2024linear,grinko2023gelfand,nguyen2023mixed}, and quantum experiments involving inverse and conjugate unitaries, yield any new progress or insights elsewhere in quantum information theory?

\vspace{0.5em}
\subsection*{Acknowledgments:}We are grateful to Laura Cui, Jonas Haferkamp, and Nicholas Hunter-Jones for valuable discussions.
T.S. acknowledges
support from the Walter Burke Institute for Theoretical Physics at Caltech.
T.S. and H.H. acknowledge support from the U.S. Department of Energy, Office of Science, National Quantum Information Science Research Centers, Quantum Systems Accelerator.
The Institute for Quantum Information and Matter, with which T.S., F.B., and H.H. are affiliated, is an NSF Physics Frontiers Center (NSF Grant PHY-2317110).
This work was done while F.M. was a postdoctoral fellow at the Simons Institute for the Theory of Computing, supported by DOE QSA grant FP00010905, NSF QLCI Grant 2016245 and DOE grant DE-SC0024124.

\vspace{3.5em}
\appendix

\addtocontents{toc}{\protect\setcounter{tocdepth}{2}}

\noindent 
\textbf{\LARGE{}Appendices}
\vspace{1em}

\noindent Our Appendices are organized as follows. In Appendix~\ref{app: designs}, we provide the full details of our results on strong unitary designs.
In Appendix~\ref{app: PRU}, we provide the full details of our results on strong pseudorandom unitaries.
In Appendix~\ref{app: LRFC}, we prove that the LRFC ensemble forms a strong unitary design and strong PRU.
In Appendix~\ref{app: gluing}, we prove our gluing results for strong unitary designs and strong PRUs.
In Appendix~\ref{sec: details mixed}, we provide additional details on the mixed Haar twirl, which are used to prove a translation lemma for the approximation errors of strong unitary designs in Appendix~\ref{app: designs}.
In Appendix~\ref{app: scrambling}, we provide full details on our applications of strong random unitaries to scrambling.

\tableofcontents

\section{Strong unitary designs} \label{app: designs}

In this Appendix, we provide the full details of our definition and constructions of strong unitary designs. Beyond the first preliminaries and definitions section, we have structured each subsection so that they can be read relatively independently of one another.

\subsection{Preliminaries and definitions} \label{app: preliminaries}

In this section, we introduce and define strong unitary $k$-designs for various notions of approximation error.
Our definitions are closely modeled on the analogous definitions for (standard) unitary $k$-designs.
With this in mind, we first review the standard definitions of unitary designs.
We then briefly discuss the limitations of these definitions in regards to fast scrambling and quantum experiments involving time-reversal.
We then introduce our strong definitions to capture such behaviors.

\subsubsection{Averaging over Haar-random unitaries}

To provide the definitions of unitary designs, let us first introduce the average (i.e.~\emph{twirl}) over a Haar-random unitary.
This will form our point of comparison in the study of unitary designs.
We consider two varieties of the twirl in our work: the Haar twirl and the mixed Haar twirl.

\paragraph{The Haar twirl.} The Haar twirl has the following definition.
\begin{definition}[The Haar twirl]
    Given a linear operator $X$ acting on $nk$ qubits, the $k$-th moment with respect to $U(2^n)$ is defined via the twirl over the unitary group:
    \begin{equation}
    \Phi^{(k)}_{H}(X) = \int dU \, U^{\otimes k} X (U^\dagger)^{\otimes k}.
    \end{equation}
\end{definition}
\noindent An explicit formula for the Haar twirl can be derived from a simple argument in representation theory~\cite{fulton2013representation, goodman2009symmetry,mele2024introduction}.
This yields the following expression.
\begin{lemma}[Explicit expression for the Haar twirl]
    For any $k \leq 2^n$. For any linear operator $X$ acting on $nk$ qubits, the $k$-th moment with respect to the unitary group can be written in the form 
    \begin{equation} \label{eq: Haar twirl}
    \Phi_{H}^{(k)}(X) = \sum_{\pi, \tilde{\pi} \in S_k} \Wg_{\pi, \tilde{\pi}} \cdot \tr(X \pi^{-1}) \cdot \tilde{\pi},
    \end{equation}
    where $\pi,\tilde{\pi} \in S_k$ permute the $k$ copies of the $n$-qubit Hilbert space, and the Weingarten matrix elements $\Wg_{\pi,\tilde{\pi}}$ depend on $k$ and the Hilbert space dimension $2^n$.
\end{lemma}

\paragraph{The mixed Haar twirl.} To incorporate quantum experiments that can query the inverse and conjugate of a random unitary, we will also make use of the mixed Haar twirl in our work.
Here, the $k$ copies of the unitary $U$ are replaced by $p$ copies of $U$ and $q$ copies of its conjugate $U^*$.
\begin{definition}[The mixed Haar twirl]
    Given a linear operator $X$ acting on $nk$ qubits, the $(p,q)$-th moment with respect to $U(2^n)$ is defined via the mixed twirl over the unitary group:
    \begin{equation}
    \Phi^{(p,q)}_{H}(X) = \int dU \, (U^{\otimes p} \otimes U^{*,\otimes q}) X (U^{\dagger, \otimes p} \otimes U^{T,\otimes q}).
    \end{equation}
\end{definition}
\noindent The mixed Haar twirl can be obtained from the standard Haar twirl by taking the partial transpose of the last $q$ registers before and after $(U^*)^{\otimes q}$ is applied. Thus, the expression for the mixed Haar twirl is fully determined by the expression for the standard Haar twirl.
\begin{lemma}[Explicit expression for the mixed Haar twirl] \label{lemma: explicit mixed twirl}
    For any $p+q \leq 2^n$. For any linear operator $X$ acting on $nk$ qubits, the $(p,q)$-th moment with respect to the unitary group can be written
    \begin{equation} \label{eq: mixed Haar twirl}
    \Phi^{(p,q)}_{H}(X) = \sum_{\pi, \tilde{\pi} \in S_k} \Wg_{\pi, \tilde{\pi}} \cdot \tr(X (\pi^{-1})^\Gamma) \cdot \tilde{\pi}^\Gamma,
    \end{equation}
    where $\Gamma$ is the partial transpose on the final $q$ registers.
\end{lemma}
\begin{proof}
    The expression follows from the equality $\Phi^{(p,q)}_H(X) = \Phi_H^{(k)}(X^\Gamma)^\Gamma$. The partial transpose $\Gamma$ is then transferred from $X$ to $\pi^{-1}$ inside the trace using $\tr(A^\Gamma B) = \tr(A B^\Gamma)$.
\end{proof}

\subsubsection{Approximate unitary designs}

Let us now turn to unitary designs.
A unitary ensemble is an exact unitary $k$-design if it exactly replicates the first $k$ moments of a Haar-random unitary.
\begin{definition}[Exact unitary $k$-design] \label{def: exact design}
    An ensemble of unitaries $\mathcal{E}$ is an \emph{exact unitary $k$-design} if it exactly reproduces the first $k$ moments of the Haar measure
    \begin{equation}
        \Phi^{(k)}_\mathcal{E} = \Phi^{(k)}_H
    \end{equation}
    where we have used the abbreviated notation
    \begin{equation}
    \Phi^{(k)}_\mathcal{E}(X) = \E_{U \sim \mathcal{E}} U^{\otimes k} X (U^\dagger)^{\otimes k}
    \end{equation}
    to denote the $k$-th moment over the unitary ensemble $\mathcal E$.
\end{definition}
\noindent An exact design is the strongest notion of unitary design.
It guarantees that any experiment that queries the unitary $U$ (or $U^\dagger$ or $U^T$ or $U^*$ or controlled versions of any of these quantities) up to $k$ times exactly reproduces the output of the same experiment querying a Haar-random unitary\footnote{For experiments involving controlled queries, one should also assume that the unitary ensemble $\mathcal{E}$ is invariant under a random global phase $e^{i\phi}$. This guarantees that all un-matched moments, such as $\E_{U \sim \mathcal{E}} [U^{\otimes k} (\cdot) U^{\otimes k'}]$ for $k \neq k'$, vanish. Equality of the matched moments (Definition~\ref{def: exact design}) then guarantees that all controlled queries to a Haar-random unitary are exactly reproduced by controlled queries to a unitary sampled from $\mathcal{E}$.}.

In practice, exact constructions of unitary designs are extremely scarce beyond very low moments $k \leq 3$.
This motivates the notion of an \emph{approximate} design.
Three forms of approximation error for unitary designs are common.

\paragraph{Additive error.} The simplest form of approximation error is the additive or diamond-norm error.
\begin{definition}[Unitary $k$-design with additive error]
    Let $\varepsilon > 0$.
    An ensemble of unitaries $\mathcal{E}$ is an approximate unitary $k$-design with additive error $\varepsilon$ if
    \begin{equation}
    \norm{\Phi^{(k)}_\mathcal{E} - \Phi^{(k)}_H}_\diamond \leq \varepsilon,
    \end{equation}
    where $\lVert \Phi - \Phi' \rVert_\diamond \equiv \max_\rho \lVert \Phi(\rho) - \Phi'(\rho) \rVert_1$ is the diamond norm. The maximization is over all states $\rho$ on $nk+m$ qubits, where the number $m$ of ancilla qubits may be arbitrarily large.
\end{definition}
\noindent Physically, the additive error is equivalent to security under parallel queries to the unitary $U$.
Namely, an ensemble $\mathcal{E}$ is an approximate unitary $k$-design up to additive error $\varepsilon$ if and only if for any quantum algorithm making a single query to $U^{\otimes k}$, i.e.~$k$ parallel queries to $U$, the output states when $U$ is sampled from $\mathcal{E}$ versus the Haar ensemble are $\varepsilon$-close in trace distance~\cite{cui2025unitary}.
This follows immediately from the definition of the diamond norm.

\paragraph{Measurable error.} The additive error has a significant drawback, in that it can only capture experiments in which $U$ is applied $k$ times in parallel.
To address this, Ref.~\cite{cui2025unitary} introduced a stronger notion of approximation error, which guarantees that an ensemble is indistinguishable from Haar-random in \emph{any} quantum experiment that queries $U$ up to $k$ times.
This is termed the \emph{measurable error} owing to its physical motivation.
\begin{definition}[Unitary $k$-design with measurable error]
    Let $\varepsilon > 0$. 
    An ensemble of unitaries $\mathcal{E}$ is an approximate unitary $k$-design with measurable error $\varepsilon$ if for any quantum experiment with $k$ queries to $U$, the output states when $U$ is sampled from $\mathcal{E}$ versus the Haar ensemble are $\varepsilon$-close in trace distance,
    \begin{equation}
        \sup_{W_1 \cdots W_{k+1}} \norm{\rho_\mathcal{E} - \rho_H}_1 \leq \varepsilon,
    \end{equation}
    where we have used the notation 
    \begin{equation} \nonumber
        \rho_\mathcal{E} = \E_{U \sim \mathcal{E}} \left[ W_{k+1} [U \otimes \mathbbm{1}_m] W_k \cdots W_2 [U \otimes \mathbbm{1}_m] W_1 |0^{n+m} \rangle\! \langle 0^{n+m} | W_1^\dagger [U^\dagger \otimes \mathbbm{1}_m] W_2^\dagger \cdots W_k^\dagger [U^\dagger \otimes \mathbbm{1}_m] W_{k+1}^\dagger 
        \right]
    \end{equation}
    to denote the expected output state of a general quantum experiment that queries $U$ $k$ times.
    Each $W_i$ is an arbitrary unitary on $n+m$ qubits, where the number $m$ of ancilla qubits may be arbitrarily large.
\end{definition}
\noindent In the definition, we can assume without loss of generality that each $U$ is applied in sequence on the same subsystem $A$ of $n$ qubits. If the unitaries are in fact applied in parallel, this is equivalent to performing the first unitary $U$, then using $W_1$ to swap $A$ with $n$ ancilla qubits, then performing the second unitary $U$, then using $W_2$ to swap back $A$ and the $n$ ancilla qubits, and so on.

The measurable error is in some sense the most natural notion of approximation error for unitary designs.
It captures precisely the features of a random unitary that can be measured in physical experiments that query the unitary.

\paragraph{Relative error.} The strongest notion of approximation error for unitary designs is the relative error. 
Unlike the additive or measurable errors, the relative error is sensitive to properties that cannot be efficient measured in any quantum experiment.
It has the following definition.
\begin{definition}[Unitary $k$-design with relative error]
    Let $\varepsilon > 0$. Then an ensemble of unitaries $\mathcal{E}$ is an approximate unitary $k$-design up to relative error $\varepsilon$ if
    \begin{equation}
    (1 - \varepsilon) \Phi^{(k)}_H \preceq \Phi^{(k)}_{\mathcal{E}} \preceq (1+\varepsilon) \Phi^{(k)}_H,
    \end{equation}
    where $\mathcal{A} \preceq \mathcal{B}$ denotes that $\mathcal{B}-\mathcal{A}$ is completely positive.
\end{definition}
\noindent Physically, an ensemble $\mathcal{E}$ is an approximate unitary $k$-design up to relative error $\varepsilon$ if and only if for any quantum experiment with $k$ queries to $U$ (or $U^T$), the expectation value of any positive-valued operator $\chi$ is equal to its Haar value to within multiplicative precision, $(1-\varepsilon) \text{tr}( \chi \rho_H ) \leq \text{tr}( \chi \rho_\mathcal{E} ) \leq (1+\varepsilon) \text{tr}( \chi \rho_H )$.
This holds even if the expectation value is exponentially small and cannot be efficient measured.

\paragraph{Translating between unitary design approximation errors.}
We can translate between different notions of approximation error for unitary designs as follows.
\begin{lemma}[Translating between different approximation errors~\cite{brandao2016local,schuster2024random,cui2025unitary}] \label{lemma:translating}
    The additive error $\varepsilon_a$ is upper bounded by the measurable error $\varepsilon_m$, which is in turn upper bounded by twice the relative error $\varepsilon_r$,
    \begin{equation}
        \varepsilon_a \leq \varepsilon_m \leq 2\varepsilon_r.
    \end{equation}
    Conversely, the relative error is bounded by the additive error times an exponentially large pre-factor,
    \begin{equation}
        \eps_r \leq 2^{nk} {2^n+k-1 \choose k} \varepsilon_a \leq \left( \frac{4^{nk}}{k!} \right) \left( 1+\frac{k^2}{2^n} \right) \eps_a,
    \end{equation}
    where the second inequality holds for $k^2 \leq 2^n$.
\end{lemma}

\subsubsection{Strong approximate unitary designs}

We can now introduce our strong notions of approximation error for unitary designs. To do so, let us first highlight the weakness of  standard definitions towards capturing experiments that query the inverse or conjugate of a random unitary. We then introduce our definitions to resolve this.

\paragraph{Weakness of standard definitions of approximate unitary designs.} We can illustrate the weakness of standard definitions of approximate unitary designs with two examples. The examples involve experiments that perform one query to $U$ and one query to either the inverse or the conjugate. We show that both experiments can easily distinguish any low-depth unitary from Haar-random. Since low-depth circuits can form relative error unitary designs, this implies that such designs are not sufficient for bounding properties of such experiments. At a formal level, the translation from the relative error to the trace-norm error in the output of an experiment querying the inverse or conjugate incurs an exponential factor of $2^n$, which ruins the error bound when the number of qubits $n$ is large.

Our first and simplest example involves the inverse unitary~\cite{schuster2023learning,cotler2023emergent}.
Consider a quantum experiment that prepares the zero state on all qubits, scrambles the system via $U$, then applies a Pauli operator $X_1$ on the first qubit, then attempts to un-scramble the system via $U^\dagger$.
If $U$ is Haar-random, the Pauli perturbation completely disrupts the time-reversal, resulting in an effectively Haar-random state after $U^\dagger$ is applied.
On the other hand, if $U$ is low-depth, all qubits outside of the light-cone of the first qubit return to the zero state under $U^\dagger$.
Hence, one can measure e.g.~the fidelity for the last qubit to return to the zero state to easily distinguish any low-depth unitary from Haar-random.

\begin{figure}[t]
    \centering
    \includegraphics[width=0.75\textwidth]{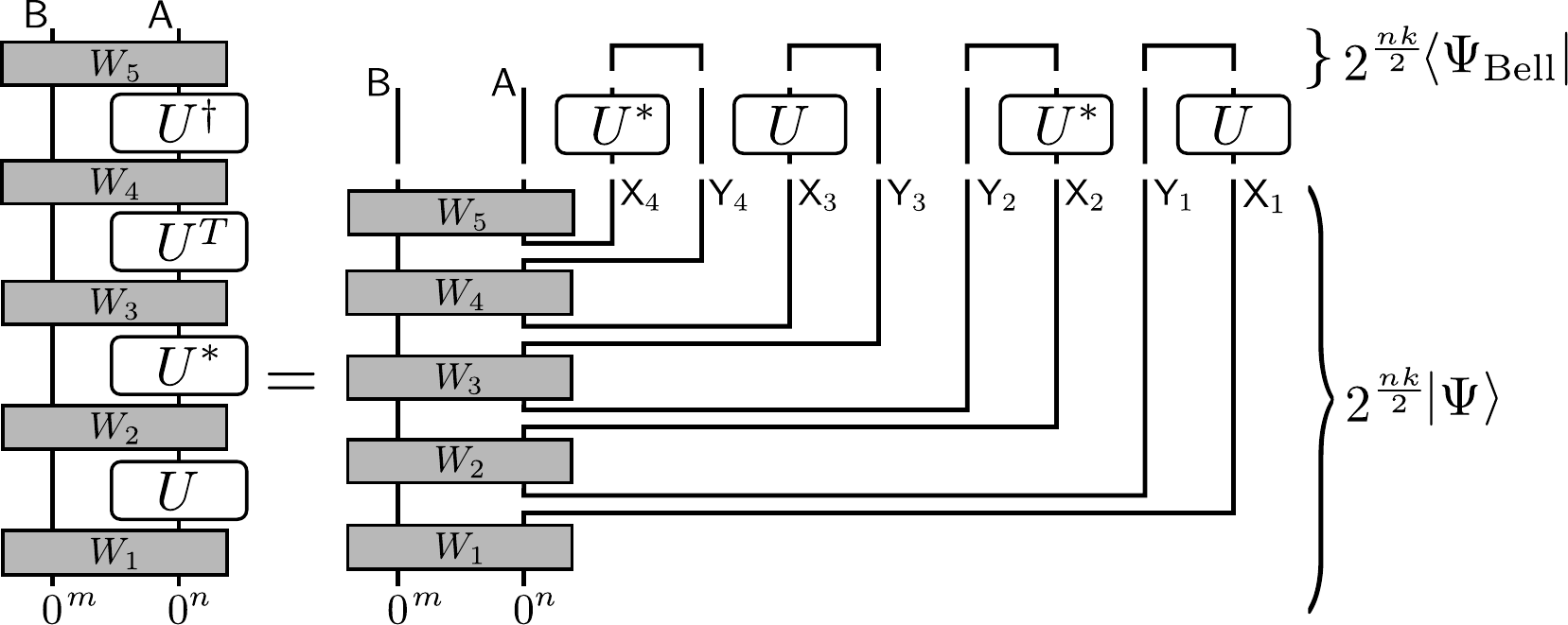}
    \caption{Reformulation of any quantum experiment that makes any $k$ queries to $U$, $U^*$, $U^T$, or $U^\dagger$ as an alternative experiment that makes a single query to $U^{\otimes p} \otimes U^{*,\otimes q}$ and performs $k$ post-selections on $n$-qubit Bell states. Here, $p$ counts the number of applications of $U$ and $U^T$ and $q$ counts the number of applications of $U^*$ and $U^\dagger$. The subsystem labels match the notation used in our proofs based on the path-recording framework (Appendix~\ref{sec: proof strong gluing} and~\ref{app: LRFC}).}
    \label{fig:reformulation}
\end{figure}

Our second example is similar but replaces the inverse unitary with the conjugate.
Consider an experiment that prepars the EPR state between two copies of $n$ qubits, applies a Pauli operator $X_1$ on the first qubit of the left side, then applies $U$ and $U^*$ in parallel to the left and right side.
This yields the state $(U \otimes U^*)(X_1 \otimes \mathbbm{1})\ket{\Psi_{\text{EPR}}} = (U X_1 U^\dagger \otimes \mathbbm{1})\ket{\Psi_{\text{EPR}}}$, where we use that $(\mathbbm{1} \otimes O)\ket{\Psi_{\text{EPR}}} = (O^T \otimes \mathbbm{1})\ket{\Psi_{\text{EPR}}}$ for any operator $O$.
When $U$ is Haar-random, the operator $UX_1 U^\dagger$ is a seemingly random operator on all $n$ qubits.
This implies that the fidelity for any pair of qubits between the left and right side to be in the EPR state is close to its maximally mixed value, $1/4$
On the other hand, when $U$ is low depth, the operator $U X_1 U^\dagger$ has support only on a small light-cone around the first qubit.
Thus, the fidelity of e.g.~the last pair of qubits to remain in the EPR state is equal to $1$.
Hence, one can again easily distinguish any low-depth unitary from Haar-random.

\paragraph{Strong measurable error.} We can now introduce our definitions of strong unitary designs. We introduce each definition in order of its prominence in our work. 

We begin with the strong analog of the measurable error. We say that a unitary ensemble is a strong unitary design with measurable error if it is indistinguishable from a Haar-random unitary in any quantum experiment with any combination $k$ queries to $U$ or $U^\dagger$ or $U^T$ or $U^*$.
\begin{definition}[Strong unitary $k$-design with measurable error]
    Let $\varepsilon > 0$. 
    An ensemble of unitaries $\mathcal{E}$ is a strong approximate unitary $k$-design with measurable error $\varepsilon$ if for any quantum experiment with any combination of $k$ queries to $U$ or $U^\dagger$ or $U^T$ or $U^*$, the output states when $U$ is sampled from $\mathcal{E}$ versus the Haar ensemble are $\varepsilon$-close in trace distance,
    \begin{equation}
        \sup_{W_1 \cdots W_{k+1}, U_1 \cdots U_k} \norm{\rho_\mathcal{E} - \rho_H}_1 \leq \varepsilon,
    \end{equation}
    where we have used the notation 
    \begin{equation} \nonumber
        \rho_\mathcal{E} = \E_{U \sim \mathcal{E}} \left[ W_{k+1} [U_k \otimes \mathbbm{1}_m] \cdots W_2 [U_1 \otimes \mathbbm{1}_m] W_1 |0^{n+m} \rangle\! \langle 0^{n+m} | W_1^\dagger [U_1^\dagger \otimes \mathbbm{1}_m] W_2^\dagger \cdots [U_k^\dagger \otimes \mathbbm{1}_m] W_{k+1}^\dagger 
        \right]
    \end{equation}
    to denote the expected output state of a general quantum experiment, where $W_i$ are arbitrary unitaries for each $i = 1,\ldots, k+1$ and $U_i \in \{U,U^\dagger,U^T,U^*\}$ for each $i = 1,\ldots, k$.
\end{definition}
\noindent Due to its natural definition, we will adopt the strong measurable error as our default error metric for strong unitary designs. Hence, in later sections, we will often refer to an \emph{approximate unitary $k$-design with measurable error $\varepsilon$} as simply an \emph{$\varepsilon$-approximate unitary $k$-design}.

\paragraph{Strong relative error.} We can also define analogs of the relative error and additive error for strong unitary designs. These share similar drawbacks to the definitions of the relative and additive error for standard unitary designs (as discussed in the previous section). Nonetheless, they will be convenient for select technical purposes in our analysis.
For convenience, we state the definitions in terms of the number of queries $p$ to $U$ or $U^T$ and the number of queries $q$ to $U^*$ or $U^\dagger$, instead of the total number of queries $k = p+q$.
\begin{definition}[Unitary $(p,q)$-design with relative error]
    Let $\varepsilon > 0$. Then an ensemble of unitaries $\mathcal{E}$ is an approximate unitary $(p,q)$-design with relative error $\varepsilon$ if
    \begin{equation}
    (1 - \varepsilon) \Phi^{(p,q)}_H \preceq \Phi^{(p,q)}_{\mathcal{E}} \preceq (1+\varepsilon) \Phi^{(p,q)}_H.
    \end{equation}
\end{definition}
\noindent The relative error guarantees that the expectation value of any positive-valued operator $\chi$ in the output state of any quantum experiment that performs any combination of $p$ queries to $U$ or $U^\dagger$ and $q$ queries to $U^T$ or $U^*$, is equal to its Haar value to within multiplicative precision. 
This follows from the fact that any such experiment can be reformulated as an experiment that performs $p$ and $q$ parallel queries to $U$ and $U^*$ and post-selects on the EPR state (Fig.~\ref{fig:reformulation}).

\paragraph{Strong additive error.} Finally, we define the additive error for strong unitary designs as follows.
\begin{definition}[Unitary $(p,q)$-design with additive error]
    Let $\varepsilon > 0$.
    An ensemble of unitaries $\mathcal{E}$ is an approximate unitary $(p,q)$-design with additive error $\varepsilon$ if
    \begin{equation}
    \norm{\Phi^{(p,q)}_\mathcal{E} - \Phi^{(p,q)}_H}_\diamond \leq \varepsilon.
    \end{equation}
\end{definition}
\noindent Physically, the additive error is equivalent to security under $p$ parallel queries to $U$ and $q$ parallel queries to $U^*$.
In principle, one could extend this definition to include parallel queries to $U^\dagger$ and $U^T$ as well.
However, the only use of the additive error in our work is as a stepping stone to prove small relative and measurable errors, and the current definition will be sufficient for these purposes.

\paragraph{Translating between strong unitary design approximation errors.} 
As in the standard case, it is possible translate between different notions of approximation error for strong unitary designs.
We formalize this in the following lemma.
The first statement of the lemma follows immediately from the definitions of the approximation errors.
The second statement is a significant result of our work. 
We provide the proof of the second statement in Section~\ref{sec: approx twirl}.
\begin{lemma}[Translating between different strong approximation errors] \label{lemma:translating-strong}
    Let $k = p+q$.
    The strong additive error $\varepsilon^{(p,q)}_a$ is upper bounded by the strong measurable error $\varepsilon^{(k)}_m$, which is in turn upper bounded by twice the strong relative error $\varepsilon^{(p,q)}_r$,
    \begin{equation}
        \varepsilon^{(p,q)}_a \leq \varepsilon^{(k)}_m \leq 2\varepsilon^{(p,q)}_r.
    \end{equation}
    Conversely, the strong relative error is bounded by the strong additive error as follows,
    \begin{equation}
        \eps_r^{(p,q)} \leq \left(\frac{4^{n(p+q)}}{p!q!} \right) 2 \eps^{(p,q)}_a + \frac{2(p+q)^2}{2^n},
    \end{equation}
    for any $2k^2 = 2(p+q)^2 \leq 2^n$.
\end{lemma}
In Section~\ref{sec: strong linear random}, we apply the second statement of the lemma to translate existing results on the spectral gap of one-dimensional random circuits into the statement that such circuits form strong unitary designs with small relative error. 

\subsection{The mixed Haar twirl and strong unitary designs} \label{sec: approx twirl}

In this section, we build a basic framework for understanding the structure of the mixed Haar twirl.
We first use this framework to derive a simpler approximate formula for the mixed Haar twirl.
We then use this approximate formula to prove the second statement of Lemma~\ref{lemma:translating-strong}, which allows one to perform tight translations between the additive and relative errors of strong unitary designs.

\subsubsection{The approximate mixed Haar twirl} 

Let us begin by re-printing the formula for the mixed Haar twirl from Lemma~\ref{lemma: explicit mixed twirl},
\begin{equation} 
    \Phi^{(p,q)}_{H}(X) = \sum_{\sigma, \tau \in S^\Gamma_k} \widetilde{\Wg}_{\sigma, \tau} \cdot \text{tr}(X \sigma^{\dagger}) \cdot \tau,
\end{equation}
where $\sigma \equiv \pi^\Gamma$ and $\tau \equiv \tilde{\pi}^\Gamma$ are summed over the \emph{partially transposed permutations}, $S^\Gamma_k$.
Here, we define $\widetilde{\Wg}_{\sigma, \tau} \equiv \Wg_{\pi,\tilde{\pi}}$.
While this formula is simple to write down, it  provides fairly little intuition about the properties and behavior of the mixed Haar twirl.
We begin this section by reformulating this expression in a more intuitive manner.
We will then derive our \emph{approximate} expression for the mixed Haar twirl from this reformulation.
For brevity, we defer the proofs of several facts stated in our reformulation of the mixed Haar twirl to the later Appendix~\ref{sec: details mixed}.
The proofs are straightforward but require many detailed steps.

Unlike the permutation operators, which appear in the standard Haar twirl, the partially transposed permutations are not necessarily unitary.
Indeed, a key role in the mixed Haar twirl is played by a subset of partially transposed permutations that are \emph{projectors}.
Consider any permutation $\pi$ that is equal to a tensor product of (i) identity elements, and (ii) swap operations between a copy on the left side and a copy on the right side.
Here, the ``left side'' corresponds to the $p$ copies that $U$ acts on, and the ``right side'' to the $q$ copies that $U^*$ acts on.
When $\pi$ is partially transposed, it results in a permutation $\sigma = \pi^\Gamma$ that is a tensor product of (i) identity elements, and (ii) EPR projectors between a copy on the left side and a copy on the right side.
This follows because the partial transpose of a swap operator is proportional to an EPR projector, $\mathcal{S}_{ij}^\Gamma = 2^n P_{\text{EPR},ij}$ for copies $i,j$.
If we label the set of pairs between the left and right side as $\alpha_\sigma = \{(i_1,j_1),(i_2,j_2),\ldots,(i_{|\alpha_\sigma|},j_{|\alpha_\sigma|})\}$, then $P_{\alpha_\sigma} \equiv \sigma / 2^{n|\alpha_\sigma|} $ projects onto EPR states on every pair in $\alpha_\sigma$ and acts trivially on the remaining copies.

In Appendix~\ref{sec: details mixed}, we show that these projectors break up the Hilbert space $\mathcal{H}^{\otimes p} \otimes \mathcal{H}^{\otimes q}$ into a tensor sum of distinct components,
\begin{equation} \label{eq: H decomposition}
    \mathcal{H}^{\otimes p} \otimes \mathcal{H}^{\otimes q} \cong \bigoplus_{\ell = 0}^{\min(p,q)} \left( \left[ \mathcal{H}^{\otimes (p-\ell)} \otimes \mathcal{H}^{\otimes (q-\ell)} \right]_{\text{nE}} \otimes \mathcal{A}_\ell \right).
\end{equation}
Here, $\left[ \mathcal{H}^{\otimes (p-\ell)} \otimes \mathcal{H}^{\otimes (q-\ell)} \right]_{\text{nE}}$ denotes the subspace of $\mathcal{H}^{\otimes (p-\ell)} \otimes \mathcal{H}^{\otimes (q-\ell)}$ that is orthogonal to every EPR projector between a copy $i$ on the left side and $j$ on the right side.
Meanwhile, $\mathcal{A}_\ell$ is an auxiliary Hilbert space of dimension $|\mathcal{A}_\ell| = {p\choose \ell}{q \choose \ell}\ell!$. This equals the number of distinct sets of $\ell$ pairs between the left and right side.
Intuitively, the $\ell$-th subspace in the tensor sum contains all states in $\mathcal{H}^{\otimes p} \otimes \mathcal{H}^{\otimes q}$ that have exactly $\ell$ EPR pairs between the left and right sides.
There are ${p\choose \ell}{q \choose \ell}\ell!$ ways to place the $\ell$ EPR pairs, which are indexed by the subsystem $\mathcal{A}_\ell$.
Once placed, the remaining $p-\ell$ and $q-\ell$ copies are free to be in any quantum state that has zero EPR pairs between the left and right.
We provide a detailed derivation of Eq.~(\ref{eq: H decomposition}) in Appendix~\ref{sec: details mixed} using only basic properties of the partially transposed permutations.

A crucial property of the EPR state is that it is invariant under the action of $U \otimes U^*$ for any $U$.
This follows because $(U \otimes \mathbbm{1})P_{\text{EPR}} = (\mathbbm{1} \otimes U^T)P_{\text{EPR}}$ and $U^* U^T = \mathbbm{1}$.
Hence, when a mixed unitary $(U)^{\otimes p} \otimes (U^*)^{\otimes q}$ is applied to a subspace with exactly $\ell$ EPR pairs, its action on each of the $\ell$ EPR pairs becomes trivial.
This leaves $p-\ell$ copies of $U$ and $q-\ell$ copies of $U^*$ remaining. 
To formalize this, for each value of $\ell$ we define a partial isometry $\tilde{\mathcal{I}}_\ell$ (see Appendix~\ref{sec: details mixed} for details) that maps each $\ell$-EPR subspace $\mathcal{H}^{\otimes p} \otimes \mathcal{H}^{\otimes q}$ to the  $\ell$-th Hilbert space on the right side of Eq.~(\ref{eq: H decomposition}).
The partial isometries provide an orthogonal decomposition of the full Hilbert space, $\mathbbm{1} = \sum_\ell \tilde{I}_\ell^\dagger \tilde{I}_\ell$, where each $\tilde{I}_\ell^\dagger \tilde{I}_\ell$ projects onto the subspace of exactly $\ell$ EPR pairs.
We then show that one can re-write the mixed Haar twirl as,
\begin{equation} \nonumber
\begin{split}
    \Phi^{(p,q)}_H(X) &
     = \E_{U \sim H}  \sum_{\ell}  \tilde{I}_\ell^\dagger 
    \left((U)^{\otimes (p-\ell)} \otimes (U^*)^{\otimes (q-\ell)} \otimes \mathbbm{1} \right)
    \tilde{I}^{}_\ell \, X \, \tilde{I}_{\ell}^\dagger \left( (U^\dagger)^{\otimes (p-\ell)} \otimes (U^T)^{\otimes (q-\ell)}  \otimes \mathbbm{1} \right) \tilde{I}^{}_{\ell},
\end{split}
\end{equation}
where the action of the mixed unitary twirl is ``pulled'' inside each partial isometry to act on the Hilbert space $\left[ \mathcal{H}^{\otimes (p-\ell)} \otimes \mathcal{H}^{\otimes (q-\ell)} \right]_{\text{nE}}$.
One can then compute each twirl explicitly, which yields
\begin{equation} \label{eq: exact mixed twirl subspace}
\begin{split}
    \Phi^{(p,q)}_H(X) & = \sum_{\ell} \tilde{I}_{\ell}^\dagger \left[ \sum_{\pi_L  \pi_R} \sum_{\tilde{\pi}_L  \tilde{\pi}_R}
    \tr( \tilde{I}^{}_\ell \, X \, \tilde{I}_{\ell}^\dagger (\pi_L \otimes \pi_R)^{-1} ) \cdot \Wg^{(p+q-2\ell)}_{\pi_L \otimes \pi_R, \tilde{\pi}_L \otimes \tilde{\pi}_L} \cdot  \, (\tilde{\pi}_L \otimes \tilde{\pi}_R) \right] \tilde{I}_{\ell}.
\end{split}
\end{equation}
Here, we apply the formula in Lemma~\ref{lemma: explicit mixed twirl} for the mixed Haar twirl for each $\ell$.
The only partially transposed permutations that contribute correspond to tensor products $\sigma \equiv \pi_L \otimes \pi_R$ and $\tau \equiv \tilde{\pi}_L \otimes \tilde{\pi}_R$.
Every partially transposed permutation besides these contains at least one EPR projector and hence vanishes inside $\tilde{I}^\dagger_\ell(\cdot)\tilde{I}_\ell$.
The trace is partial over $\left[ \mathcal{H}^{\otimes (p-\ell)} \otimes \mathcal{H}^{\otimes (q-\ell)}\right]_{\text{nE}}$ and does not act on $\mathcal{A}_\ell$. 
We refer to Appendix~\ref{sec: details mixed} for full details.

This completes our exact reformulation of the mixed Haar twirl.
To provide a more intuitive picture of its behavior, we will now derive our simpler approximate expression.
Before doing so, let us first recall the approximate formula for the standard Haar twirl from Ref.~\cite{schuster2024random}:

\vspace{2mm}
\noindent \textbf{Lemma~1 of Ref.~\cite{schuster2024random}} (Approximation for Haar twirl)\textbf{.}
\emph{
For any $k^2 \leq D$, the Haar twirl is approximated by,
\begin{equation} \label{eq: Phi a def}
\Phi_a^{(k)}( X ) \equiv \frac{1}{2^{nk}} \sum_\pi \tr( X \pi^{-1} ) \cdot \pi,
\end{equation}
up to relative error, $(1-\varepsilon) \Phi^{(k)}_a \preceq \Phi^{(k)}_H \preceq (1+\varepsilon) \Phi^{(k)}_a$, for $\varepsilon = k^2/2D/(1-k^2/2D)$.
}
\vspace{2mm}

\noindent The approximate formula is accurate up to exponentially small error and enables much easier applications due to its simplicity. 
In effect, it replaces the Weingarten matrix $\Wg_{\pi,\tilde{\pi}}$ in the exact Haar twirl with the identity matrix $(1/2^{nk})\delta_{\pi,\tilde{\pi}}$.

We can now provide a similar simplification for the mixed Haar twirl.
\begin{lemma}[Approximation for the mixed Haar twirl] \label{lemma: approx mixed Haar twirl}
    For any $k^2 \leq D$, the mixed Haar twirl is approximated by,
    \begin{equation} 
    \begin{split}
        \Phi^{(p,q)}_H(X) & = \sum_{\ell} \tilde{I}_{\ell}^\dagger \left[ \frac{1}{2^{n(p+q-2\ell)}} \sum_{\pi_L  \pi_R}
        \tr( \tilde{I}^{}_\ell \, X \, \tilde{I}_{\ell}^\dagger (\pi_L \otimes \pi_R)^{-1} ) \cdot  (\pi_L \otimes \pi_R) \right] \tilde{I}_{\ell}.
    \end{split}
    \end{equation}
    up to relative error, $(1-\varepsilon) \Phi_a^{(p,q)} \preceq \Phi^{(p,q)}_H \preceq (1+\varepsilon) \Phi_a^{(p,q)}$, for $\varepsilon = (p+q)^2/D$.
    The approximate mixed Haar twirl can be equivalently written as,
    \begin{equation} \label{eq: Phi a p q def}
    \Phi_a^{(p,q)}
    = \sum_\ell \left( \mathcal{I}^\dagger_\ell \circ \big[ \Phi_a^{(p-\ell)} \otimes \Phi_a^{(q-\ell)} \otimes \mathbbm{1}_{\mathcal{A}_\ell} \big] \circ \mathcal{I}_\ell \right)  
    \end{equation}
    where $\Phi_a^{(p-\ell)}$ and $\Phi_a^{(q-\ell)}$ are the approximate standard Haar twirls [Eq.~(\ref{eq: Phi a def})] and $\mathcal{I}_\ell(\cdot) \equiv \tilde{I}_\ell (\cdot) \tilde{I}^\dagger_\ell$ denotes conjugation by the partial isometry.
\end{lemma}
\noindent Similar to the standard Haar twirl, our approximation in effect replaces the Weingarten matrix in Eq.~(\ref{eq: exact mixed twirl subspace}) by the product of identity matrices, $(1/D^{p-\ell}) \delta_{\pi_L, \tilde{\pi}_L}$ and $(1/D^{q-\ell}) \delta_{\pi_R, \tilde{\pi}_R}$.
We provide a proof of the lemma in Appendix~\ref{sec: proof approx mixed},

In addition to its technical convenience, the approximate mixed Haar twirl provides an intuitive picture for the behavior of random unitaries in experiments involving time-reversal.
Consider a quantum experiment querying $U$ and $U^\dagger$ (or $U^*$ or $U^T$) many times.
First, within the experiment, some subset of the applications of $U$ might cancel with applications of $U^\dagger$ (or $U^*$, etc.).
The number of cancellations corresponds to the $\ell$ in our earlier discussion, and the pairing of the cancellations corresponds to the register $\mathcal{A}_\ell$.
In principle, a quantum experiment could be a superposition of many different pairings, hence $\mathcal{A}_\ell$ is a quantum register.
Then, among the un-cancelled applications, $U$ and $U^\dagger$ behave indistinguishably from two \emph{independent} random unitaries $U$ and $V$.
This is the statement of Lemma~\ref{lemma: approx mixed Haar twirl}: After pulling $U$ and $U^*$ inside the partial isometry, their twirl is equal to the tensor product of two standard Haar twirls $\Phi^{(p-\ell)}_a$ and $\Phi^{(q-\ell)}_a$ up to small relative error.
So the behavior of time-reversal experiments is in some sense very simple: $U$ and $U^\dagger$ either cancel or they act completely independently.
The function of the original Weingarten matrix elements and the ensuing partial isometries  is solely to keep track of all the different ways that $U$ and $U^\dagger$ might cancel.

\subsubsection{Bounding the relative error of strong unitary designs}

We now provide a handful of useful techniques for bounding the relative error of strong unitary designs.
Our strategy is as follows. 
We first establish a technical lemma (Lemma~\ref{lemma: relative error EPR strong}) that allows one to bound the relative error between any mixed unitary ensemble and the approximate mixed Haar twirl.
We then use this lemma to bound the relative error between the exact mixed Haar twirl and the approximate mixed Haar twirl, which proves Lemma~\ref{lemma: approx mixed Haar twirl}.
We then combine Lemma~\ref{lemma: approx mixed Haar twirl} with Lemma~\ref{lemma: relative error EPR strong} to prove Lemma~\ref{lemma:translating-strong}, which allows one to translate additive to relative errors for strong unitary designs.


Our first technical lemma is as follows. It is inspired by Lemma~7 of Ref.~\cite{schuster2024random} for the standard approximate Haar twirl.
\begin{lemma}[Relative error to the approximate mixed Haar twirl] \label{lemma: relative error EPR strong}
    Consider a unitary ensemble $\mathcal{E}$ and its mixed twirl $\Phi^{(p,q)}_{\mathcal{E}}$.
    The mixed twirl is approximated by $\Phi^{(p,q)}_a$ up to relative error,
    \begin{equation} \label{eq: relative error EPR normal}
        \varepsilon = \frac{4^{n(p+q)}}{p! q!} \big\lVert \big[ \delta \Phi \otimes \mathbbm{1} \big]( [\noEPR \otimes \mathbbm{1}] P_{\text{EPR}} [\noEPR \otimes \mathbbm{1}] )\big\rVert_\infty,
    \end{equation}
    where $\delta \Phi \equiv \Phi^{(p,q)}_{\mathcal{E}} - \Phi^{(p,q)}_a$ and  $P_{\text{EPR}}$ is the projector onto the EPR state on $(\mathcal{H}^{\otimes p} \otimes \mathcal{H}^{\otimes q})^{\otimes 2}$.
    
    Further, let $\noEPR \equiv \tilde{I}_{\ell=0}^\dagger \tilde{I}_{\ell=0}$ project to the subspace of $\mathcal{H}^{\otimes p} \otimes \mathcal{H}^{\otimes q}$ that is orthogonal to all EPR projectors between a copy on the left side and a copy on the right side.
    On the no-EPR subspace $\noEPR$, the mixed twirl is approximated by $\Phi^{(p)}_a \otimes \Phi^{(q)}_a$ up to relative error,
    \begin{equation} \label{eq: relative error EPR noEPR}
        \varepsilon = \frac{4^{n(p+q)}}{p! q!} \big\lVert \big[ \delta \tilde{\Phi} \otimes \mathbbm{1} \big](  \tilde{P}_{\text{EPR}}  )\big\rVert_\infty,
    \end{equation}
    where $\delta \tilde{\Phi}(\cdot) \equiv \Phi^{(p,q)}_{\mathcal{E}}(\noEPR(\cdot)\noEPR) - \Phi^{(p,q)}_a(\noEPR(\cdot)\noEPR)$ is the difference between channels restricted to the no-EPR subspace, and $\tilde{P}_{\text{EPR}} \equiv [\noEPR \otimes \mathbbm{1}] P_{\text{EPR}} [\noEPR \otimes \mathbbm{1}]$ is the EPR state on the no-EPR subspace.
    %
\end{lemma}
\noindent We prove Lemma~\ref{lemma: relative error EPR strong} below, and apply it to prove Lemma~\ref{lemma: explicit mixed twirl} and Lemma~\ref{lemma:translating-strong} in the following sections.

\begin{proof} 
We will prove a slightly more general version of the lemma that encapsulates both Eq.~(\ref{eq: relative error EPR normal}) and Eq.~(\ref{eq: relative error EPR noEPR}).
Consider any projectors $\Pi_1$ and $\Pi_2$ such that $\Pi_1 \otimes \Pi_2$ commutes with $[\Phi^{(p,q)}_a \otimes \mathbbm{1}](P_\text{EPR})$.
We will set $\Pi_1 = \Pi_2 = \mathbbm{1}$ for Eq.~(\ref{eq: relative error EPR normal}), and $\Pi_1 =  \noEPR$ and $\Pi_2 = \mathbbm{1}$ for Eq.~(\ref{eq: relative error EPR noEPR}). 

We follow a similar general strategy to the proof of Lemma~2 in Ref.~\cite{schuster2024random}. 
Let
\begin{equation}
    \rho \equiv [( \Pi_1 \circ \Phi^{(p,q)}_\mathcal{E} \circ \Pi_2) \otimes \mathbbm{1}](P_{\text{EPR}}) = [\Pi_1 \otimes \Pi_2] \cdot [\Phi^{(p,q)}_\mathcal{E} \otimes \mathbbm{1}](P_{\text{EPR}}) \cdot [\Pi_1 \otimes \Pi_2],
\end{equation}
where we define the channel $\Pi_\alpha(X) \equiv \Pi_\alpha X \Pi_\alpha$ whose action will be clear from context.
Let $\rho_a$ denote the same expression with $\Phi$ replaced by $\Phi_a^{(p,q)}$.
Note that $\Phi_a^{(p,q)}(\noEPR(\cdot)\noEPR) = [\Phi_a^{(p)}\otimes \Phi_a^{(q)}](\noEPR(\cdot)\noEPR)$ for the proof of Eq.~(\ref{eq: relative error EPR noEPR}). We have
\begin{equation} \label{eq: rho a rel error proof}
\begin{split}
    \rho_a & \equiv [( \Pi_1 \circ \Phi \circ \Pi_2) \otimes \mathbbm{1}](P_{\text{EPR}}) \\
    & = [\Pi_1 \otimes \Pi_2] \cdot \sum_\ell \frac{D_\ell}{D} \left( [\tilde{I}_\ell^\dagger \otimes \tilde{I}_\ell^\dagger] \cdot [(\Phi_a^{(p-\ell)} \otimes \Phi_a^{(q-\ell)}) \otimes \mathbbm{1}](P^\ell_{\text{EPR}}) \cdot [\tilde{I}_\ell \otimes \tilde{I}_\ell] \right) \cdot [\Pi_1 \otimes \Pi_2]\\
    & = [\Pi_1 \otimes \Pi_2] \cdot \sum_\ell \frac{D_\ell}{D^{2p+2q-2\ell}} \left( [\tilde{I}_\ell^\dagger \otimes \tilde{I}_\ell^\dagger] \cdot \sum_{\pi} [\pi \otimes \pi] \cdot [\tilde{I}_\ell \otimes \tilde{I}_\ell] \right) \cdot [\Pi_1 \otimes \Pi_2] \\
    & = \sum_\ell  \frac{D_\ell (p-\ell)! (q-\ell)!}{D^{3p+3q-4\ell}}  \left( [\Pi_1 \otimes \Pi_2] \cdot  [\tilde{I}_\ell^\dagger \otimes \tilde{I}_\ell^\dagger] \cdot \frac{1}{(p-\ell)!(q-\ell)!} \sum_{\pi} [\pi \otimes \pi] \cdot [\tilde{I}_\ell \otimes \tilde{I}_\ell]  \cdot [\Pi_1 \otimes \Pi_2] \right)
\end{split}
\end{equation}
where $D_\ell \equiv \text{tr}(\tilde{I}^\dagger_\ell \tilde{I}_\ell) \leq D^{p+q-2\ell} {p \choose \ell} {q \choose \ell} \ell!$ is the rank of the $\ell$-EPR subspace, $D = 2^n$, and we abbreviate $\pi \equiv \pi_L \otimes \pi_R$, where $\pi_L \in S_{p-\ell}$ and $\pi_R \in S_{q-\ell}$.
We use that the partial isometries $\tilde{I}_\ell$ are real (Appendix~\ref{sec: details mixed}), so that $[\tilde{I}_\ell \otimes \mathbbm{1}] P_{\text{EPR}} = [\mathbbm{1} \otimes \tilde{I}_\ell^\dagger] P_{\text{EPR}}$.
We proceed in four steps.

\vspace{3mm}
\noindent (1) The term in parentheses in the final line of Eq.~(\ref{eq: rho a rel error proof}) is a projector, i.e.~it squares to itself.
This follows because the sum over permutations inside the partial isometry is a projector (onto the tensor product of the symmetric subspace~\cite{harrow2013church} of the left and right copy),
$\Pi_1 \otimes \Pi_2$ is a projector (by definition), 
and these two projectors commute with one another (by assumption).
Hence, their product is also a projector.
Thus, 
$\rho_a$ is equal to a sum of projectors with coefficients,
\begin{equation}
    \frac{D_\ell (p-\ell)! (q-\ell)!}{D^{3p+3q-4\ell}} \leq \frac{1}{\ell!} \frac{p! q!}{D^{2p+2q-2\ell}}.
\end{equation}
applying our upper bound on $D_\ell$.
The minimum non-zero eigenvalue of $\rho_a$ is given by the $\ell = 0$ coefficient, and is equal to $p! q! / D^{2(p+q)}$.

\vspace{3mm}
\noindent (2) The state $\rho$ has support entirely within the support of $\rho_a$. 
This follows because, first, the support of $\rho_a$ is equal to the support of $\sum_\ell [\Pi_1 \otimes \Pi_2] Q_\ell$ (note that $[\Pi_1 \otimes \Pi_2] Q_\ell [\Pi_1 \otimes \Pi_2] = [\Pi_1 \otimes \Pi_2] Q_\ell$ since the two projectors commute), where $Q_\ell$ is defined as the term in between the two copies of $\Pi_1 \otimes \Pi_2$ in Eq.~(\ref{eq: rho a rel error proof}). Second, the latter operator stabilizes $\rho$,
\begin{equation}
    \sum_\ell [\Pi_1 \otimes \Pi_2] Q_\ell \cdot \rho
    =
    [\Pi_1 \otimes \Pi_2] \cdot \rho
    = \rho,
\end{equation}
where the first step follows because $Q_\ell$ commutes with $\Pi_1 \otimes \Pi_2$ and with $(U)^{\otimes p} \otimes (U^*)^{\otimes q}$ for any $U \sim \mathcal{E}$, and $\sum_\ell Q_\ell P_{\text{EPR}} = P_{\text{EPR}}$.
The latter statement follows because $\pi \otimes \pi$ stabilizes the EPR state $[\tilde{I}_\ell \otimes \tilde{I}_\ell ]P_{\text{EPR}} [\tilde{I}^\dagger_\ell \otimes \tilde{I}^\dagger_\ell ]$ for any $\pi$. Hence, the sum over $\pi$ can be eliminated and the remaining sum over $\ell$ yields $\sum_\ell [\tilde{I}^\dagger_\ell \tilde{I}_\ell \otimes \tilde{I}^\dagger_\ell \tilde{I}_\ell] P_{\text{EPR}} = \sum_\ell [\tilde{I}^\dagger_\ell \tilde{I}_\ell \tilde{I}^\dagger_\ell \tilde{I}_\ell \otimes \mathbbm{1}] P_{\text{EPR}} = \sum_\ell [\tilde{I}^\dagger_\ell \tilde{I}_\ell  \otimes \mathbbm{1}] P_{\text{EPR}} = P_{\text{EPR}}$.

\vspace{3mm}
\noindent (3) Steps (1) and (2) immediately imply that the twirl has relative error $\varepsilon$ [Eq.~(\ref{eq: relative error EPR normal}) or Eq.~(\ref{eq: relative error EPR noEPR})] on the EPR state.

\vspace{3mm}
\noindent (4) The relative error on the EPR state upper bounds the relative error on any state.
This follows because we can express $\Phi(\chi) = D^{2(p+q)}\tr_2( (\mathbbm{1} \otimes \chi^T) [ \Phi \otimes \mathbbm{1} ](P_{\text{EPR}}))$ for any $\Phi$, where  the trace is over the second copy of $\mathcal{H}^{\otimes p} \otimes \mathcal{H}^{\otimes q}$.

\vspace{3mm}
\noindent This completes the proof.
\end{proof}

\subsubsection{Proof of Lemma~\ref{lemma: approx mixed Haar twirl}: The approximate mixed Haar twirl} \label{sec: proof approx mixed}

From Eq.~(\ref{eq: exact mixed twirl subspace}), we have the following expression for the mixed Haar twirl,
\begin{equation}
\begin{split}
    \Phi^{(p,q)}_H(\rho) & = \sum_{\ell} \tilde{I}_{\ell}^\dagger \left[ \sum_{\pi_L \pi_R} \sum_{\tilde{\pi}_L,  \tilde{\pi}_R}
    \tr( \tilde{I}^{}_\ell \, \rho \, \tilde{I}_{\ell}^\dagger (\pi_L \otimes \pi_R)^{-1} ) \cdot \Wg^{(p+q-2\ell)}_{\pi_L \otimes \pi_R, \tilde{\pi}_L \otimes \tilde{\pi}_L} \cdot \, (\tilde{\pi}_L \otimes \tilde{\pi}_R) \right] \tilde{I}_{\ell},
\end{split}
\end{equation}
    We can compare this with our desired approximation, 
    \begin{equation} \label{eq: approx twirl ell}
        \Phi_a^{(p,q)}(\rho) 
        = \sum_\ell \tilde{I}_\ell^\dagger \left[
        \frac{1}{2^{n(p+q-2\ell)}} \sum_{\pi_L \otimes \pi_R}
        \tr( \tilde{I}^{}_\ell \, \rho \, \tilde{I}_{\ell}^\dagger (\pi_L \otimes \pi_R)^{-1} ) \cdot \, (\pi_L \otimes \pi_R) \right] \tilde{I}_{\ell}.
    \end{equation}
    Since both the mixed Haar twirl and our desired approximation are a tensor sum over $\ell$, the two are close in relative error if and only if they are close in relative error within each partial isometry $\tilde{I}_\ell$.
    Within each partial isometry, the channels act solely on the no-EPR subspace of $\mathcal{H}^{\otimes (p-\ell)} \otimes \mathcal{H}^{\otimes (q-\ell)}$ by definition.
    Hence, to prove the proposition, it suffices to show that $\Phi_a^{(p-\ell)} \otimes \Phi_a^{(q-\ell)}$ and $\Phi_H^{(p-\ell,q-\ell)}$ are close in relative error on the no-EPR subspace.

    Observing Lemma~\ref{lemma: relative error EPR strong}, we can bound this relative error in terms of the spectral norm
    %
    \begin{equation}
    \begin{split}
        & [\delta \Phi \otimes \mathbbm{1}] ((\noEPR \otimes \mathbbm{1} )P_{\text{EPR}} (\noEPR \otimes \mathbbm{1} )) \\
        & \quad  =
        \frac{1}{4^{n(p+q-2\ell)}}
        \sum_{\pi_L \pi_R} \sum_{\tilde{\pi}_L   \tilde{\pi}_R}
        \left( 2^{n(p+q-2\ell)} \Wg^{(p+q-2\ell)}_{\pi_L \otimes \pi_R, \tilde{\pi}_L \otimes \tilde{\pi}_L} - \delta_{\pi_L \otimes \pi_R, \tilde{\pi}_L \otimes \tilde{\pi}_L}  \right)
        \cdot
        ( \pi_L \otimes \pi_R ) \otimes  (\tilde{\pi}_L \otimes \tilde{\pi}_R).
    \end{split}
    \end{equation}
    Applying the triangle inequality, we find
    \begin{equation} \label{eq: spectral norm bound phi tilde}
    \begin{split}
        4^{n(p+q-2\ell)} & \lVert [\delta \Phi^{(p-\ell,q-\ell)} \otimes \mathbbm{1}]((\noEPR \otimes \mathbbm{1} )P_{\text{EPR}} (\noEPR \otimes \mathbbm{1} )) \rVert_\infty\\
        & \hspace{15mm} \leq 
        \sum_{\pi_L \pi_R} \sum_{\tilde{\pi}_L   \tilde{\pi}_R}
        \left| 2^{n(p+q-2\ell)} \Wg^{(p+q-2\ell)}_{\pi_L \otimes \pi_R, \tilde{\pi}_L \otimes \tilde{\pi}_L} - \delta_{\pi_L \otimes \pi_R, \tilde{\pi}_L \otimes \tilde{\pi}_L} \right| \\
        & \hspace{15mm} \equiv \sum_{\pi_L \pi_R} \sum_{\tilde{\pi}_L   \tilde{\pi}_R} A_{\pi_L \otimes \pi_R, \tilde{\pi}_L \otimes \tilde{\pi}_L},
    \end{split}
    \end{equation}
    where we define the $(p-\ell)!(q-\ell)! \times (p-\ell)!(q-\ell)!$ matrix $\hat A$ elementwise via the absolute values in the preceding expression.
    
    Since $\hat A$ has all positive entries, its maximal eigenvalue is achieved by a vector with all positive entries.
    Moreover, it is invariant under multiplication by any permutation, $\pi_L \otimes \pi_R \rightarrow (\pi_L \otimes \pi_R) (\tilde{\pi}_L \otimes \tilde{\pi}_R)$, $\tilde{\pi}_L \otimes \tilde{\pi}_R \rightarrow (\tilde{\pi}_L \otimes \tilde{\pi}_R) (\tilde{\pi}_L \otimes \tilde{\pi}_R)$.
    Thus, without loss of generality, we can take its maximal eigenvector to be both positive and permutation invariant.
    This implies that the maximal eigenvalue is achieved by the constant vector.
    The sum in Eq.~(\ref{eq: spectral norm bound phi tilde}) is thus equal to the maximal eigenvalue multiplied by the matrix dimension, $(p-\ell)!(q-\ell)!$.

    To determine the maximum eigenvalue of $\hat A$, we note that $\hat A$ is a sub-matrix of the $(p+q-2\ell)! \times (p+q-2\ell)!$ matrix with entries, $| \delta W|_{\sigma,\tau} \equiv | 2^{(p+q-2\ell)} \Wg_{\sigma,\tau} - \delta_{\sigma,\tau} |$, where $\sigma,\tau \in S_{p+q-2\ell}$.
    The spectral norm of a sub-matrix is upper bounded by the spectral norm of the matrix, and hence $\lVert \hat A \rVert_\infty \leq \lVert | \delta \hat{W} | \rVert_\infty$.
    From Ref.~\cite{harrow2023approximate}, the latter spectral norm is upper bounded by $(p+q-2\ell)^2/2^{n} \leq (p+q)^2/2^n$.
    Applying Lemma~\ref{lemma: relative error EPR strong} completes our proof. \qed

\subsubsection{Proof of Lemma~\ref{lemma:translating-strong}: Translating between different strong approximation errors}

As aforementioned, the first statement of the lemma follows trivially from definitions.
The combination of Lemma~\ref{lemma: approx mixed Haar twirl} and Lemma~\ref{lemma: relative error EPR strong} immediately proves the second statement.
We have $\big\lVert \big[ (\Phi^{(p,q)}_\mathcal{E} - \Phi^{(p,q)}_H) \otimes \mathbbm{1} \big]( {P}_{\text{EPR}}  )\big\rVert_\infty \leq \varepsilon_a$ from the definition of the additive error.
We also have $\frac{4^{n(p+q)}}{p! q!} \big\lVert \big[ (\Phi^{(p,q)}_H - \Phi^{(p)}_a \otimes \Phi^{(q)}_a) \otimes \mathbbm{1} \big]( {P}_{\text{EPR}}  )\big\rVert_\infty \leq (p+q)^2/2^n$ from Lemma~\ref{lemma: approx mixed Haar twirl}.
Together, these imply that $\frac{4^{n(p+q)}}{p! q!} \big\lVert \big[ (\Phi^{(p,q)}_\mathcal{E} - \Phi^{(p)}_a \otimes \Phi^{(q)}_a) \otimes \mathbbm{1} \big]( {P}_{\text{EPR}}  )\big\rVert_\infty \leq \frac{4^{n(p+q)}}{p! q!} \varepsilon_a + (p+q)^2/2^n \equiv \varepsilon'$.
From Lemma~\ref{lemma: approx mixed Haar twirl} and Lemma~\ref{lemma: relative error EPR strong}, this proves that 
\begin{equation}\nonumber
    \Phi^{(p,q)}_\mathcal{E} \preceq (1+\varepsilon')\Phi^{(p,q)}_a \preceq \frac{1+\varepsilon'}{1-\frac{(p+q)^2}{2^n}}\Phi^{(p,q)}_H 
    \preceq \left(1+\left(\frac{4^{n(p+q)}}{p! q!}\right) \varepsilon_a + \frac{(p+q)^2}{2^n}\right)\left(1+\frac{2(p+q)^2}{2^n} \right) \Phi^{(p,q)}_H,
\end{equation}
and similar in the reverse direction. Here, we assumed $2(p+q)^2 \leq 2^n$ in the third inequality. This completes the proof. \qed

\subsection{Proof of Lemma~\ref{lemma: strong 2-designs}: Strong unitary 2-designs from blocked fast scrambling random circuits}

In this section, we prove Lemma~\ref{lemma: strong 2-designs} showing that the blocked fast scrambling circuit is a strong unitary 2-design.
The circuit has depth $\mathcal{O}(\log n \cdot \log n/\varepsilon)$ when each small random unitary is instantiated with a 1D random circuit of depth $\mathcal{O}(\log n/\varepsilon)$.
This forms a key building block in our  random circuit construction of strong unitary $k$-designs, as described in the main text.




Our proof of Lemma~\ref{lemma: strong 2-designs} proceeds in two steps.
First, we introduce and prove a simple proposition that quantifies when any unitary ensemble forms a strong approximate unitary 2-design.
Our proposition leverages a well-known mapping from 2-designs to classical Markov processes on the set of Pauli operators~\cite{nahum2018operator,hunter2019unitary}.
We show that the measurable error of a strong approximate 2-design is upper bounded by the total variational distance between the Markov process associated to the circuit ensemble, and that associated to the Haar ensemble.
Second, we prove that this total variational distance is small for the blocked fast scrambling circuit, which completes the proof.

Let us first review the mapping from unitary 2-designs to Markov processes, and formally state our proposition. 
Consider the action of the unitary $U \otimes U^*$ on $\mathcal{H} \otimes \mathcal{H}$, where $U$ is drawn from any ensemble $\mathcal{E}$ that is invariant under random single-qubit rotations at the input and output.
To understand the action of the unitary, we consider the complete basis of states, 
\begin{equation}
    \ket{P} \equiv (P \otimes \mathbbm{1}) \ket{\Psi_{\text{EPR}}},
\end{equation}
 where $\ket{\Psi_{\text{EPR}}}$ is the EPR state and $P$ runs over all $4^n$ Pauli operators.
 The invariance of $\mathcal{E}$ under random single-qubit rotations guarantees that the twirl over $\mathcal{E}$ can be written in the following form,
    \begin{equation}
        \E_{U \sim \mathcal{E}} \left[ (U\otimes U^*) \rho (U^\dagger \otimes U^T) \right]
        = \bra{ \mathbbm{1} } \rho \ket{ \mathbbm{1} } \cdot \dyad{ \mathbbm{1} }
        + \sum_{P,Q \neq \mathbbm{1}} p_{\mathcal{E}}(Q;P) \bra{ P } \rho \ket{ P } \cdot \dyad{Q},
    \end{equation}
where $p_{\mathcal{E}}(Q ; P)$ is a normalized probability distribution over Pauli operators $Q$, for each Pauli operator $P$.
For a Haar-random unitary, we have $p_{H}(Q;P) = 1/(4^n-1)$.

Our proposition is as follows.
\begin{proposition} \label{lemma: 2-design TVD}
    Consider any unitary ensemble $\mathcal{E}$ is invariant under conjugation, transposition, and random single-qubit Pauli rotations at the input and output circuit layer. 
    The ensemble forms a strong approximate unitary 2-design with measurable error,
    \begin{equation}
        \varepsilon = \max_P \,\, \text{\emph{TVD}}\big( \, p_{\mathcal{E}}(Q;P) \, ,\, p_H(Q;P) \, \big),
    \end{equation}
    in any quantum experiment that queries one of $U$ or $U^T$ and one of $U^*$ or $U^\dagger$.
\end{proposition}
\noindent The proposition does not cover experiments that query two of $U$ or $U^T$ and neither of $U^*$ or $U^\dagger$ (and vice versa). We will address such experiments using standard unitary design features later on.
\begin{proof}[Proof of Proposition~\ref{lemma: 2-design TVD}]
    Without loss of generality, we consider an experiment that first queries $U^\dagger$ and then queries $U$. The remaining seven classes of experiments follow by symmetric arguments\footnote{In particular, experiments that query $U^*$ and then $U$ follow the exact same proof exchanging $\gsX$ and $\gsY$. The remaining six classes of experiments follow from these two by exchanging replacements $U \leftrightarrow U^T$ or $U \leftrightarrow U^*$ or $U \leftrightarrow U^\dagger$. These replacements are allowed because the ensemble $\mathcal{E}$ is invariant under conjugation and transposition.}.
    The output state $\rho$ of any such experiment can be written as (Fig.~\ref{fig:reformulation})
    \begin{equation} \nonumber
        \rho_U = 4^n (\mathbbm{1}_{\gsB \gsA} \otimes \bra{\Psi_{\text{Bell}}}_{\gsX \gsY} ) (\mathbbm{1}_{\gsB} \otimes U_\gsA \otimes U^*_\gsX \otimes \mathbbm{1}_{\gsY}) \dyad{\Psi}_{\gsB \gsA \gsX \gsY} (\mathbbm{1}_{\gsB} \otimes U^\dagger_\gsA \otimes U^T_\gsX \otimes \mathbbm{1}_{\gsY}) (\mathbbm{1}_{\gsB \gsA} \otimes \ket{\Psi_{\text{Bell}}}_{\gsX \gsY} ),
    \end{equation}
    where $\ket{\Psi}$ is a normalized quantum state.
    Taking the twirl over $\mathcal{E}$ yields,
    \begin{equation} \nonumber
    \begin{split}
        \E_{U \sim \mathcal{E}} \left[ \rho_U \right] 
        & = 4^n \bra{ \mathbbm{1}}_{\gsX \gsY}  \dyad{\mathbbm{1}}_{\gsA \gsX} \dyad{\Psi}_{\gsB \gsA \gsX \gsY} \dyad{\mathbbm{1}}_{\gsA \gsX} \ket{ \mathbbm{1}}_{\gsX \gsY} \\
        & \,\,\,\,\,\,\, + 4^n \sum_{P,Q \neq \mathbbm{1}} p_{\mathcal{E}}(Q;P) \cdot \bra{ \mathbbm{1}}_{\gsX \gsY}  \dyad{Q}{P}_{\gsA \gsX} \dyad{\Psi}_{\gsB \gsA \gsX \gsY} \dyad{P}{Q}_{\gsA \gsX} \ket{ \mathbbm{1} }_{\gsX \gsY} \\
        & = S_{\gsY \rightarrow \gsA} ( \mathbbm{1}_{A'} \otimes \bra{\mathbbm{1}}_{\gsA \gsX} \otimes \mathbbm{1}_{\gsY} ) \dyad{\Psi}_{\gsB \gsA \gsX \gsY} ( \mathbbm{1}_{A'} \otimes \ket{\mathbbm{1}}_{\gsA \gsX} \otimes \mathbbm{1}_{\gsY} ) S^\dagger_{\gsY \rightarrow \gsA} \\
        & \,\,\,\,\,\,\, + \sum_{P,Q \neq \mathbbm{1}} p_{\mathcal{E}}(Q;P) \cdot S_{\gsY \rightarrow \gsA} ( \mathbbm{1}_{A'} \otimes \bra{P}_{\gsA \gsX} \otimes Q_{\gsY} ) \dyad{\Psi}_{\gsB \gsA \gsX \gsY} ( \mathbbm{1}_{A'} \otimes \ket{P}_{\gsA \gsX} \otimes Q_{\gsY} ) S^\dagger_{\gsY \rightarrow \gsA} \\
    \end{split}
    \end{equation}
    where $S_{\gsY \rightarrow \gsA}$ swaps register $\gsY$ to register $\gsA$, and is obtained from the multiplying the Bell pairs, $\bra{ \mathbbm{1}}_{\gsX \gsY}  \ket{Q}_{\gsA \gsX} = (1/2^n) S_{\gsY \rightarrow \gsA} Q_{\gsY}$.

    Let us denote each state in the final line above as $\rho_{PQ}$ which acts on $\gsA \gsB$.
    With this notation, the final line  becomes
    \begin{equation}
        \E_{U \sim \mathcal{E}} \left[ \rho_U \right] = \rho_{\mathbbm{1} \mathbbm{1}} + \sum_{P,Q\neq \mathbbm{1}} p_{\mathcal{E}}(Q;P) \cdot \rho_{PQ}.
    \end{equation}
    We also let $\delta p(P,Q) = p_{\mathcal{E}}(P,Q) - p_H(P,Q)$.
    We can bound the additive error as follows,
    \begin{equation}
    \begin{split}
        \left\lVert \E_{U \sim \mathcal{E}} \left[ \rho_U \right] - \E_{U \sim H} \left[ \rho_U \right] \right\rVert_1 
        & = \bigg\lVert \sum_{P,Q \neq \mathbbm{1}} \delta p(Q;P) \cdot \rho_{PQ} \bigg\rVert_1 \\ 
        & \leq \sum_{P,Q \neq \mathbbm{1}} | \delta p(Q;P) | \cdot \lVert \rho_{PQ} \rVert_1 \\ 
        & \leq \sum_{P \neq \mathbbm{1}} \lVert \rho_{P \mathbbm{1}} \rVert^{}_1 \cdot \sum_{Q \neq \mathbbm{1}} | \delta p(Q;P) | \\ 
        & = \sum_{P \neq \mathbbm{1}} \lVert \rho_{P \mathbbm{1}} \rVert^{}_1 \cdot \text{TVD}\big( \, p_{\mathcal{E}}(Q;P) \, ,\, p_H(Q;P) \, \big) \\ 
        & \leq \max_P \, \text{TVD}\big( \, p_{\mathcal{E}}(Q;P) \, ,\, p_H(Q;P) \, \big) \cdot \sum_{P \neq \mathbbm{1}} \lVert \rho_{P \mathbbm{1}} \rVert^{}_1 \\ 
        & \leq \max_P \, \text{TVD}\big( \, p_{\mathcal{E}}(Q;P) \, ,\, p_H(Q;P) \, \big). \\ 
    \end{split}
    \end{equation}  
    The first step uses the triangle inequality, the second step uses $\lVert \rho_{PQ} \rVert_1 = \lVert \rho_{P \mathbbm{1}} \rVert_1$ since the 1-norm is invariant under Pauli rotations, the third step uses the definition of the total variational distance, the fourth step uses Holder's inequality, and the final step uses $\sum_{P\neq \mathbbm{1}} \lVert \rho_{P\mathbbm{1}} \rVert_1 \leq \sum_{P} \lVert \rho_{P\mathbbm{1}} \rVert_1 = \sum_P \tr(\rho_{P\mathbbm{1}}) = \tr(\dyad{\Psi}) = 1$.
    This completes the proof.
\end{proof}

We can now proceed to the proof of Lemma~\ref{lemma: strong 2-designs}.

\begin{proof}[Proof of Lemma~\ref{lemma: strong 2-designs}]
There are sixteen classes of experiments that a strong unitary 2-design must capture, corresponding to whether the first query is to $U$, $U^T$, $U^*$, or $U^\dagger$ and similar for the second query.
Let us first consider the eight classes of experiment that query two of $U, U^T$ or two of $U^*, U^\dagger$.
From the gluing Lemma~3 of Ref.~\cite{schuster2024random}, the blocked fast scrambling circuit with small Haar-random unitaries forms a standard approximate unitary 2-design with relative error $4m/2^\xi$.
This implies that the ensemble has measurable error at most $8m/2^\xi$ in any of the eight experiments.
Replacing each small random unitary with a relative error $\frac{\varepsilon}{n}$-approximate unitary 2-design yields an additional measurable error $2 m \varepsilon / n = 2\varepsilon/\xi$. Here, we only need to glue $2m$ unitaries together in order to realize a standard design via Lemma~3 of Ref.~\cite{schuster2024random}; hence, we only pick up an additional error linear in $m$, instead of $m\log_2 m$.
%
Thus, the blocked fast scrambling circuit has measurable error at most $8 m /2^\xi + m \varepsilon / n$ in any of these eight experiments.
Setting $\xi \geq \min( \log_2(3n/\varepsilon),4)$ yields an error less than $\varepsilon$.

Let us now turn to the remaining eight experiments, which query one of $U, U^T$ and one of $U^*, U^\dagger$.
We will leverage Proposition~\ref{lemma: 2-design TVD} to bound the measurable error.
To begin, we consider the blocked fast scrambling circuit in which each small unitary is Haar-random.
Let $P$ be any Pauli operator. 
When a Pauli operator is acted on by a small random unitary within its support, the resulting Pauli operator has support on both patches of the small random unitary with probability $1 - (2 \cdot 4^\xi - 1)/(4^{2\xi}-1) \geq 1 - 2/4^{\xi}$, where the latter term counts the fraction of non-identity Pauli operators on $2\xi$ qubits that have support on only a single patch.
Consider any single patch in the support of $P$, and all the small random unitaries within the light-cone of the patch.
There are at most $1+2+4+\cdots+m = 2m-1$ such unitaries.
Therefore, the probability that \emph{every} small random unitary in the light-cone spreads the Pauli operator to both patches in its support is at least $1 - (2 m-1) (2/4^\xi)$.
The probability that the time-evolved operator $U P U^\dagger$ has support on every patch of $\xi$ qubits is hence at least this value.
We denote this probability as $1-\delta_{\mathcal{E}}$, with $\delta_{\mathcal{E}} \leq 4 m/4^\xi$.
%
%

Let us now turn to the action of an $n$-qubit Haar-random unitary.
Under an $n$-qubit Haar-random unitary, the time-evolved operator $U P U^\dagger$ is drawn from the flat distribution on the set of $4^n-1$ non-identity Pauli operators.
Therefore, it has support on every patch of $\xi$ qubits with probability $1-\delta_H = (4^\xi-1)^m/(4^n-1) \geq (4^n - m 4^{n-\xi})/(4^n-1) \geq 1 - m/4^\xi$.

The probability distributions $p_{\mathcal{E}}(P,Q)$ and $p_H(P,Q)$ are both flat among the Pauli operators $Q$ that have support on every patch of $\xi$ qubits, since both $\mathcal{E}$ and $H$ are invariant under composition with small random unitaries on each patch.
The TVD between the two distributions picks up contributions from two sources.
The $Q$ with full support contribute a factor of $| \delta_H - \delta_{\mathcal{E}} | \leq \delta_H + \delta_{\mathcal{E}}$, proportional to the difference in the flat value of $p_{\mathcal{E}}(P,Q)$ and $p_H(P,Q)$ on such $Q$.
Meanwhile, the other $Q$ contribute at most a factor of $\delta_H + \delta_{\mathcal{E}}$, since the two distributions have at most this support on such operators.
Thus, the TVD between the two distributions is less than,
\begin{equation}
    \text{TVD}\big( \, p_{\mathcal{E}}(Q;P) \, ,\, p_H(Q;P) \, \big) \leq 2(\delta_H + \delta_{\mathcal{E}}) \leq 2 \big( 4m /4^\xi + m/4^\xi \big),
\end{equation}
for any $P$. 
Applying Proposition~\ref{lemma: 2-design TVD} upper bounds the measurable error of the blocked fast scrambling circuit with small Haar-random unitaries.
Replacing each small Haar-random unitary in the light-cone with an $\frac{\varepsilon}{n}$-approximate strong unitary 2-design yields a measurable error less than 
\begin{equation}
    8m /4^\xi + 2 m/4^\xi + \left( \frac{2m}{n} \right) \varepsilon.
\end{equation}
Setting $\xi \geq \min( \frac{1}{2} \log_2(5n/\varepsilon) , 4 )$ yields an error less than 
\begin{equation}
    (8/5) (m/n) \varepsilon + (2/5) (m/n) \varepsilon  + 2 (m/n) \varepsilon \leq \varepsilon,
\end{equation}
which completes our proof.
\end{proof}

\subsection{Proof of Lemma~\ref{lemma: strong design 1D RUC}: Strong unitary $k$-designs from 1D random circuits} \label{sec: strong linear random}

In this section, we provide  our proof that one-dimensional local random circuits form strong unitary $k$-designs with small relative error in linear depth (Theorem~\ref{lemma: strong design 1D RUC} of the main text).
Our proof follows the same approach used to show that 1D random circuits form standard unitary designs~\cite{brandao2016local}.
This is enabled by our new Lemma~\ref{lemma:translating-strong} which allows one to translate from spectral gaps~\cite{brandao2016local,haferkamp2022random,chen2024incompressibility} to relative error strong unitary designs.


Let us first review a few of the key definitions used in the spectral gap literature~\cite{brandao2016local,haferkamp2021improved,chen2024incompressibility}.
The central object studied in these works is the so-called moment operator,
\begin{equation}
    M_{\mathcal{E}} = \E_{U \sim \mathcal{E}} \big[ U^{\otimes k} \otimes (U^*)^{\otimes k} \big],
\end{equation}
of a random unitary ensemble $\mathcal{E}$.
To quantify the closeness of the moment operator to the Haar-random moment operator, we let $\delta M = M_{\mathcal{E}} - M_H$ and define the \emph{essential norm},
\begin{equation}
    g(\mathcal{E}) = \left\lVert \delta M \right\rVert_\infty.
\end{equation}
Ref.~\cite{chen2024incompressibility} proves that the ensemble of 1D local random brickwork circuits of depth $d$ has essential norm $g(\mathcal{E}) \leq \exp \big( \!-\Omega(d \log^7 \! k) \big)$.

The relationship between the essential norm and the approximation errors of standard unitary designs is well-known.
One can bound $\lVert \delta \Phi(\rho) \rVert_\infty \leq g(\mathcal{E})$ for any $\rho$ when $\delta \Phi = \Phi^{(k)}_{\mathcal{E}} - \Phi^{(k)}_{H}$ is defined with respect to the standard twirl~\cite{brandao2016local}.
This allows one to bound the relative error via Lemma~\ref{lemma:translating}.
In what follows, we show that this approach proceeds identically for strong unitary designs after replacing Lemma~\ref{lemma:translating} with our Lemma~\ref{lemma:translating-strong}.

\begin{proof}[Proof of Lemma~\ref{lemma: strong design 1D RUC}]
Let $\delta \Phi = \Phi^{(p,q)}_{\mathcal{E}} - \Phi^{(p,q)}_{H}$ and $p+q=k$.
The relationship between the essential norm and the spectral norm when $\delta \Phi$ is applied to a state follows from a straightforward series of equalities,
\begin{equation} \label{eq: essential norm to spectral norm}
    \lVert \delta \Phi ( \rho ) \rVert_\infty \leq \lVert \delta \Phi ( \rho ) \rVert_2 = \sqrt{\tr( \delta \Phi ( \rho)^2 )} = \sqrt{ \bra{\rho} \delta M^2 \ket{\rho} } \leq \lVert \delta M \rVert_\infty \sqrt{ \langle \rho | \rho \rangle } = \lVert \delta M \rVert_\infty =  g(\mathcal{E}),
\end{equation}
where $\ket{\rho}$ denotes the vectorization of $\rho$.
Hence, an upper bound on the essential norm immediately translates to an identical upper bound on the spectral norm of $\delta \Phi(\rho)$ for any state $\rho$.
The two inequalities are saturated when $\rho$ is pure and equal to the extremal eigenvector of $\delta M$.

In previous works, the equalities above are applied when $\delta \Phi$ is defined with respect to the standard twirl~\cite{brandao2016local}.
In this case, the $k$ copies of $U$ in $\delta M$ act on the left (``ket'') side of the vectorization of $\rho$, and the $k$ copies of $U^*$ act on the right (``bra'') side.
However, an identical bound holds for the mixed twirl as well.
In this case, we have $p$ copies of $U$ and $q$ copies of $U^*$ that act on the left side of the vectorization of $\rho$, and $q$ copies of $U$ and $p$ copies of $U^*$ that act on the right side.
Since there $k$ copies of both $U$ and $U^*$ in total, the remaining steps proceed identically regardless of the value of $p,q$.

qawFrom Ref.~\cite{chen2024incompressibility} and Eq.~(\ref{eq: essential norm to spectral norm}), we have $\lVert \delta \Phi(\rho) \rVert_\infty \leq \exp \big( \!-\Omega(\log(k)^7 d) \big)$ for any state $\rho$.
Taking $\rho = P_{\text{EPR}}$ as in Lemma~\ref{lemma: relative error EPR strong}, we find from Lemma~\ref{lemma: relative error EPR strong} that $\mathcal{E}$ is a strong unitary design with relative error 
\begin{equation}
    \varepsilon \leq \left(1+\frac{k^2}{2^n} \right) \cdot \frac{4^{nk}}{p!q!} \cdot \exp \big( \!-\Omega(\log(k)^7 d) \big) + \frac{k^2}{2^n},
\end{equation}
where again, $k = p+q$.
The factors of $k^2/2^n$ arise when converting from the approximate mixed Haar twirl (used in Lemma~\ref{lemma: relative error EPR strong}) to the exact Haar twirl (in the definition of $\delta \Phi$ above).
For any $\varepsilon \geq 2 k^2/2^n$, we can set $d = \Omega \big( \! \log(k)^7 ( nk + \log(1/\varepsilon) ) \big)$ to ensure that $\mathcal{E}$ is a strong unitary design with relative error $\varepsilon$, as claimed. 
\end{proof}

\subsection{Proof of Theorem~\ref{thm:strong-design-depth}}

The combination of Theorem~\ref{thm:LRFC-design}, Theorem~\ref{thm:two-layer-design}, Lemma~\ref{lemma: strong 2-designs}, and Lemma~\ref{lemma: strong design 1D RUC}  immediately yield Theorem~\ref{thm:strong-design-depth} on the circuit depth of strong unitary designs. 
We refer to Ref.~\cite{cui2025unitary} for a detailed derivation of the circuit depths and resources required to implement the LRFC ensemble with $k$-wise independent functions.

\begin{proof}[Proof of Theorem~\ref{thm:strong-design-depth}]
The first two statements of Theorem~\ref{thm:strong-design-depth} follow immediately from Theorem~\ref{thm:LRFC-design} and Theorem~\ref{thm:two-layer-design} as described in the main text. 
The third statement of Theorem~\ref{thm:strong-design-depth} follows immediately from Theorem~\ref{thm:two-layer-design} (and its extension to the blocked fast scrambling circuit via Lemma~\ref{lemma: strong 2-designs}) and Lemma~\ref{lemma: strong design 1D RUC}.
\end{proof}

\subsection{Lower bounds on the depth of strong unitary designs} \label{sec: lower bound}

In this section, we prove our two lower bounds on the circuit depths of strong unitary designs.
As discussed in the main text, our first lower bound applies to any circuit ensemble and any notion of approximation error (additive, measurable, or relative).
It shows that strong unitary designs require depth $\Omega(n)$ in 1D circuits and $\Omega(\log n)$ in all-to-all connected circuits.
Our second lower bound is specific to local random circuits and the \emph{relative error} metric.
It shows that local random circuits require depth $\Omega(n)$ to form relative error strong unitary designs in \emph{any} circuit geometry.
This contrasts with the measurable error, where we achieved local random circuit designs in depth $\mathcal{O}(\log^2 n)$ in Theorem~\ref{thm:strong-design-depth}.

\subsubsection{Lower bound for any random unitary ensemble}

Our first lower bound is Proposition~\ref{prop: lower bound design} in the main text. The proof follows from a simple light-cone argument.

\begin{proof}
We consider the state, $U^\dagger Z_0 U \ket{0^n}$, where $Z_0$ is a local Pauli operator on the first qubit.
We then measure the state in the computational basis, and count the average fraction of bits that are 1.
This corresponds to the expectation value of the operator, $M = \sum_{\textbf{s} \in \{0,1\}^n} ( | \textbf{s} | / n ) \cdot \dyad{\textbf{s}}$, where $| \textbf{s} |$ is the Hamming weight of a bitstring $\textbf{s}$. 
When $U$ is Haar-random, the expected number of bits flipped is just above $1/2$,
\begin{equation}
    \mathds{E}_{U \sim H}[ \bra{0^n} U^\dagger Z_0 U \cdot M \cdot U^\dagger Z_0 U\ket{0^n}] \geq \frac{1}{2}.
\end{equation}
Therefore, any strong $\varepsilon$-approximate unitary 2-design must have expectation value at least $1/2 - \varepsilon$.

Let us now turn to circuits with bounded depth.
The operator $U^\dagger Z_0 U$ can have support on at most $L$ qubits, where $L$ is the size of the light-cone of $U$.
For 1D circuits, we have $L \leq 2d$ for any $U$, and for all-to-all circuits, we have $L \leq 2^d$.
This implies that the fraction of bits flipped in any individual unitary $U$ is at most $L/n$, which yields
\begin{equation}
    \mathds{E}_{U \sim \mathcal{E}} \big[ \bra{0^n} U^\dagger Z_0 U \cdot M \cdot U^\dagger Z_0 U\ket{0^n} \big] \leq \frac{L}{n}.
\end{equation}
Hence, we require $d \geq n (1/2-\varepsilon)$ in 1D, and $d \geq \log_2(n(1/2-\varepsilon))$ in all-to-all circuits. 
We set $\varepsilon = \mathcal{O}(1)$ to obtain our stated bounds.
\end{proof}

\subsubsection{Lower bound for local random circuit ensembles}

Our second lower bound is specific to circuit ensembles composed of independent random gates.

\begin{proposition} \label{prop: lower bound rc design}
    {\emph{(Depth lower bound for strong unitary designs from local random circuits)}}
    Any quantum circuit ensemble composed of independent local Haar-random gates requires circuit depth
    \begin{itemize}
    \item $d = \Omega \big( \! \min \! \big\{ \! \log( 1 / \varepsilon) , n \big\} \big)$, to form a strong $\varepsilon$-approximate unitary $2$-design,
    \item $d = \Omega ( n )$, to form a strong $\varepsilon$-approximate unitary $2$-design with relative error.
    \end{itemize}
\end{proposition}
\noindent Both of the statements above hold on any circuit geometry.
To prove the first statement, we show that any local operator retains a small, exponentially-decaying memory of its initial value under any local random circuit.
To prove the latter statement, we show that this value must be exponentially small in $n$, $\mathcal{O}(4^{-n})$, in a strong unitary 2-design with relative error.

\begin{proof}
We consider an experiment which prepares the first qubit in the zero state and all other qubits in the maximally mixed state, applies $U$, and measures the probability for the first qubit to remain to the zero state.
This yields an expectation value,
\begin{equation}
\tr( Z_0 U \left( \dyad{0} \otimes (\mathbbm{1}/2)^{\otimes n-1} \right) U^\dagger )
= \frac{1}{2^n} \tr( Z_0 U Z_0 U^\dagger ).
\end{equation}
In the second equality, we expand, $\dyad{0} = (\mathbbm{1}+Z_0)/2$, and note that the identity term vanishes after the trace.
The expectation value will be zero, on average, whenever $U$ is drawn from a unitary 1-design.
To this end, our quantity of interest is the square of the expectation value.
When $U$ is Haar-random, we have
\begin{equation}
 \mathds{E}_{U \sim H} \left[ \frac{1}{4^n} \tr( Z_0 U Z_0 U^\dagger )^2 \right] = \frac{1}{4^n-1},
\end{equation}
because the time-evolved operator $U Z_0 U^\dagger$ has equal probability to be any of $4^n-1$ non-identity Pauli operators.

Let us now consider when $U$ is drawn from a circuit of independent local Haar-random gates.
For convenience, let us re-write the squared expectation value as
\begin{equation}
    \frac{1}{4^n} \tr( Z_0 U Z_0 U^\dagger )^2 = \bra{Z_0} (U \otimes U^*) \dyad{Z_0} (U^\dagger \otimes U^T) \ket{Z_0},
\end{equation}
where $\ket{Z_0} \equiv (Z_0 \otimes \mathbbm{1})\ket{E}$, and $\ket{E}$ is the EPR state.
To proceed, we decompose the $d$ layers of the circuit as, $U = U_d U_{d-1} \ldots U_1$, and denote the gate acting on the first qubit in each layer as $G_d$.
We assume the gates act on at most $r = \mathcal{O}(1)$ qubits.
After a single layer, $U_1$, of the circuit, we have
\begin{equation}
\begin{split}
    \mathds{E}_{U_1} \big[ (U_1 \otimes U_1^*) \dyad{Z_0} (U_1^\dagger \otimes U_1^T) \big] & = \mathds{E}_{G_1} \big[ (G_1 \otimes G_1^*) \dyad{Z_0} (G_1^\dagger \otimes G_1^T) \big] \\
    & = \frac{1}{4^r-1} \sum_{P \in \text{supp}(G_1), P \neq \mathbbm{1}} \dyad{P}, \\
\end{split}
\end{equation}
where the sum is over all non-identity Pauli operators in the support of $G_1$.
To proceed, we note that the state above is strictly greater than the initial state divided by $4^r-1$,
\begin{equation} \label{eq: U1 op ineq}
    \mathds{E}_{U_1} \big[ (U_1 \otimes U_1^*) \dyad{Z_0} (U_1^\dagger \otimes U_1^T) \big] = \frac{1}{4^r-1} \sum_{P \in \text{supp}(G_1), P \neq \mathbbm{1}} \dyad{P} \geq \frac{1}{4^r-1} \dyad{Z_0}.
\end{equation}
Intuitively, this is because the Pauli operator $Z_0$ has probability $1/(4^r-1)$ to return to itself under the $r$-qubit random gate $G_1$.

We can now iterate.
After the second circuit layer, we have
\begin{equation}
\begin{split}
    \mathds{E}_{U_2, U_1} \big[ (U_2 U_1 \otimes U_2^* U_1^*) \dyad{Z_0} (U_1^\dagger U_2^\dagger \otimes U_1^T U_2^T) \big]  & \geq 
    \mathds{E}_{U_2} \big[ (U_2 \otimes U_2^*) \dyad{Z_0} (U_2^\dagger \otimes U_2^T) \big] \\
    & \geq \left( \frac{1}{4^r-1} \right)^2 \dyad{Z_0}, \\
\end{split}
\end{equation}
where in the first inequality we apply Eq.~(\ref{eq: U1 op ineq}) for the unitary $U_1$, and in the second inequality we apply Eq.~(\ref{eq: U1 op ineq}) for the unitary $U_2$.
Proceeding through all $d$ circuit layers of $U$, we find
\begin{equation}
    \mathds{E}_{U \sim \mathcal{E}} \big[ (U \otimes U^*) \dyad{Z_0} (U^\dagger \otimes U^T) \big] \geq \left( \frac{1}{4^r-1} \right)^d  \dyad{Z_0}.
\end{equation}
This yields a lower bound on our quantity of interest,
\begin{equation}
    \mathds{E}_{U \sim \mathcal{E}} \left[ \frac{1}{4^n} \tr( Z_0 U Z_0 U^\dagger )^2 \right]
    = \mathds{E}_{U \sim \mathcal{E}} \big[ \bra{Z_0} (U \otimes U^*) \dyad{Z_0} (U^\dagger \otimes U^T) \ket{Z_0} \big] 
    \geq  \left( \frac{1}{4^r-1} \right)^d.
\end{equation}
Hence, for the ensemble $\mathcal{E}$ to form a strong $\varepsilon$-approximate unitary 2-design, we must have $1/(4^r-1)^d \leq 1/(4^n-1) + \varepsilon$.
This requires either $d = \Omega(\log(1/\varepsilon))$ or $d = \Omega(n)$.
For the ensemble $\mathcal{E}$ to form a strong $\varepsilon$-approximate unitary 2-design with relative error, we must have $1/(4^r-1)^d \leq (1+\varepsilon)/(4^n-1)$.
This requires $d = \Omega(n)$.
This completes the proof.
\end{proof}

\section{Strong pseudorandom unitaries} \label{app: PRU}

In this section, we introduce strong pseudorandom unitaries (PRUs) and extend the proof of~\cite{ma2025construct} to allow queries to the conjugate and transpose. 
This yields strong PRUs on $n$ qubits in $\poly n$ circuit depth.
We show how to reduce the circuit depth to $\mathcal{O}(\log n)$ using our new constructions in later sections.

\subsection{Definitions}

To define pseudorandom unitaries (PRU), we first define oracle adversaries that can query an $n$-qubit unitary oracle $\mathcal{O}$ in multiple ways. The oracle adversary is the quantum algorithm that aims to attack the pseudorandom unitary construction by distinguishing it from Haar-random unitaries. We consider the strongest possible adversary with access to all four fundamental operations: $U$, $U^\dagger$, $U^T$, and $U^*$.

\begin{definition}[Oracle adversaries]
\label{def:oracle-adversary}
A $t$-query oracle adversary $\mathcal{A}$ is parameterized by a sequence of $(n+m)$-qubit unitaries $(W_1,\dots,W_{t+1})$ acting on registers $(\sA,\sB)$, where $\mathsf{A}$ is the $n$-qubit query register and $\gsB$ is an $m$-qubit ancilla, and a sequence of oracle queries $U_1, \ldots, U_t$ where each $U_i \in \{U, U^\dagger, U^T, U^*\}$. The state after $t$ queries is
\begin{align}
    \ket{\mathcal{A}_t^{U}}_{\gsA\gsB} \coloneqq W_{t+1} [U_t \otimes \mathbbm{1}_m] \cdots W_2 [U_1 \otimes \mathbbm{1}_m] W_1 \ket{0^{n+m}}_{\gsA\gsB}.
\end{align}
\end{definition}

\begin{definition}[Pseudorandom unitaries]
\label{def:pru}
We say $\{\mathcal{U}_n\}_{n \in \mathbb{N}}$ is a secure PRU if, for all $n \in \mathbb{N}$, $\mathcal{U}_n = \{U_k\}_{k \in \mathcal{K}_n}$ is a set of $n$-qubit unitaries where $\mathcal{K}_n$ denotes the keyspace, satisfying:
\begin{itemize}
    \item \textbf{Efficient computation:} There exists a $\mathrm{poly}(n)$-time quantum algorithm that implements the $n$-qubit unitary $U_k$ for all $k \in \mathcal{K}_n$.
    \item \textbf{Indistinguishability from Haar:} For any oracle adversary $\mathcal{A}$ that runs in time $\mathrm{poly}(n)$ and measures a two-outcome observable $D_{\mathcal{A}}$ with eigenvalues $\{0,1\}$ after the queries, we have
    \begin{equation}
        \left| \mathbb{E}_{\mathcal{O} \leftarrow \mathcal{U}_n} \Tr\left( D_{\mathcal{A}} \cdot \ketbra{\mathcal{A}^{\mathcal{O}}}_{\gsA\gsB} \right) - \mathbb{E}_{\mathcal{O} \sim H} \Tr\left( D_{\mathcal{A}} \cdot \ketbra{\mathcal{A}^{\mathcal{O}}}_{\gsA\gsB} \right) \right| \leq \mathsf{negl}(n),
    \end{equation}
    where $\mathsf{negl}(n)$ is any function that is $o(1/n^c)$ for all $c > 0$.
\end{itemize}
We distinguish between two security levels based on the adversary's query capabilities:
\begin{itemize}
    \item A \textbf{standard PRU} achieves indistinguishability against adversaries that can only make forward queries, i.e., where each $U_i = U$ in Definition~\ref{def:oracle-adversary}.
    \item A \textbf{strong PRU} achieves indistinguishability against adversaries with access to all four query types: $U$, $U^\dagger$, $U^*$, and $U^T$, as defined in Definition~\ref{def:oracle-adversary}.
\end{itemize}
\end{definition}

\noindent The strong PRU definition considered here is stronger than that of \cite{ma2024construct} due to the ability for the oracle adversary to query the transpose $U^T$ and complex conjugation $U^*$.
We next show how to enhance the proof in \cite{ma2024construct} to handle transpose and complex conjugation.

\subsection{Preliminaries}

This section establishes the notation, definitions, and basic results from \cite{ma2024construct} that underpin our analysis. For completeness, we present all necessary background material.

\subsubsection{Notations}

\paragraph{Basic notation.} Let $N \coloneqq 2^n$ where $n$ denotes the number of qubits. We write $[N] \coloneqq \{1, \ldots, N\}$ and identify $[N]$ with $\{0,1\}^n$ via the binary representation of $i-1$ for each $i \in [N]$. For any $1 \leq t \leq N$, let $[N]^t_{\text{dist}}$ denote the set of length-$t$ sequences of distinct integers from $[N]$:
\begin{equation}
    [N]^t_{\text{dist}} \coloneqq \{(x_1,\dots,x_t) \in [N]^t: x_i \neq x_j \text{ for all } i \neq j \}.
\end{equation}
For $t = 0$, we set $[N]^0_{\text{dist}} \coloneqq \{ () \}$. For any permutation $\pi \in S_t$, define the unitary $S_{\pi}$ acting on $(\mathbb{C}^N)^{\otimes t}$:
\begin{equation}
    S_{\pi}: \ket{x_1,\dots,x_t} \mapsto \ket{x_{\pi^{-1}(1)},\dots,x_{\pi^{-1}(t)}}.
\end{equation}
We let $x = x_< \| x_> \in \{0, 1\}^n$, where $x_<, x_> \in \{0, 1\}^{n/2}$ and $\|$ denotes bitstring concatenation.

\paragraph{Quantum registers.} We use capital sans-serif letters for quantum registers. For register $\mathsf{A}$, the associated Hilbert space is $\mathcal{H}_{\mathsf{A}}$. States on multiple registers $(\sA, \sB)$ belong to $\mathcal{H}_{\mathsf{A}} \otimes \mathcal{H}_{\sB}$. We sometimes include register labels as subscripts, e.g., $\ket{\psi}_{\gsA \gsB}$. When an operator $U$ acts only on subsystem $\mathsf{A}$, we write $U_{\gsA}$ and extend it trivially to larger systems. To reduce notation, we often omit identity operators and write $U_{\gsA} \ket{\psi}_{\gsA\gsB}$ instead of $(U_{\gsA} \otimes \mathbbm{1}_{\gsB}) \ket{\psi}_{\gsA\gsB}$.

Given a projector $\Pi$ on register $\mathsf{A}$, we say state $\ket{\psi} \in \mathcal{H}_{\mathsf{A}}$ is in the image of $\Pi$ if $\Pi \ket{\psi} = \ket{\psi}$. For $\ket{\psi} \in \mathcal{H}_{\mathsf{A}} \otimes \mathcal{H}_{\gsB}$, we say $\ket{\psi}$ is in the image of $\Pi_{\gsA}$ if $\Pi_{\gsA} \ket{\psi}_{\gsA\gsB} = \ket{\psi}_{\gsA\gsB}$. We denote partial traces as $\Tr_{\gsB}(\ketbra{\psi})$ or $\Tr_{-\mathsf{A}}(\ketbra{\psi})$ when tracing out all systems except $\mathsf{A}$.

\paragraph{Relations} A \emph{relation} $R = \{(x_1, y_1), \dots, (x_t, y_t)\}$ is a multiset of ordered pairs $(x_i, y_i) \in [N]^2$. The \emph{size} $|R|$ equals the number of pairs counting multiplicities.

\begin{definition}[Sets of relations]
Let $\mathcal{R}$ denote the set of all relations and $\mathcal{R}_t$ the set of all size-$t$ relations. Define
\begin{align}
    \Dom(R) &= \{x \in [N]: \exists y \text{ such that } (x,y) \in R\},\\
    \Im(R) &= \{y \in [N]: \exists x \text{ such that } (x,y) \in R\},\\
    \Dom_<(R) &= \{x_< \in [\sqrt{N}]: \exists x_>, y \text{ such that } (x_< \| x_>, y) \in R\},\\
    \Dom_>(R) &= \{x_> \in [\sqrt{N}]: \exists x_<, y \text{ such that } (x_< \| x_>, y) \in R\},\\
    \Im_<(R) &= \{y_< \in [\sqrt{N}]: \exists x, y_> \text{ such that } (x, y_< \| y_>) \in R\},\\
    \Im_>(R) &= \{y_> \in [\sqrt{N}]: \exists x, y_< \text{ such that } (x, y_< \| y_>) \in R\}.
\end{align}
\end{definition}

Each relation $R$ corresponds to a \emph{relation state} in the symmetric subspace.

\begin{definition}[Relation states] 
\label{def:relation-states}
For relation $R = \{(x_1,y_1),\dots,(x_t,y_t)\}$, define
\begin{align}
    \ket{R} \coloneqq \frac{\sum_{\pi \in S_t} \ket{x_{\pi(1)},y_{\pi(1)},\dots,x_{\pi(t)},y_{\pi(t)}}}{\sqrt{t! \cdot \prod_{(x,y) \in [N]^2} \num(R,(x,y))!}},
\end{align}
where $\num(R,(x,y))$ denotes the multiplicity of pair $(x,y)$ in $R$.
\end{definition}

The relation states form an orthonormal basis for the symmetric subspace of $(\mathbb{C}^{N^2})^{\otimes t}$. When all pairs in $R$ are distinct, the normalization simplifies to $1/\sqrt{t!}$.

\begin{definition}[Restricted relation sets]
We define the restricted relation sets as follows.
\begin{itemize}
    \item $\mathcal{R}_t^{\text{inj}}$: injective relations where $(y_1,\dots,y_t) \in [N]^t_{\text{dist}}$
    \item $\mathcal{R}_t^{\text{bij}}$: bijective relations where $(x_1,\dots,x_t), (y_1,\dots,y_t) \in [N]^t_{\text{dist}}$
\end{itemize}
We also define $\mathcal{R}^{\text{inj}} = \bigcup_{t=0}^N \mathcal{R}_t^{\text{inj}}$ and $\mathcal{R}^{\text{bij}} = \bigcup_{t=0}^N \mathcal{R}_t^{\text{bij}}$.
\end{definition}

\paragraph{Variable-length registers} For each $t \geq 0$, let $\mathsf{R}^{(t)}$ be a register with Hilbert space $\mathcal{H}_{\mathsf{R}^{(t)}} = (\mathbb{C}^N \otimes \mathbb{C}^N)^{\otimes t}$. Define the variable-length register $\mathsf{R}$ with infinite-dimensional Hilbert space
\begin{align}
    \mathcal{H}_{\mathsf{R}} \coloneqq \bigoplus_{t=0}^\infty \mathcal{H}_{\mathsf{R}^{(t)}} = \bigoplus_{t=0}^\infty (\mathbb{C}^N \otimes \mathbb{C}^N)^{\otimes t}.
\end{align}
We decompose $\mathsf{R}^{(t)} = (\mathsf{R}^{(t)}_{\mathsf{X}}, \mathsf{R}^{(t)}_{\mathsf{Y}})$ where $\mathsf{R}^{(t)}_{\mathsf{X}}$ contains $\ket{x_1,\dots,x_t}$ and $\mathsf{R}^{(t)}_{\mathsf{Y}}$ contains $\ket{y_1,\dots,y_t}$. States of different lengths are orthogonal by the direct sum structure.

\begin{definition}[Projectors and extensions]
Define the projector onto relation states of size $t$:
\begin{align}
    \Pi^{\mathcal{R}}_t \coloneqq \sum_{R \in \mathcal{R}_t} \ketbra{R} = \Pi^{N^2,t}_{\ssym},
\end{align}
where $\Pi^{N^2,t}_{\ssym}$ projects onto the symmetric subspace of $(\mathbb{C}^{N^2})^{\otimes t}$. The projector onto all relation states~is
\begin{align}
    \Pi^{\mathcal{R}} \coloneqq \sum_{t=0}^\infty \Pi^{\mathcal{R}}_t = \sum_{R \in \mathcal{R}} \ketbra{R}.
\end{align}
For any operator $O$ on $\mathcal{H}_{\mathsf{R}^{(t)}}$, we extend it to $\mathcal{H}_{\mathsf{R}}$ by acting as the zero operator on $\mathcal{H}_{\mathsf{R}^{(t')}}$ for $t' \neq t$.
\end{definition}

\begin{definition}[Variable-length tensor powers]
For any unitary $U \in \mathcal{U}(N^2)$, define
\begin{align}
    U^{\otimes *} \coloneqq \sum_{t=0}^{\infty} U^{\otimes t}
\end{align}
acting on $\mathcal{H}_{\mathsf{R}}$.
\end{definition}

\paragraph{Pairs of variable-length registers} For constructions involving two variable-length registers $\mathsf{L}$ and $\mathsf{R}$, we introduce additional notation.

\begin{definition}[Length projectors] \label{notation:pi-leq-t}
For integers $\ell,r \geq 0$, let $\Pi_{\ell,r}$ project onto $\mathcal{H}_{\mathsf{L}^{(\ell)}} \otimes \mathcal{H}_{\mathsf{R}^{(r)}}$. For integer $t \geq 0$, let $\Pi_{\leq t}$ project onto $\bigoplus_{\ell,r \geq 0: \ell + r \leq t} \mathcal{H}_{\mathsf{L}^{(\ell)}} \otimes \mathcal{H}_{\mathsf{R}^{(r)}}$.
\end{definition}

\begin{definition}[Length-restricted operators]
For operator $B$ acting on registers $\mathsf{L}$ and $\mathsf{R}$, define
\begin{align}
    B_{\ell,r} &\coloneqq B \cdot \Pi_{\ell,r},\\
    B_{\leq t} &\coloneqq B \cdot \Pi_{\leq t}.
\end{align}
We adopt the convention that $B_{\leq t}^\dagger = (B_{\leq t})^\dagger$.
\end{definition}

\paragraph{Bounding distances between quantum states}
In our analysis of pseudorandom unitaries, we will make use of the following helpful inequalities for bounding distances between quantum states.
For any pure states $\ket{u}, \ket{v}$ with  $\langle u | u \rangle, \langle v | v \rangle \leq 1$,
\begin{equation} \label{eq: state to op bound}
    \big\lVert \dyad{u} - \dyad{v} \big\rVert_1 \leq 2 \big\lVert \ket{u} - \ket{v} \big\rVert_2.
\end{equation}
We also have the gentle measurement lemma,
\begin{equation} \label{eq: gentle measurement lemma}
    \big\lVert \Pi \rho \Pi - \rho \big\rVert_1 \leq 2 \sqrt{ 1 - \tr( \Pi \rho ) }.
\end{equation}
Finally, we will use the following variant on gentle measurement lemma from~\cite{ma2024construct}.
\begin{lemma}[Sequential gentle measurement; Lemma 2.3 of \cite{ma2024construct}]  \label{lem:seq-gentleM-pure}
    Let $\ket*{\psi}$ be a normalized state, $P_1,\dots,P_t$ be projectors, and $U_1,\dots,U_t$ be unitaries.
    \begin{align}
        \norm{U_t \ldots U_1 \ket*{\psi} -  P_t U_t \ldots P_{1} U_1 \ket*{\psi}}_2 \leq t \sqrt{1 - \norm{P_t U_t \ldots P_{1} U_1 \ket*{\psi}}_2^2}.
    \end{align}
\end{lemma}

\paragraph{Formulas for the 2-design twirl}

Finally, we will make use of the following formulas for the twirl over an exact 2-design $\mathfrak{D}$.
Let $\alpha$ denote any subset of $n$ qubits, and $\bar \alpha$ its complement.
Also let $\Pi^{\mathsf{eq}} = \sum_{x} \dyad{x} \otimes \dyad{x}$ denote the projector onto bitstrings that are equal between two copies, and $\Pi^{\mathsf{neq}} = 1 - \Pi^{\mathsf{eq}}$ its complement.
We have
\begin{equation} \label{eq: clifford twirl eq neq 1}
    \E_{U \sim \mathfrak{D}} \left[ 
    ( U \otimes U )^\dagger 
    \cdot
    \Pi^{\mathsf{eq}}_{\alpha}
    \Pi^{\mathsf{neq}}_{\bar \alpha}
    \cdot  
    ( U \otimes U ) \right]
    =
    \frac{N_\alpha N_{\bar \alpha} (N_{\bar \alpha}-1) }{N^2}
    \cdot \mathbbm{1}
    \preceq
    \frac{1}{N_\alpha } \cdot \mathbbm{1},
\end{equation}
and
\begin{equation} \label{eq: clifford twirl eq neq 2}
    \E_{U \sim \mathfrak{D}} \left[ 
    ( U \otimes \bar U )^\dagger 
    \cdot
    \Pi^{\mathsf{eq}}_{\alpha}
    \Pi^{\mathsf{neq}}_{\bar \alpha}
    \cdot  
    ( U \otimes \bar U ) \right]
    =
    \frac{N_\alpha  N_{\bar \alpha} (N_{\bar \alpha}-1) }{N^2}
    \cdot \mathbbm{1}
    \preceq
    \frac{1}{N_\alpha } \cdot \mathbbm{1}.
\end{equation}
From this, for any approximate unitary 2-design with additive error $\varepsilon$, we have
\begin{equation} \label{eq: clifford twirl eq neq 1 approx}
    \left\lVert \E_{U \sim \mathfrak{D}} \left[ 
    ( U \otimes U )^\dagger 
    \cdot
    \Pi^{\mathsf{eq}}_{\alpha}
    \Pi^{\mathsf{neq}}_{\bar \alpha}
    \cdot  
    ( U \otimes U ) \right] \right\rVert_\infty
    \leq
    \frac{1}{N_\alpha } + \varepsilon,
\end{equation}
and for any strong approximate unitary 2-design with additive error $\varepsilon$, we have
\begin{equation} \label{eq: clifford twirl eq neq 2 approx}
    \left\lVert \E_{U \sim \mathfrak{D}} \left[ 
    ( U \otimes \bar U )^\dagger 
    \cdot
    \Pi^{\mathsf{eq}}_{\alpha}
    \Pi^{\mathsf{neq}}_{\bar \alpha}
    \cdot  
    ( U \otimes \bar U ) \right] \right\rVert_\infty
    \leq
    \frac{1}{N_\alpha } + \varepsilon.
\end{equation}

\subsection{The purified permutation-function oracle}
\label{sec:PF3-oracle}

We now introduce the strong PFC ensemble~\cite{metger2024simple}.
In the remaining subsections of this section, we will present our extension of the proof of \cite{ma2024construct} to include the conjugate and transpose.
We also extend the proof to allow for any strong approximate unitary 2-design instead of an exact unitary 2-design.

We analyze the view of an adversary that can make standard, inverse, complex-conjugated, and transposed queries to an oracle $P_{\pi} \cdot F_{f}$, for uniformly random $\pi \sim \sSym_N$ and a random \textbf{ternary} function $f \sim \{0,1,2\}^N$. This will motivate an extension of the definition of the path recording oracle $V$ proposed in \cite{ma2024construct} to include its conjugate $\overline{V}$.

\begin{definition}[Purified permutation-function oracle] 
    The purified permutation-function oracle $\spfo$ is a unitary acting on registers $\sA,\sP,\sF$, where
    \begin{itemize}
        \item $\sP$ is a register associated with the Hilbert space $\calH_{\sP}$,  defined to be the span of the orthonormal states $\ket*{\pi}$ for all $\pi \in \sSym_N$. 
        \item $\sF$ is a register associated with the Hilbert space $\calH_{\sF}$, defined to be the span of the orthonormal states $\ket*{f}$ for all $f \in \{0,1,2\}^N$. 
    \end{itemize}
    The unitary $\spfo$ is defined to act as follows:
    \begin{align}
        \spfo_{\gsA\gsP\gsF} \ket*{x}_{\gsA} \ket*{\pi}_{\gsP} \ket*{f}_{\gsF} & \coloneqq \omega_3^{f(x)} \ket*{\pi(x)}_{\gsA} \ket*{\pi}_{\gsP} \ket*{f}_{\gsF}, \label{eq:pfo-definition-strong}\\
        &= \sum_{y \in [N]} \ket*{y}_{\gsA} \delta_{\pi(x) = y} \ket*{\pi} \omega_3^{f(x)} \ket*{f},
    \end{align}
    for all $x \in [N], \pi \in \sSym_N,$ and $f \in \{0, 1, 2\}^N$.
    Here, $\omega_3 = \exp( 2 \pi i / 3)$.
\end{definition}

The action of $\spfo^\dagger$ is
\begin{align}
    \spfo^\dagger \ket*{y}_{\gsA} \ket*{\pi} \ket*{f} &=  \sum_{x \in [N]} \ket*{x}_{\gsA} \delta_{\pi(x) = y} \ket*{\pi} \omega_3^{-f(x)} \ket*{f}.
\end{align}
The action of $\spfo^*$ is
\begin{align}
    \spfo^* \ket*{x}_{\gsA} \ket*{\pi} \ket*{f} &=  \sum_{y \in [N]} \ket*{y}_{\gsA} \delta_{\pi(x) = y} \ket*{\pi} \omega_3^{-f(x)} \ket*{f}.
\end{align}
The action of $\spfo^T$ is
\begin{align}
    \spfo^T \ket*{y}_{\gsA} \ket*{\pi} \ket*{f} &=  \sum_{x \in [N]} \ket*{x}_{\gsA} \delta_{\pi(x) = y} \ket*{\pi} \omega_3^{-f(x)} \ket*{f}.
\end{align}
Consider $P_\pi \coloneq \sum_{x \in [N]} \ketbra{\pi(x)}{x}$ and $F_f \coloneq \sum_{x \in [N]} \omega_3^{f(x)} \ketbra{x}{x}$. We have the equivalence:

\begin{fact}[Equivalence of purified and standard oracles] \label{claim:equiv-purified-vs-standard}
    For any oracle adversary, the following oracle instantiations are perfectly indistinguishable:
    \begin{itemize}
        \item (Queries to a random $P_\pi \cdot F_f$) Sample a uniformly random $\pi \sim \sSym_N, f \sim \{0,1, 2\}^N$. On each query, apply $(P_\pi \cdot F_f), (P_\pi \cdot F_f)^\dagger, (P_\pi \cdot F_f)^*, (P_\pi \cdot F_f)^T$ to register $\sA$.
        \item (Queries to $\spfo$) Initialize registers $\sP,\sF$ to $\frac{1}{\sqrt{N!}} \sum_{\pi \in \sSym_N} \ket*{\pi}_{\gsP} \otimes \frac{1}{\sqrt{2^N}} \sum_{f \in \{0,1, 2\}^N} \ket*{f}_{\gsF}$. At each query, apply $\spfo, \spfo^\dagger, \spfo^*$ or $\spfo^T$ to registers $\sA,\sP,\sF$.
    \end{itemize}
\end{fact}

\begin{definition}[$\mathsf{pf}$-relation state]
    For relation $L = \{(x_1, y_1), \dots, (x_\ell, y_\ell)\} \in \calR_\ell$ and relation $R = \{(x'_1,y'_1),\dots,(x'_r,y'_r)\} \in \calR_r$, where $\ell$ and $r$ are non-negative integers such that $\ell + r \leq N$, let
    \begin{align}
        \ket*{\pf_{L,R}} \coloneqq \frac{1}{\sqrt{3^N (N-\ell-r)!}} \sum_{\pi \in \sSym_N} \delta_{\pi,L \cup R} \ket*{\pi} \sum_{f \in \{0,1,2\}^N} \omega_3^{f(x_1) + \cdots + f(x_\ell) - (f(x'_1) + \cdots + f(x'_r))} \ket*{f},
    \end{align}
    where $\delta_{\pi,L \cup R}$ is an indicator variable that equals $1$ if $\pi(x) = y$ for all $(x,y) \in L \cup R$, and is $0$ otherwise.
\end{definition}

\begin{definition}
    Let $\calR^{2,\dist}$ be the set of all ordered pairs of relations $(L,R) \in \calR^2$ where $L \cup R = \{(x_1, y_1), \ldots, (x_t, y_t))\}$ is a bijective relation, i.e., $x_1, \ldots, x_t$ are distinct and $y_1, \ldots, y_t$ are distinct.
\end{definition}

\begin{lemma}[Orthonormality of $\mathsf{pf}$-relation states; From Claim 7 of \cite{ma2024construct}]
\label{claim:phi-L-R-orthogonal}
    $\{\ket*{\pf_{L,R}}\}_{(L,R) \in \calR^{2,\dist}}$ is an orthonormal set of vectors. 
\end{lemma}

\begin{definition}
    Define the partial isometry $\Compress^{\mathsf{PF}}: \calH_{\sP} \otimes \calH_{\sF} \rightarrow \calH_{\sL} \otimes \calH_{\sR}$ to be
    \begin{align}
        \Compress^{\mathsf{PF}} \coloneqq \sum_{(L,R) \in \calR^{2,\dist}} \ket*{L}_{\gsL} \otimes \ket*{R}_{\gsR} \cdot \bra{\pf_{L, R}}_{\gsP \gsF}.
    \end{align}
\end{definition}

We can use the action of $\spfo^*, \spfo^T$ to strengthen Claim 8 of \cite{ma2024construct} to obtain the following.

\begin{lemma}[Action of $\spfo$; From Claim 8 of \cite{ma2024construct}]
\label{claim:ternary-pfo-action}
    For any $(L,R) \in \calR^{2,\dist}$ and $x \in [N]$ such that $x \not\in \Dom(L \cup R)$, we have
    \begin{align}
        \spfo \ket*{x}_{\gsA} \ket*{\pf_{L,R}}_{\gsP \gsF} = \frac{1}{\sqrt{N-\abs{L \cup R}}} \sum_{\substack{y \in [N]:\\ y\not\in \Im(L \cup R)}} \ket*{y}_{\gsA} \ket*{\pf_{L \cup \{(x,y)\},R}}_{\gsP \gsF}. \label{eq:tpfo-map}
    \end{align}
    For any $(L,R) \in \calR^{2,\dist}$ and $y \in [N]$ such that $y \not\in \Im(L \cup R)$, we have
    \begin{align}
        \spfo^\dagger \ket*{y}_{\gsA} \ket*{\pf_{L,R}}_{\gsP \gsF} = \frac{1}{\sqrt{N-\abs{L \cup R}}} \sum_{\substack{x \in [N]:\\ x\not\in \Dom(L \cup R)}} \ket*{x}_{\gsA} \ket*{\pf_{L,R \cup \{(x,y)\}}}_{\gsP \gsF}. \label{eq:tpfo-inverse-map}
    \end{align}
    For any $(L,R) \in \calR^{2,\dist}$ and $x \in [N]$ such that $x \not\in \Dom(L \cup R)$, we have
    \begin{align}
        \spfo^* \ket*{x}_{\gsA} \ket*{\pf_{L,R}}_{\gsP \gsF} = \frac{1}{\sqrt{N-\abs{L \cup R}}} \sum_{\substack{y \in [N]:\\ y\not\in \Im(L \cup R)}} \ket*{y}_{\gsA} \ket*{\pf_{L , R\cup \{(x,y)\}}}_{\gsP \gsF}. \label{eq:tpfo-map-conj}
    \end{align}
    For any $(L,R) \in \calR^{2,\dist}$ and $y \in [N]$ such that $y \not\in \Im(L \cup R)$, we have
    \begin{align}
        \spfo^T \ket*{y}_{\gsA} \ket*{\pf_{L,R}}_{\gsP \gsF} = \frac{1}{\sqrt{N-\abs{L \cup R}}} \sum_{\substack{x \in [N]:\\ x\not\in \Dom(L \cup R)}} \ket*{x}_{\gsA} \ket*{\pf_{L\cup \{(x,y)\},R }}_{\gsP \gsF}. \label{eq:tpfo-inverse-map-conj}
    \end{align}
\end{lemma}

From the above lemma, we can see that $\spfo^*, \spfo^T$ closely mimics $\spfo, \spfo^\dagger$ with the only difference being the relations $L$ and $R$ are switched.

\subsection{The path-recording oracle $V$ and its conjugate $\overline{V}$}

The path-recording oracle $V$ proposed in \cite{ma2024construct} efficiently simulates a Haar-random unitary $U$ under both queries to $U$ (via query to $V$) and $U^\dagger$ (via query to $V^\dagger$).
In this work, we extend the path-recording framework to enable queries to $U, U^\dagger, U^*,$ and $U^\dagger$ by defining a new operator $\overline{V}$ that serves as the conjugate of $V$.
While $\overline{V}$ is not the actual complex conjugate of $V$, we will show that $V, V^\dagger, \overline{V}, \overline{V}^\dagger$ efficiently simulates a Haar-random unitary $U$ under queries to $U, U^\dagger, U^*, U^T$, respectively.

To define the path-recording oracle $V$, we need to define the left part $V^L$ and right part $V^R$ of $V$.

\begin{definition}[Left and right parts of $V$] \label{def:V-sym-PRO}
    Let $V^L$ be the linear operator that acts as follows. For $x \in [N]$ and $(L,R) \in \mathcal{R}^{2,\leq N-1}$,
    \begin{equation}
        V^L \cdot \ket*{x}_{\gsA} \ket*{L}_{\gsL} \ket*{R}_{\gsR} \coloneqq \sum_{\substack{y \in [N]:\\ y\not\in \Im(L \cup R)}} \frac{1}{\sqrt{N - \abs{\Im(L \cup R)}}} \ket*{y}_{\gsA} \ket*{L \cup \{(x,y)\}}_{\gsL} \ket*{R}_{\gsR}.
    \end{equation}
    Define $V^R$ to be the linear operator such that for all $y \in [N]$ and $(L,R) \in \mathcal{R}^{2,\leq N-1}$,
    \begin{equation}
        V^R \cdot \ket*{y}_{\gsA} \ket*{L}_{\gsL} \ket*{R}_{\gsR} \coloneqq \sum_{\substack{x \in [N]:\\ x\not\in \Dom(L \cup R)}} \frac{1}{\sqrt{N - \abs{\Dom(L \cup R)}}} \ket*{x}_{\gsA} \ket*{L}_{\gsL} \ket*{R \cup \{(x, y)\} }_{\gsR}.
    \end{equation}
    By construction, $V^L$ and $V^R$ take states in $\Id_{\gsA} \otimes \Pi^{\calR^2}_{\leq i, \gsL \gsR}$ to $\Id_{\gsA} \otimes \Pi^{\calR^2}_{\leq i+1, \gsL \gsR}$.
\end{definition}

\begin{lemma}[Claim 14 \cite{ma2024construct}]
$V^L$ and $V^R$ are partial isometries.
\end{lemma}

\begin{definition}[Path-recording oracle $V$]
\label{def:symmetric-V}
    The path-recording oracle is the operator $V$ defined as
    \begin{align}
        V &:= V^L \cdot (\Id - V^R \cdot V^{R,\dagger}) + (\Id - V^L \cdot V^{L,\dagger}) \cdot V^{R,\dagger}.
    \end{align}
    By construction, $V$ and $V^\dagger$ take states in $\Id_{\gsA} \otimes \Pi^{\calR^2}_{\leq i, \gsL \gsR}$ to $\Id_{\gsA} \otimes \Pi^{\calR^2}_{\leq i+1, \gsL \gsR}$ for any integer $i \geq 0$.
\end{definition}

\begin{lemma}[Claim 15 \cite{ma2024construct}]
$V$ is a partial isometry.
\end{lemma}

We now define the conjugate $\overline{V}$ of $V$. This operator is not the actual complex conjugation of the path-recording oracle $V$. $\overline{V}$ is a new object proposed to simulate the action of the complex conjugation $U^*$ of a Haar-random unitary $U$. And the conjugate transpose $\overline{V}^\dagger$ of $\overline{V}$ is designed to simulate the action of the transpose $U^T$ of a Haar-random unitary $U$.

\begin{definition}[Conjugated left and right parts of $V$] \label{def:V-sym-PRO-conj}
    Let $\overline{V}^L$ be the linear operator that acts as follows. For $x \in [N]$ and $(L,R) \in \mathcal{R}^{2,\leq N-1}$,
    \begin{equation}
        \overline{V}^L \cdot \ket*{x}_{\gsA} \ket*{L}_{\gsL} \ket*{R}_{\gsR} \coloneqq \sum_{\substack{y \in [N]:\\ y\not\in \Im(L \cup R)}} \frac{1}{\sqrt{N - \abs{\Im(L \cup R)}}} \ket*{y}_{\gsA} \ket*{L}_{\gsL} \ket*{R \cup \{(x,y)\}}_{\gsR}.
    \end{equation}
    Define $\overline{V}^R$ to be the linear operator such that for all $y \in [N]$ and $(L,R) \in \mathcal{R}^{2,\leq N-1}$,
    \begin{equation}
        \overline{V}^R \cdot \ket*{y}_{\gsA} \ket*{L}_{\gsL} \ket*{R}_{\gsR} \coloneqq \sum_{\substack{x \in [N]:\\ x\not\in \Dom(L \cup R)}} \frac{1}{\sqrt{N - \abs{\Dom(L \cup R)}}} \ket*{x}_{\gsA} \ket*{L \cup \{(x, y)\} }_{\gsL} \ket*{R }_{\gsR}.
    \end{equation}
    By construction, $\overline{V}^L$ and $\overline{V}^R$ take states in $\Id_{\gsA} \otimes \Pi^{\calR^2}_{\leq i, \gsL \gsR}$ to $\Id_{\gsA} \otimes \Pi^{\calR^2}_{\leq i+1, \gsL \gsR}$.
\end{definition}

\begin{definition}[Conjugated path-recording oracle $\overline{V}$]
\label{def:symmetric-V-conj}
    The conjugated path-recording oracle is the operator $\overline{V}$ defined as
    \begin{align}
        \overline{V} &:= \overline{V}^L \cdot (\Id - \overline{V}^R \cdot \overline{V}^{R,\dagger}) + (\Id - \overline{V}^L \cdot \overline{V}^{L,\dagger}) \cdot \overline{V}^{R,\dagger}.
    \end{align}
    By construction, $\overline{V}$ and $\overline{V}^\dagger$ take states in $\Id_{\gsA} \otimes \Pi^{\calR^2}_{\leq i, \gsL \gsR}$ to $\Id_{\gsA} \otimes \Pi^{\calR^2}_{\leq i+1, \gsL \gsR}$ for any integer $i \geq 0$.
\end{definition}

Because the main change in $\overline{V}$ over $V$ is in swapping $L$ and $R$, using the same proof as Claim 14 and 15 of \cite{ma2024construct}, we have the following lemma.

\begin{lemma} \label{lem:partial-iso-overline}
$\overline{V}^L$, $\overline{V}^R$, and $\overline{V}$ are partial isometries.
\end{lemma}

We next present a central property of the path-recording oracle $V$ and its conjugate $\overline{V}$.

\begin{definition} \label{def:multi-rot-Q}
    For any $n$-qubit unitary $C,D$, define
    \begin{align}
        Q[C,D] &\coloneqq (C \otimes D^T)^{\otimes *}_{\gsL} \otimes (C^* \otimes D^{\dagger})^{\otimes *}_{\gsR},
    \end{align}
    where $C^*$ is the complex conjugate of $C$.
\end{definition}

Because $\overline{V}$ corresponds to swapping $L$ and $R$ in $V$, using the same proof, we can obtain the following two-sided unitary invariance property for $V$ and $\overline{V}$.

\begin{lemma}[Two-sided unitary invariance; Claim 16 of \cite{ma2024construct}]
\label{claim:two-sided-invariance}
    For any integer $0 \leq t \leq N-1$ and any pair of $n$-qubit unitaries $C,D$,
    \begin{align}
        \norm{D_{\gsA} \cdot V_{\leq t} \cdot C_{\gsA} \otimes Q[C,D]_{\gsL \gsR}
        - Q[C,D]_{\gsL \gsR} \cdot V_{\leq t}}_{\infty} & \leq 16\sqrt{\frac{2t(t+1)}{N}},\\
        \norm{C_{\gsA}^\dagger \cdot (V^\dagger)_{\leq t} \cdot D_{\gsA}^\dagger \otimes Q[C,D]_{\gsL \gsR}
        - Q[C,D]_{\gsL \gsR} \cdot (V^\dagger)_{\leq t}}_{\infty} & \leq 16\sqrt{\frac{2t(t+1)}{N}},\\
        \norm{D^*_{\gsA} \cdot \overline{V}_{\leq t} \cdot C^*_{\gsA} \otimes Q[C,D]_{\gsL \gsR}
        - Q[C,D]_{\gsL \gsR} \cdot \overline{V}_{\leq t}}_{\infty} & \leq 16\sqrt{\frac{2t(t+1)}{N}},\\
        \norm{C_{\gsA}^T \cdot (\overline{V}^\dagger)_{\leq t} \cdot D_{\gsA}^T \otimes Q[C,D]_{\gsL \gsR}
        - Q[C,D]_{\gsL \gsR} \cdot (\overline{V}^\dagger)_{\leq t}}_{\infty} & \leq 16\sqrt{\frac{2t(t+1)}{N}},
    \end{align}
\end{lemma}

\subsection{Partial path-recording oracle $W$ and its conjugate $\overline{W}$}

A very useful object proposed in \cite{ma2024construct} is the partial path-recording oracle $W$, which is a restricted version of the path-recording oracle $V$. The operator $W$ only acts nontrivially on a subspace and maps the orthogonal subspace to zero. The subspace is defined based on $\calR^{2,\dist}$.

Similar to $V$, the partial path-recording oracle $W$ contains a left part $W^L$ and a right part $W^R$.

\begin{definition}[$W^L$ and $W^R$]
\label{def:ternary-W-action}
    Define $W^L$ to be the linear map such that for any $(L,R) \in \calR^{2,\dist}$ and $x \in [N]$ such that $x\not\in \Dom(L \cup R)$,
    \begin{align}
        W^L \cdot \ket*{x}_{\gsA} \ket*{L}_{\gsL} \ket*{R}_{\gsR} \coloneqq \frac{1}{\sqrt{N-\abs{L \cup R}}} \sum_{\substack{y \in [N]:\\ y\not\in \Im(L \cup R)}} \ket*{y}_{\gsA} \ket*{L \cup \{(x,y)\}}_{\gsL} \ket*{R}_{\gsR}. \label{eq:WL-def}
    \end{align}
    Similarly, define $W^R$ be the linear map such that for any $(L,R) \in \calR^{2,\dist}$ and $y \in [N]$ such that $y\not\in \Im(L \cup R)$,
    \begin{align}
        W^R \cdot \ket*{y}_{\gsA} \ket*{L}_{\gsL} \ket*{R}_{\gsR} \coloneqq \frac{1}{\sqrt{N-\abs{L \cup R}}} \sum_{\substack{x \in [N]:\\ x\not\in \Dom(L \cup R)}} \ket*{x}_{\gsA} \ket*{L}_{\gsL} \ket*{R \cup \{(x,y)\}}_{\gsR}. \label{eq:WR-def}
    \end{align}
\end{definition}

\begin{definition} \label{def:symmetric-W}
    The partial path-recording oracle is the operator $W$ defined as
    \begin{equation}
        W \coloneqq W^L + W^{R,\dagger}.
    \end{equation}
\end{definition}

We now extend the definition of $W$ in \cite{ma2024construct} to define its conjugate $\overline{W}$. Note that $\overline{W}$ is not the exact complex conjugate of $W$. However, intuitively speaking, for an oracle adversary, $\overline{W}$ will behave like the complex conjugation of $W$ similar to how $\overline{V}$ behave like the complex conjugation of $V$.

\begin{definition}[Conjugates of $W^L, W^R, $ and $W$]
\label{def:ternary-W-action-conj}
    Define $\overline{W}^L$ to be the linear map such that for any $(L,R) \in \calR^{2,\dist}$ and $x \in [N]$ such that $x\not\in \Dom(L \cup R)$,
    \begin{align}
        \overline{W}^L \cdot \ket*{x}_{\gsA} \ket*{L}_{\gsL} \ket*{R}_{\gsR} \coloneqq \frac{1}{\sqrt{N-\abs{L \cup R}}} \sum_{\substack{y \in [N]:\\ y\not\in \Im(L \cup R)}} \ket*{y}_{\gsA} \ket*{L }_{\gsL} \ket*{R \cup \{(x,y)\} }_{\gsR}. \label{eq:conjWL-def}
    \end{align}
    Similarly, define $\overline{W}^R$ be the linear map such that for any $(L,R) \in \calR^{2,\dist}$ and $y \in [N]$ such that $y\not\in \Im(L \cup R)$,
    \begin{align}
        \overline{W}^R \cdot \ket*{y}_{\gsA} \ket*{L}_{\gsL} \ket*{R}_{\gsR} \coloneqq \frac{1}{\sqrt{N-\abs{L \cup R}}} \sum_{\substack{x \in [N]:\\ x\not\in \Dom(L \cup R)}} \ket*{x}_{\gsA} \ket*{L\cup \{(x,y)\}}_{\gsL} \ket*{R}_{\gsR}. \label{eq:conjWR-def}
    \end{align}
    The conjugate of the partial path-recording oracle $W$ is the operator $\overline{W}$ defined as
    \begin{equation}
        \overline{W} \coloneqq \overline{W}^L + \overline{W}^{R,\dagger}.
    \end{equation}
\end{definition}

We instantiate the following definitions of projectors.

\begin{definition}[Bijective-relation projectors] \label{def:bij-proj}
    Define the projectors
    \begin{align}
        \Pi^{\bij}_{\gsL \gsR} \coloneq \sum_{(L,R) \in \mathcal{R}^{2,\dist}} \ketbra*{L}_{\gsL} \otimes \ketbra*{R}_{\gsR}, \quad\quad \Pi^{\bij}_{\leq t, \gsL \gsR} \coloneq \Pi^{\bij}_{\gsL \gsR} \cdot \Pi_{\leq t, \gsL \gsR} = \Pi_{\leq t, \gsL \gsR} \cdot \Pi^{\bij}_{\gsL \gsR},
    \end{align}
    where the projector $\Pi_{\leq t, \gsL \gsR}$ is the maximum-length projector defined in \cref{notation:pi-leq-t}.
\end{definition}

\begin{definition}
    For a partial isometry $G$, let $\calD(G)$ and $\calI(G)$ denote its domain and image. Let $\Pi^{\calD(G)} = G^\dagger \cdot G$ and $\Pi^{\calI(G)} = G \cdot G^\dagger$ denote the orthogonal projectors onto $\calD(G)$ and $\calI(G)$.
\end{definition}

Because the conjugated versions of $W^L, W^R,$ and $W$ amounts to swapping the $L$ and $R$ register, the proofs in \cite{ma2024construct} can be combined with the action of $\spfo^*$ and $\spfo^T$ to establish the following lemmas.

\begin{lemma}[$W$ is a restriction of $\spfo$ up to isometry; From Claim 13 of \cite{ma2024construct}]
\label{claim:relate-W-and-spfo}
We have
    \begin{align}
        W &= \Compress \cdot \spfo \cdot \Compress^\dagger \cdot \Pi^{\calD(W)},\label{eq:compress-proof-spfo-goal-1}\\
        W^\dagger &= \Compress \cdot \spfo^\dagger \cdot \Compress^\dagger \cdot \Pi^{\calI(W)} \label{eq:compress-proof-spfo-goal-2}\\
        \overline{W} &= \Compress \cdot \spfo^* \cdot \Compress^\dagger \cdot \Pi^{\calD(\overline{W})},\label{eq:compress-proof-spfo-goal-1-conj}\\
        \overline{W}^\dagger &= \Compress \cdot \spfo^\dagger \cdot \Compress^\dagger \cdot \Pi^{\calI(\overline{W})}. \label{eq:compress-proof-spfo-goal-2-conj}
    \end{align}
\end{lemma}

\begin{lemma}[From Fact 5 of \cite{ma2024construct}] \label{fact:WLWR-space-leqi}
    For any integer $i \geq 0$, $W^L, W^R, \overline{W}^L, \overline{W}^R$ map states in the subspace associated to the projector $\Id_{\gsA} \otimes \Pi^{\bij}_{\leq i, \gsL \gsR}$ into the subspace associated with the projector $\Id_{\gsA} \otimes \Pi^{\bij}_{\leq i+1, \gsL \gsR}$.
\end{lemma}

\begin{lemma}[From Claim 9 and Claim 11 of \cite{ma2024construct}]
    $W^L, W^R, W, \overline{W}^L, \overline{W}^R, \overline{W}$ are partial isometries.
\end{lemma}

\begin{lemma}[From Fact 8 of \cite{ma2024construct}]\label{claim:W-partial-isometry}
    The domain and image of the partial isometry $W$ are given by
    \begin{align}
        \Pi^{\calD(W)} &= \Pi^{\calD(W^L)} + \Pi^{\calI(W^R)}, \label{eq:expand-DW}\\
        \Pi^{\calI(W)} &= \Pi^{\calD(W^R)} + \Pi^{\calI(W^L)}. \label{eq:expand-IW}
    \end{align}
    The domain and image of the partial isometry $\overline{W}$ are given by
    \begin{align}
        \Pi^{\calD(\overline{W})} &= \Pi^{\calD(\overline{W}^L)} + \Pi^{\calI(\overline{W}^R)}, \label{eq:expand-DW-conj}\\
        \Pi^{\calI(\overline{W})} &= \Pi^{\calD(\overline{W}^R)} + \Pi^{\calI(\overline{W}^L)}. \label{eq:expand-IW-conj}
    \end{align}
\end{lemma}

\begin{lemma}[From Claim 10 and Claim 12 of \cite{ma2024construct}] \label{claim:pileqt-commutes-with-WLWRDomIm}
For all integers $t \geq 0$,
    $\Pi_{\leq t}$ commutes with $\Pi^{\calD(W^L)}$, $\Pi^{\calI(W^L)}$, $\Pi^{\calD(W^R)}$, $\Pi^{\calI(W^R)}$, $\Pi^{\calD(W)}$, $\Pi^{\calI(W)}$, $\Pi^{\calD(\overline{W}^L)}$, $\Pi^{\calI(\overline{W}^L)}$, $\Pi^{\calD(\overline{W}^R)}$, $\Pi^{\calI(\overline{W}^R)}$, $\Pi^{\calD(\overline{W})}$, and $\Pi^{\calI(\overline{W})}$.
\end{lemma}

\begin{lemma}[$W$ is a restriction of $V$; From Claim 17 of \cite{ma2024construct}]
\label{claim:relate-W-and-V}
We have
    \begin{align}
        W &= V \cdot \Pi^{\calD(W)}, \label{eq:relate-W-and-V-goal1}\\
        W^\dagger &= V^\dagger \cdot \Pi^{\calI(W)}, \label{eq:relate-W-and-V-goal2}\\
        \overline{W} &= \overline{V} \cdot \Pi^{\calD(\overline{W})}, \label{eq:relate-W-and-V-goal1-conj}\\
        \overline{W}^\dagger &= \overline{V}^\dagger \cdot \Pi^{\calI(\overline{W})}. \label{eq:relate-W-and-V-goal2-conj}
    \end{align}
\end{lemma}

\begin{lemma}[From Corollary 8.3 of \cite{ma2024construct}] \label{claim:relate-W-V-projector}
We have
    \begin{align}
        W^\dagger \cdot V &= \Pi^{\calD(W)},\\
        W \cdot V^\dagger &= \Pi^{\calI(W)},\\
        \overline{W}^\dagger \cdot \overline{V} &= \Pi^{\calD(\overline{W})},\\
        \overline{W} \cdot \overline{V}^\dagger &= \Pi^{\calI(\overline{W})}.
    \end{align}
\end{lemma}

\begin{lemma}[Twirling by strong approximate unitary $2$-design; From Lemma 9.2 of \cite{ma2024construct}] \label{lem:twirling-strongPRU}
For any strong approximate unitary $2$-design $\mathfrak{D}$ with additive error $\varepsilon$, and any integer $0 \leq t \leq N-1$, we have
\begin{align}
    \norm{  \E_{C,D \sim \mathfrak{D}} (C_{\gsA} \otimes Q[C,D]_{\gsL \gsR})^\dagger \cdot \Big( \Pi^{\bij}_{\leq t, \gsL \gsR} - \Pi^{\calD(W)}_{\leq t, \gsA \gsL \gsR}\Big) \cdot (C_{\gsA} \otimes Q[C,D]_{\gsL \gsR}) }_{\infty} &\leq 6t \sqrt{\frac{t}{N}} + 2t\varepsilon,\\
    \norm{  \E_{C,D \sim \mathfrak{D}} (D^\dagger_{\gsA} \otimes Q[C,D]_{\gsL \gsR})^\dagger \cdot \Big( \Pi^{\bij}_{\leq t, \gsL \gsR} - \Pi^{\calI(W)}_{\leq t, \gsA \gsL \gsR}\Big) \cdot (D^\dagger_{\gsA} \otimes Q[C,D]_{\gsL \gsR}) }_{\infty} &\leq 6t \sqrt{\frac{t}{N}} + 2t\varepsilon,\\
    \norm{  \E_{C,D \sim \mathfrak{D}} (C^*_{\gsA} \otimes Q[C,D]_{\gsL \gsR})^\dagger \cdot \Big( \Pi^{\bij}_{\leq t, \gsL \gsR} - \Pi^{\calD(\overline{W})}_{\leq t, \gsA \gsL \gsR}\Big) \cdot (C^*_{\gsA} \otimes Q[C,D]_{\gsL \gsR}) }_{\infty} &\leq 6t \sqrt{\frac{t}{N}} + 2t\varepsilon,\\
    \norm{  \E_{C,D \sim \mathfrak{D}} (D^T_{\gsA} \otimes Q[C,D]_{\gsL \gsR})^\dagger \cdot \Big( \Pi^{\bij}_{\leq t, \gsL \gsR} - \Pi^{\calI(\overline{W})}_{\leq t, \gsA \gsL \gsR}\Big) \cdot (D^T_{\gsA} \otimes Q[C,D]_{\gsL \gsR}) }_{\infty} &\leq 6t \sqrt{\frac{t}{N}} + 2t\varepsilon.
\end{align}
\end{lemma}
\begin{proof}
    The proof follows from the proof of Lemma~9.2 in~\cite{ma2025construct} with one replacement.
    In Claims 29 and 30 and Eq.~(11.82), \cite{ma2025construct} uses the following spectral norm bound for the twirl over an \emph{exact} unitary 2-design $\mathfrak{D}'$,
    \begin{align}
        \left\lVert \E_{U \sim \mathfrak{D}'} \left[ (U \otimes U)^\dagger \cdot \Pi^{\mathsf{eq}} \cdot (U \otimes U) \right] \right\rVert_\infty & \leq \frac{2}{N+1}, \\
        \left\lVert \E_{U \sim \mathfrak{D}'} \left[ (U \otimes U^*) \cdot ( \Pi^{\mathsf{eq}} - \Pi^{\mathsf{EPR}} ) \cdot (U \otimes U^*) \right] \right\rVert_\infty & \leq \frac{1}{N+1}.
    \end{align}
    Here, we replace these bounds with analogous bounds for the twirl over an \emph{approximate} unitary 2-design $\mathfrak{D}$.
    For any (standard) approximate unitary 2-design with additive error $\varepsilon$, the first inequality becomes,
    \begin{align}
        \left\lVert \E_{U \sim \mathfrak{D}} \left[ (U \otimes U)^\dagger \cdot \Pi^{\mathsf{eq}} \cdot (U \otimes U) \right] \right\rVert_\infty & \leq \frac{2}{N+1}  + \varepsilon.
    \end{align}
    Meanwhile, for any strong approximate unitary 2-design with additive error $\varepsilon$, the second inequality becomes,
    \begin{align}
        \left\lVert \E_{U \sim \mathfrak{D}'} \left[ (U \otimes U^*) \cdot ( \Pi^{\mathsf{eq}} - \Pi^{\mathsf{EPR}} ) \cdot (U \otimes U^*) \right] \right\rVert_\infty & \leq \frac{1}{N+1} + \varepsilon.
    \end{align}
    To derive the first inequality, we abbreviate $\Phi_\mathfrak{D}(X) \equiv \E_{U \sim \mathfrak{D}}\left[ (U \otimes U) X (U \otimes U)^\dagger \right]$ and $\delta\Phi_\mathfrak{D} \equiv \Phi_\mathfrak{D} - \Phi_{\mathfrak{D}'}$, and apply
    $\lVert \delta\Phi^\dagger_\mathfrak{D}(O) \rVert_\infty = \max_\rho \text{tr}( \delta\Phi^\dagger_\mathfrak{D}(O) \rho) = \max_\rho \text{tr}( O \delta\Phi_\mathfrak{D}(\rho) ) \leq \max_\rho \lVert O \rVert_\infty \cdot \lVert \Phi_\mathfrak{D}(\rho) \rVert_1 \leq \lVert O \rVert_\infty \cdot \varepsilon$.
    An identical series of steps derives the second inequality. 
    Propagating this replacement through the remainder of the proof yields Lemma~\ref{lem:twirling-strongPRU}.
\end{proof}

Note that in the statement of~\cref{lem:twirling-strongPRU}, $\Pi^{\bij}_{\leq t, \gsL \gsR}$ is shorthand for $\Id_{\gsA} \otimes \Pi^{\bij}_{\leq t, \gsL \gsR}$, and thus the operators inside the spectral norm $\norm{\cdot}_{\infty}$ act on $\sA,\sL,\sR$.

\subsection{$V, \overline{V}$ approximates Haar-random unitary $U$ under $U, U^\dagger, U^*, U^T$}

Because all lemmas generalize to complex conjugation and transpose using the suitably defined $\overline{W}$ and $\overline{V}$, we can follow the same proof of \cite{ma2024construct} to show that the extended path recording oracle $V, \overline{V}$ approximates the following random unitary ensemble under $U, U^\dagger, U^*, U^T$.

\begin{definition}[$\mathsf{sPFC}(\mathfrak{D})$ distribution]
    For any distribution $\mathfrak{D}$ supported on $\calU(N)$, define the distribution $\mathsf{sPFC}({\mathfrak{D}})$ as follows:
    \begin{enumerate}
        \item Sample a uniformly random permutation $\pi \sim \sSym_{N}$, a uniformly random $f \sim \{0,1, 2\}^N$, and two independently sampled $n$-qubit unitaries $C, D \sim \mathfrak{D}$. Following the definitions in~\cref{sec:PF3-oracle}, 
        \begin{align}
            F_f \coloneqq \sum_{x \in [N]} e^{2 \pi \cdot f(x) \cdot i/3} \ketbra*{x} \quad \text{and} \quad P_{\pi} \coloneqq \sum_{x \in [N]} \ketbra*{\pi(x)}{x}.
        \end{align}
        \item Output the $n$-qubit unitary $\calO \coloneqq D \cdot P_\pi \cdot F_f \cdot C$.
    \end{enumerate} 
\end{definition}

\begin{definition}[Global state after queries to $V, \overline{V}$]
    For a $t$-query oracle adversary $\mathcal{A}$ that can perform queries to $U, U^\dagger, U^*, U^T$, where $b_i \in \{0, 1\}$ and $c_i \in \{0, 1\}$ denote the four choices (\,$U \rightarrow b_i=0, c_i=0; U^\dagger \rightarrow b_i=1, c_i=0; U^* \rightarrow b_i=0, c_i=1; U^T \rightarrow b_i=1, c_i=1$), and any $0 \leq i \leq t$, let
    \begin{align}
        \ket*{\calA_i^{V, \overline{V}}}_{\gsA \gsB \gsL \gsR} \coloneqq \prod_{i = 1}^t &\Bigg( \Big( (1-c_i)((1-b_i) \cdot V_{\gsA \gsL \gsR} + b_i \cdot V_{\gsA \gsL \gsR}^\dagger)\\
        &+ c_i((1-b_i) \cdot \overline{V}_{\gsA \gsL \gsR} + b_i \cdot \overline{V}_{\gsA \gsL \gsR}^\dagger) \Big) \cdot A_{i,\gsA \gsB} \Bigg) \ket*{0^{n+m}}_{\gsA \gsB} \otimes \ket*{\varnothing}_{\gsL} \ket*{\varnothing}_{\gsR}
    \end{align}
    denote the global state on registers $\sA,\sB,\sL,\sR$ after $\calA$ makes $i$ queries to $V$.
\end{definition}

We will also consider the global purified state after queries to $W, \overline{W}$, where we twirl the input and the output states by two independent random unitaries sampled from any unitary $2$-design. For this purpose, we define the purification of two random unitaries $C, D$.

\begin{definition}
\label{def:init-D-state}
    For any distribution $\frakD$ over $n$-qubit unitaries, define the state
    \begin{align}
    \ket*{\init(\mathfrak{D})}_{\gsC \gsD} \coloneq \int_{C,D} \sqrt{ d\mu_{\mathfrak{D}}(C) d\mu_{\mathfrak{D}}(D)} \ket*{C}_{\gsC} \otimes \ket*{D}_{\gsD},
\end{align}
where $\mu_{\mathfrak{D}}(C)$ is the probability measure for which $C$ is sampled from $\mathfrak{D}$.
\end{definition}

\begin{definition}[Controlled $C,D$ and $Q$]
\label{def:controlled-CDQ}
    Define the following operators
    \begin{align}
        &\mathsf{cC} \coloneq \int_{C} C_{\gsA} \otimes \ketbra*{C}_{\gsC}, \quad \mathsf{cD} \coloneq \int_{D} D_{\gsA} \otimes \ketbra*{D}_{\gsD},\\
    &\mathsf{cQ} \coloneq \int_{C, D} Q[C, D]_{\gsL {\color{gray} ,} \gsR} \otimes \ketbra*{C}_{\gsC} \otimes \ketbra*{D}_{\gsD}.
\end{align}
\end{definition} 

\begin{definition}[Global state after queries to Twirled $W, \overline{W}$]
\label{def:twirled-W-state}
    For a $t$-query adversary $\mathcal{A}$ that can perform queries to $U, U^\dagger, U^*, U^T$, where $b_i \in \{0, 1\}$ and $c_i \in \{0, 1\}$ denote the four choices (\,$U \rightarrow b_i=0, c_i=0; U^\dagger \rightarrow b_i=1, c_i=0; U^* \rightarrow b_i=0, c_i=1; U^T \rightarrow b_i=1, c_i=1$), let
    \begin{align}
        \ket*{\mathcal{A}_0^{W, \overline{W},\frakD}} &\coloneqq \ket*{0^n}_{\gsA} \ket*{0^m}_{\gsB} \ket*{\varnothing}_{\gsL} \ket*{\varnothing}_{\gsR} \ket*{\mathsf{init}(\mathfrak{D})}_{\gsC \gsD}.
    \end{align}
    For $i$ from $1$ to $t$, let
    \begin{align}
        \ket*{\mathcal{A}_i^{W, \overline{W},\frakD}} \coloneqq &\Big( (1-c_i)((1-b_i) \cdot (\scD \cdot W \cdot \scC) + b_i \cdot (\scD \cdot W \cdot \scC)^\dagger) +\\
        &\quad c_i ((1-b_i) \cdot (\scD^* \cdot \overline{W} \cdot \scC^*) + b_i \cdot (\scD^* \cdot \overline{W} \cdot \scC^* )^\dagger) \Big) \cdot A_i \cdot \ket*{\mathcal{A}^{W,\frakD}_{i-1}}.
    \end{align}
\end{definition}

\begin{lemma}[$W$ is indistinguishable from $V$ after twirling; From Lemma 9.3 of \cite{ma2024construct}] \label{lem:closeness-AWD-and-PhiVt}
Let $\mathfrak{D}$ be any strong approximate unitary $2$-design with additive error $\varepsilon$. For any $t$-query oracle adversary $\calA$ that can query $\mathcal{O}, \mathcal{O}^\dagger, \mathcal{O}^*, \mathcal{O}^T$,
    \begin{align}
        \TD( \Tr_{- \sA \sB} \ketbra*{\mathcal{A}^{W, \overline{W}, \mathfrak{D}}_t}_{\gsA \gsB \gsL \gsR \gsC \gsD}, \Tr_{- \sA \sB} \ketbra*{\mathcal{A}^{V, \overline{V}}_t}_{\gsA \gsB \gsL \gsR} ) \leq \frac{9t}{N^{1/8}} + 2 t^{1/4} \varepsilon^{1/4}. \label{eq:TD-phi-psi}
    \end{align}
\end{lemma}
\begin{proof}
    The proof follows from the proof of Lemma~9.3 in~\cite{ma2025construct} with one replacement. In the application of Lemma~9.2 in Eq.~(9.45) of the proof of Claim~18, we apply Lemma~\ref{lem:twirling-strongPRU} for the twirl over a strong $\varepsilon$-approximate unitary 2-design instead. This modifies the right hand side of the statement of Claim~18 to $1 - 35t^2/N^{1/4} - \sqrt{2 t \varepsilon}$. Propagating this replacement through the rest of the proof of Lemma~9.3 and applying the inequality $\sqrt{x+y} \leq \sqrt{x} + \sqrt{y}$ yields Lemma~\ref{lem:closeness-AWD-and-PhiVt}.
\end{proof}

\begin{lemma}[$\mathsf{sPFC}(\mathfrak{D})$ is indistinguishable from $V$; From Lemma 9.1 of \cite{ma2024construct}]\label{lemma:pfd-cho-strong}
     Let $\mathfrak{D}$ be any strong approximate unitary $2$-design with additive error $\varepsilon$. For any $t$-query oracle adversary $\calA$ that can query $\mathcal{O}, \mathcal{O}^\dagger, \mathcal{O}^*, \mathcal{O}^T$,
     \begin{align}
        \TD\left(\E_{\calO \sim \mathsf{sPFC}(\mathfrak{D})} \ketbra*{\mathcal{A}_t^{\calO}}_{\gsA \gsB}, \,\,\, \Tr_{\sL \sR}\left( \ketbra*{\mathcal{A}^{V, \overline{V}}_t}_{\gsA \gsB \gsL \gsR} \right) \right) \leq \frac{9t(t+1)}{N^{1/8}} + 4 t^{5/4}\varepsilon^{1/4}. \label{eq:intermediate-step-strong-main-thm}
    \end{align}
\end{lemma}
\begin{proof}
    The proof follows from the proof of Lemma~9.1 in~\cite{ma2025construct}. The right hand side of the statement of Lemma~9.4 is modified to $1-70t^2/N^{1/4} - 2\sqrt{2t\varepsilon}$ following the modification of Claim~18. From this, the right hand side of the statement of Lemma~9.5 is modified to $9t^2/N^{1/8} + t \sqrt{2\sqrt{2t\varepsilon}} \leq 9t^2/N^{1/8} + 2 t^{5/4} \varepsilon^{1/4}$. Inserting this modification and that of Lemma~25 into the proof of Lemma~9.1 in~\cite{ma2025construct} yields Lemma~\ref{lemma:pfd-cho-strong}.
\end{proof}

\begin{theorem}[$V$ is indistinguishable from a Haar-random unitary; From Theorem 8 of \cite{ma2024construct}]\label{theorem:haar-cho-strong}
     For any $t$-query oracle adversary $\mathcal{A}$ that can can query $\mathcal{O}, \mathcal{O}^\dagger, \mathcal{O}^*, \mathcal{O}^T$,
     \begin{align}
        \TD\left(\E_{\calO \sim \mu_{\mathsf{Haar}}} \ketbra*{\mathcal{A}_t^{\calO}}_{\gsA \gsB}, \,\,\, \Tr_{\sL \sR}\left( \ketbra*{\mathcal{A}^{V, \overline{V}}_t}_{\gsA \gsB \gsL \gsR} \right) \right) \leq \frac{9t(t+1)}{N^{1/8}}.
    \end{align}
\end{theorem}

To simplify the notation, we will often denote $\ket*{\mathcal{A}^{V, \overline{V}}_t}_{\gsA \gsB \gsL \gsR}$ as simply $\ket*{\mathcal{A}^{V}_t}_{\gsA \gsB \gsL \gsR}$. Similarly, we will denote $\ket*{\mathcal{A}^{W, \overline{W}, \mathfrak{D}}_t}_{\gsA \gsB \gsL \gsR \gsC \gsD}$ as simply $ \ket*{\mathcal{A}^{W, \mathfrak{D}}_t}_{\gsA \gsB \gsL \gsR \gsC \gsD}$ when appropriate.

\section{The Luby-Rackoff-Function-Clifford (LRFC) ensemble} \label{app: LRFC}

In this section we present the central random unitary construction of this work, the Luby-Rackoff-Function-Clifford (LRFC) ensemble, which only uses random unitary $2$-designs and random functions. The LRFC ensemble does not require the use of random permutations, which is useful for generating minimum depth strong random unitaries, due to the lack of known low-depth constructions of quantum-secure pseudorandom permutations. The best known construction of quantum-secure strong pseudorandom permutations is given in \cite{zhandry2016note}, which requires $\poly(n)$ circuit depth for $n$-qubit systems.

In what follows, we first define the LRFC ensemble and then proceed step-by-step through our proof that it is indistinguishable from a Haar-random unitary.

\subsection{Definition of the ensemble}

We begin by formally defining the LRFC ensemble and its key ingredients. 

\paragraph{Feistel network.} Let $n$ be the number of qubits and $N \coloneqq 2^n$. Our construction utilizes a simple variant of the Feistel network, also known as the Luby-Rackoff construction \cite{luby1988construct}. For a function $h: \{0, 1\}^{n/2} \rightarrow \{0, 1\}^{n/2}$, we define the \emph{Left} and \emph{Right} Luby-Rackoff function as follows:
\begin{align}
     \mathsf{L}_{h}(x_< \| x_>) &\coloneqq (x_< \oplus h(x_>)) \| x_>, \\
     \mathsf{R}_{h}(x_< \| x_>) &\coloneqq x_< \| (x_> \oplus h(x_<)),
\end{align}
where $x = x_< \| x_> \in \{0, 1\}^n$, and $\|$ denotes bitstring concatenation.

\paragraph{Quantum oracles.} We define the following $n$-qubit quantum oracles:
\begin{align}
    \mathcal{O}^{f} &\coloneqq \sum_{x \in \{0,1\}^n} e^{2 \pi i f(x) / 3} \ketbra{x}, \\
    \mathcal{O}^{\mathsf{L}, h_1} &\coloneqq \sum_{x \in \{0,1\}^n} \ketbra{\mathsf{L}_{h_1}(x)}{x} = \sum_{x \in \{0,1\}^n} \ketbra{(x_< \oplus h_1(x_>)) \| x_>}{x}, \\
    \mathcal{O}^{\mathsf{R}, h_2} &\coloneqq \sum_{x \in \{0,1\}^n} \ketbra{\mathsf{R}_{h_2}(x)}{x} = \sum_{x \in \{0,1\}^n} \ketbra{x_< \| (x_> \oplus h_2(x_<))}{x}.
\end{align}
These are identical to the operators $F$, $S_L$, $S_R$ defined in the main text.

\paragraph{Construction.} Let $C,D$ be two $n$-qubit random unitaries sampled independently from a unitary 2-design, such as a random Clifford circuit. A random unitary $U$ sampled from the LRFC ensemble is given by:
\begin{equation}
    U \coloneqq D \cdot \mathcal{O}^{\mathsf{R}, h_2} \cdot \mathcal{O}^{\mathsf{L}, h_1} \cdot \mathcal{O}^{f} \cdot C.
\end{equation}
Putting everything together, we have the following definition for LRFC ensemble.

\begin{definition}[LRFC ensemble] \label{def:LRFC}
    Suppose $h_1$ and $h_2$ are drawn uniformly randomly from functions on $\{0,1\}^{n/2} \rightarrow \{0,1\}^{n/2}$,  $f$ is drawn uniformly randomly from ternary functions on $\{0,1\}^n$, and $C, D$ are drawn uniformly from a unitary $2$-design on $n$ qubits.
    Then the Luby-Rackoff-Function-Clifford (LRFC) ensemble is given by the family of $n$-qubit unitaries:
    \begin{equation}
        U := \sum_{x \in \{0, 1\}^n } e^{2 \pi i f(x) / 3} \cdot D \cdot \ketbra{x_< \oplus h_1(x_>) \| x_> \oplus h_2(x_< \oplus h_1(x_>))}{x} \cdot C,
    \end{equation}
    where $\lVert$ denotes the concatenation of two $\frac{n}{2}$-bit strings and $x = x_< \lVert x_>$.
\end{definition}

\subsection{Purified Luby-Rackoff-Function oracle}

In this section, we analyze the view of an adversary that makes queries to an oracle implementing the Luby-Rackoff-Function construction with random functions $h_1$, $h_2$, and a random ternary function $f$. We will do this by analyzing the \emph{purified Luby-Rackoff-Function oracle}, which uses a purification of these random functions.

\begin{definition}[Purified Luby-Rackoff-Function oracle] 
    The purified Luby-Rackoff-Function oracle $\mathsf{lrfO}$ is a unitary acting on registers $\sA$, $\mathsf{H}_{\mathsf{1}}$, $\mathsf{H}_{\mathsf{2}}$, $\sF$, where
    \begin{itemize}
        \item $\mathsf{H}_{\mathsf{1}}$ is a register associated with the Hilbert space $\mathcal{H}_{\mathsf{H}_{\mathsf{1}}}$, defined to be the span of the orthonormal states $\ket{h_1}$ for all $h_1: \{0,1\}^{n/2} \rightarrow \{0,1\}^{n/2}$.
        \item $\mathsf{H}_{\mathsf{2}}$ is a register associated with the Hilbert space $\mathcal{H}_{\mathsf{H}_{\mathsf{2}}}$, defined to be the span of the orthonormal states $\ket{h_2}$ for all $h_2: \{0,1\}^{n/2} \rightarrow \{0,1\}^{n/2}$.
        \item $\sF$ is a register associated with the Hilbert space $\mathcal{H}_{\mathsf{F}}$, defined to be the span of the orthonormal states $\ket{f}$ for all $f: \{0,1\}^n \rightarrow \{0,1,2\}$. 
    \end{itemize}
    The unitary $\mathsf{lrfO}$ is defined to act as follows:
    \begin{align}
        &\mathsf{lrfO}_{\gsA\gsH_{\gsL}\gsH_{\gsR}\gsF} \ket{x}_{\gsA} \ket{h_1}_{\gsH_{{ \mathsf{\color{gray} 1}}}} \ket{h_2}_{\gsH_{{\mathsf{\color{gray} 2}}}} \ket{f}_{\gsF} \\
        &\coloneqq \omega_3^{f(x)} \ket{x_< \oplus h_1(x_>) \| x_> \oplus h_2(x_< \oplus h_1(x_>))}_{\gsA} \ket{h_1}_{\gsH_{\color{gray} 2}} \ket{h_2}_{\gsH_{\mathsf{\color{gray} 2}}} \ket{f}_{\gsF},
    \end{align}
    for all $x = x_{<} \| x_{>} \in \{0,1\}^n$ with $x_{<}, x_{>} \in \{0,1\}^{n/2}$, and for all functions $h_1, h_2: \{0,1\}^{n/2} \rightarrow \{0,1\}^{n/2}$, $f: \{0,1\}^n \rightarrow \{0, 1, 2\}$. Here, $\omega_3 = \exp(2 \pi i / 3)$.
\end{definition}

The action of $\lrfo^*$ is
\begin{align}
    &\lrfo^* \ket{x}_{\gsA} \ket{h_1}_{\gsH_{{ \mathsf{\color{gray} 1}}}} \ket{h_2}_{\gsH_{{\mathsf{\color{gray} 2}}}} \ket{f}_{\gsF} \\
    &= \omega_3^{-f(x)} \ket{x_< \oplus h_1(x_>) \| (x_> \oplus h_2(x_< \oplus h_1(x_>)))}_{\gsA} \ket{h_1}_{\gsH_{\color{gray} 2}} \ket{h_2}_{\gsH_{\mathsf{\color{gray} 2}}} \ket{f}_{\gsF}.
\end{align}
The action of $\lrfo^\dagger$ is
\begin{align}
    &\lrfo^\dagger \ket{y}_{\gsA} \ket{h_1}_{\gsH_{{ \mathsf{\color{gray} 1}}}} \ket{h_2}_{\gsH_{{\mathsf{\color{gray} 2}}}} \ket{f}_{\gsF} \\
    &= \omega_3^{-f(x)} \ket{ (y_< \oplus h_1(y_> \oplus h_2(y_<))) \| y_> \oplus h_2(y_<) }_{\gsA} \ket{h_1}_{\gsH_{\color{gray} 2}} \ket{h_2}_{\gsH_{\mathsf{\color{gray} 2}}} \ket{f}_{\gsF}.
\end{align}
The action of $\lrfo^T$ is
\begin{align}
    &\lrfo^T \ket{y}_{\gsA} \ket{h_1}_{\gsH_{{ \mathsf{\color{gray} 1}}}} \ket{h_2}_{\gsH_{{\mathsf{\color{gray} 2}}}} \ket{f}_{\gsF} \\
    &= \omega_3^{f(x)} \ket{ (y_< \oplus h_1(y_> \oplus h_2(y_<))) \| y_> \oplus h_2(y_<) }_{\gsA} \ket{h_1}_{\gsH_{\color{gray} 2}} \ket{h_2}_{\gsH_{\mathsf{\color{gray} 2}}} \ket{f}_{\gsF}.
\end{align}
Because $\lrfo$ is constructed by purifying the randomness inf $h_1, h_2, f$, the output state of any oracle adversary that queries the purified oracle after tracing out $\sH_1, \sH_2, \sF$ is equivalent to the output state of the adversary that queries the standard oracle $\mathcal{O}^{\mathsf{R}, h_2} \cdot \mathcal{O}^{\mathsf{L}, h_1} \cdot \mathcal{O}^{f}$, for uniformly random $h_1, h_2 \sim \{0, 1\}^{n/2 \cdot \sqrt{N}}$ and $f \sim \{0,1,2\}^N$.

\begin{fact}[Equivalence of purified and standard oracles] \label{claim:purified-vs-standard-ternary-LRFO}
    For any oracle adversary, the following oracle instantiations are perfectly indistinguishable:
    \begin{itemize}
        \item (Queries to a random $\mathcal{O}^{\mathsf{R}, h_2} \cdot \mathcal{O}^{\mathsf{L}, h_1} \cdot \mathcal{O}^{f}$) Sample a uniformly random $h_1, h_2 \sim \{0, 1\}^{n/2 \cdot \sqrt{N}}, f \sim \{0, 1, 2\}^N$. On each query, apply $U = \mathcal{O}^{\mathsf{R}, h_2} \cdot \mathcal{O}^{\mathsf{L}, h_1} \cdot \mathcal{O}^{f}, U^\dagger, U^*$ or $U^T$ to register $\sA$.
        \item (Queries to $\lrfo$) Initialize registers $\sH_1,\sH_2, \sF$ to $\frac{1}{\sqrt{N}^{\sqrt{N}}} \sum_{h_1, h_2 : \{0, 1\}^{n/2} \rightarrow \{0, 1\}^{n/2} } \ket{h_1}_{\gsH_{\color{gray} 1}} \otimes \ket{h_2}_{\gsH_{\color{gray} 2}} \otimes \frac{1}{\sqrt{3^N}} \sum_{f \in \{0,1, 2\}^N} \ket*{f}_{\gsF}$. On each query, apply $\lrfo, \lrfo^\dagger, \lrfo^*$ or $\lrfo^T$ to registers $\sA, \sH_1,\sH_2, \sF$.
    \end{itemize}
\end{fact}

Next, we define the relation states for the LRF oracle. 

\begin{definition}[$\mathsf{lrf}$-relation state]
    For the relations
    $L = \{(x_1, y_1), \dots, (x_\ell, y_\ell)\} \in \mathcal{L}_\ell$ and $R = \{(x'_1,y'_1),\dots,(x'_r,y'_r)\} \in \mathcal{R}_r$,
    where $\ell$ and $r$ are non-negative integers such  that $\ell + r \leq 2^{n/2}$, let
    \begin{align}
        \ket{\mathrm{lrf}_{L,R}}_{\gsH_{\mathsf{\color{gray} 1}} \gsH_{\mathsf{\color{gray} 2}} \gsF} &\coloneqq \frac{1}{\sqrt{\sqrt{N}^{\sqrt{N} - \ell - r}}} \sum_{h_1: \{0,1\}^{n/2} \rightarrow \{0,1\}^{n/2}} \delta_{h_1,L \cup R} \ket{h_1}_{\gsH_{\color{gray} 1}} \\
        &\otimes \frac{1}{\sqrt{\sqrt{N}^{\sqrt{N} - \ell - r}}} \sum_{h_2: \{0,1\}^{n/2} \rightarrow \{0,1\}^{n/2}} \delta'_{h_2, L \cup R} \ket{h_2}_{\gsH_{\mathsf{\color{gray} 2}}} \\
        &\otimes \frac{1}{\sqrt{3^{N}}} \sum_{f: \{0,1\}^n \rightarrow \{0,1,2\}} \omega_3^{\sum_{(x,y) \in L} f(x) - \sum_{(x',y') \in R} f(x')} \ket{f}_{\gsF}.
    \end{align}
    Here, $\delta_{h_1, L \cup R}$ is an indicator variable that equals $1$ if $y_< = (x_< \oplus h_1(x_>))$ for all $(x,y) \in L \cup R$, and $0$ otherwise; $\delta'_{h_2, L \cup R}$ equals $1$ if $(y_> \oplus h_2(y_<)) = x_>$ for all $(x,y) \in L \cup R$, and $0$ otherwise.
\end{definition}

\subsection{Projecting onto the local distinct subspace}

Our analysis of the LRFC ensemble centers upon a projection onto the \emph{local distinct subspace} on the registers $\mathsf{L}$ and $\mathsf{R}$.
In this section, we first define this subspace and then show that we can project onto this subspace (and close variants of it) throughout any quantum experiment that queries a Haar-random unitary or the LRFC ensemble.
These latter steps form the core technical results behind our proof that the LRFC ensemble is indistinguishable from a Haar-random unitary.

We define the local distinct subspace as follows.
\begin{definition}
    Let $\calR^{2,\lcdist}$ be the set of all ordered pairs of relations $(L,R) \in \calR^2$ where $L \cup R = \{(x_1,y_1),\dots,(x_t,y_t)\}$ satisfies $x_{1, >}, \ldots, x_{t, >}$ are all distinct and $y_{1, <}, \ldots, y_{t, <}$ are all distinct.
\end{definition}

The states $\ket{\mathrm{lrf}_{L,R}}$ possess several nice properties when $L$ and $R$ are locally distinct.
For example, by expanding the definition of $\ket{\mathrm{lrf}_{L,R}}$, we obtain the following facts.
\begin{fact}[Orthonormality]
\label{fact:phi_lrf-orthogonal}
    $\{\ket{\mathsf{lrf}_{L,R}}\}_{(L,R) \in \calR^{2,\lcdist}}$ forms an orthonormal set of vectors.
\end{fact}
\begin{fact}[Action of $\mathsf{lrfO}$]
\label{fact:ternary-lrfo-action}
    For any $(L,R) \in \calR^{2,\lcdist}$ and $x \in [N]$ such that $x_> \not\in \Dom_>(L \cup R)$,
    \begin{align}
        \mathsf{lrfO} \ket*{x}_{\gsA} \ket*{\mathsf{lrf}_{L,R}}_{\gsH_{\mathsf{\color{gray} 1}} \gsH_{\mathsf{\color{gray} 2}} \gsF} &= \frac{1}{\sqrt{N}} \sum_{y \in [N]} \ket*{y}_{\gsA} \ket*{\mathsf{lrf}_{L \cup \{(x, y)\}, R}}_{\gsH_{\mathsf{\color{gray} 1}} \gsH_{\mathsf{\color{gray} 2}} \gsF}, \label{eq:tlrfo-map}\\
        \mathsf{lrfO}^* \ket*{x}_{\gsA} \ket*{\mathsf{lrf}_{L,R}}_{\gsH_{\mathsf{\color{gray} 1}} \gsH_{\mathsf{\color{gray} 2}} \gsF} &= \frac{1}{\sqrt{N}} \sum_{y \in [N]} \ket*{y}_{\gsA} \ket*{\mathsf{lrf}_{L, R \cup \{(x, y)\}}}_{\gsH_{\mathsf{\color{gray} 1}} \gsH_{\mathsf{\color{gray} 2}} \gsF}. \label{eq:tlrfo-conj-map}
    \end{align}
    Similarly, for any $(L,R) \in \calR^{2,\lcdist}$ and $y \in [N]$ such that $y_< \not\in \Im_<(L \cup R)$,
    \begin{align}
        \mathsf{lrfO}^\dagger \ket*{y}_{\gsA} \ket*{\mathsf{lrf}_{L,R}}_{\gsH_{\mathsf{\color{gray} 1}} \gsH_{\mathsf{\color{gray} 2}} \gsF} &= \frac{1}{\sqrt{N}} \sum_{x \in [N]} \ket*{x}_{\gsA} \ket*{\mathsf{lrf}_{L, R \cup \{(x, y)\}}}_{\gsH_{\mathsf{\color{gray} 1}} \gsH_{\mathsf{\color{gray} 2}} \gsF}, \label{eq:tlrfo-inverse-map}\\
        \mathsf{lrfO}^T \ket*{y}_{\gsA} \ket*{\mathsf{lrf}_{L,R}}_{\gsH_{\mathsf{\color{gray} 1}} \gsH_{\mathsf{\color{gray} 2}} \gsF} &= \frac{1}{\sqrt{N}} \sum_{x \in [N]} \ket*{x}_{\gsA} \ket*{\mathsf{lrf}_{L \cup \{(x, y)\}, R}}_{\gsH_{\mathsf{\color{gray} 1}} \gsH_{\mathsf{\color{gray} 2}} \gsF}. \label{eq:tlrfo-T-map}
    \end{align}
\end{fact}

We can also define a partial isometry between the states $\ket{\mathrm{lrf}_{L,R}}$, and the states $\ket{L} \otimes \ket{R}$, when $L$ and $R$ are locally distinct.
\begin{definition}
    Define the partial isometry $\Compress_{\mathsf{LRF}}: \mathcal{H}_{\mathsf{H}_{\mathsf{1}}} \otimes \mathcal{H}_{\mathsf{H}_{\mathsf{2}}} \otimes \mathcal{H}_{\mathsf{F}} \rightarrow \calH_{\sL} \otimes \calH_{\sR}$ to be
    \begin{align}
        \Compress_{\mathsf{LRF}} \coloneqq \sum_{(L,R) \in \calR^{2,\lcdist}} \ket*{L}_{\gsL} \otimes \ket*{R}_{\gsR} \cdot \bra{\mathsf{lrf}_{L, R}}_{\gsH_{\mathsf{\color{gray} 1}} \gsH_{\mathsf{\color{gray} 2}} \gsF}.
    \end{align}
\end{definition}

\noindent Note that $\Compress$ is a partial isometry by~\cref{fact:phi_lrf-orthogonal}.

For our later analysis, it will be convenient to define projectors that keep one in the locally distinct subspace.
We do so as follows.
First, we recall the definition of the projector onto  \emph{bijective} relation states,
\begin{equation}
    \Pi^{\mathsf{bij}}_{\gsL \gsR} 
    \ket{ L }_{\gsL}
   \ket{ R }_{\gsR} 
   = 
   \begin{cases}
       \ket{ L }_{\gsL} \ket{ R }_{\gsR}, & \text{if } \mathsf{Dom}(L \cup R) \in \mathsf{ distinct }, \\
       & \text{and } \mathsf{Im}(L \cup R) \in \mathsf{ distinct } \\
       0, & \text{else}. \\
   \end{cases}
\end{equation}
%
In a similar fashion, we define the not-in-domain projector on $\mathsf{ALR}$ as,
\begin{equation}
    \Pi^{\notin \mathsf{Dom}}_{\gsA \gsL \gsR} 
    \ket{ x }_{\gsA}
    \ket{ L }_{\gsL}
   \ket{ R }_{\gsR} 
   = 
   \begin{cases}
       \ket{ x }_{\gsA}\ket{ L }_{\gsL} \ket{ R }_{\gsR}, & \text{if } x \notin \mathsf{Dom}(L \cup R), \\
       0, & \text{else}, \\
   \end{cases}
\end{equation}
and the not-in-image projector as,
\begin{equation}
    \Pi^{\notin \mathsf{Im}}_{\gsA \gsL \gsR} 
    \ket{ y }_{\gsA}
    \ket{ L }_{\gsL}
   \ket{ R }_{\gsR} 
   = 
   \begin{cases}
       \ket{ y }_{\gsA}\ket{ L }_{\gsL} \ket{ R }_{\gsR}, & \text{if } y \notin \mathsf{Im}(L \cup R), \\
       0, & \text{else}. \\
   \end{cases}
\end{equation}
Turning to the local distinct projectors, we define the \emph{locally bijective} relation states via the projector,
\begin{equation}
    \Pi^{\mathsf{locbij}}_{\gsL \gsR} 
    \ket{ L }_{\gsL}
   \ket{ R }_{\gsR} 
   = 
   \begin{cases}
       \ket{ L }_{\gsL} \ket{ R }_{\gsR}, & \text{if } \mathsf{Dom}_>(L \cup R) \in \mathsf{ distinct }, \\
       & \text{and } \mathsf{Im}_<(L \cup R) \in \mathsf{ distinct } \\
       0, & \text{else}, \\
   \end{cases}
\end{equation}
and the not-in-local-domain and not-in-local-image projectors as,
\begin{align}
    \Pi^{\notin \mathsf{locDom}}_{\gsA \gsL \gsR} 
    \ket{ x }_{\gsA}
    \ket{ L }_{\gsL}
   \ket{ R }_{\gsR} 
   & = 
   \begin{cases}
       \ket{ x }_{\gsA}\ket{ L }_{\gsL} \ket{ R }_{\gsR}, & \text{if } x_> \notin \mathsf{Dom}_>(L \cup R), \\
       0, & \text{else}, \\
   \end{cases} \\
    \Pi^{\notin \mathsf{locIm}}_{\gsA \gsL \gsR} 
    \ket{ y }_{\gsA}
    \ket{ L }_{\gsL}
   \ket{ R }_{\gsR} 
   & = 
   \begin{cases}
       \ket{ y }_{\gsA}\ket{ L }_{\gsL} \ket{ R }_{\gsR}, & \text{if } y_< \notin \mathsf{Im}_<(L \cup R), \\
       0, & \text{else}. \\
   \end{cases}
\end{align}
We also let,
\begin{align}
    \Pi^{\notin \mathsf{Dom} \rightarrow \notin \mathsf{locDom}}_{\gsA \gsL \gsR} & = 
    1 
    - \Pi^{\notin \mathsf{Dom}}_{\gsA \gsL \gsR}
    + \Pi^{\notin \mathsf{locDom}}_{\gsA \gsL \gsR}  \\
     \Pi^{\notin \mathsf{Im} \rightarrow \notin \mathsf{locIm}}_{\gsA \gsL \gsR} & = 
    1 
    - \Pi^{\notin \mathsf{Im}}_{\gsA \gsL \gsR}
    + \Pi^{\notin \mathsf{locIm}}_{\gsA \gsL \gsR},
\end{align}
denote projectors which do nothing if $x \in \mathsf{Dom}(L\cup R)$, and project to the not-in-local-domain subspace if $x \notin \mathsf{Dom}(L\cup R)$ (and similar for $\mathsf{Im}$).

\subsubsection{Twirled $W$ is indistinguishable from twirled projected $W$}

In this and the following two sections, we use the local distinct subspace projectors can be inserted, up to small error, throughout any quantum experiment that queries the LRFC ensemble or a Haar-random unitary.
We begin in this section by analyzing the insertion of the local distinct subspace projectors for experiments that involve the path-recording oracles $W$ and $\overline{W}$.

We define projected versions of the path-recording oracles $W$ and $\overline{W}$ as follows,
\begin{align}
    W'
    & \equiv 
    \Pi^{\mathsf{locbij}}
    \cdot W \cdot
    \Pi^{\notin \mathsf{Dom} \rightarrow \notin \mathsf{locDom}} \cdot \Pi^{\mathsf{locbij}} \\
    (W^\dagger)'
    & \equiv 
    \Pi^{\mathsf{locbij}}
    \cdot W^\dagger \cdot
    \Pi^{\notin \mathsf{Im} \rightarrow \notin \mathsf{locIm}} \cdot \Pi^{\mathsf{locbij}} \\
    \overline{W}'
    & \equiv 
    \Pi^{\mathsf{locbij}}
    \cdot \overline{W} \cdot
    \Pi^{\notin \mathsf{Dom} \rightarrow \notin \mathsf{locDom}} \cdot \Pi^{\mathsf{locbij}} \\
    (\overline{W}^\dagger)'
    & \equiv 
    \Pi^{\mathsf{locbij}}
    \cdot \overline{W}^\dagger \cdot
    \Pi^{\notin \mathsf{Im} \rightarrow \notin \mathsf{locIm}} \cdot \Pi^{\mathsf{locbij}}.
\end{align}
In the remainder of this section, we show that the projected oracles are indistinguishable from the original oracles whenever the oracles are surrounded by random unitary 2-designs.
\begin{lemma}
[Twirled $W$ is indistinguishable from twirled  $W'$] \label{lem:closeness-AWD-and-AWpD}
Let $\mathfrak{D}$ be any strong approximate unitary 2-design with additive error $\varepsilon$. For any $t$-query oracle adversary $\mathcal{A}$, we have
\begin{equation} \label{eq: TDt 1}
    \mathsf{TD}\left( 
      \dyad*{ \mathcal{A}^{W, \mathfrak{D}}_t }_{\gsA \gsL \gsR \gsC \gsD}
    , 
    \dyad*{ \mathcal{A}^{W', \mathfrak{D}}_t }_{\gsA \gsL \gsR \gsC \gsD}
    \right)
    \leq 
    \frac{\sqrt{70} t(t-1)}{N^{1/8}} + \frac{2 t(t+1)}{N^{1/4}} + 4 t^{5/4} \varepsilon^{1/4}.
\end{equation}
\end{lemma}
\noindent Here and in the remainder of the manuscript, we abbreviate $\ket*{ \mathcal{A}^{W,\overline{W}, \mathfrak{D}}_t }$ as simply $\ket*{ \mathcal{A}^{W, \mathfrak{D}}_t }$.
All of our analysis includes queries to the conjugate and transpose.

\begin{proof}
Let $\mathsf{TD}_t$ denote the trace distance in Eq.~(\ref{eq: TDt 1}). We will prove the theorem by induction.
The statement holds trivially at $t=0$.
To prove the inductive step, suppose that the Eq.~(\ref{eq: TDt 1}) holds up to time $t-1$ for any $t \geq 1$.
Without loss of generality, we assume that the oracle $W$ is applied at time $t$.
The case when $W^\dagger$, $\overline{W}$, and $\overline{W}^\dagger$ are applied follow by symmetric arguments.
The states at time $t$ are obtained from the states at time $t-1$ as follows,
\begin{align}
    \ket*{ \mathcal{A}^{W, \mathfrak{D}}_{t} }
    & =
    \mathsf{cD}
    \cdot
    W
    \cdot
    \mathsf{cC}
    \cdot
    A_{t}
    \cdot
    \ket*{ \mathcal{A}^{W, \mathfrak{D}}_{t-1} } \\
    \ket*{ \mathcal{A}^{W', \mathfrak{D}}_{t} }
    & =
    \mathsf{cD}
    \cdot
    \Pi^{\mathsf{locbij}}
    \cdot W \cdot
    \Pi^{\notin \mathsf{Dom} \rightarrow \notin \mathsf{locDom}}
    \cdot
    \mathsf{cC}
    \cdot
    A_{t}
    \cdot
    \ket*{ \mathcal{A}^{W', \mathfrak{D}}_{t-1} }
\end{align}
In the second line, we use that $\Pi^{\mathsf{locbij}} \ket*{\mathcal{A}^{W',\mathfrak{D}}_{t-1}} = \ket*{\mathcal{A}^{W',\mathfrak{D}}_{t-1}}$ to eliminate the final projector in $W'$. 
This follows because the projector $\Pi^{\mathsf{locbij}}$ has already been applied after the prior $(t-1)$-th query.
We have 
\begin{equation} \label{eq: step 1 1}
\begin{split}
    \mathsf{TD}_{t} \leq \mathsf{TD}_{t-1} 
    & + 2 \left\lVert \big( 1- \Pi^{\notin \mathsf{Dom} \rightarrow \notin \mathsf{locDom}} \big)
    \cdot \mathsf{cC} \cdot A_{t} 
    \ket*{ \mathcal{A}^{W, \mathfrak{D}}_{t-1} }
     \right\rVert_2 \\
    & \quad + 2 \left\lVert
    \big( 1 - \Pi^{\mathsf{locbij}} \big)
    \cdot W \cdot
    \Pi^{\notin \mathsf{Dom} \rightarrow \notin \mathsf{locDom}}
    \cdot
    \mathsf{cC}
    \cdot
    A_{t}
    \cdot
    \ket*{ \mathcal{A}^{W', \mathfrak{D}}_{t-1} }
    \right\rVert_2, \\
\end{split}
\end{equation}
where the first term accounts for the error up to time $t-1$, the second term for the error induced by the projector $\Pi^{\notin \mathsf{Dom} \rightarrow \notin \mathsf{locDom}}$ [using Eq.~(\ref{eq: state to op bound})], and the third term for the error induced by the projector $\Pi^{\mathsf{locbij}}$ [again using Eq.~(\ref{eq: state to op bound})].

We bound the second term as follows.
From Eq.~(9.55) and Claim 18 of Ref.~\cite{ma2024construct} (see also the modification of Claim 18 to strong approximate designs in the proof of Lemma~\ref{lem:closeness-AWD-and-PhiVt}), we have
\begin{equation}
    \left\lVert 
    \ket*{ \mathcal{A}^{W, \mathfrak{D}}_{t-1} }
    -
    \mathsf{cQ} \cdot \ket*{ \mathcal{A}^{V}_{t-1} }
    \right\rVert_2 \leq \frac{\sqrt{70}(t-1)}{N^{1/8}} + 2 t^{1/4} \varepsilon^{1/4}.
\end{equation}
This yields,
\begin{equation}
\begin{split}
    & \left\lVert \big( 1- \Pi^{\notin \mathsf{Dom} \rightarrow \notin \mathsf{locDom}} \big)
    \cdot \mathsf{cC} \cdot A_{t} \cdot
    \ket*{ \mathcal{A}^{W, \mathfrak{D}}_{t-1} }
     \right\rVert_2  \\
     & \quad \quad \quad  \quad \leq 
     \left\lVert \big( 1- \Pi^{\notin \mathsf{Dom} \rightarrow \notin \mathsf{locDom}} \big) \cdot
    \mathsf{cC} \cdot A_{t} \cdot
    \mathsf{cQ} \cdot \ket*{ \mathcal{A}^{V}_{t-1} }
     \right\rVert_2
     +
     \frac{\sqrt{70}(t-1)}{N^{1/8}}  + 2 t^{1/4} \varepsilon^{1/4}.
\end{split}
\end{equation}
The latter state norm can be written out explicitly, as
\begin{equation}
\begin{split}
    & \left\lVert \big( 1- \Pi^{\notin \mathsf{Dom} \rightarrow \notin \mathsf{locDom}} \big) \cdot
    \mathsf{cC} \cdot A_{t} \cdot
    \mathsf{cQ} \cdot \ket*{ \mathcal{A}^{V}_{t-1} }
     \right\rVert_2  \\
     & \quad \quad  \quad \quad \quad  \quad =
     \sqrt{
     \bra*{ \mathcal{A}^{V}_{t-1} }
      \cdot A_{t}^\dagger \cdot 
     \mathsf{cQ}^\dagger \cdot 
     \mathsf{cC}^\dagger \cdot 
     \big( 1- \Pi^{\notin \mathsf{Dom} \rightarrow \notin \mathsf{locDom}} \big) 
     \cdot \mathsf{cC} \cdot 
    \mathsf{cQ}   \cdot A_{t} \cdot  \ket*{ \mathcal{A}^{V}_{t-1} }
     }.
\end{split}
\end{equation}
where we used that $A_{t}$ and $\mathsf{cQ}$ act on distinct registers to commute them past one another.

To proceed, we first apply the operator inequality,
\begin{equation}
    1 - \Pi^{\notin \mathsf{Dom} \rightarrow \notin \mathsf{locDom}}  \preceq 
    \sum_\ell \sum_{i \in [\ell]} 
    \Pi^{\mathsf{eq}}_{\mathsf{A}_{>} \mathsf{L}_{\mathsf{X_>, i}}^{(\ell)}}
    \Pi^{\mathsf{neq}}_{\mathsf{A}_{<} \mathsf{L}_{\mathsf{X_{<}, i}}^{(\ell)}}
    +
    \sum_r \sum_{j \in [r]} 
    \Pi^{\mathsf{eq}}_{\mathsf{A}_{>} \mathsf{R}_{\mathsf{X_>, j}}^{(r)}}
    \Pi^{\mathsf{neq}}_{\mathsf{A}_{<} \mathsf{R}_{\mathsf{X_{<}, j}}^{(r)}},
\end{equation} 
where $\Pi^{\mathsf{eq}}_{\mathsf{A}_{>} \mathsf{L}_{\mathsf{X_>, i}}^{(\ell)}}$ projects onto states with the same bitstring on $\mathsf{A}_>$ as on $\mathsf{L}_{\mathsf{X_{>}, i}}^{(\ell)}$, and 
\begin{equation}
    \Pi^{\mathsf{neq}}_{\mathsf{A}_{<} \mathsf{L}_{\mathsf{X_{<}, i}}^{(\ell)}} \equiv 1 - \Pi^{\mathsf{eq}}_{\mathsf{A}_{<} \mathsf{L}_{\mathsf{X_{<}, i}}^{(\ell)}}
\end{equation} 
does the reverse on $<$.
For each individual term with $i \in [\ell]$, we have
\begin{equation} \label{eq: EV eq neq 1 1}
\begin{split}
    & \bra*{ \mathcal{A}^{V}_{t-1} }
      \cdot A_{t}^\dagger \cdot 
     \mathsf{cQ}_{\gsC \gsD \gsL \gsR}^\dagger \cdot 
     \mathsf{cC}_{\gsC \gsA}^\dagger \cdot 
     \Pi^{\mathsf{eq}}_{\mathsf{A}_{>} \mathsf{L}_{\mathsf{X_>, i}}^{(\ell)}}
    \Pi^{\mathsf{neq}}_{\mathsf{A}_{<} \mathsf{L}_{\mathsf{X_{<}, i}}^{(\ell)}}
     \cdot \mathsf{cC}_{\gsC \gsA} \cdot 
    \mathsf{cQ}_{\gsC \gsD \gsL \gsR}   \cdot A_{t} \cdot  \ket*{ \mathcal{A}^{V}_t }  \\
     & \quad \quad  \quad \quad  \quad =
     \bra*{ \mathcal{A}^{V}_{t-1} }
      \cdot A_{t}^\dagger \cdot 
     \mathsf{cC}_{\gsC 
    {\textcolor{gray}{\mathsf{L}_{\mathsf{X_{<}, i}}^{(\ell)}}}}^\dagger \cdot 
     \mathsf{cC}_{\gsC \gsA}^\dagger \cdot 
     \Pi^{\mathsf{eq}}_{\mathsf{A}_{>} \mathsf{L}_{\mathsf{X_>, i}}^{(\ell)}}
    \Pi^{\mathsf{neq}}_{\mathsf{A}_{<} \mathsf{L}_{\mathsf{X_{<}, i}}^{(\ell)}} 
     \cdot 
     \mathsf{cC}_{\gsC \gsA} \cdot 
    \mathsf{cC}_{\gsC 
    {\textcolor{gray}{\mathsf{L}_{\mathsf{X_{<}, i}}^{(\ell)}}}}
    \cdot A_{t} \cdot  \ket*{ \mathcal{A}^{V}_{t-1} },
\end{split}
\end{equation}
where all but one of the Clifford unitaries in $\mathsf{cQ}$ cancel, since the middle term in the expectation value acts only on register $\mathsf{L}_{\mathsf{X_>, i}}$. 
For terms with $j \in [r]$, we have instead
\begin{equation} \label{eq: EV eq neq 2 1}
\begin{split}
    & \bra*{ \mathcal{A}^{V}_{t-1} }
      \cdot A_{t}^\dagger \cdot 
     \mathsf{cQ}_{\gsC \gsD \gsL \gsR}^\dagger \cdot 
     \mathsf{cC}_{\gsC \gsA}^\dagger \cdot 
     \Pi^{\mathsf{eq}}_{\mathsf{A}_{>} \mathsf{R}_{\mathsf{X_>, j}}^{(r)}}
    \Pi^{\mathsf{neq}}_{\mathsf{A}_{<} \mathsf{R}_{\mathsf{X_{<}, j}}^{(r)}}
     \cdot \mathsf{cC}_{\gsC \gsA} \cdot 
    \mathsf{cQ}_{\gsC \gsD \gsL \gsR}   \cdot A_{t} å\cdot  \ket*{ \mathcal{A}^{V}_{t-1} }  \\
     & \quad \quad  \quad \quad  \quad =
     \bra*{ \mathcal{A}^{V}_{t-1} }
      \cdot A_{t}^\dagger \cdot 
     \mathsf{c}\overline{\mathsf{C}}_{\gsC 
    {\textcolor{gray}{\mathsf{R}_{\mathsf{X_{<}, j}}^{(r)}}}}^\dagger \cdot 
     \mathsf{cC}_{\gsC \gsA}^\dagger \cdot 
     \Pi^{\mathsf{eq}}_{\mathsf{A}_{>} \mathsf{R}_{\mathsf{X_>, j}}^{(r)}}
    \Pi^{\mathsf{neq}}_{\mathsf{A}_{<} \mathsf{R}_{\mathsf{X_{<}, j}}^{(r)}} 
     \cdot 
     \mathsf{cC}_{\gsC \gsA} \cdot 
    \mathsf{c}\overline{\mathsf{C}}_{\gsC 
    {\textcolor{gray}{\mathsf{R}_{\mathsf{X_{<}, j}}^{(r)}}}}
    \cdot A_{t} \cdot  \ket*{ \mathcal{A}^{V}_{t-1} }.
\end{split}
\end{equation}
We can upper bound the latter expectation values by performing the twirl over $C$.
From Eq.~(\ref{eq: clifford twirl eq neq 1 approx}) and Eq.~(\ref{eq: clifford twirl eq neq 2 approx}), this yields an upper bound of $1/N^{1/2}+\varepsilon$ on both Eq.~(\ref{eq: EV eq neq 1 1}) and Eq.~(\ref{eq: EV eq neq 2 1}).
Therefore, in total, we have an upper bound
\begin{equation}
    \left\lVert \big( 1- \Pi^{\notin \mathsf{Dom} \rightarrow \notin \mathsf{locDom}} \big) \cdot
    \mathsf{cC} \cdot A_{t} \cdot
    \mathsf{cQ} \cdot \ket*{ \mathcal{A}^{V}_{t-1} }
     \right\rVert_2
     \leq 
     \sqrt{ (\ell+r)^2 (1/N^{1/2}+\varepsilon)}
     \leq t/N^{1/4} + t \varepsilon^{1/2}.
\end{equation}

The third term in Eq.~(\ref{eq: step 1 1}) is simpler to bound.
The input state to $W$ lies in the subspace $\Pi^{\mathsf{locbij}}_{\leq t}$ by construction.
Therefore, the output of $W$ lies in the subspace $\Pi^{\mathcal{I}( W \Pi^{\mathsf{locbij}}_{\leq t})}$.
This latter subspace is spanned by two classes of states.
The first class is,
\begin{equation}
    \ket{ y }_{\gsA}
    \ket{ L }_{\gsL}
    \ket{ R }_{\gsR},
\end{equation}
for any $\ell+r \leq t$, where $\mathsf{Dom}_>(L \cup R)$ is distinct, $\mathsf{Im}_<(L \cup R)$ is distinct, and $y_< \notin \mathsf{Im}_<(L \cup R)$.
 These arise if the $W^{R,\dagger}$ branch of $W$ is applied. 
 The second class is,
\begin{equation}
    \frac{1}{\sqrt{N-\ell-r}}
    \sum_{y \notin \mathsf{Im}(L\cup R)}
    \ket{ y }_{\gsA}
    \ket{ L \cup ( x , y ) }_{\gsL}
    \ket{ R }_{\gsR},
\end{equation}
for $\ell+r \leq t$,  where $\mathsf{Dom}_>(L \cup R)$ is distinct, $\mathsf{Im}_<(L \cup R)$ is distinct, and $x_> \notin \mathsf{Dom}_>(L \cup R)$.
These arise if the $W^L$ branch of $W$ is applied.
The states above are mutually orthogonal to one another as well as between different $\ell,r$.

The first class of states is invariant under $\Pi^{\mathsf{locbij}}$.
Thus, the projector $\Pi^{\mathsf{locbij}}$ acts trivially and incurs no error.
Meanwhile, on the second class of states, we have
\begin{equation}
\begin{split}
    & \Pi^{\mathsf{locbij}}
    \frac{1}{\sqrt{N-\ell-r}}
    \sum_{y \notin \mathsf{Im}(L\cup R)}
    \ket{ y }_{\gsA}
    \ket{ L \cup ( x , y ) }_{\gsL}
    \ket{ R }_{\gsR} \\
    & \quad \quad \quad \quad \quad \quad \quad = 
    \frac{1}{\sqrt{N-\ell-r}}
    \sum_{y_< \notin \mathsf{Im}_<(L\cup R)}
    \ket{ y }_{\gsA}
    \ket{ L \cup ( x , y ) }_{\gsL}
    \ket{ R }_{\gsR}.
\end{split}
\end{equation}
The final state is orthogonal to the first class of states, as well as between different $\ell,r$.
The state has norm $( N^{1/2}(N^{1/2} - \ell - r) ) / (N-\ell-r) \geq 1 - t / N^{1/2}$.
The above analysis establishes that $\Pi^{\mathsf{locbij}} \Pi^{\mathcal{I}( W \Pi^{\mathsf{locbij}}_{\leq t})}$ is block diagonal between the two classes of input and output states, as well as between different $\ell, r$.
Therefore, the desired error is given by the maximum error within each block.
From the above, the maximum is achieved at $\ell + r = t$, which yields,
\begin{equation}
    \left\lVert \big( 1 - \Pi^{\mathsf{locbij}} \big)
    \cdot W \cdot
    \Pi^{\notin \mathsf{Dom} \rightarrow \notin \mathsf{locDom}}
    \cdot
    \mathsf{cC}
    \cdot
    A_{t}
    \cdot
    \ket*{ \mathcal{A}^{\overline{W}, \mathfrak{D}}_{t-1} }
    \right\rVert_2 \leq \sqrt{t/N^{1/2}}.
\end{equation}

In total, we have shown that the error in Eq.~(\ref{eq: step 1 1}) is upper bounded by,
\begin{equation} \nonumber
    \mathsf{TD}_{t} \leq \mathsf{TD}_{t-1} 
    + \frac{2\sqrt{70}(t-1)}{N^{1/8}} + 4 t^{1/4} \varepsilon^{1/4} + 2 t/N^{1/4} + 2 \sqrt{t/N^{1/2}}
    \leq
    \frac{\sqrt{70} t(t-1)}{N^{1/8}} + 4 t^{5/4} \varepsilon^{1/4} + \frac{2 t(t+1)}{N^{1/4}}
    , \\
\end{equation}
applying the inductive hypothesis.
This completes our proof.
\end{proof}

\subsubsection{Projected $W$ is indistinguishable from projected $\mathsf{lrfO}$}

We will now show that the path-recording oracle $W$ and the LRF oracle $\lrfo$ are equal on the locally distinct subspace.
To show this, let us adopt the general notation,
\begin{align}
    \widetilde{\Pi} & \equiv \mathsf{Compress}_{\mathsf{LRF}}^\dagger \cdot \Pi \cdot \mathsf{Compress}_{\mathsf{LRF}}
\end{align}
for any projector $\Pi$ on $\mathsf{ALR}$.
We will also let $\widetilde{\Pi}^{\calI(\lrfo \Pi^{\notin \mathsf{locDom}} \Pi^{\mathsf{locbij}})}$ denote the projector onto the subspace spanned by states of the form Eq.~(\ref{eq:tlrfo-map}), $\widetilde{\Pi}^{\calI(\lrfo^* \Pi^{\notin \mathsf{locDom}} \Pi^{\mathsf{locbij}})}$ analogously for Eq.~(\ref{eq:tlrfo-conj-map}), $\widetilde{\Pi}^{\calI(\lrfo^\dagger \Pi^{\notin \mathsf{locIm}} \Pi^{\mathsf{locbij}})}$ for Eq.~(\ref{eq:tlrfo-inverse-map}), and $\widetilde{\Pi}^{\calI(\lrfo^T \Pi^{\notin \mathsf{locIm}} \Pi^{\mathsf{locbij}})}$ for Eq.~(\ref{eq:tlrfo-T-map}).
As indicated in the notation, these project onto the image of the $\lrfo$ oracles when they are applied to locally distinct states.

Leveraging these projectors, we can define projected versions of the $\lrfo$ oracle as follows,
\begin{align}
    \lrfo'
    & \equiv 
    \widetilde{\Pi}^{\mathsf{locbij}}
    \cdot \lrfo \cdot
    \left( \widetilde{\Pi}^{\notin \mathsf{locDom}}  + \widetilde{\Pi}^{\calI(\lrfo^\dagger \Pi^{\notin \mathsf{locIm}}  \Pi^{\mathsf{locbij}})} \cdot \widetilde{\Pi}^{\calD(W_R^\dagger)}  \right) \cdot \widetilde{\Pi}^{\mathsf{locbij}} \label{eq:lrfop} \\
    (\lrfo^\dagger)'
    & \equiv 
    \widetilde{\Pi}^{\mathsf{locbij}}
    \cdot \lrfo^\dagger \cdot
    \left( \widetilde{\Pi}^{\notin \mathsf{locIm}}  + \widetilde{\Pi}^{\calD(\lrfo \Pi^{\notin \mathsf{locDom}}  \Pi^{\mathsf{locbij}})} \cdot \widetilde{\Pi}^{\calD(W_L^\dagger)}  \right) \cdot \widetilde{\Pi}^{\mathsf{locbij}} \\
    (\lrfo^*)'
    & \equiv 
    \widetilde{\Pi}^{\mathsf{locbij}}
    \cdot \lrfo^* \cdot
    \left( \widetilde{\Pi}^{\notin \mathsf{locDom}}  + \widetilde{\Pi}^{\calI(\lrfo^T \Pi^{\notin \mathsf{locIm}}  \Pi^{\mathsf{locbij}})} \cdot \widetilde{\Pi}^{\calD(\overline{W}_R^\dagger)}  \right) \cdot \widetilde{\Pi}^{\mathsf{locbij}} \\
    (\lrfo^T)'
    & \equiv 
    \widetilde{\Pi}^{\mathsf{locbij}}
    \cdot \lrfo^T \cdot
    \left( \widetilde{\Pi}^{\notin \mathsf{locIm}}  + \widetilde{\Pi}^{\calD(\lrfo^* \Pi^{\notin \mathsf{locDom}}  \Pi^{\mathsf{locbij}})} \cdot \widetilde{\Pi}^{\calD(\overline{W}_L^\dagger)}  \right) \cdot \widetilde{\Pi}^{\mathsf{locbij}}.
\end{align}
The first key result of this section is that $W'$ and $\lrfo'$ are nearly equal up to the compress isometry.
\begin{lemma}[$W'$ and $\lrfo'$ are nearly equal up to isometry]
\label{claim:relate-W-and-lrfo}
We have
    \begin{align}
        \left\lVert \Pi_{\leq t} \left( W' - \Compress_{\mathsf{LRF}} \cdot \lrfo' \cdot \Compress^\dagger_{\mathsf{LRF}} \right) \Pi_{\leq t}\right\rVert_\infty & \leq t/N,\label{eq:compress-proof-spfo-goal-1-x}\\
        \left\lVert \Pi_{\leq t} \left( (W^\dagger)' - \Compress_{\mathsf{LRF}} \cdot (\lrfo^\dagger)' \cdot \Compress^\dagger_{\mathsf{LRF}}  \right) \Pi_{\leq t} \right\rVert_\infty & \leq t/N \label{eq:compress-proof-spfo-goal-2-x}\\
        \left\lVert \Pi_{\leq t} \left( \overline{W}' - \Compress_{\mathsf{LRF}} \cdot (\lrfo^*)' \cdot \Compress^\dagger_{\mathsf{LRF}} \right) \Pi_{\leq t} \right\rVert_\infty & \leq t/N\label{eq:compress-proof-spfo-goal-1-x-conj}\\
        \left\lVert \Pi_{\leq t} \left( (\overline{W}^\dagger)' - \Compress_{\mathsf{LRF}} \cdot (\lrfo^T)' \cdot \Compress^\dagger_{\mathsf{LRF}} \right) \Pi_{\leq t} \right\rVert_\infty & \leq t/N. \label{eq:compress-proof-spfo-goal-2-x-conj}
    \end{align}
\end{lemma}
\begin{proof}
We focus on the first equality without loss of generality. The remaining three equalities following by symmetric arguments.
The $W'$ and $\lrfo'$ oracles act on states in two domains, corresponding to the two terms in parentheses in Eq.~(\ref{eq:lrfop}).
The first is the domain of $\widetilde{\Pi}^{\notin \mathsf{locDom}}  \widetilde{\Pi}^{\mathsf{locbij}}$.
For states in this domain, $\lrfo'$ acts as
\begin{align}
    \lrfo' \ket*{x}_{\gsA} \ket*{\mathsf{lrf}_{L,R}}_{\gsH_{\mathsf{\color{gray} 1}} \gsH_{\mathsf{\color{gray} 2}} \gsF} &= \frac{1}{\sqrt{N}} \sum_{y \in [N]} \delta_{y_< \notin \mathsf{Im}_<(L \cup R)} \ket*{y}_{\gsA} \ket*{\mathsf{lrf}_{L \cup \{(x, y)\}, R}}_{\gsH_{\mathsf{\color{gray} 1}} \gsH_{\mathsf{\color{gray} 2}} \gsF}, \label{eq:tlrfo-map-prime}
\end{align}
where $(L,R) \in \mathcal{R}^{2,\text{lcdist}}$ and $x_> \notin \mathsf{Dom}_>(L \cup R)$. Meanwhile, on the un-compressed versions of the same states, $W'$ acts as
\begin{align}
    W' \ket*{x}_{\gsA} \ket*{L}_{\gsL} \ket*{R}_{\gsL} &= \frac{1}{\sqrt{N-\ell-r}} \sum_{y \notin \mathsf{Im}(L\cup R)} \delta_{y_< \notin \mathsf{Im}_<(L \cup R)} \ket*{y}_{\gsA} \ket*{L \cup \{(x, y)\}}_\gsL \ket{R}_\gsR,
\end{align}
After compression by $\mathsf{Compress}_{\mathsf{LRF}}$, the actions of the two oracles are identical aside from a normalization ratio of $\sqrt{1-(\ell+r)/N}$.

The second is the domain of $\widetilde{\Pi}^{\calD(W_R^\dagger)}  \widetilde{\Pi}^{\mathsf{locbij}}$.
This domain is spanned by states of the form,
\begin{align}
    \sum_{x \notin \mathsf{Dom}(L \cup R)} \delta_{x_> \notin \mathsf{Dom}_>(L \cup R)} \ket*{x}_{\gsA} \ket*{\mathsf{lrf}_{L, R \cup \{(x, y)\}}}_{\gsH_{\mathsf{\color{gray} 1}} \gsH_{\mathsf{\color{gray} 2}} \gsF},
\end{align}
where $(L,R) \in \mathcal{R}^{2,\text{lcdist}}$ and $y_< \notin \mathsf{Im}_<(L \cup R)$, and we leave the state un-normalized for brevity.
For states in this domain, $\lrfo'$ acts as
\begin{align}
    \lrfo' \sum_{x \notin \mathsf{Dom}(L \cup R)} & \delta_{x_> \notin \mathsf{Dom}_>(L \cup R)} \ket*{x}_{\gsA} \ket*{\mathsf{lrf}_{L, R \cup \{(x, y)\}}}_{\gsH_{\mathsf{\color{gray} 1}} \gsH_{\mathsf{\color{gray} 2}} \gsF} \\
    & = N^{1/2} \frac{N-\ell-r}{N} \frac{N^{1/2}(N^{1/2}-\ell-r)}{N-\ell-r} \ket*{y}_{\gsA} \ket*{\mathsf{lrf}_{L, R}}_{\gsH_{\mathsf{\color{gray} 1}} \gsH_{\mathsf{\color{gray} 2}} \gsF},
\end{align}
where the second ratio arises from the action of $\widetilde{\Pi}^{\calD(W_R^\dagger)}$, the first ratio arises from the action of $\widetilde{\Pi}^{\calI(\lrfo^\dagger \Pi^{\notin \mathsf{locIm}}  \Pi^{\mathsf{locbij}})}$, and then factor of $N^{1/2}$ arises from applying Eq.~(\ref{eq:tlrfo-inverse-map}).
Meanwhile, on the un-compressed versions of the same states, $W'$ acts as
\begin{align}
    W' \sum_{x \notin \mathsf{Dom}(L \cup R)} & \delta_{x_> \notin \mathsf{Dom}_>(L \cup R)} \ket*{x}_{\gsA} \ket*{L}_{\gsL} \ket*{R \cup \{(x, y)\}}_{\gsR} \\
    & = (N-\ell-r)^{1/2} \frac{N^{1/2}(N^{1/2}-\ell-r)}{N-\ell-r} \ket*{y}_{\gsA} \ket*{L}_\gsL \ket{R}_\gsR, 
\end{align}
After compression by $\mathsf{Compress}_{\mathsf{LRF}}$, the actions of the two oracles are identical aside from a normalization ratio of $\sqrt{1-(\ell+r)/N}$.

The $W'$ and $\lrfo'$ oracles are block-diagonal between different values of $\ell, r$. 
Hence, the spectral norm of the difference between the two operators (after compressing $\lrfo'$) is bounded by the maximum spectral norm of the difference for each $\ell,r$.
From the above analysis, the two oracles are related by a constant re-scaling $\sqrt{1-(\ell+r)/N} \geq 1-(\ell+r)/N$ within each $\ell,r$.
Hence, the spectral norm of their difference is at most $(\ell+r)/N$. Applying $\ell+r \leq t$ completes the proof.
\end{proof}

Using Lemma~\ref{claim:relate-W-and-lrfo},
we can then show that $W'$ is indistinguishable from $\mathsf{lrfO}'$ by any adversary.
\begin{lemma}
[$W'$ is indistinguishable from $\mathsf{lrfO}'$] \label{lem:closeness-AWpD-and-LRFpD}
For any $t$-query oracle adversary $\mathcal{A}$, we have
\begin{equation}
    \left\lVert \Tr_{\gsL \gsR \gsC \gsD} \big( \dyad*{ \mathcal{A}^{W', \mathfrak{D}}_t }_{ \gsA \gsB \gsL \gsR \gsC \gsD} \big) - \Tr_{\gsH_{\color{gray} \mathsf{1}} \gsH_{\color{gray} \mathsf{2}} \gsF \gsC \gsD} \big( \dyad*{ \mathcal{A}^{\lrfo', \mathfrak{D}}_t }_{ \gsA \gsB \gsH_{\color{gray} \mathsf{1}} \gsH_{\color{gray} \mathsf{2}} \gsF \gsC \gsD} \big) \right\rVert_1 \leq \frac{t(t+1)}{2N}.
\end{equation}
\end{lemma}
\begin{proof}
    From Lemma~\ref{claim:relate-W-and-lrfo}, every application of $W'$ in $\ket*{ \mathcal{A}^{W', \mathfrak{D}}_t }$ can be replaced by an application of $\mathsf{Compress}_{\mathsf{LFR}} \cdot \lrfo' \cdot \mathsf{Compress}^\dagger_{\mathsf{LFR}}$ up to total trace norm error $\sum_{s=1}^t s/N = t(t+1)/2N$.
    Since the $\mathsf{Compress}^\dagger_{\mathsf{LFR}}$ operations act only the $\mathsf{L}$ and $\mathsf{R}$ registers, they commute with all other objects in $\ket*{ \mathcal{A}^{W', \mathfrak{D}}_t }$ (namely, $\mathsf{cC}$ and $\mathsf{cD}$ and $A_s$).
    This implies that all compress operations between adjacent applications of the $\lrfo$ oracle cancel one another, leaving only a first application of $\mathsf{Compress}^\dagger_{\mathsf{LFR}}$ before the first query to $\lrfo$ and a last application of $\mathsf{Compress}_{\mathsf{LFR}}$ following the $t$-th application.
    The first application acts trivially because $L$ and $R$ are empty in the initial state.
    The final application has no effect since $\mathsf{L}$ and $\mathsf{R}$ are traced out in the final state.
    Hence, all applications of $\mathsf{Compress}_{\mathsf{LFR}}$ and $\mathsf{Compress}^\dagger_{\mathsf{LFR}}$ vanish which yields the state $\ket*{ \mathcal{A}^{\lrfo', \mathfrak{D}}_t }$.
\end{proof}

\subsubsection{Twirled projected $\mathsf{lrfO}$ is indistinguishable from twirled $\mathsf{lrfO}$}

Finally, we can leverage our results thus far to show that the projected oracle $\lrfo'$ is indistinguishable from the original oracle $\lrfo$ by any adversary.
\begin{lemma}
[Twirled $\mathsf{lrfO}'$ is indistinguishable from twirled $\mathsf{lrfO}$] \label{lem:closeness-LRFpD-and-LRFD}
Let $\mathfrak{D}$ be any strong approximate unitary 2-design with additive error $\varepsilon$. For any $t$-query oracle adversary $\mathcal{A}$, we have
\begin{equation} \nonumber
    \left\lVert \dyad*{ \mathcal{A}^{\lrfo', \mathfrak{D}}_t }  - \dyad*{ \mathcal{A}^{\lrfo, \mathfrak{D}}_t } \right\rVert_1 = \mathcal{O}(t^2/N^{1/16}) + \mathcal{O}(t^{5/8} \varepsilon^{1/8}).
\end{equation}
\end{lemma}
\begin{proof}
Let us rewrite the projected oracle as follows,
\begin{align}\nonumber
    \lrfo'
    & \equiv 
    \widetilde{\Pi}^{\mathsf{locbij}}
    \cdot \lrfo \cdot
    \left( \widetilde{\Pi}^{\notin \mathsf{locDom}}  + \widetilde{\Pi}^{\calI(\lrfo^\dagger \Pi^{\notin \mathsf{locIm}}  \Pi^{\mathsf{locbij}})}    \right) \cdot
    \left( \widetilde{\Pi}^{\notin \mathsf{locDom}}  + \widetilde{\Pi}^{\calD(W_R^\dagger)}  \right) \cdot \widetilde{\Pi}^{\mathsf{locbij}},
\end{align}
and similar for $(\lrfo^\dagger)'$, $(\lrfo^*)'$, and $(\lrfo^T)'$.
This decomposition follows from the original definition because $\widetilde{\Pi}^{\notin \mathsf{locDom}}$ acts on an orthogonal subspace to $\widetilde{\Pi}^{\calI(\lrfo^\dagger \Pi^{\notin \mathsf{locIm}}  \Pi^{\mathsf{locbij}})}$ and $\widetilde{\Pi}^{\calD(W_R^\dagger)}$.
This fact also implies that each sum in parentheses above is a projector.
Hence, $\lrfo'$ is equal to the product of $\lrfo$ and four projectors.

From Eq.~(\ref{eq: state to op bound}) and the sequential gentle measurement lemma (Lemma~\ref{lem:seq-gentleM-pure}), the trace distance of interest is upper bounded as,
\begin{equation} \nonumber
    \left\lVert \dyad*{ \mathcal{A}^{\lrfo', \mathfrak{D}}_t }  - \dyad*{ \mathcal{A}^{\lrfo, \mathfrak{D}}_t } \right\rVert_1
    \leq 
    2 
    \left\lVert \ket*{ \mathcal{A}^{\lrfo', \mathfrak{D}}_t }  - \ket*{ \mathcal{A}^{\lrfo, \mathfrak{D}}_t }
    \right\rVert_2
    \leq 
    8 t \sqrt{1 - \langle
    \mathcal{A}^{\lrfo', \mathfrak{D}}_t 
    \big|
    \mathcal{A}^{\lrfo', \mathfrak{D}}_t 
    \rangle}.
\end{equation}
To bound the normalization on the right hand side, we apply Lemma~\ref{lem:closeness-AWpD-and-LRFpD} and Lemma~\ref{lem:closeness-AWD-and-AWpD} and Lemma~\ref{lem:closeness-AWD-and-PhiVt} and Theorem~\ref{theorem:haar-cho-strong} to find 
\begin{equation} \nonumber
    \left| \langle
    \mathcal{A}^{\lrfo', \mathfrak{D}}_t 
    \big|
    \mathcal{A}^{\lrfo', \mathfrak{D}}_t 
    \rangle
    -
    \E_{U \sim H} \langle \mathcal{A}^{U, \mathfrak{D}}_t 
    \big|
    \mathcal{A}^{U, \mathfrak{D}}_t \rangle
    \right|
    \leq \frac{t(t+1)}{2N} + \frac{\sqrt{70}t(t-1)}{N^{1/8}} + \frac{2t(t+1)}{N^{1/4}} + \frac{9t(t+2)}{N^{1/8}} + 8t^{5/4} \varepsilon^{1/4}.
\end{equation}
We have $\langle \mathcal{A}^{U, \mathfrak{D}}_t 
    \big|
    \mathcal{A}^{U, \mathfrak{D}}_t \rangle = 1$ since $U$ is unitary.
Hence, the trace distance is bounded above by $8t$ multiplied by the square root of the right hand side of the above equation.
The right hand side is $\mathcal{O}(t^2/N^{1/8}) + \mathcal{O}(t^{5/4} \varepsilon^{1/4})$.
Hence, the trace distance is $\mathcal{O}(t^2/N^{1/16}) + \mathcal{O}(t^{5/8} \varepsilon^{1/8})$ as claimed.
\end{proof}

\subsection{LRFC  is indistinguishable from a Haar-random unitary}

We now prove our main result, that the LRFC ensemble is indistinguishable from a Haar-random unitary.
This follows relatively quickly from the results in the previous sections.
\begin{theorem}
[LRFC is indistinguishable from Haar-random] \label{thm:LRFC}
Let $\mathfrak{D}$ be any strong approximate unitary 2-design with additive error $\varepsilon$. For any $t$-query oracle adversary $\mathcal{A}$, we have
\begin{equation} \nonumber
    \left\lVert \E_{U \sim \text{\emph{LRFC}}} \left( \dyad*{ \mathcal{A}^{U}_t } \right) - \E_{U \sim H} \left( \dyad*{ \mathcal{A}^{U}_t } \right) \right\rVert_1 = \mathcal{O}(t^2/N^{1/16}) + \mathcal{O}(t^{5/8} \varepsilon^{1/8}).
\end{equation}
\end{theorem}
\begin{proof}
Our proof follows from the results in the preceding subsections in four steps.
At each step, we bound the trace distance between two density matrices.
The density matrices are,
\begin{align}
    \rho^{(0)} & = \E_{U \sim H}  \big( \dyad*{ \mathcal{A}^{U}_t }_{ \gsA \gsB} \big) \\
    \rho^{(1)} & = \Tr_{\gsL \gsR \gsC \gsD} \big( \dyad*{ \mathcal{A}^{W, \mathfrak{D}}_t }_{ \gsA \gsB \gsL \gsR \gsC \gsD} \big) \\
    \rho^{(2)} & = \Tr_{\gsL \gsR \gsC \gsD} \big( \dyad*{ \mathcal{A}^{W', \mathfrak{D}}_t }_{ \gsA \gsB \gsL \gsR \gsC \gsD} \big) \\
    \rho^{(3)} & = \Tr_{\gsH_{\color{gray} \mathsf{1}} \gsH_{\color{gray} \mathsf{2}} \gsF \gsC \gsD} \big( \dyad*{ \mathcal{A}^{\lrfo', \mathfrak{D}}_t }_{ \gsA \gsB \gsH_{\color{gray} \mathsf{1}} \gsH_{\color{gray} \mathsf{2}} \gsF \gsC \gsD} \big) \\
    \rho^{(4)} & = \E_{U \sim \text{LRFC}}  \big( \dyad*{ \mathcal{A}^{U}_t }_{ \gsA \gsB} \big) 
    = \Tr_{\gsH_{\color{gray} \mathsf{1}} \gsH_{\color{gray} \mathsf{2}} \gsF \gsC \gsD} \big( \dyad*{ \mathcal{A}^{\lrfo, \mathfrak{D}}_t }_{ \gsA \gsB \gsH_{\color{gray} \mathsf{1}} \gsH_{\color{gray} \mathsf{2}} \gsF \gsC \gsD} \big)
\end{align}
The first density matrix is the expected output state of an experiment that queries a Haar-random unitary $U$.
The last density matrix is the expected output of an experiment that queries the a random LRFC unitary.
The intermediary density matrices denote the output state of experiments in which the action of each unitary is replaced with a path-recording oracle from the previous sections.

In the previous sections, we have already bounded the trace distance between each pair of density matrices. These are:
\begin{enumerate}
\setcounter{enumi}{0}

\item $\lVert \rho^{(0)} - \rho^{(1)} \rVert_1 \leq 18t(t+1)/N^{1/8}$ \hspace*{0pt}\hfill  (Lemma~\ref{lem:closeness-AWD-and-PhiVt} and Theorem~\ref{theorem:haar-cho-strong})

\item $\lVert \rho^{(1)} - \rho^{(2)} \rVert_1 \leq \sqrt{70} t(t-1)/N^{1/8} + 2 t(t+1)/N^{1/4} + 4 t^{5/4} \varepsilon^{1/4}$ \hspace*{0pt}\hfill 
 (Lemma~\ref{lem:closeness-AWD-and-AWpD})

\item $\lVert \rho^{(2)} - \rho^{(3)} \rVert_1  \leq t(t+1)/2N$ \hspace*{0pt}\hfill (Lemma~\ref{lem:closeness-AWpD-and-LRFpD})

\item $\lVert \rho^{(3)} - \rho^{(4)} \rVert_1 = \mathcal{O}(t^2/N^{1/16}) + \mathcal{O}(t^{5/8} \varepsilon^{1/8})$ \hspace*{0pt}\hfill (Lemma~\ref{lem:closeness-LRFpD-and-LRFD})

\end{enumerate}
\noindent By the triangle inequality, the total trace distance, $\lVert \rho^{(0)} - \rho^{(4)} \rVert_1$, is less than the sum of the four distances above.
This yields $\lVert \rho^{(0)} - \rho^{(4)} \rVert_1 = \mathcal{O}(t^2/N^{1/16})$ as claimed. 
\end{proof}

\subsection{Proof of Theorem~\ref{thm:LRFC-design}: LRFC is a strong unitary design}

We let $\mathfrak{D}$ be an exact unitary 2-design~\cite{webb2015clifford,cleve2015near} with additive error zero.
By definition, the output of any quantum experiment that queries any combination of $U$, $U^\dagger$, $U^*$, $U^T$ up to $k$ times is identical whether $f, h_1, h_2$ are $2k$-wise independent random functions versus truly random functions.
From Theorem~\ref{thm:LRFC}, the output of any quantum experiment that queries the truly random LRFC ensemble $k$ times is close to the output of the same experiment that queries a Haar-random unitary, up to trace distance $\mathcal{O}(k^2/N^{1/6})$ where $N = 2^n$.
Hence, the $2k$-wise independent variant of the LRFC ensemble forms an $\varepsilon$-approximate strong unitary $k$-design with $\varepsilon = \mathcal{O}(t^2/N^{1/6})$. \qed

\subsection{Proof of Theorem~\ref{thm:LRFC-PRU}: LRFC is a strong PRU}

We let $\mathfrak{D}$ be an exact unitary 2-design~\cite{webb2015clifford,cleve2015near} with additive error zero.
By definition, no subexponential-time quantum experiment can distinguish whether $f,h_1,h_2$ are PRFs (with security against any subexponential-time quantum adversary) versus truly random functions.
From Theorem~\ref{thm:LRFC}, the output of any quantum experiment that queries the truly random LRFC ensemble $k$ times is close to the output of the same experiment that queries a Haar-random unitary, up to trace distance $\mathcal{O}(k^2/N^{1/6})$ where $N=2^n$.
This is negligibly small for any $k$ subexponential in $n$.
Hence, the pseudorandom variant of the LRFC ensemble forms a strong PRU with security against any subexponential-time quantum adversary. \qed

\section{Gluing strong random unitaries} \label{app: gluing}

In this section, we provide a proof of the strong gluing lemma (Lemma~\ref{lemma: strong gluing}).
We then apply the strong gluing lemma to prove our Theorems~\ref{thm:two-layer-design} and~\ref{thm:two-layer-PRU} on the scrambled two-layer circuit ensemble.

\subsection{Proof of Lemma~\ref{lemma: strong gluing}: Gluing strong random unitaries} \label{sec: proof strong gluing}

Our proof is long but straightforward, and uses the path-recording framework introduced in Ref.~\cite{ma2024construct}.
We refer the reader to Appendices~\ref{app: PRU} and~\ref{app: LRFC} for a complete introduction to this framework and key notation.
In what follows, we begin in Appendix~\ref{sec: preliminaries} by introducing several new objects within the path-recording framework that will be useful in the strong gluing proof.
We then provide a summary of our proof in Appendix~\ref{sec: summary}.
Each step in the proof summary is then proven individually in the following Appendices~\ref{sec: twirled WABC to twirled projected WABC},~\ref{sec: projected WABC to projected WAB WBC}, and~\ref{sec: twirled projected WAB WBC to twirled WAB WBC}.
At a high-level, our proof follows a roughly similar approach to our analysis of the LRFC ensemble in Appendix~\ref{app: LRFC}.

\subsubsection{Preliminaries} \label{sec: preliminaries}

In this subsection, we provide a short overview of the new notation used in our proof.
%

\paragraph{Registers.} Our proof will apply the path-recording framework to the unitaries $U_{\mathsf{abc}}$ and $U_{\mathsf{bc}} U_{\mathsf{ab}}$.
To each Haar-random unitary, the path-recording framework associate two ancilla registers.
We denote these registers as $\mathsf{L}_{\mathsf{abc}}$ and $\mathsf{R}_{\mathsf{abc}}$, $\mathsf{L}_{\mathsf{bc}}$ and $\mathsf{R}_{\mathsf{bc}}$, and $\mathsf{L}_{\mathsf{ab}}$ and $\mathsf{R}_{\mathsf{ab}}$, for the three unitaries in consideration.
We denote the system register as $A = \mathsf{a} \cup \mathsf{b} \cup \mathsf{c}$.
As in Ref.~\cite{ma2024construct}, we also allow an arbitrary-sized physical register $\mathsf{B}$, as well as ancilla registers $\mathsf{C}$ and $\mathsf{D}$ which purify the twirl over $C, D \sim \mathfrak{D}$.

The registers $\mathsf{L}_{\mathsf{abc}}$, $\mathsf{R}_{\mathsf{abc}}$, $\mathsf{L}_{\mathsf{bc}}$, $\mathsf{R}_{\mathsf{bc}}$, $\mathsf{L}_{\mathsf{ab}}$, $\mathsf{R}_{\mathsf{ab}}$ contain relation states.
For $\mathsf{L}_{\mathsf{abc}}$ and $\mathsf{R}_{\mathsf{abc}}$, a relation state takes the form,
\begin{align}
    \ket{ L }_{\gsLabc} & = \ket{ \big\{ \big( x^i_\mathsf{a} x^i_\mathsf{b} x^i_\mathsf{c} , y^i_\mathsf{a} y^i_\mathsf{b} y^i_\mathsf{c} \big) : i \in [\ell] \big\} }_{\gsLabc} \\
    \ket{ R }_{\gsRabc} & = \ket{ \big\{ \big( x^j_\mathsf{a} x^j_\mathsf{b} x^j_\mathsf{c} , y^j_\mathsf{a} y^j_\mathsf{b} y^j_\mathsf{c} \big) : j \in [r] \big\} }_{\gsRabc},
\end{align}
where $x^i_\alpha$ is bitstring on subsystem $\alpha \in \{ \mathsf{a} , \mathsf{b}, \mathsf{c} \}$ (and similar for $y^i_\alpha$, $x^j_\alpha$, and $y^j_\alpha$).
Here, $\ell$ and $r$ denote the length of the relation state registers.
For $\mathsf{L}_{\mathsf{bc}}$ and $\mathsf{R}_{\mathsf{bc}}$ and $\mathsf{L}_{\mathsf{ab}}$ and $\mathsf{R}_{\mathsf{ab}}$, we write,
\begin{align}
    \ket{ L_{\mathsf{ab}} }_{\gsLab} \ket*{ L_{\mathsf{bc}} }_{\gsLbc} 
    & = \ket*{ \big\{ \big( x^i_\mathsf{a} x^i_\mathsf{b} , y^i_\mathsf{a} z^i_\mathsf{b} \big) : i \in [\ell] \big\} }_{\gsLab}
    \ket*{ \big\{ \big( z'^i_\mathsf{b} x^i_\mathsf{c} , y^i_\mathsf{b} y^i_\mathsf{c} \big) : i \in [\ell'] \big\} }_{\gsLbc}\\
    \ket{ R_{\mathsf{ab}} }_{\gsRab} \ket*{ R_{\mathsf{bc}} }_{\gsRbc} 
    & = \ket*{ \big\{ \big( x^j_\mathsf{a} x^j_\mathsf{b} , y^j_\mathsf{a} z^j_\mathsf{b} \big) : j \in [r] \big\} }_{\gsRab}
    \ket*{ \big\{ \big( z'^j_\mathsf{b} x^j_\mathsf{c} , y^j_\mathsf{b} y^j_\mathsf{c} \big) : j \in [r'] \big\} }_{\gsRbc},
\end{align}
where $z^i_\mathsf{b}$, $z'^i_\mathsf{b}$, $z^j_\mathsf{b}$, $z'^j_\mathsf{b}$ are bitstrings on subsystem $\mathsf{b}$.
We use a different character, $z$, for these bitstrings, because they will correspond to the bitstrings that appear ``between'' $U_{\mathsf{bc}}$ and $U_{\mathsf{ab}}$ in $U_{\mathsf{bc}} U_{\mathsf{ab}}$, and will play a special role in our proof.

\paragraph{Projectors on $\mathsf{L}_\mathsf{abc}$ and $\mathsf{R}_\mathsf{abc}$.}

We will often wish to restrict attention to certain subsets of the relation states.
As in Appendix~\ref{app: LRFC}, we have the projector onto bijective relation states,
\begin{equation}
    \Pi^{\mathsf{bij}}_{\gsLabc \gsRabc} 
    \ket{ L }_{\gsLabc}
   \ket{ R }_{\gsRabc} 
   = 
   \begin{cases}
       \ket{ L }_{\gsLabc} \ket{ R }_{\gsRabc}, & \text{if } \mathsf{Dom}(L \cup R) \in \mathsf{ distinct }, \\
       & \text{and } \mathsf{Im}(L \cup R) \in \mathsf{ distinct } \\
       0, & \text{else}. \\
   \end{cases}
\end{equation}
Here, $\mathsf{Dom}(L \cup R) = \{ x^i_{\mathsf{a}} x^i_{\mathsf{b}} x^i_{\mathsf{c}} : i \in [\ell] \} \cup \{ x^j_{\mathsf{a}} x^j_{\mathsf{b}} x^j_{\mathsf{c}} : j \in [r] \}$,
and 
$\mathsf{Im}(L \cup R) = \{ y^i_{\mathsf{a}} y^i_{\mathsf{b}} y^i_{\mathsf{c}} : i \in [\ell] \} \cup \{ y^j_{\mathsf{a}} y^j_{\mathsf{b}} y^j_{\mathsf{c}} : j \in [r] \}$.
Identical to before, we also define the not-in-domain projector on $\mathsf{ALR}$ as,
\begin{equation}
    \Pi^{\notin \mathsf{Dom}}_{\gsA \gsLabc \gsRabc} 
    \ket{ x_\mathsf{a} x_\mathsf{b} x_\mathsf{c} }_{\gsA}
    \ket{ L }_{\gsLabc}
   \ket{ R }_{\gsRabc} 
   = 
   \begin{cases}
       \ket{ x_\mathsf{a} x_\mathsf{b} x_\mathsf{c} }_{\gsA}\ket{ L }_{\gsLabc} \ket{ R }_{\gsRabc}, & \text{if } x_\mathsf{a} x_\mathsf{b} x_\mathsf{c} \notin \mathsf{Dom}(L \cup R), \\
       0, & \text{else}, \\
   \end{cases}
\end{equation}
and the not-in-image projector as,
\begin{equation}
    \Pi^{\notin \mathsf{Im}}_{\gsA \gsLabc \gsRabc} 
    \ket{ y_\mathsf{a} y_\mathsf{b} y_\mathsf{c} }_{\gsA}
    \ket{ L }_{\gsLabc}
   \ket{ R }_{\gsRabc} 
   = 
   \begin{cases}
       \ket{ y_\mathsf{a} y_\mathsf{b} y_\mathsf{c} }_{\gsA}\ket{ L }_{\gsLabc} \ket{ R }_{\gsRabc}, & \text{if } y_\mathsf{a} y_\mathsf{b} y_\mathsf{c} \notin \mathsf{Im}(L \cup R), \\
       0, & \text{else}. \\
   \end{cases}
\end{equation}
For the strong gluing proof, we will also introduce new local variants of the above projectors.
These are extremely similar to those defined in Appendix~\ref{app: LRFC}.
For this reason, we use the same notation here as in Appendix~\ref{app: LRFC}, even though the precise definitions are slightly different.
We define the \emph{locally bijective} relation states via the projector,
\begin{equation}
    \Pi^{\mathsf{locbij}}_{\gsLabc \gsRabc} 
    \ket{ L }_{\gsLabc}
   \ket{ R }_{\gsRabc} 
   = 
   \begin{cases}
       \ket{ L }_{\gsLabc} \ket{ R }_{\gsRabc}, & \text{if } \mathsf{Dom}(L \cup R)_\alpha \in \mathsf{ distinct }, \forall \alpha = \mathsf{a}, \mathsf{b}, \mathsf{c} \\
       & \text{and } \mathsf{Im}(L \cup R)_\alpha \in \mathsf{ distinct }, \forall  \alpha = \mathsf{a}, \mathsf{b}, \mathsf{c} \\
       0, & \text{else}, \\
   \end{cases}
\end{equation}
where $\mathsf{Dom}(L \cup R)_\alpha = \{ x^i_\alpha : i \in [\ell] \} \cup \{ x^j_\alpha : j \in [r] \}$,
and 
$\mathsf{Im}(L \cup R) = \{ y^i_\alpha : i \in [\ell] \} \cup \{ y^j_\alpha : j \in [r] \}$.
Similarly, we define the not-in-local-domain and not-in-local-image projectors as,
\begin{align}
    \Pi^{\notin \mathsf{locDom}}_{\gsA \gsLabc \gsRabc} 
    \ket{ x_\mathsf{a} x_\mathsf{b} x_\mathsf{c} }_{\gsA}
    \ket{ L }_{\gsLabc}
   \ket{ R }_{\gsRabc} 
   & = 
   \begin{cases}
       \ket{ x_\mathsf{a} x_\mathsf{b} x_\mathsf{c} }_{\gsA}\ket{ L }_{\gsLabc} \ket{ R }_{\gsRabc}, & \text{if } x_\alpha \notin \mathsf{Dom}(L \cup R)_\alpha, \forall  \alpha = \mathsf{a}, \mathsf{b}, \mathsf{c} \\
       0, & \text{else}, \\
   \end{cases} \\
    \Pi^{\notin \mathsf{locIm}}_{\gsA \gsLabc \gsRabc} 
    \ket{ y_\mathsf{a} y_\mathsf{b} y_\mathsf{c} }_{\gsA}
    \ket{ L }_{\gsLabc}
   \ket{ R }_{\gsRabc} 
   & = 
   \begin{cases}
       \ket{ y_\mathsf{a} y_\mathsf{b} y_\mathsf{c} }_{\gsA}\ket{ L }_{\gsLabc} \ket{ R }_{\gsRabc}, & \text{if } y_\alpha \notin \mathsf{Im}(L \cup R)_\alpha, \,\, \forall \alpha = \mathsf{a}, \mathsf{b}, \mathsf{c} \\
       0, & \text{else}. \\
   \end{cases}
\end{align}
We also let,
\begin{align}
    \Pi^{\notin \mathsf{Dom} \rightarrow \notin \mathsf{locDom}}_{\gsA \gsLabc \gsRabc} & = 
    1 
    - \Pi^{\mathsf{Dom}}_{\gsA \gsLabc \gsRabc}
    + \Pi^{\mathsf{locDom}}_{\gsA \gsLabc \gsRabc}  \\
     \Pi^{\notin \mathsf{Im} \rightarrow \notin \mathsf{locIm}}_{\gsA \gsLabc \gsRabc} & = 
    1 
    - \Pi^{\mathsf{Im}}_{\gsA \gsLabc \gsRabc}
    + \Pi^{\mathsf{locIm}}_{\gsA \gsLabc \gsRabc},
\end{align}
denote projectors which do nothing if $x_\mathsf{a} x_\mathsf{b} x_\mathsf{c} \in \mathsf{Dom}(L\cup R)$, and project to the not-in-local-domain subspace if $x_\mathsf{a} x_\mathsf{b} x_\mathsf{c} \notin \mathsf{Dom}(L\cup R)$ (and similar for $\mathsf{Im}$).

\paragraph{Projectors on $\mathsf{L}_\mathsf{ab}$, $\mathsf{R}_\mathsf{ab}$, $\mathsf{L}_\mathsf{bc}$, $\mathsf{R}_\mathsf{bc}$ and the $\mathsf{Compress}$ partial isometry.}

Let us now turn to the registers $\mathsf{L}_{\mathsf{bc}}$, $\mathsf{R}_{\mathsf{bc}}$, $\mathsf{L}_{\mathsf{ab}}$, $\mathsf{R}_{\mathsf{ab}}$.
A key component of our proof is to construct a partial isometry between relation states on $\mathsf{L}_{\mathsf{abc}} \otimes \mathsf{R}_{\mathsf{abc}}$ and those on $\mathsf{L}_{\mathsf{ab}} \otimes \mathsf{R}_{\mathsf{ab}} \otimes \mathsf{L}_{\mathsf{bc}} \otimes \mathsf{R}_{\mathsf{bc}}$.
Namely, for any state in $\Pi^{\mathsf{locbij}}_{\gsLabc \gsRabc}$ on $\mathsf{L}_{\mathsf{abc}} \otimes \mathsf{R}_{\mathsf{abc}}$, we define an ``un-compressed'' state,
\begin{equation}
\begin{split}
    \mathsf{Compress}^\dagger & \cdot 
    \ket*{ \big\{ \big( x^i_\mathsf{a} x^i_\mathsf{b} x^i_\mathsf{c} , y^i_\mathsf{a} y^i_\mathsf{b} y^i_\mathsf{c} \big) : i \in [\ell] \big\} }_{\gsLabc}
    \ket*{ \big\{ \big( x^j_\mathsf{a} x^j_\mathsf{b} x^j_\mathsf{c} , y^j_\mathsf{a} y^j_\mathsf{b} y^j_\mathsf{c} \big) : j \in [r] \big\} }_{\gsRabc} \\
    & = 
    \frac{1}{\sqrt{ N_\mathsf{b}^{\ell+r} }} \sum_{z_\ell, z_r}
    \ket*{ \big\{ \big( x^i_\mathsf{a} x^i_\mathsf{b} , y^i_\mathsf{a} z^i_\mathsf{b} \big) : i \in [\ell] \big\} }_{\gsLab}
    \ket*{ \big\{ \big( z^i_\mathsf{b} x^i_\mathsf{c} , y^i_\mathsf{b} y^i_\mathsf{c} \big) : i \in [\ell] \big\} }_{\gsLbc}
    \\
    & \quad \quad \quad \quad \quad \quad \quad \quad \otimes 
    \ket*{ \big\{ \big( x^j_\mathsf{a} x^j_\mathsf{b} , y^j_\mathsf{a} z^j_\mathsf{b} \big) : j \in [r] \big\} }_{\gsRab}
    \ket*{ \big\{ \big( z^j_\mathsf{b} x^j_\mathsf{c} , y^j_\mathsf{b} y^j_\mathsf{c} \big) : j \in [r] \big\} }_{\gsRbc},
\end{split}
\end{equation}
where we abbreviate $z_\ell \equiv \{ z^i_\mathsf{b} : i \in [\ell] \}$, $z_r \equiv \{ z^j_\mathsf{b} : j \in [r] \}$.
The final state is a valid relation state because $\{ y^i_\mathsf{a} : i \in [\ell] \} \cup \{ y^j_\mathsf{a} : j \in [r] \}$ and $\{ x^i_\mathsf{c} : i \in [\ell] \} \cup \{ x^j_\mathsf{c} : j \in [r] \}$ are distinct by assumption.
We define $\mathsf{Compress}$ as the adjoint of this operation.

The range of $\mathsf{Compress}^\dagger$ on $\mathsf{L}_{\mathsf{ab}} \otimes \mathsf{R}_{\mathsf{ab}} \otimes \mathsf{L}_{\mathsf{bc}} \otimes \mathsf{R}_{\mathsf{bc}}$ consists of the \emph{paired} relation states,
\begin{equation}
\begin{split}
    \frac{1}{\sqrt{ N_\mathsf{b}^{\ell+r} }} \sum_{z_\ell, z_r} &
    \ket{ L_{\mathsf{ab}}^{z_\ell} }_{\gsLab}
    \ket{ L_{\mathsf{bc}}^{z_\ell} }_{\gsLbc}
    \ket{ R_{\mathsf{ab}}^{z_r} }_{\gsRab}
    \ket{ R_{\mathsf{bc}}^{z_r} }_{\gsRbc}  \\
    & \quad \equiv \frac{1}{\sqrt{ N_\mathsf{b}^{\ell+r} }} \sum_{z_\ell, z_r}
    \ket*{ \big\{ \big( x^i_\mathsf{a} x^i_\mathsf{b} , y^i_\mathsf{a} z^i_\mathsf{b} \big) : i \in [\ell] \big\} }_{\gsLab}
    \ket*{ \big\{ \big( z^i_\mathsf{b} x^i_\mathsf{c} , y^i_\mathsf{b} y^i_\mathsf{c} \big) : i \in [\ell] \big\} }_{\gsLbc}
    \\
    & \quad \quad \quad \quad \quad \quad \quad \quad \quad \otimes 
    \ket*{ \big\{ \big( x^j_\mathsf{a} x^j_\mathsf{b} , y^j_\mathsf{a} z^j_\mathsf{b} \big) : j \in [r] \big\} }_{\gsRab}
    \ket*{ \big\{ \big( z^j_\mathsf{b} x^j_\mathsf{c} , y^j_\mathsf{b} y^j_\mathsf{c} \big) : j \in [r] \big\} }_{\gsRbc}, \\
\end{split}
\end{equation}
for any locally bijective $L \equiv \{ (x^i_\mathsf{a} x^i_\mathsf{b} x^i_\mathsf{c} , y^i_\mathsf{a} y^i_\mathsf{b} y^i_\mathsf{c} ) : i \in [\ell] \}$
and $R \equiv \{ (x^j_\mathsf{a} x^j_\mathsf{b} x^j_\mathsf{c} , y^j_\mathsf{a} y^j_\mathsf{b} y^j_\mathsf{c} ) : j \in [r] \}$.
We let $\Pi^{\mathsf{paired}}_{\gsLab \gsLbc \gsRab \gsRbc}$ denote the projector onto the set of states above.
We have
\begin{align}
    \mathsf{Compress} \cdot \mathsf{Compress}^\dagger 
    & = 
    \Pi^{\mathsf{locbij}}_{\gsLabc \gsRabc}
     \\
    \mathsf{Compress}^\dagger \cdot \mathsf{Compress} & = 
     \Pi^{\mathsf{paired}}_{\gsLab \gsLbc \gsRab \gsRbc},
\end{align}
by construction.

\subsubsection{Proof overview} \label{sec: summary}

Our proof proceeds in five steps.
At each step, we bound the trace distance between two density matrices.
The density matrices are,
\begin{align}
    \rho^{(0)} & = \E_{U_{\mathsf{abc}} \sim H} \big( \dyad*{ \mathcal{A}^{U_{\mathsf{abc}}, \mathfrak{D}}_t }_{ \gsA \gsB \gsC \gsD} \big) \\
    \rho^{(1)} & = \Tr_{\gsLabc \gsRabc \gsC \gsD} \big( \dyad*{ \mathcal{A}^{W_{\mathsf{abc}}, \mathfrak{D}}_t }_{ \gsA \gsB \gsLabc \gsRabc \gsC \gsD} \big) \\
    \rho^{(2)} & = \Tr_{\gsLabc \gsRabc \gsC \gsD} \big( \dyad*{ \mathcal{A}^{W_{\mathsf{abc}}', \mathfrak{D}}_t }_{ \gsA \gsB \gsLabc \gsRabc \gsC \gsD} \big) \\
    \rho^{(3)} & = \Tr_{\gsLab \gsLbc \gsRab \gsRbc \gsC \gsD} \big( \dyad*{ \mathcal{A}^{(W_{\mathsf{bc}} W_{\mathsf{ab}})', \mathfrak{D}}_t }_{ \gsA \gsB \gsLab \gsLbc \gsRab \gsRbc \gsC \gsD} \big) \\
    \rho^{(4)} & = \Tr_{\gsLab \gsLbc \gsRab \gsRbc \gsC \gsD} \big( \dyad*{ \mathcal{A}^{V_{\mathsf{bc}} V_{\mathsf{ab}}, \mathfrak{D}}_t }_{ \gsA \gsB \gsLab \gsLbc \gsRab \gsRbc \gsC \gsD} \big) \\
    \rho^{(5)} & = \E_{U_{\mathsf{ab}}, U_{\mathsf{bc}} \sim H } \big( \dyad*{ \mathcal{A}^{U_{\mathsf{bc}} U_{\mathsf{ab}}, \mathfrak{D}}_t }_{ \gsA \gsB \gsC \gsD} \big)
\end{align}
The first density matrix is the expected output state of an experiment that queries a Haar-random unitary $U_{\mathsf{abc}}$ (and its inverse, conjugate, and transpose).
The last density matrix is the expected output of an experiment that queries $U_{\mathsf{bc}} U_{\mathsf{ab}}$.

The intermediary density matrices denote the output state of experiments in which the action of each unitary is replaced with a path-recording oracle.
In particular, in the third step, $\ket*{ \mathcal{A}^{W_{\mathsf{abc}}', \mathfrak{D}}_t}$ denotes the state in which each application of $W^{}_{\mathsf{abc}}$, $W_{\mathsf{abc}}^\dagger$, $\overline{W}_{\mathsf{abc}}$, $\overline{W}_{\mathsf{abc}}^\dagger$ is replaced by the operators,
\begin{align}
    W_{\mathsf{abc}}' 
    & \equiv 
    \Pi^{\mathsf{locbij}}
    \cdot W_{\mathsf{abc}} \cdot
    \Pi^{\notin \mathsf{Dom} \rightarrow \notin \mathsf{locDom}} \\
    (W^\dagger_{\mathsf{abc}})'
    & \equiv 
    \Pi^{\mathsf{locbij}}
    \cdot W^\dagger_{\mathsf{abc}} \cdot
    \Pi^{\notin \mathsf{Im} \rightarrow \notin \mathsf{locIm}} \\
    \overline{W}_{\mathsf{abc}}' 
    & \equiv 
    \Pi^{\mathsf{locbij}}
    \cdot \overline{W}_{\mathsf{abc}} \cdot
    \Pi^{\notin \mathsf{Dom} \rightarrow \notin \mathsf{locDom}} \\
    (\overline{W}^\dagger_{\mathsf{abc}})'
    & \equiv 
    \Pi^{\mathsf{locbij}}
    \cdot \overline{W}^\dagger_{\mathsf{abc}} \cdot
    \Pi^{\notin \mathsf{Im} \rightarrow \notin \mathsf{locIm}},
\end{align}
respectively.
Similarly, in the fourth step, $\ket*{ \mathcal{A}^{(W_{\mathsf{bc}} W_{\mathsf{ab}})', \mathfrak{D}}_t }$ denotes the state in which each application of $W^{}_{\mathsf{bc}} W^{}_{\mathsf{ab}}$, $W^{\dagger}_{\mathsf{ab}} W^{\dagger}_{\mathsf{bc}}$, $\overline{W}_{\mathsf{bc}} \overline{W}_{\mathsf{ab}}$, $ \overline{W}^{\dagger}_{\mathsf{ab}} \overline{W}^{\dagger}_{\mathsf{bc}}$ is replaced by,
\begin{align}
    (W_{\mathsf{bc}} W_{\mathsf{ab}})' 
    & \equiv 
    \Pi^{\mathsf{paired}} \cdot
    W_{\mathsf{bc}} W_{\mathsf{ab}} \cdot
    \widetilde{\Pi}^{\mathcal{D}(W_{\mathsf{abc}})} \cdot
    \widetilde{\Pi}^{\notin \mathsf{Dom} \rightarrow \notin \mathsf{locDom}} \\
    (W^\dagger_{\mathsf{ab}} W^\dagger_{\mathsf{bc}})' 
    & \equiv 
    \Pi^{\mathsf{paired}} \cdot
    W^\dagger_{\mathsf{ab}} W^\dagger_{\mathsf{bc}}  \cdot
    \widetilde{\Pi}^{\mathcal{D}(W^\dagger_{\mathsf{abc}})} \cdot
    \widetilde{\Pi}^{\notin \mathsf{Im} \rightarrow \notin \mathsf{locIm}} \\
    (\overline{W}_{\mathsf{bc}} \overline{W}_{\mathsf{ab}})' 
    & \equiv 
    \Pi^{\mathsf{paired}} \cdot
    \overline{W}_{\mathsf{bc}} \overline{W}_{\mathsf{ab}} \cdot
    \widetilde{\Pi}^{\mathcal{D}(\overline{W}_{\mathsf{abc}})} \cdot
    \widetilde{\Pi}^{\notin \mathsf{Dom} \rightarrow \notin \mathsf{locDom}} \\
    (\overline{W}^\dagger_{\mathsf{ab}} \overline{W}^\dagger_{\mathsf{bc}})' 
    & \equiv 
    \Pi^{\mathsf{paired}} \cdot
    \overline{W}^\dagger_{\mathsf{ab}} \overline{W}^\dagger_{\mathsf{bc}}  \cdot
    \widetilde{\Pi}^{\mathcal{D}(\overline{W}^\dagger_{\mathsf{abc}})} \cdot
    \widetilde{\Pi}^{\notin \mathsf{Im} \rightarrow \notin \mathsf{locIm}},
\end{align}
respectively. Here, we let $\widetilde{\Pi} \equiv \mathsf{Compress}^\dagger \cdot \Pi \cdot \mathsf{Compress}$, for any projector $\Pi$.

We bound the trace distance between each pair of density matrices as follows.
\begin{enumerate}
\setcounter{enumi}{0}

\item $\lVert \rho^{(0)} - \rho^{(1)} \rVert_1 \leq 9t(t+2)/N_{\mathsf{abc}}^{1/8} + 2t^{1/4}\varepsilon^{1/4}$ \hspace*{0pt}\hfill  (Theorem~\ref{theorem:haar-cho-strong} and Lemma~\ref{lem:closeness-AWD-and-PhiVt})

\item $\lVert \rho^{(1)} - \rho^{(2)} \rVert_1 \leq 17 t^2/N_{\mathsf{abc}}^{1/8} + 7 t^{3/2}/ (\min_\alpha N_\alpha)^{1/2} + 6 t^{5/4} \varepsilon^{1/4}$ \hspace*{0pt}\hfill 
 (Section~\ref{sec: twirled WABC to twirled projected WABC})

\item $\lVert \rho^{(2)} - \rho^{(3)} \rVert_1 \leq t^2 / N_{\mathsf{ab}} + t^2 / N_{\mathsf{bc}}$ \hspace*{0pt}\hfill (Section~\ref{sec: projected WABC to projected WAB WBC})

\item $\lVert \rho^{(3)} - \rho^{(4)} \rVert_1 \leq 2t \sqrt{ \frac{17t^2}{N_{\mathsf{abc}}^{1/8}} + \frac{7 t^{3/2}}{(\min_\alpha N_\alpha)^{1/2}}  + \frac{2t^2}{N_{\mathsf{ab}}} + \frac{2t^2}{N_{\mathsf{bc}}} +\frac{9 t}{N_{\mathsf{abc}}^{1/8}}+ 8 t^{5/4} \varepsilon^{1/4}}$ \hspace*{0pt}\hfill (Section~\ref{sec: twirled projected WAB WBC to twirled WAB WBC})

\item $\lVert \rho^{(4)} - \rho^{(5)} \rVert_1 \leq 9t(t+1)/N_{\mathsf{ab}}^{1/8} + 9t(t+1)/N_{\mathsf{bc}}^{1/8}$ \hspace*{0pt}\hfill (Theorem~\ref{theorem:haar-cho-strong})

\end{enumerate}
\noindent By the triangle inequality, the total trace distance, $\lVert \rho^{(0)} - \rho^{(5)} \rVert_1$, is less than the sum of the five distances above.
If each local Hilbert space has dimension at least $N_\alpha \geq 2^{\xi}$, then we have $\lVert \rho^{(0)} - \rho^{(5)} \rVert_1 \leq \mathcal{O}\big( t^2/2^{(3/16) \xi} \big) + \mathcal{O}\big( t^{5/8} \varepsilon^{1/8} \big)$ as claimed. \qed

\subsubsection{Twirled $W_{\mathsf{abc}}$ is indistinguishable from twirled projected $W_{\mathsf{abc}}$} \label{sec: twirled WABC to twirled projected WABC}

We will prove that
\begin{equation} \label{eq: TDt}
    \mathsf{TD}_t \equiv 
    \mathsf{TD}\left( 
      \dyad*{ \mathcal{A}^{W_{\mathsf{abc}}, \mathfrak{D}}_t }
    , 
    \dyad*{ \mathcal{A}^{W_{\mathsf{abc}}', \mathfrak{D}}_t }
    \right)
    \leq 
    \frac{2 \sqrt{70} t^2}{N_{\mathsf{abc}}^{1/8}} + \frac{4\sqrt{3} t^{3/2}}{(\min_\alpha N_\alpha)^{1/2}}  + 6 t^{5/4} \varepsilon^{1/4}.
\end{equation}
This implies that $\lVert \rho^{(1)} - \rho^{(2)} \rVert_1$ is less than the same value since the 1-norm cannot increase after tracing out $\mathsf{L}_{\mathsf{abc}}  \mathsf{R}_{\mathsf{abc}}$.
The claim follows since $2\sqrt{70} < 17$ and $4\sqrt{3} < 7$.

We proceed by induction.
The statement holds trivially at $t=0$.
To prove the inductive step, suppose that the Eq.~(\ref{eq: TDt}) holds up to time $t-1$ for any $t \geq 1$.
Without loss of generality, we assume that the forward unitary is applied at time $t$.
The case when the inverse or conjugate or transpose are applied follow by symmetric arguments.
The states at time $t$ are obtained from the states at time $t-1$ as follows,
\begin{align}
    \ket*{ \mathcal{A}^{W_{\mathsf{abc}}, \mathfrak{D}}_{t} }
    & =
    \mathsf{cD}
    \cdot
    W_{\mathsf{abc}}
    \cdot
    \mathsf{cC}
    \cdot
    A_{t}
    \cdot
    \ket*{ \mathcal{A}^{W_{\mathsf{abc}}, \mathfrak{D}}_{t-1} } \\
    \ket*{ \mathcal{A}^{W_{\mathsf{abc}}', \mathfrak{D}}_{t} }
    & =
    \mathsf{cD}
    \cdot
    \Pi^{\mathsf{locbij}}
    \cdot W_{\mathsf{abc}} \cdot
    \Pi^{\notin \mathsf{Dom} \rightarrow \notin \mathsf{locDom}}
    \cdot
    \mathsf{cC}
    \cdot
    A_{t}
    \cdot
    \ket*{ \mathcal{A}^{W_{\mathsf{abc}}', \mathfrak{D}}_{t-1} }
\end{align}
We have 
\begin{equation} \label{eq: step 1}
\begin{split}
    \mathsf{TD}_{t} \leq \mathsf{TD}_{t} 
    & + 2 \left\lVert \big( 1- \Pi^{\notin \mathsf{Dom} \rightarrow \notin \mathsf{locDom}} \big) \cdot
    \mathsf{cC} \cdot A_{t} 
    \ket*{ \mathcal{A}^{W_{\mathsf{abc}}, \mathfrak{D}}_{t-1} }
     \right\rVert_2 \\
    & \quad + 2 \left\lVert
    \big( 1 - \Pi^{\mathsf{locbij}} \big)
    \cdot W_{\mathsf{abc}} \cdot
    \Pi^{\notin \mathsf{Dom} \rightarrow \notin \mathsf{locDom}}
    \cdot
    \mathsf{cC}
    \cdot
    A_{t}
    \cdot
    \ket*{ \mathcal{A}^{W_{\mathsf{abc}}', \mathfrak{D}}_{t-1} }
    \right\rVert_2, \\
\end{split}
\end{equation}
where the first term accounts for the error up to time $t$, the second term for the error induced by the projector $\Pi^{\notin \mathsf{Dom} \rightarrow \notin \mathsf{locDom}}$ [using Eq.~(\ref{eq: state to op bound})], and the third term for the error induced by the projector $\Pi^{\mathsf{locbij}}$ [again using Eq.~(\ref{eq: state to op bound})].

We bound the second term as follows.
From Eq.~(9.55) and Claim 18 of Ref.~\cite{ma2024construct} (see also the modification of Claim 18 to strong approximate designs in the proof of Lemma~\ref{lem:closeness-AWD-and-PhiVt}), we have
\begin{equation}
    \left\lVert 
    \ket*{ \mathcal{A}^{W_{\mathsf{abc}}, \mathfrak{D}}_{t-1} }
    -
    \mathsf{cQ} \cdot \ket*{ \mathcal{A}^{V_{\mathsf{abc}}}_{t-1} }
    \right\rVert_2 \leq \frac{\sqrt{70}t}{N_{\mathsf{abc}}^{1/8}} +  2 t^{1/4} \varepsilon^{1/4}.
\end{equation}
This yields,
\begin{equation}
\begin{split}
    & \left\lVert \big( 1- \Pi^{\notin \mathsf{Dom} \rightarrow \notin \mathsf{locDom}} \big)
    \cdot \mathsf{cC} \cdot A_{t} \cdot
    \ket*{ \mathcal{A}^{W_{\mathsf{abc}}, \mathfrak{D}}_{t-1} }
     \right\rVert_2  \\
     & \quad  \quad \quad \quad \quad  \quad \leq 
     \left\lVert \big( 1- \Pi^{\notin \mathsf{Dom} \rightarrow \notin \mathsf{locDom}} \big) \cdot
    \mathsf{cC} \cdot A_{t} \cdot
    \mathsf{cQ} \cdot \ket*{ \mathcal{A}^{V_{\mathsf{abc}}}_{t-1} }
     \right\rVert_2
     +
     \frac{\sqrt{70}t}{N_{\mathsf{abc}}^{1/8}} + +  2 t^{1/4} \varepsilon^{1/4}.
\end{split}
\end{equation}
The latter state norm can be written out explicitly, as
\begin{equation}
\begin{split}
    & \left\lVert \big( 1- \Pi^{\notin \mathsf{Dom} \rightarrow \notin \mathsf{locDom}} \big) \cdot
    \mathsf{cC} \cdot A_{t} \cdot
    \mathsf{cQ} \cdot \ket*{ \mathcal{A}^{V_{\mathsf{abc}}}_{t-1} }
     \right\rVert_2  \\
     & \quad \quad  \quad \quad \quad  \quad =
     \sqrt{
     \bra*{ \mathcal{A}^{V_{\mathsf{abc}}}_{t-1} }
      \cdot A_{t}^\dagger \cdot 
     \mathsf{cQ}^\dagger \cdot 
     \mathsf{cC}^\dagger \cdot 
     \big( 1- \Pi^{\notin \mathsf{Dom} \rightarrow \notin \mathsf{locDom}} \big) 
     \cdot \mathsf{cC} \cdot 
    \mathsf{cQ}   \cdot A_{t} \cdot  \ket*{ \mathcal{A}^{V_{\mathsf{abc}}}_{t-1} }
     }.
\end{split}
\end{equation}
where we used that $A_{t}$ and $\mathsf{cQ}$ act on distinct registers to commute them past one another.

To proceed, we first apply the operator inequality,
\begin{equation}
    1 - \Pi^{\notin \mathsf{Dom} \rightarrow \notin \mathsf{locDom}}  \preceq 
    \sum_\alpha \sum_{i \in [\ell]} 
    \Pi^{\mathsf{eq}}_{\mathsf{A}_{\alpha} \mathsf{L}_{\mathsf{X_\alpha, i}}^{(\ell)}}
    \Pi^{\mathsf{neq}}_{\mathsf{A}_{\bar \alpha} \mathsf{L}_{\mathsf{X_{\bar \alpha}, i}}^{(\ell)}}
    +
    \sum_\alpha \sum_{j \in [r]} 
    \Pi^{\mathsf{eq}}_{\mathsf{A}_{\alpha} \mathsf{R}_{\mathsf{X_\alpha, j}}^{(r)}}
    \Pi^{\mathsf{neq}}_{\mathsf{A}_{\bar \alpha} \mathsf{R}_{\mathsf{X_{\bar \alpha}, j}}^{(r)}},
\end{equation} 
where $\Pi^{\mathsf{eq}}_{\mathsf{A}_{\alpha} \mathsf{L}_{\mathsf{X_\alpha, i}}^{(\ell)}}$ projects onto states with the same bitstring on $\mathsf{A}_\alpha$ as on $\mathsf{L}_{\mathsf{X_{\bar \alpha}, i}}^{(\ell)}$, and 
\begin{equation}
    \Pi^{\mathsf{neq}}_{\mathsf{A}_{\bar \alpha} \mathsf{L}_{\mathsf{X_{\bar \alpha}, i}}^{(\ell)}} \equiv 1 - \Pi^{\mathsf{eq}}_{\mathsf{A}_{\bar \alpha} \mathsf{L}_{\mathsf{X_{\bar \alpha}, i}}^{(\ell)}}
\end{equation} 
does the reverse on $\bar \alpha$ (where we define $\bar \alpha \equiv \mathsf{bc}$ if $\alpha = \mathsf{a}$, and analogous for $\alpha = \mathsf{b},\mathsf{c}$).
For each individual term with $i \in [\ell]$, we have
\begin{equation} \label{eq: EV eq neq 1}
\begin{split}
    & \bra*{ \mathcal{A}^{V_{\mathsf{abc}}}_{t-1} }
      \cdot A_{t}^\dagger \cdot 
     \mathsf{cQ}_{\gsC \gsD \gsLabc \gsRabc}^\dagger \cdot 
     \mathsf{cC}_{\gsC \gsA}^\dagger \cdot 
     \Pi^{\mathsf{eq}}_{\mathsf{A}_{\alpha} \mathsf{L}_{\mathsf{X_\alpha, i}}^{(\ell)}}
    \Pi^{\mathsf{neq}}_{\mathsf{A}_{\bar \alpha} \mathsf{L}_{\mathsf{X_{\bar \alpha}, i}}^{(\ell)}}
     \cdot \mathsf{cC}_{\gsC \gsA} \cdot 
    \mathsf{cQ}_{\gsC \gsD \gsLabc \gsRabc}   \cdot A_{t} \cdot  \ket*{ \mathcal{A}^{V_{\mathsf{abc}}}_{t-1} }  \\
     & \quad \quad  \quad \quad  \quad =
     \bra*{ \mathcal{A}^{V_{\mathsf{abc}}}_{t-1} }
      \cdot A_{t}^\dagger \cdot 
     \mathsf{cC}_{\gsC 
    {\textcolor{gray}{\mathsf{L}_{\mathsf{X_{\bar \alpha}, i}}^{(\ell)}}}}^\dagger \cdot 
     \mathsf{cC}_{\gsC \gsA}^\dagger \cdot 
     \Pi^{\mathsf{eq}}_{\mathsf{A}_{\alpha} \mathsf{L}_{\mathsf{X_\alpha, i}}^{(\ell)}}
    \Pi^{\mathsf{neq}}_{\mathsf{A}_{\bar \alpha} \mathsf{L}_{\mathsf{X_{\bar \alpha}, i}}^{(\ell)}} 
     \cdot 
     \mathsf{cC}_{\gsC \gsA} \cdot 
    \mathsf{cC}_{\gsC 
    {\textcolor{gray}{\mathsf{L}_{\mathsf{X_{\bar \alpha}, i}}^{(\ell)}}}}
    \cdot A_{t} \cdot  \ket*{ \mathcal{A}^{V_{\mathsf{abc}}}_{t-1} },
\end{split}
\end{equation}
where all but one of the Clifford unitaries in $\mathsf{cQ}$ cancel, since the middle term in the expectation value acts only on register $\mathsf{L}_{\mathsf{X_\alpha, i}}$. 
For terms with $j \in [r]$, we have instead
\begin{equation} \label{eq: EV eq neq 2}
\begin{split}
    & \bra*{ \mathcal{A}^{V_{\mathsf{abc}}}_{t-1} }
      \cdot A_{t}^\dagger \cdot 
     \mathsf{cQ}_{\gsC \gsD \gsLabc \gsRabc}^\dagger \cdot 
     \mathsf{cC}_{\gsC \gsA}^\dagger \cdot 
     \Pi^{\mathsf{eq}}_{\mathsf{A}_{\alpha} \mathsf{R}_{\mathsf{X_\alpha, j}}^{(r)}}
    \Pi^{\mathsf{neq}}_{\mathsf{A}_{\bar \alpha} \mathsf{R}_{\mathsf{X_{\bar \alpha}, j}}^{(r)}}
     \cdot \mathsf{cC}_{\gsC \gsA} \cdot 
    \mathsf{cQ}_{\gsC \gsD \gsLabc \gsRabc}   \cdot A_{t} \cdot  \ket*{ \mathcal{A}^{V_{\mathsf{abc}}}_{t-1} }  \\
     & \quad \quad  \quad \quad  \quad =
     \bra*{ \mathcal{A}^{V_{\mathsf{abc}}}_{t-1} }
      \cdot A_{t}^\dagger \cdot 
     \mathsf{c}\overline{\mathsf{C}}_{\gsC 
    {\textcolor{gray}{\mathsf{R}_{\mathsf{X_{\bar \alpha}, j}}^{(r)}}}}^\dagger \cdot 
     \mathsf{cC}_{\gsC \gsA}^\dagger \cdot 
     \Pi^{\mathsf{eq}}_{\mathsf{A}_{\alpha} \mathsf{R}_{\mathsf{X_\alpha, j}}^{(r)}}
    \Pi^{\mathsf{neq}}_{\mathsf{A}_{\bar \alpha} \mathsf{R}_{\mathsf{X_{\bar \alpha}, j}}^{(r)}} 
     \cdot 
     \mathsf{cC}_{\gsC \gsA} \cdot 
    \mathsf{c}\overline{\mathsf{C}}_{\gsC 
    {\textcolor{gray}{\mathsf{R}_{\mathsf{X_{\bar \alpha}, j}}^{(r)}}}}
    \cdot A_{t} \cdot  \ket*{ \mathcal{A}^{V_{\mathsf{abc}}}_{t-1} }.
\end{split}
\end{equation}
We can upper bound the latter expectation values by performing the twirl over $C$.
From Eq.~(\ref{eq: clifford twirl eq neq 1}) and Eq.~(\ref{eq: clifford twirl eq neq 2}), this yields an upper bound of $1/N_\alpha$ on both Eq.~(\ref{eq: EV eq neq 1}) and Eq.~(\ref{eq: EV eq neq 2}).
Therefore, in total, we have an upper bound
\begin{equation} \nonumber
    \left\lVert \big( 1- \Pi^{\notin \mathsf{Dom} \rightarrow \notin \mathsf{locDom}} \big) \cdot
    \mathsf{cC} \cdot A_{t} \cdot
    \mathsf{cQ} \cdot \ket*{ \mathcal{A}^{V_{\mathsf{abc}}}_{t-1} }
     \right\rVert_2
     \leq 
     \sqrt{ (\ell+r) \sum_\alpha (1/N_\alpha + \varepsilon)}
     \leq \sqrt{3 t / \min_\alpha N_\alpha + 3t\varepsilon}.
\end{equation}

The third term in Eq.~(\ref{eq: step 1}) is simpler to bound.
The input state to $W_{\mathsf{abc}}$ lies in the subspace $\Pi^{\mathsf{locbij}}_{\leq t}$ by construction.
Therefore, the output of $W_{\mathsf{abc}}$ lies in the subspace $\Pi^{\mathcal{I}( W_{\mathsf{abc}} \Pi^{\mathsf{locbij}}_{\leq t})}$.
This latter subspace is spanned by two classes of states.
The first class is,
\begin{equation}
    \ket{ y_\mathsf{a} y_\mathsf{b} y_\mathsf{c} }_{\gsA}
    \ket{ L }_{\gsLabc}
    \ket{ R }_{\gsRabc},
\end{equation}
for any $\ell+r \leq t$, where $\mathsf{Dom}(L \cup R)_\alpha$ is distinct, $\mathsf{Im}(L \cup R)_\alpha$ is distinct, and $y_\alpha \notin \mathsf{Im}(L \cup R)_\alpha$.
 These arise if the $W_{\mathsf{abc}}^{R,\dagger}$ branch of $W_{\mathsf{abc}}$ is applied. 
 The second class is,
\begin{equation}
    \frac{1}{\sqrt{N_{\mathsf{abc}}-\ell-r}}
    \sum_{y_\mathsf{a} y_\mathsf{b} y_\mathsf{c} \notin \mathsf{Im}(L\cup R)}
    \ket{ y_\mathsf{a} y_\mathsf{b} y_\mathsf{c} }_{\gsA}
    \ket{ L \cup ( x_\mathsf{a} x_\mathsf{b} x_\mathsf{c} , y_\mathsf{a} y_\mathsf{b} y_\mathsf{c} ) }_{\gsLabc}
    \ket{ R }_{\gsRabc},
\end{equation}
for $\ell+r \leq t$,  where $\mathsf{Dom}(L \cup R)_\alpha$ is distinct, $\mathsf{Im}(L \cup R)_\alpha$ is distinct, and $x_\alpha \notin \mathsf{Dom}(L \cup R)_\alpha$.
These arise if the $W_{\mathsf{abc}}^L$ branch of $W_{\mathsf{abc}}$ is applied.
The states above are mutually orthogonal to one another as well as between different $\ell,r$.

The first class of states is invariant under $\Pi^{\mathsf{locbij}}$.
Thus, the projector $\Pi^{\mathsf{locbij}}$ acts trivially and incurs no error.
Meanwhile, on the second class of states, we have
\begin{equation}
\begin{split}
    & \Pi^{\mathsf{locbij}}
    \frac{1}{\sqrt{N_{\mathsf{abc}}-\ell-r}}
    \sum_{y_\mathsf{a} y_\mathsf{b} y_\mathsf{c} \notin \mathsf{Im}(L\cup R)}
    \ket{ y_\mathsf{a} y_\mathsf{b} y_\mathsf{c} }_{\gsA}
    \ket{ L \cup ( x_\mathsf{a} x_\mathsf{b} x_\mathsf{c} , y_\mathsf{a} y_\mathsf{b} y_\mathsf{c} ) }_{\gsLabc}
    \ket{ R }_{\gsRabc} \\
    & \quad \quad \quad \quad \quad \quad \quad = 
    \frac{1}{\sqrt{N_{\mathsf{abc}}-\ell-r}}
    \sum_{\substack{ y_\mathsf{a} \notin \mathsf{Im}(L\cup R)_\mathsf{a} \\ y_\mathsf{b} \notin \mathsf{Im}(L\cup R)_\mathsf{b} \\ y_\mathsf{c} \notin \mathsf{Im}(L\cup R)_\mathsf{c}}}
    \ket{ y_\mathsf{a} y_\mathsf{b} y_\mathsf{c} }_{\gsA}
    \ket{ L \cup ( x_\mathsf{a} x_\mathsf{b} x_\mathsf{c} , y_\mathsf{a} y_\mathsf{b} y_\mathsf{c} ) }_{\gsLabc}
    \ket{ R }_{\gsRabc}.
\end{split}
\end{equation}
The final state is orthogonal to the first class of states, as well as between different $\ell,r$.
The state has norm $( \prod_\alpha (N_\alpha - \ell - r) ) / (N-\ell-r) \geq 1 - 3 t / \min_\alpha N_\alpha$.
The above analysis establishes that $\Pi^{\mathsf{locbij}} \Pi^{\mathcal{I}( W_{\mathsf{abc}} \Pi^{\mathsf{locbij}}_{\leq t})}$ is block diagonal between the two classes of input and output states, as well as between different $\ell, r$.
Therefore, the desired error is given by the maximum error within each block.
From the above, the maximum is achieved at $\ell + r = t$, which yields,
\begin{equation}
    \left\lVert \big( 1 - \Pi^{\mathsf{locbij}} \big)
    \cdot W_{\mathsf{abc}} \cdot
    \Pi^{\notin \mathsf{Dom} \rightarrow \notin \mathsf{locDom}}
    \cdot
    \mathsf{cC}
    \cdot
    A_{t}
    \cdot
    \ket*{ \mathcal{A}^{W_{\mathsf{abc}}', \mathfrak{D}}_{t-1} }
    \right\rVert_2 \leq \sqrt{3t/\min_\alpha N_\alpha}.
\end{equation}

In total, we have shown that the error in Eq.~(\ref{eq: step 1}) is upper bounded by,
\begin{equation} 
    \mathsf{TD}_{t} \leq \mathsf{TD}_{t} 
    + \frac{2\sqrt{70}t}{N_{\mathsf{abc}}^{1/8}} + 2t^{1/4} \varepsilon^{1/4} + 4 \sqrt{\frac{3t}{\min_\alpha N_\alpha}} + 2\sqrt{3t\varepsilon}
    \leq
    \frac{2 \sqrt{70} t^2}{N_{\mathsf{abc}}^{1/8}} + \frac{4\sqrt{3} t^{3/2}}{(\min_\alpha N_\alpha)^{1/2}} + 6 t^{5/4} \varepsilon^{1/4} 
    , \\
\end{equation}
applying the inductive hypothesis.
This completes our proof.

\subsubsection{Projected $W_{\mathsf{abc}}$ is indistinguishable from projected $W_{\mathsf{bc}} W_{\mathsf{ab}}$} \label{sec: projected WABC to projected WAB WBC}

We will prove that 
\begin{equation} \label{eq: 2-norm bound 2to3}
    \left\lVert 
    \mathsf{Compress}^\dagger 
    \ket*{ \mathcal{A}^{W_{\mathsf{abc}}', \mathfrak{D}}_t }
    -
    \ket*{ \mathcal{A}^{(W_{\mathsf{bc}} W_{\mathsf{ab}})', \mathfrak{D}}_t }
    \right\rVert_2 \leq \frac{t(t-1)}{2N_{\mathsf{ab}}}+\frac{t(t-1)}{2N_{\mathsf{bc}}} 
\end{equation}
Using Eq.~(\ref{eq: state to op bound}), this implies that 
\begin{equation}
    \mathsf{TD}\left( 
    \mathsf{Compress}^\dagger  \dyad*{ \mathcal{A}^{W_{\mathsf{abc}}', \mathfrak{D}}_t }
     \! \mathsf{Compress} 
    , 
    \dyad*{ \mathcal{A}^{(W_{\mathsf{bc}} W_{\mathsf{ab}})', \mathfrak{D}}_t }
    \right)
    \leq 
    \frac{t^2}{N_{\mathsf{ab}}}+\frac{t^2}{N_{\mathsf{bc}}} 
\end{equation}
which implies that $\lVert \rho^{(2)} - \rho^{(3)} \rVert_1$ is less than the same value since the 1-norm cannot increase after tracing out $\mathsf{L}_{\mathsf{ab}} \mathsf{L}_{\mathsf{bc}} \mathsf{R}_{\mathsf{ab}} \mathsf{R}_{\mathsf{bc}} \mathsf{C} \mathsf{D}$.

We proceed by induction.
The statement holds trivially at $t=0$.
For the inductive step, we assume that Eq.~(\ref{eq: 2-norm bound 2to3}) holds up to time $t-1$.
We will show that the claim holds for time $t$ as well.
Without loss of generality, we assume that the forward unitary is applied at time $t$.
The case when the inverse or conjugate or transpose are applied follow by symmetric arguments.

The states at time $t$ are obtained from the states at time $t-1$ as follows,
\begin{align}
    \ket*{ \mathcal{A}^{W_{\mathsf{abc}}', \mathfrak{D}}_{t} }
    & =
    \mathsf{cD}
    \cdot
    W_{\mathsf{abc}}'
    \cdot
    \mathsf{cC}
    \cdot
    A_{t}
    \cdot
    \ket*{ \mathcal{A}^{W_{\mathsf{abc}}', \mathfrak{D}}_{t-1} } \\
    \ket*{ \mathcal{A}^{(W_{\mathsf{bc}} W_{\mathsf{ab}})', \mathfrak{D}}_{t} }
    & =
    \mathsf{cD}
    \cdot
    (W_{\mathsf{bc}} W_{\mathsf{ab}})'
    \cdot
    \mathsf{cC}
    \cdot
    A_{t}
    \cdot
    \ket*{ \mathcal{A}^{(W_{\mathsf{bc}} W_{\mathsf{ab}})', \mathfrak{D}}_{t-1} }
\end{align}
The final projection in $W_{\mathsf{abc}}'$ and $(W_{\mathsf{bc}} W_{\mathsf{ab}})'$ guarantees that the input state to the $t$-th application of $W_{\mathsf{abc}}'$ and $(W_{\mathsf{bc}} W_{\mathsf{ab}})'$ obeys,
\begin{align}
    \Pi^{\mathsf{locbij}}_{\gsLabc \gsRabc}
    \cdot \mathsf{cC} \cdot
    A_{t} \cdot
    \ket*{ \mathcal{A}^{W_{\mathsf{abc}}', \mathfrak{D}}_{t-1} }
    & =
    \mathsf{cC} \cdot
    A_{t} \cdot
    \ket*{ \mathcal{A}^{W_{\mathsf{abc}}', \mathfrak{D}}_{t-1} } \\
    \Pi^{\mathsf{paired}}_{\gsLab \gsLbc \gsRab \gsRbc}
    \cdot \mathsf{cC} \cdot
    A_{t} \cdot
    \ket*{ \mathcal{A}^{(W_{\mathsf{bc}} W_{\mathsf{ab}})', \mathfrak{D}}_{t-1} }
    & =
    \mathsf{cC} \cdot
    A_{t} \cdot
    \ket*{ \mathcal{A}^{(W_{\mathsf{bc}} W_{\mathsf{ab}})', \mathfrak{D}}_{t-1} }.
\end{align}
We will now analyze the action of first $W_{\mathsf{abc}}'$ and then $(W_{\mathsf{bc}} W_{\mathsf{ab}})'$.

The domain of $W_{\mathsf{abc}}' \Pi^{\mathsf{locbij}}$ contains two classes of states, corresponding to the domain of
\begin{equation}
    \Pi^{\mathcal{D}(W_{\mathsf{abc}}')} 
    \Pi^{\mathsf{locbij}}
    = 
    \Pi^{\mathcal{D}(W_{\mathsf{abc}})} 
    \Pi^{\notin \mathsf{Dom} \rightarrow \notin \mathsf{locDom}} 
    \Pi^{\mathsf{locbij}}
    =
    \Pi^{\notin \mathsf{locDom}} 
    \Pi^{\mathsf{locbij}} 
    +
    \Pi^{\mathcal{I}(W^R_{\mathsf{abc}})}
    \Pi^{\mathsf{locbij}}.
\end{equation}
The second equality follows from the domain of $W_{\mathsf{abc}}$,
\begin{equation}
    \Pi^{\mathcal{D}(W_{\mathsf{abc}})} = 
    \Pi^{\notin \mathsf{Dom}}
    \Pi^{\mathsf{bij}}
    +
    \Pi^{\mathcal{I}(W^R_{\mathsf{abc}})}
    .
\end{equation}
We will consider the action on each class of states separately.
Focusing on the $\mathsf{A} \mathsf{L}_{\mathsf{abc}} \mathsf{R}_{\mathsf{abc}}$ registers, a complete basis for the first class of states is given by, 
\begin{equation}
    \ket{ x_\mathsf{a} x_\mathsf{b} x_\mathsf{c} }_{\gsA}
    \ket{ L }_{\gsLabc}
    \ket{ R }_{\gsRabc},
\end{equation}
where $\mathsf{Dom}(L \cup R)_\alpha$ is distinct, $\mathsf{Im}(L \cup R)_\alpha$ is distinct, and $x_\alpha \notin \mathsf{Dom}(L \cup R)_\alpha$.
We then have,
\begin{equation} \label{eq: abc result 1 before compress}
\begin{split}
    W_{\mathsf{abc}}' &
    \ket{ x_\mathsf{a} x_\mathsf{b} x_\mathsf{c} }_{\gsA}
    \ket{ L }_{\gsLabc}
    \ket{ R }_{\gsRabc}
    = 
    \Pi^{\mathsf{locbij}} W_{\mathsf{abc}}
    \ket{ x_\mathsf{a} x_\mathsf{b} x_\mathsf{c} }_{\gsA}
    \ket{ L }_{\gsLabc}
    \ket{ R }_{\gsRabc} \\
    & = 
    \Pi^{\mathsf{locbij}}
    \frac{1}{\sqrt{N_{\mathsf{abc}}-\ell-r}}
    \sum_{y_\mathsf{a} y_\mathsf{b} y_\mathsf{c} \notin \mathsf{Im}(L\cup R)}
    \ket{ y_\mathsf{a} y_\mathsf{b} y_\mathsf{c} }_{\gsA}
    \ket{ L \cup ( x_\mathsf{a} x_\mathsf{b} x_\mathsf{c} , y_\mathsf{a} y_\mathsf{b} y_\mathsf{c} ) }_{\gsLabc}
    \ket{ R }_{\gsRabc} \\
    & = 
    \frac{1}{\sqrt{N_{\mathsf{abc}}-\ell-r}}
    \sum_{\substack{ y_\mathsf{a} \notin \mathsf{Im}(L\cup R)_\mathsf{a} \\ y_\mathsf{b} \notin \mathsf{Im}(L\cup R)_\mathsf{b} \\ y_\mathsf{c} \notin \mathsf{Im}(L\cup R)_\mathsf{c}}}
    \ket{ y_\mathsf{a} y_\mathsf{b} y_\mathsf{c} }_{\gsA}
    \ket{ L \cup ( x_\mathsf{a} x_\mathsf{b} x_\mathsf{c} , y_\mathsf{a} y_\mathsf{b} y_\mathsf{c} ) }_{\gsLabc}
    \ket{ R }_{\gsRabc}. \\
\end{split}
\end{equation}

Meanwhile, a complete basis for the second class of states is given by 
\begin{equation}
    \frac{1}{\sqrt{\prod_\alpha (N_\alpha - \ell - r)}} 
    \sum_{\substack{ x_\mathsf{a} \notin \mathsf{Dom}(L\cup R)_\mathsf{a} \\ x_\mathsf{b} \notin \mathsf{Dom}(L\cup R)_\mathsf{b} \\ x_\mathsf{c} \notin \mathsf{Dom}(L\cup R)_\mathsf{c}}}
    \ket{ x_\mathsf{a} x_\mathsf{b} x_\mathsf{c} }_{\gsA}
    \ket{ L }_{\gsLabc}
    \ket{ R \cup ( x_\mathsf{a} x_\mathsf{b} x_\mathsf{c} , y_\mathsf{a} y_\mathsf{b} y_\mathsf{c} ) }_{\gsRabc},
\end{equation}
where $\mathsf{Dom}(L \cup R)_\alpha$ is distinct, $\mathsf{Im}(L \cup R)_\alpha$ is distinct, and $y_\alpha \notin \mathsf{Im}(L \cup R)_\alpha$.
We then have,
\begin{equation} \label{eq: abc result 2 before compress}
\begin{split}
    & W_{\mathsf{abc}}'
    \cdot \frac{1}{\sqrt{\prod_\alpha (N_\alpha - \ell - r)}} 
    \sum_{\substack{ x_\mathsf{a} \notin \mathsf{Dom}(L\cup R)_\mathsf{a} \\ x_\mathsf{b} \notin \mathsf{Dom}(L\cup R)_\mathsf{b} \\ x_\mathsf{c} \notin \mathsf{Dom}(L\cup R)_\mathsf{c}}}
    \ket{ x_\mathsf{a} x_\mathsf{b} x_\mathsf{c} }_{\gsA}
    \ket{ L }_{\gsLabc}
    \ket{ R \cup ( x_\mathsf{a} x_\mathsf{b} x_\mathsf{c} , y_\mathsf{a} y_\mathsf{b} y_\mathsf{c} ) }_{\gsRabc} \\
    & \,\,\,\, = 
    \Pi^{\mathsf{locbij}} 
    W^{R,\dagger}_{\mathsf{abc}} 
    \cdot \frac{1}{\sqrt{\prod_\alpha (N_\alpha - \ell - r)}} 
    \sum_{\substack{ x_\mathsf{a} \notin \mathsf{Dom}(L\cup R)_\mathsf{a} \\ x_\mathsf{b} \notin \mathsf{Dom}(L\cup R)_\mathsf{b} \\ x_\mathsf{c} \notin \mathsf{Dom}(L\cup R)_\mathsf{c}}}
    \ket{ x_\mathsf{a} x_\mathsf{b} x_\mathsf{c} }_{\gsA}
    \ket{ L }_{\gsLabc}
    \ket{ R \cup ( x_\mathsf{a} x_\mathsf{b} x_\mathsf{c} , y_\mathsf{a} y_\mathsf{b} y_\mathsf{c} ) }_{\gsRabc} \\
    & \,\,\,\, = 
    \left( \frac{\sqrt{\prod_\alpha (N_\alpha - \ell - r)}}{\sqrt{N_{\mathsf{abc}}-\ell-r}} \right)
    \Pi^{\mathsf{locbij}} 
    \ket{ y_\mathsf{a} y_\mathsf{b} y_\mathsf{c} }_{\gsA}
    \ket{ L }_{\gsLabc}
    \ket{ R }_{\gsRabc} \\
    & \,\,\,\, = 
    \left( \frac{\sqrt{\prod_\alpha (N_\alpha - \ell - r)}}{\sqrt{N_{\mathsf{abc}}-\ell-r}} \right)
    \ket{ y_\mathsf{a} y_\mathsf{b} y_\mathsf{c} }_{\gsA}
    \ket{ L }_{\gsLabc}
    \ket{ R }_{\gsRabc}, \\
\end{split}
\end{equation}
where the factor in parentheses arises from the decrease in normalization when the projector $\Pi^{\mathcal{I}(W^R_{\mathsf{abc}})}$ is applied (through $W^{R,\dagger}_{\mathsf{abc}} = W^{R,\dagger}_{\mathsf{abc}} \Pi^{\mathcal{I}(W^R_{\mathsf{abc}})}$).

Let us now turn to $(W_{\mathsf{bc}} W_{\mathsf{ab}})'$.
The input state to $W_{\mathsf{bc}} W_{\mathsf{ab}}$ is contained within the domain of
\begin{equation}
    \widetilde{\Pi}^{\mathcal{D}(W_{\mathsf{abc}})} 
    \widetilde{\Pi}^{\notin \mathsf{Dom} \rightarrow \notin \mathsf{locDom}} 
    \Pi^{\mathsf{paired}}
    =
    \widetilde{\Pi}^{\notin \mathsf{locDom}} 
    \Pi^{\mathsf{paired}}
    +
    \widetilde{\Pi}^{\mathcal{I}(W^R_{\mathsf{abc}})}
    \Pi^{\mathsf{paired}}.
\end{equation}
As before, we must consider two classes of input states, corresponding to the domain of each term on the right side above.
Focusing on the $\mathsf{L}_{\mathsf{ab}} \mathsf{L}_{\mathsf{bc}} \mathsf{R}_{\mathsf{ab}} \mathsf{R}_{\mathsf{bc}}$ registers, a complete basis for the first class of states is given by, 
\begin{equation}
    \frac{1}{\sqrt{ N_\mathsf{b}^{\ell+r} }} \sum_{z_\ell, z_r}
    \ket{ x_\mathsf{a} x_\mathsf{b} x_\mathsf{c} }_{\gsA}
    \ket{ L_{\mathsf{ab}}^{z_\ell} }_{\gsLab}
    \ket{ L_{\mathsf{bc}}^{z_\ell} }_{\gsLbc}
    \ket{ R_{\mathsf{ab}}^{z_r} }_{\gsRab}
    \ket{ R_{\mathsf{bc}}^{z_r} }_{\gsRbc},
\end{equation}
where $\mathsf{Dom}(L \cup R)_\alpha$ is distinct, $\mathsf{Im}(L \cup R)_\alpha$ is distinct, and $x_\alpha \notin \mathsf{Dom}(L \cup R)_\alpha$.
(We refer to Section~\ref{sec: preliminaries} for the definitions of $L_{\mathsf{ab}}^{z_\ell}$, $L_{\mathsf{bc}}^{z_\ell}$, $R_{\mathsf{ab}}^{z_r}$, $R_{\mathsf{bc}}^{z_r}$, $L$, $R$.)
Focusing on the $\mathsf{A} \mathsf{L}_{\mathsf{ab}} \mathsf{R}_{\mathsf{ab}}$ registers, the application of $W_{\mathsf{ab}}$ gives,
\begin{equation}
\begin{split}
    W_{\mathsf{ab}} 
    \ket{ x_\mathsf{a} x_\mathsf{b} x_\mathsf{c} }_{\gsA}
    &
    \ket{L^{z_\ell}_{\mathsf{ab}} }_{\gsLab}
    \ket{R^{z_r}_{\mathsf{ab}} }_{\gsRab} \\
    & = 
    \frac{1}{\sqrt{N_{\mathsf{ab}}-\ell-r}}
    \sum_{y_\mathsf{a} z_\mathsf{b} \notin \mathsf{Im}(L_{\mathsf{ab}} \cup R_{\mathsf{ab}})}
    \ket{ y_\mathsf{a} z_\mathsf{b} x_\mathsf{c} }_{\gsA}
    \ket{L^{z_\ell}_{\mathsf{ab}} \cup ( x_\mathsf{a} x_\mathsf{b} , y_\mathsf{a} z_\mathsf{b} ) }_{\gsLab}
    \ket{R^{z_r}_{\mathsf{ab}} }_{\gsRab}.
\end{split}
\end{equation}
Focusing on the $\mathsf{A} \mathsf{L}_{\mathsf{bc}} \mathsf{R}_{\mathsf{bc}}$ registers, the ensuing application of $W_{\mathsf{bc}}$ gives,
\begin{equation}
\begin{split}
    W_{\mathsf{bc}} 
    \ket{ y_\mathsf{a} z_\mathsf{b} x_\mathsf{c} }_{\gsA}
    &
    \ket{L^{z_\ell}_{\mathsf{bc}} }_{\gsLbc}
    \ket{R^{z_r}_{\mathsf{bc}} }_{\gsRbc} \\
    & = 
    \frac{1}{\sqrt{N_{\mathsf{bc}}-\ell-r}}
    \sum_{y_\mathsf{b} y_\mathsf{c} \notin \mathsf{Im}(L_{\mathsf{bc}} \cup R_{\mathsf{bc}})}
    \ket{ y_\mathsf{a} y_\mathsf{b} y_\mathsf{c} }_{\gsA}
    \ket{L^{z_\ell}_{\mathsf{bc}} \cup ( z_\mathsf{b} x_\mathsf{c} , y_\mathsf{b} y_\mathsf{c} ) }_{\gsLbc}
    \ket{R^{z_r}_{\mathsf{bc}} }_{\gsRbc}.
\end{split}
\end{equation}
In total, we have the state,
\begin{equation}
    \frac{1}{\sqrt{\mathcal{N}_1}} \sum_{ 
    \substack{ z_\ell, z_r \\ y_\mathsf{a} z_\mathsf{b} \notin \mathsf{Im}(L_{\mathsf{ab}} \cup R_{\mathsf{ab}}) \\ y_\mathsf{b} y_\mathsf{c} \notin \mathsf{Im}(L_{\mathsf{bc}} \cup R_{\mathsf{bc}})}
    }
    \ket{ y_\mathsf{a} y_\mathsf{b} y_\mathsf{c} }_\gsA
    \ket{ L_{\mathsf{ab}}^{z_\ell} \cup ( x_\mathsf{a} x_\mathsf{b} , y_\mathsf{a} z_\mathsf{b} )}_\gsLab
    \ket{ L_{\mathsf{bc}}^{z_\ell} \cup ( z_\mathsf{b} x_\mathsf{c} , y_\mathsf{b} y_\mathsf{c} ) }_\gsLbc
    \ket{ R_{\mathsf{ab}}^{z_r} }_\gsRab
    \ket{ R_{\mathsf{bc}}^{z_r} }_\gsRbc,
\end{equation}
where $\mathcal{N}_1 \equiv N_\mathsf{b}^{\ell+r} (N_{\mathsf{ab}}-\ell-r) (N_{\mathsf{bc}}-\ell-r)$.
Applying the final projection $\Pi^{\mathsf{paired}}$ forces the $y_\alpha$ to be locally distinct, which yields,
\begin{equation}
    \frac{1}{\sqrt{\mathcal{N}_1}} \sum_{ 
    \substack{ z_\ell, z_r, z_\mathsf{b} \\ y_\mathsf{a} \notin \mathsf{Im}(L \cup R)_\mathsf{a} \\ y_\mathsf{b} \notin \mathsf{Im}(L \cup R)_\mathsf{b} \\ y_\mathsf{c} \notin \mathsf{Im}(L \cup R)_\mathsf{c}}
    }
    \ket{ y_\mathsf{a} y_\mathsf{b} y_\mathsf{c} }_\gsA
    \ket{ L_{\mathsf{ab}}^{z_\ell} \cup ( x_\mathsf{a} x_\mathsf{b} , y_\mathsf{a} z_\mathsf{b} )}_\gsLab
    \ket{ L_{\mathsf{bc}}^{z_\ell} \cup ( z_\mathsf{b} x_\mathsf{c} , y_\mathsf{b} y_\mathsf{c} ) }_\gsLbc
    \ket{ R_{\mathsf{ab}}^{z_r} }_\gsRab
    \ket{ R_{\mathsf{bc}}^{z_r} }_\gsRbc,
\end{equation}
where the sum over $z_\mathsf{b}$ is unrestricted, since $y_{\mathsf{a}} \notin \mathsf{Im}(L \cup R)_\mathsf{a}$ implies that $y_\mathsf{a} z_\mathsf{b} \notin \mathsf{Im}(L_{\mathsf{ab}} \cup R_{\mathsf{ab}})$.
%
Applying $\mathsf{Compress}$ to the state yields the state in Eq.~(\ref{eq: abc result 1 before compress}) up to a normalization difference,
\begin{equation}
    \left| \sqrt{ \frac{N_{\mathsf{b}} \prod_\alpha (N_\alpha-\ell-r)}{(N_{\mathsf{ab}}-\ell-r) (N_{\mathsf{bc}}-\ell-r)} }
    -
    \sqrt{ \frac{\prod_\alpha (N_\alpha-\ell-r)}{N-\ell-r} } \right|
    \leq 
    \frac{\ell+r}{N_{\mathsf{ab}}} + \frac{\ell+r}{N_{\mathsf{bc}}},
\end{equation}
where the first inequality holds for $(\ell+r)/ N_{\mathsf{ab}} + (\ell+r)/ N_{\mathsf{bc}} \leq 1/2$.

We now turn to the second class of states.
A complete basis is given by 
\begin{equation}
    \frac{1}{\sqrt{\mathcal{N}_2}}
    \sum_{ 
    \substack{ z_\ell, z_r, z_\mathsf{b} \\ x_\mathsf{a} \notin \mathsf{Dom}(L \cup R)_\mathsf{a} \\ x_\mathsf{b} \notin \mathsf{Dom}(L \cup R)_\mathsf{b} \\ x_\mathsf{c} \notin \mathsf{Dom}(L \cup R)_\mathsf{c}}
    }
    \ket{ x_\mathsf{a} x_\mathsf{b} x_\mathsf{c} }_\gsA
    \ket{ L_{\mathsf{ab}}^{z_\ell} }_\gsLab
    \ket{ L_{\mathsf{bc}}^{z_\ell}  }_\gsLbc
    \ket{ R_{\mathsf{ab}}^{z_r} \cup ( x_\mathsf{a} x_\mathsf{b} , y_\mathsf{a} z_\mathsf{b} ) }_\gsRab
    \ket{ R_{\mathsf{bc}}^{z_r} \cup ( z_\mathsf{b} x_\mathsf{c} , y_\mathsf{b} y_\mathsf{c} ) }_\gsRbc,
\end{equation}
where $\mathsf{Dom}(L \cup R)_\alpha$ is distinct, $\mathsf{Im}(L \cup R)_\alpha$ is distinct, $y_\alpha \notin \mathsf{Im}(L \cup R)_\alpha$, and the normalization is given by $\mathcal{N}_2 \equiv N_\mathsf{b}^{\ell+r+1} \prod_\alpha (N_\alpha-\ell-r)$.
The application of $\widetilde{\Pi}^{\mathcal{I}( W^R_\mathsf{abc} )}$ gives,
\begin{equation}
    \frac{1}{\sqrt{\mathcal{N}_3}}
    \sum_{ 
    \substack{ z_\ell, z_r, z_\mathsf{b} \\ x_\mathsf{a} x_\mathsf{b} x_\mathsf{c} \notin \mathsf{Dom}(L \cup R)}
    }
    \ket{ x_\mathsf{a} x_\mathsf{b} x_\mathsf{c} }_\gsA
    \ket{ L_{\mathsf{ab}}^{z_\ell} }_\gsLab
    \ket{ L_{\mathsf{bc}}^{z_\ell}  }_\gsLbc
    \ket{ R_{\mathsf{ab}}^{z_r} \cup ( x_\mathsf{a} x_\mathsf{b} , y_\mathsf{a} z_\mathsf{b} ) }_\gsRab
    \ket{ R_{\mathsf{bc}}^{z_r} \cup ( z_\mathsf{b} x_\mathsf{c} , y_\mathsf{b} y_\mathsf{c} ) }_\gsRbc,
\end{equation}
where $\mathcal{N}_3 = N_\mathsf{b}^{\ell+r+1} (N_{\mathsf{abc}}-\ell-r)^2 / \prod_\alpha (N_\alpha-\ell-r)$.
The application of $W_{\mathsf{ab}}$ gives,
\begin{equation}
    \sqrt{  \frac{N_\mathsf{ab}-\ell-r  }{\mathcal{N}_3} }
    \sum_{ 
    \substack{ z_\ell, z_r, z_\mathsf{b} \\  x_\mathsf{c}}
    }
    \ket{ y_\mathsf{a} z_\mathsf{b} x_\mathsf{c} }_\gsA
    \ket{ L_{\mathsf{ab}}^{z_\ell} }_\gsLab
    \ket{ L_{\mathsf{bc}}^{z_\ell}  }_\gsLbc
    \ket{ R_{\mathsf{ab}}^{z_r} }_\gsRab
     \ket{ R_{\mathsf{bc}}^{z_r} \cup ( z_\mathsf{b} x_\mathsf{c} , y_\mathsf{b} y_\mathsf{c} ) }_\gsRbc,
\end{equation}
where the sum over $x_\mathsf{c}$ is unconstrained because the application of $W^{R,\dagger}_{\mathsf{ab}}$ enforces that $x_\mathsf{a} x_\mathsf{b} \notin \mathsf{Dom}(L_{\mathsf{ab}} \cup R_{\mathsf{ab}})$, which implies that $x_\mathsf{a} x_\mathsf{b} x_\mathsf{c} \notin \mathsf{Dom}(L \cup R)$.
The application of $W_{\mathsf{bc}}$ then gives,
\begin{equation}
    \sqrt{  \frac{(N_\mathsf{ab}-\ell-r)(N_\mathsf{bc}-\ell-r)}{\mathcal{N}_3} }
    \sum_{ 
    \substack{ z_\ell, z_r }
    }
    \ket{ y_\mathsf{a} y_\mathsf{b} y_\mathsf{c} }_\gsA
    \ket{ L_{\mathsf{ab}}^{z_\ell} }_\gsLab
    \ket{ L_{\mathsf{bc}}^{z_\ell}  }_\gsLbc
    \ket{ R_{\mathsf{ab}}^{z_r} }_\gsRab
     \ket{ R_{\mathsf{bc}}^{z_r}  }_\gsRbc.
\end{equation}
The state is invariant under the final projection $\Pi^{\mathsf{paired}}$.
Applying $\mathsf{Compress}$ to the state yields the state in Eq.~(\ref{eq: abc result 2 before compress}) up to a normalization difference,
\begin{equation}
    \left| \sqrt{ \frac{(N_\mathsf{ab}-\ell-r)(N_\mathsf{bc}-\ell-r) \prod_\alpha (N_\alpha - \ell -r )}{N_\mathsf{b} (N-\ell-r)^2} }
    -
    \sqrt{ \frac{\prod_\alpha (N_\alpha - \ell -r )}{(N_{\mathsf{abc}}-\ell-r)} } \right| 
    \leq \frac{\ell+r}{N} ,
\end{equation}
where the first inequality holds for $\ell + r \leq 2N_{\mathsf{abc}}$.

Note that both $\mathsf{Compress}^\dagger \cdot \overline{ W_{\mathsf{abc}} } \cdot \Pi^{\mathsf{locbij}} \cdot \mathsf{Compress}$ and $\overline{ W_{\mathsf{bc}} W_{\mathsf{ab}} } \cdot  \Pi^{\mathsf{paired}}$ are block diagonal in $\ell$ and $r$, as well as between the two classes of states considered in the above analysis.
Thus, the spectral norm of the difference of the two operators is given by the maximum spectral norm of the difference within each block.
From the above analysis, the maximum is achieved for the first class of states, at $\ell + r = t-1$. The spectral norm is thus bounded by
\begin{equation}\nonumber
    \left\lVert 
    \mathsf{Compress}^\dagger \cdot \overline{ W_{\mathsf{abc}} } \cdot \Pi^{\mathsf{locbij}} \cdot \mathsf{Compress} \cdot \Pi^{\leq t-1}  
    -
    \overline{ W_{\mathsf{bc}} W_{\mathsf{ab}} } \cdot  \Pi^{\mathsf{paired}} \cdot
    \Pi^{\leq t-1}
    \right\rVert_\infty
    \leq
    \frac{(t-1)}{N_{\mathsf{ab}}}+\frac{(t-1)}{N_{\mathsf{bc}}},
\end{equation}
where $\Pi^{\leq t-1}$ restricts to relation state register lengths $\ell + r \leq t-1$.
We have,
\begin{equation} \nonumber
\begin{split}
    & 
    \left\lVert 
    \mathsf{Compress}^\dagger 
    \ket*{ \mathcal{A}^{W_{\mathsf{abc}}', \mathfrak{D}}_{t} }
    -
    \ket*{ \mathcal{A}^{(W_{\mathsf{bc}} W_{\mathsf{ab}})', \mathfrak{D}}_{t} }
    \right\rVert_2  \\
    & \quad \quad \quad \quad  \quad   \leq \left\lVert 
    \mathsf{Compress}^\dagger 
    \ket*{ \mathcal{A}^{W_{\mathsf{abc}}', \mathfrak{D}}_{t-1} }
    -
    \ket*{ \mathcal{A}^{(W_{\mathsf{bc}} W_{\mathsf{ab}})', \mathfrak{D}}_{t-1} }
    \right\rVert_2  \\
    & \quad \quad \quad \quad  \quad \quad \quad    +
    \left\lVert 
    \mathsf{Compress}^\dagger \cdot \overline{ W_{\mathsf{abc}} } \cdot \Pi^{\mathsf{locbij}} \cdot \mathsf{Compress} \cdot \Pi^{\leq t-1}  
    -
    (W_{\mathsf{bc}} W_{\mathsf{ab}} )'\cdot  \Pi^{\mathsf{paired}} \cdot
    \Pi^{\leq t-1}
    \right\rVert_\infty
    \\
    & \quad \quad \quad \quad  \quad   \leq
    \mathsf{TD}_{t-1}
    + (t-1)/N_{\mathsf{ab}}
    + (t-1)/N_{\mathsf{bc}}
    \\
\end{split}
\end{equation}
as claimed. This completes our proof of step 2.

\subsubsection{Twirled projected $W_{\mathsf{bc}} W_{\mathsf{ab}}$ is indistinguishable from twirled $V_{\mathsf{bc}} V_{\mathsf{ab}}$} \label{sec: twirled projected WAB WBC to twirled WAB WBC}

Our proof of step 4 follows quickly from the results of steps 1-3.
We use that $W_{\mathsf{bc}}$ and $W_{\mathsf{ab}}$ are restrictions of $V_{\mathsf{bc}}$ and $V_{\mathsf{ab}}$ to write,
\begin{align}
    (W_{\mathsf{bc}} W_{\mathsf{ab}})' 
    & = 
    \Pi^{\mathsf{paired}}
    \cdot V_{\mathsf{bc}}
    \cdot \Pi^{\mathcal{D}(W_{\mathsf{bc}})} 
    \cdot V_{\mathsf{ab}}
    \cdot \Pi^{\mathcal{D}(W_{\mathsf{bc}})} 
    \cdot \widetilde{\Pi}^{\mathcal{D}(W_{\mathsf{abc}})}
    \cdot \widetilde{\Pi}^{\notin \mathsf{Dom} \rightarrow \notin \mathsf{locDom}} \\
    (W^\dagger_{\mathsf{ab}} W^\dagger_{\mathsf{bc}})' 
    & = 
    \Pi^{\mathsf{paired}}
    \cdot V^\dagger_{\mathsf{ab}}
    \cdot \Pi^{\mathcal{I}(W_{\mathsf{ab}})} \cdot 
    V^\dagger_{\mathsf{bc}} 
    \cdot \Pi^{\mathcal{I}(W_{\mathsf{bc}})} 
    \cdot
    \widetilde{\Pi}^{\mathcal{D}(W^\dagger_{\mathsf{abc}})}
    \cdot \widetilde{\Pi}^{\notin \mathsf{Im} \rightarrow \notin \mathsf{locIm}} \\
    (\overline{W}_{\mathsf{bc}} \overline{W}_{\mathsf{ab}})' 
    & = 
    \Pi^{\mathsf{paired}}
    \cdot \overline{V}_{\mathsf{bc}}
    \cdot \Pi^{\mathcal{D}(\overline{W}_{\mathsf{bc}})} 
    \cdot \overline{V}_{\mathsf{ab}}
    \cdot \Pi^{\mathcal{D}(\overline{W}_{\mathsf{bc}})} 
    \cdot \widetilde{\Pi}^{\mathcal{D}(\overline{W}_{\mathsf{abc}})}
    \cdot \widetilde{\Pi}^{\notin \mathsf{Dom} \rightarrow \notin \mathsf{locDom}} \\
    (\overline{W}^\dagger_{\mathsf{ab}} \overline{W}^\dagger_{\mathsf{bc}})' 
    & = 
    \Pi^{\mathsf{paired}}
    \cdot \overline{V}^\dagger_{\mathsf{ab}}
    \cdot \Pi^{\mathcal{I}(\overline{W}_{\mathsf{ab}})} \cdot 
    \overline{V}^\dagger_{\mathsf{bc}} 
    \cdot \Pi^{\mathcal{I}(\overline{W}_{\mathsf{bc}})} 
    \cdot
    \widetilde{\Pi}^{\mathcal{D}(\overline{W}^\dagger_{\mathsf{abc}})}
    \cdot \widetilde{\Pi}^{\notin \mathsf{Im} \rightarrow \notin \mathsf{locIm}} 
\end{align}
Therefore, the state $\ket*{ \mathcal{A}^{(W_{\mathsf{bc}} W_{\mathsf{ab}})', \mathfrak{D}}_t }$ differs from the state $\ket*{ \mathcal{A}^{V_{\mathsf{bc}} V_{\mathsf{ab}}, \mathfrak{D}}_t }$ solely by the insertion of projectors throughout the time evolution.

From Eq.~(\ref{eq: state to op bound}) and the sequential gentle measurement lemma (Lemma~\ref{lem:seq-gentleM-pure}), the difference between the two states is bounded as,
\begin{equation} \label{eq: step 4}
    \big\lVert \rho^{(3)} - \rho^{(4)} \big\rVert_1 \leq 
    2 \left\lVert 
    \ket*{ \mathcal{A}^{(W_{\mathsf{bc}} W_{\mathsf{ab}})', \mathfrak{D}}_t }
    -
    \ket*{ \mathcal{A}^{V_{\mathsf{bc}} V_{\mathsf{ab}}, \mathfrak{D}}_t }
    \right\rVert_2
    \leq 
    2 t \sqrt{1 - \langle
    \mathcal{A}^{(W_{\mathsf{bc}} W_{\mathsf{ab}})', \mathfrak{D}}_t 
    \big|
    \mathcal{A}^{(W_{\mathsf{bc}} W_{\mathsf{ab}})', \mathfrak{D}}_t 
    \rangle}.
\end{equation}
From our proofs of step 2 and step 3, we have that 
\begin{equation}
\begin{split}
    & \mathsf{TD}\left( 
     \dyad*{ \mathcal{A}^{(W_{\mathsf{bc}} W_{\mathsf{ab}})', \mathfrak{D}}_t }
     ,
    \mathsf{Compress}^\dagger  \dyad*{ \mathcal{A}^{W_{\mathsf{abc}}, \mathfrak{D}}_t }
     \! \mathsf{Compress} 
    \right) \\
    & \quad  \quad  \quad  \quad  \quad  \quad  \quad  \quad  \quad  \quad  \quad  \quad  \quad  \quad  \quad  \quad  \quad  \quad  \quad  \quad \leq 
    \frac{17t^2}{N_{\mathsf{abc}}^{1/8}} + \frac{7 t^{3/2}}{(\min_\alpha N_\alpha)^{1/2}} + \frac{2t^2}{N_{\mathsf{ab}}} + \frac{2t^2}{N_{\mathsf{bc}}}.
\end{split}
\end{equation}
Meanwhile, from Lemma~9.3 of Ref.~\cite{ma2024construct} [see in particular Eqs.~(9.58),~(9.59)], we have 
\begin{equation}
    \mathsf{TD}\left( 
    \dyad*{ \mathcal{A}^{W_{\mathsf{abc}}, \mathfrak{D}}_t } 
    , 
    \mathsf{cQ} \cdot \dyad*{ \mathcal{A}^{V_{\mathsf{abc}}, \mathfrak{D}}_t } \cdot \mathsf{cQ}
    \right)
    \leq 
    \frac{9 t}{N_{\mathsf{abc}}^{1/8}} .
\end{equation}
The state $\mathsf{cQ} \cdot \ket*{ \mathcal{A}^{V_{\mathsf{abc}}, \mathfrak{D}}_t } $ is obtained from solely unitary time-evolution, and thus has norm one.
Combining the two above equations therefore yields,
\begin{equation}
    \langle
    \mathcal{A}^{(W_{\mathsf{bc}} W_{\mathsf{ab}})', \mathfrak{D}}_t 
    \big|
    \mathcal{A}^{(W_{\mathsf{bc}} W_{\mathsf{ab}})', \mathfrak{D}}_t 
    \rangle
    \geq 1 - \left( \frac{17t^2}{N_{\mathsf{abc}}^{1/8}} + \frac{7 t^{3/2}}{(\min_\alpha N_\alpha)^{1/2}} + \frac{2t^2}{N_{\mathsf{ab}}} + \frac{2t^2}{N_{\mathsf{bc}}} \right) - \frac{9 t}{N_{\mathsf{abc}}^{1/8}}
\end{equation}
which, from Eq.~(\ref{eq: step 4}), implies that 
\begin{equation}
\begin{split}
    \big\lVert \rho^{(3)} - \rho^{(4)} \big\rVert_1
    & \leq
    2t \sqrt{ \frac{17t^2}{N_{\mathsf{abc}}^{1/8}} + \frac{7 t^{3/2}}{(\min_\alpha N_\alpha)^{1/2}} + \frac{2t^2}{N_{\mathsf{ab}}} + \frac{2t^2}{N_{\mathsf{bc}}} +\frac{9 t}{N_{\mathsf{abc}}^{1/8}} } \\
    & \leq
    \frac{2\sqrt{17}t^2}{N_{\mathsf{abc}}^{1/16}} + \frac{2\sqrt{7} t^{7/4}}{(\min_\alpha N_\alpha)^{1/4}} + \frac{2\sqrt{2}t^2}{N_{\mathsf{ab}}^{1/2}} + \frac{2\sqrt{2}t^2}{N_{\mathsf{bc}}^{1/2}} +\frac{6 t^{3/2}}{N_{\mathsf{abc}}^{1/16}} , \\
\end{split}
\end{equation}
where in the second line we use that the square root is subadditive.
This completes the proof.

\subsection{Proof of Theorems~\ref{thm:two-layer-design} and~\ref{thm:two-layer-PRU}} \label{sec: proof blocked fast scrambling} 

Our proof of Theorems~\ref{thm:two-layer-design} and~\ref{thm:two-layer-PRU} follow immediately from  Lemma~\ref{lemma: strong gluing}.
The extension of Theorem~\ref{thm:two-layer-design} to the blocked fast scrambling circuit follows immediately from Lemma~\ref{lemma: strong 2-designs}.

\begin{proof}[Proof of Theorem~\ref{thm:two-layer-design}]
Iterating Lemma~\ref{lemma: strong gluing} $m = n/\xi$ times, we can replace the two-layer circuit with a Haar-random unitary up to a measurable error $(n/\xi)(\varepsilon/n) + \mathcal{O}(nk^2/2^{(3/16)\xi}\xi)$.
The first term is less than $\varepsilon/2$ whenever $\xi \geq 2$.
The second term is less than $\varepsilon/2$ if $\xi \geq \frac{16}{3} \log_2(nk^2/\varepsilon) + \mathcal{O}(1)$.
This completes the proof.
\end{proof}

\noindent \emph{Remark.} The proof immediately extends to a modified blocked fast scrambling circuit in Section~\ref{sec: local random circuit} of the main text, in which we replace the exact $n$-qubit unitary 2-designs with blocked fast scrambling circuits composed of small strong $\frac{\varepsilon_2}{n}$-approximate unitary 2-designs.
From Lemma~\ref{lemma: strong 2-designs}, this blocked fast scrambling circuit forms a strong $\varepsilon_2$-approximate unitary 2-design when $\xi \geq \log_2(5n/\varepsilon_2)$.
Iterating Lemma~\ref{lemma: strong gluing} $m = n/\xi$ times, we can replace the two-layer circuit with a Haar-random unitary up to a measurable error $(n/\xi)(\varepsilon/n) + \mathcal{O}(nk^2/2^{(3/16)\xi}\xi) + \mathcal{O}(n k^{5/8} \varepsilon_2^{1/8})$.
The first term is less than $\varepsilon/3$ whenever $\xi \geq 3$.
The second term is less than $\varepsilon/3$ if $\xi \geq \frac{16}{3} \log_2(nk^2/\varepsilon) + \mathcal{O}(1)$.
The third term is less than $\varepsilon/3$ if $\varepsilon_2 = \mathcal{O}(\varepsilon^8 / n^8 k^5)$, which requires $\xi \geq \log_2(n^9 k^5/\varepsilon^8) + \mathcal{O}(1) = \mathcal{O}(\log nk/\varepsilon)$. \qed

\begin{proof}[Proof of Theorem~\ref{thm:two-layer-PRU}]
By assumption, each individual small PRU in the two-layer circuit is indistinguishable from a small Haar-random unitary by any $\poly n$-time quantum experiment.
Hence, following identical steps to the proof of Theorem~2 in Ref.~\cite{schuster2024random}, the scrambled two-layer circuit of small PRUs is indistinguishable from a scrambled two-layer circuit of small Haar-random unitaries.
From Theorem~\ref{thm:two-layer-PRU}, the latter ensemble forms an $\varepsilon$-approximate strong unitary $k$-design for any $k^2/\varepsilon \leq \mathcal{O}(2^{\xi}/n)$.
Setting $\xi = \omega(\log n)$ yields a design for any $k, 1/\varepsilon = \poly n$.
Hence, from the definition of strong approximate unitary $k$-designs, the scrambled two-layer circuit of small Haar-random unitaries is indistinguishable from a Haar-random unitary by any $\poly n$-time quantum experiment. 
\end{proof}

\subsection{Proof of Theorems~\ref{thm:strong-design-depth} and~\ref{thm:strong-PRU-depth}}

The combination of Theorem~\ref{thm:LRFC-design} and Theorem~\ref{thm:two-layer-design}  immediately yield Theorem~\ref{thm:strong-design-depth} on the circuit depth of strong unitary designs. 

\begin{proof}[Proof of Theorem~\ref{thm:strong-design-depth}]
The second and third statements of Theorem~\ref{thm:strong-design-depth} follows immediately from Theorem~\ref{thm:LRFC-design} and the circuit depth required to implement the LRFC ensemble  with $2k$-wise independent functions~\cite{cui2025unitary}. 
The first statement of Theorem~\ref{thm:strong-design-depth} follows from Theorem~\ref{thm:two-layer-design} and Lemma~\ref{lemma: strong design 1D RUC}.
\end{proof}

\noindent The combination of Theorem~\ref{thm:LRFC-PRU} and Theorem~\ref{thm:two-layer-PRU}  immediately yield Theorem~\ref{thm:strong-PRU-depth} on the circuit depth of strong pseudorandom unitaries. 

\begin{proof}[Proof of Theorem~\ref{thm:strong-PRU-depth}]
The first statement of Theorem~\ref{thm:strong-PRU-depth} follows immediately from Theorem~\ref{thm:LRFC-PRU} and the circuit depth required to implement the LRFC ensemble  with pseudorandom functions~\cite{schuster2024random}. 
The second statement of Theorem~\ref{thm:strong-PRU-depth} follows from Theorem~\ref{thm:LRFC-PRU} and Theorem~\ref{thm:two-layer-PRU}.
The derivation of the circuit depth is described in the main text.
\end{proof}

\section{Ancilla-free pseudorandom unitaries} \label{app: no-ancilla}

In this section, we give the first constructions of ancilla-free (strong) pseudorandom unitarites. Our main crytographic building block will be pseudorandom functions \cite{goldreich1986construct} computable in the complexity class ``logspace-uniform $\mathsf{TC}^1$.'' A function is computable in logspace-uniform $\mathsf{TC}^1$ if (1) it is computable by a family of $O(\log n)$-depth circuits with large fan-in threshold gates and (2) this family of circuits is output by a logspace Turing machine on the input $1^n$. Crucially, it is known that such PRFs exist under the LWE assumption \cite{banerjee2012pseudorandom}, and that this construction is post-quantum secure \cite{zhandry2021PRF}. 

Our main technical result about ancilla-free computation is as follows. 

\begin{theorem}\label{thm:ancilla-free-implementation}
    Let $f: \{0,1\}^n \rightarrow \mathbb Z_q^m$ be any logspace-uniform $\mathsf{TC}^1$-computable function, where $q = O(1)$. Then, there is a $\mathsf{poly}(n, m)$-size reversible circuit implementing the permutation
    \begin{equation} (x,y,a) \mapsto (x, y + f(x) \pmod q, a), 
    \end{equation}
    where $a$ denotes an arbitrary setting of the ancilla register. 
\end{theorem}

By combining \cref{thm:ancilla-free-implementation} with the LRFC construction of \cref{sec: LRFC}, obtain the following intermediate result: a PRU family with an efficient ancilla-free implementation of the unitary $U_k \otimes \mathsf{Id}$ (rather than $U_k$ alone). We call such PRUs ``ancilla-independent.''

\begin{theorem}\label{thm:ancilla-independent-PRU}
    Assuming polynomially secure (respectively, sub-exponentially secure) post-quantum PRFs computable in logspace-uniform $\mathsf{TC}^1$, there exist polynomially secure (respectively, sub-exponentially secure) ancilla-independent strong PRUs. 
\end{theorem}

Finally, combining \cref{thm:ancilla-free-PRU} with the strong PRU gluing lemma (and its corollary, \cref{thm:two-layer-PRU}), we obtain ancilla-free strong PRUs. This approach works because when gluing ancilla-independent PRU implementations, one can use the ancilla register of one unitary as part of the \emph{input} register of another unitary. In the end, the resulting glued circuit will compute a PRU on its entire domain.  

\begin{theorem}\label{thm:ancilla-free-PRU}
    Assuming post-quantum PRFs computable in $\mathsf{TC}^1$, there exist ancilla-free strong PRUs. Moreover, assuming sub-exponentially secure post-quantum PRFs computable in $\mathsf{TC}^1$, there exist ancilla-free strong PRUs computable in depth $\poly(\log n)$ with all-to-all circuits. 
\end{theorem}

Since logspace-uniform $\mathsf{TC}^1$-computable PRFs are known under the standard LWE assumption \cite{banerjee2012pseudorandom}, we obtain instantiations of our results under LWE.

\begin{corollary}\label{cor:LWE-ancilla-free-PRUs}
    Assuming the post-quantum hardness of LWE, there exist ancilla-free PRUs. Assuming the sub-exponential post-quantum hardness of LWE, there exist ancilla-free strong PRUs computable in depth $\poly(\log n)$ with all-to-all circuits. 
\end{corollary}

The rest of this section is devoted to proving \cref{thm:ancilla-free-implementation,thm:ancilla-independent-PRU,thm:ancilla-free-PRU}. 

\subsection{Ancilla-preserving reversible computation of functions}
Let $f: \mathcal X \rightarrow \mathbb Z_q$ denote a function with $q = O(1)$. We study different ``ancilla-respecting'' reversible circuit implementations of the computation
\begin{equation}
    \ket{x, y} \mapsto \ket{x, y+ f(x) \pmod q}.
\end{equation}
We will use bra-ket notation to describe the action of these reversible circuits, but we note that they only use Toffoli gates and thus correspond to classical reversible computation, i.e., permutations. 

Our reversible circuits will operate on four registers:
\begin{itemize}
    \item Let $\mathsf X$ denote an $n$-qubit register whose standard basis states correspond to $\mathcal X$.
    \item Let $\mathsf Y$ denote an output register with standard basis in bijection with $\mathbb Z_q$.
    \item Let $\mathsf A$ denote an $\ell$-qubit \emph{trusted} ancilla register. This means that (at least initially), our computations will rely on the ancilla being initialized to the $\ket*{0^\ell}$ state. 
    \item Let $\mathsf W$ denote an \emph{untrusted} (or catalytic) ancilla register of size $\poly n$. 
\end{itemize}

\begin{definition}
\label{def:ap-imp-f}
    We say that a reversible circuit $C$ acting on $(\mathsf{X},  \mathsf{Y}, \mathsf{A}, \mathsf{W})$ is an \emph{ancilla-preserving} implementation of $f$ with trusted space $\sA$ and untrusted space $\sW$ if it maps
    \begin{align}
        \ket{x,y,a,w} \mapsto \ket{x,y + g(x,a,w),a,w},
    \end{align}
    where $g$ is any function that agrees with $f$ when $a = 0$. That is, $g(x,0^\ell,w) = f(x)$, but otherwise $g(x,a,w)$ may be arbitrary.
\end{definition}
Here, ``ancilla-preserving'' refers to the fact that $C$ never changes the values stored in the $\sA$ and $\sW$ register, regardless of what they are initialized to. We say that $\sA$ is ``trusted space'', since $g(x,a,w)$ is only guaranteed to compute the output $f(x)$ correctly when $\sA$ is initialized properly to $0^\ell$. $\sW$ is ``untrusted'' because we require that $g(x,0,w)$ correctly compute $f(x)$ for any choice of $w$. 

As we will show, the following lemma is an easy consequence of recent work on logspace catalytic classical computation \cite{STOC:BCKLS14-catalytic,TQC:BFMSSST25-catalytic}.


\begin{lemma}\label{lemma:TC1-implementation}
    For every function $f: \{0,1\}^n \rightarrow \mathbb Z_q$ computable in logspace-uniform $\mathsf{TC}^1$, there is a $\poly(n)$-size ancilla-preserving reversible circuit $C$ that implements $f$ with an $\ell = O(\log n)$-size trusted ancilla space $\sA$ and a $\poly(n)$-size untrusted ancilla space $\sW$. 
\end{lemma}

\begin{proof}
    According to~\cite{TQC:BFMSSST25-catalytic}, for $f: \{0,1\}^n \rightarrow \mathbb Z_q$ computable in logspace-uniform $\mathsf{TC}^1$, there is a $\poly(n)$-size reversible circuit $C'$ that maps
    \begin{align}
        \ket*{x,y,a,w} \mapsto \ket*{p(x,y,a,w)}
    \end{align}
    where $p(x,0,0^{\ell},w) = (x',f(x),a',w)$, but on other inputs, $p(x,y,a,w)$ may be arbitrary. Given such a circuit $C'$, we can implement $C$ satisfying~\cref{def:ap-imp-f}, where the trusted ancilla is now size $\ell+O(1)$, using standard techniques from reversible computation:
    \begin{enumerate}
        \item On input $\ket{x,y,a,w}$, parse $a$ as $(y',a')$, where $y'\in \mathbb Z_q$ and $a'$ is length $\ell$. Run $C'$ on $(x,y',a',w)$. This step does not affect the $y$ register, and the effect on $\ket{x,a,w}$ is:
        \begin{align}
            \ket*{x,a= (y',a'),w} \mapsto \ket*{p(x,y',a',w)}.
        \end{align}
        \item Parse $p(x,y',a',w)$ as $p(x,y',a',w) = (x_{\mathrm{out}},y'_{\mathrm{out}},a'_{\mathrm{out}},w_{\mathrm{out}})$. Add the value of $y'_{\mathrm{out}}$ onto $\ket{y}$.
        \item Apply the inverse of $C'$.
    \end{enumerate}
    The result is that we have performed the map
    \begin{align}
        \ket{x,y,a,w} \mapsto \ket*{x,y + y'_{\mathrm{out}},a,w},
    \end{align}
    where $y'_{\mathrm{out}}$ is a function of $(x,a=(y',a'),w)$, which equals $f(x)$ when $a = (y',a') = (0,0^{\ell})$.
\end{proof}
We will now define two further restricted classes of ancilla-preserving implementations, and we will show how to compile circuits satisfying~\cref{def:ap-imp-f} into circuits satisfying these more restricted notions.

Recall that an ancilla-preserving implementation of $f$ is a circuit $C$ that maps $(x,y,a,w)$ to $(x,y + g(x,a,w),a,w)$, where $g$ is \emph{any function} that agrees with $f$ when $a = 0$. We now define two more restrictive notions where we impose further requirements on $g$.

\begin{definition}[Stronger versions of ancilla-preserving implementations] Let $C$ be an ancilla-preserving implementation of $f$, and let $g$ be as in~\cref{def:ap-imp-f}. Then we say that $C$ is:
    \begin{itemize}
        \item an \textbf{ancilla-controlled} implementation of $f$ if $g(x,a,w) = f(x) \cdot \chi_{a = 0^{\ell}}$ and is $0$ otherwise. 
        \item an \textbf{ancilla-independent} implementation of $f$ if $g(x,a,w) = f(x)$ for all $a,w$.
    \end{itemize}
\end{definition}

\subsubsection{Converting ancilla-preserving to ancilla-controlled}
Next, we show that an ancilla-preserving implementation of a function $f$,
\begin{equation*}
    \ket{x, y, a, w}\mapsto \ket{x, y + g(x,a,w), a, w},
\end{equation*}
can be efficiently converted into an ancilla-controlled implementation of $f$,

\begin{equation*}
    \ket{x, y, a, w'} \mapsto \ket{x,y + \chi_{a = 0^\ell}\cdot f(x), a, w'},
\end{equation*}
where the trusted ancilla space $\mathcal A$ is unchanged and $\mathcal W' = \mathcal W \times \mathbb Z_q$. 

\begin{lemma}\label{lemma:preserving-to-controlled}
    Let $g(x, a, w)$ be any function such that $g(x, 0, w) = f(x)$ for all $(x, w)$. Then, if there is a circuit $C$ implementing the map $\ket{x, y, a, w}\mapsto \ket{x, y + g(x,a,w), a, w}$, there is another circuit $C'$ implementing the map $\ket{x, y, a, w'} \mapsto \ket{x,y + \chi_{a = 0^\ell}\cdot f(x), a, w'}$. 
    
    Moreover, the size and depth of $C'$ is bounded in terms of $C$:
    \begin{itemize}
        \item $|C'| = O(|C|) + 2^{O(\ell)}$,~\footnote{We believe it should be possible to improve the $2^{O(\ell)}$ dependence to $\poly(\ell)$, but due to a subsequent $2^{O(\ell)}$ overhead, it makes no difference for our purposes.} and
        \item $\mathrm{depth}(C') = O(\mathrm{depth}(C)) + 2^{O(\ell)}$
    \end{itemize}
\end{lemma}

\renewcommand{\Id}{\mathsf{Id}}

\begin{proof}
    Let us introduce a new register $\sQ$ supported on states $\ket{z}$ for $z \in \mathbb{Z}_q$. Suppose we could implement the following maps:
    \begin{align}
        O:& \ket{z}_{\sQ} \ket{x,y,a,w} \mapsto \ket{z}_{\sQ} \ket{x,y + z \cdot g(x,a,w),a,w}\\
        W:& \ket{z}_{\sQ} \ket{x,y,a,w} \mapsto \ket{z - \chi_{a = 0^{\ell}}}_{\sQ}\ket{x,y,a,w},
    \end{align}
    where $\chi_{a= 0^{\ell}}$ is the element $1 \in \mathbb{Z}_q$ if $a = 0^{\ell}$ and is $0 \in \mathbb{Z}_q$ otherwise.
    Since
    \begin{align}
        WO:& \ket{z}\ket{x,y,a,w} \mapsto \ket{z-\chi_{a = 0^{\ell}}}  \ket{x,y + z \cdot g(x,a,w),a,w},\\
        W^\dagger O^\dagger:& \ket{z}\ket{x,y,a,w} \mapsto \ket{z+\chi_{a = 0^{\ell}}}  \ket{x,y - z \cdot g(x,a,w),a,w},
    \end{align}
    we can compose these operations to obtain the desired ancilla-controlled implementation:
    \begin{align}
        W^\dagger O^\dagger W O: \ket{z}\ket{x,y,a,w} &\mapsto \ket{z}  \ket{x,y + z \cdot  g(x,a,w) - (z - \chi_{a=0^\ell})  \cdot g(x,a,w),a,w}\\
        &= \ket{z} \ket{x,y +  \chi_{a = 0^{\ell}} f(x),a,w}.
    \end{align}
    In the last equality, we used the fact that $g$ satisfies $\chi_{a=0^{\ell}} \cdot f(x) = \chi_{a = 0^{\ell}} \cdot g(x,a,w)$. It remains to show how to implement $O$ and $W$.

    \vspace{0.5em}
    \noindent \textbf{Implementing O.} For $i \in \mathbb{Z}_q$ let $C^{(i)}$ be the circuit that applies $C$ controlled on element $z$ in the $\sQ$ register satisfying $z \leq i$. Given an implementation of $C$, we can implement $C^{(i)}$ by replacing each gate $g$ with the gate $g^{(i)}$, which implements $g$ controlled on the element $z$ in the $\sQ$ register satisfying $z \leq i$. Since $q = O(1)$, this only requires a constant number of additional elementary gates. Then $O$ can be implemented as $O = C^{(q)} C^{(q-1)} \cdots C^{(1)}$.

    \vspace{0.5em}
    \noindent \textbf{Implementing $W$}. $W$ is an $\ell + O(1)$ qubit unitary, so it has a $2^{O(\ell)}$-size ancilla-free implementation. \qedhere

\end{proof}

\subsubsection{Converting ancilla-controlled to ancilla-independent}
Finally, we give a simple transformation from ancilla-controlled implementations to ancilla-independent implementations of $f$. The transformation preserves the ancilla registers but has a size and depth blowup of $2^\ell$. 

\begin{lemma}\label{lemma:controlled-to-independent}
    Let $C$ be an ancilla-controlled implementation of $f$, i.e., it implements the map $\ket{x, y,a, w} \mapsto \ket{x, y + \chi_{a = 0^\ell} \cdot f(x), a, w}$. There is an ancilla-independent implementation of $f$, i.e., a circuit $C'$ that maps $\ket{x, y,a, w} \mapsto \ket{x, y + f(x), a, w}$, with size $\Theta(2^\ell \cdot \abs{C})$ and depth $\Theta(2^\ell \cdot \mathrm{depth}(C))$.
\end{lemma}

\begin{proof}
Let $R\coloneqq \sum_{a \in \mathcal A} \ketbra{a+1}{a}$ be the increment operator acting on $\sA$. Then we claim that $C' = (R \cdot C)^{2^{\ell}}$ is an ancilla-independent implementation of $f$. This works because $R \cdot C$ maps  $\ket{x,y,a,w}$ to 
\begin{align}
    \ket{x,y + \chi_{a = 0^{\ell}} \cdot f(x) ,a+1,w},
\end{align}
so repeating this $2^{\ell}$ times maps $\ket{x,y,a,w}$ to
\begin{align}
    \ket*{x,y + \sum_{i \in \calA} \chi_{(a+i) = 0^{\ell}} \cdot f(x),a,w} = \ket*{x,y + f(x),a,w}.
\end{align}
\end{proof}

\subsubsection{Completing the proof of \cref{thm:ancilla-free-implementation}}
By combining \cref{lemma:TC1-implementation,lemma:preserving-to-controlled,lemma:controlled-to-independent}, we obtain \cref{thm:ancilla-free-implementation} in the special case of $m=1$. That is, we have an ancilla-free implementation of any logspace-uniform $\mathsf{TC}^1$ function $f: \{0,1\}^n \rightarrow \mathbb Z_q$ with size $\poly n$. To complete the proof of \cref{thm:ancilla-free-implementation}, it suffices to extend this to $f: \{0,1\}^n \rightarrow \mathbb Z_q$ with $m > 1$. However, this turns out to be simple: compute each output symbol one-at-a-time in sequence (using a different output register for each symbol). It is clear that the ancilla-independent implementations compose, completing the proof.  

\subsection{Constructing ancilla-independent PRUs}
In this section, we proceed from studying ancilla-independent implementations of \emph{functions} to studying ancilla-independent implementations of \emph{unitaries}. We say that $C$ is an ancilla-independent implementation of a unitary $U$ (acting on Hilbert space $\sX$) if $C$ implements the map $U_{\sX} \otimes \Id_{\sA}$. 

\begin{definition}[Ancilla-independent PRU]
We say that a PRU family $\{U_k\}$ is an ancilla-independent PRU family if for every key $k$, there is a polynomial-size quantum circuit $C_k$ that is an ancilla-independent implementation of $U_k$. Moreover, we require that $C_k$ is efficiently computable from $k$. 
\end{definition}

Our goal in this section is to prove the following result
\begin{theorem}\label{thm:ancilla-independent-PRU-from-function}
    Assume the existence of a post-quantum PRF that has a polynomial-size ancilla-independent implementation. Then, there exists an ancilla-independent PRU family. Moreover, if the PRF family is sub-exponentially secure, so is the PRU.
\end{theorem}

To prove this theorem, we wish to make use of the LRFC construction (\cref{thm:LRFC-PRU}). To do so, we must first give an ancilla-independent implementation of a pseudorandom \emph{ternary phase oracle}. This is achieved via the following lemma. 

\begin{lemma}\label{lemma:bit-flip-to-phase}
    Let $f: \{0,1\}^n \rightarrow \mathbb Z_q$ be a function with $q = O(1)$. Given an ancilla-independent implementation $C$ of $\ket{x, y} \mapsto \ket{x, y+f(x)}$, there is an ancilla-independent implementation $C'$ of $\ket{x}\mapsto \omega_q^{f(x)} \ket{x}$. Moreover, we have that $|C'| = O(|C|) $ and $\mathrm{depth}(C') = O(\mathrm{depth}(C))$. 
\end{lemma}

\begin{proof}
    Let $C$ be an ancilla-independent circuit implementation of $\ket{x}\ket{y}\mapsto \ket{x}\ket{y+f(x)}$ with ancilla register $\sA$. Then, to implement $\ket{x}\mapsto \omega_q^{f(x)} \ket{x}$, we use an ancilla register $\sA' = \sA \otimes \sY$. We then implement the map

    \begin{equation}
        O = (\Id_{\sX, \sA} \otimes F_q) C (\Id_{\sX, \sA} \otimes F_q^\dagger)
    \end{equation}
    where $F_q$ denotes the $q$-ary Fourier transform on register $\sY$. This is equivalent to the map

    \begin{equation}
        \ket{x, y, a} \mapsto \ket{x} \otimes \frac 1 {\sqrt q} F_q \sum_{z\in \mathbb F_q} \omega_q^{-z\cdot y} \ket{z + f(x)} \otimes \ket{a} = \omega_q^{y\cdot f(x)} \ket{x, y, a}.
    \end{equation}
    Finally, we observe that, letting $R_{\sY}$ denote the increment operator on $\sY$, 
    \begin{equation}
        R_{\sY}\cdot O^\dagger \cdot R_{\sY}^\dagger \cdot O = \sum_{x, y, a} \omega_q^{f(x)} \ketbra{x, y, a},
    \end{equation}
    yielding the desired ancilla-independent implementation of $\ket{x}\mapsto \omega_q^{f(x)} \ket{x}$. 
\end{proof}

\begin{proof}[Proofs of \cref{thm:ancilla-independent-PRU-from-function,thm:ancilla-independent-PRU}]
We now prove \cref{thm:ancilla-independent-PRU-from-function} by appealing to the LRFC construction of \cref{thm:LRFC-PRU}. That is, unitaries in the PRU family have the form

\begin{equation}
    U = D \cdot S_R \cdot F \cdot S_L \cdot C.
\end{equation}
where $D$ and $C$ are $2$-designs, $F$ is a pseudorandom ternary phase oracle, and $S_L, S_R$ are pseudorandom bit-flip permutations $\ket{x_L, x_R} \mapsto \ket{x_L\oplus f(x_R), x_R}$ and $\ket{x_L, x_R}\mapsto \ket{x_L, x_R\oplus g(x_L)}$. By \cref{thm:ancilla-free-implementation,lemma:bit-flip-to-phase}, we know that there exist ancilla-independent pseudorandom instantiations of $S_L, S_R, F$. Moreover, it is known that approximate $2$-designs have efficient ancilla-independent (even ancilla-free) implementations \cite{PRA:DCEL09-designs}. Thus, by invoking \cref{thm:LRFC-PRU}, we obtain \cref{thm:ancilla-independent-PRU-from-function}.
Moreover, by combining \cref{thm:ancilla-independent-PRU-from-function} with \cref{thm:ancilla-free-implementation}, we immediately obtain \cref{thm:ancilla-independent-PRU}. 
\end{proof}

\subsection{Constructing ancilla-free PRUs}

Finally, we prove \cref{thm:ancilla-free-PRU} by combining \cref{thm:ancilla-independent-PRU} with \cref{thm:two-layer-PRU}. Recall that \cref{thm:two-layer-PRU} states that the following construction yields a strong PRU when its building block unitaries are instantiated with strong PRUs:
\begin{itemize}
    \item Apply a $2$-design.
    \item Apply two brickwork layers of building block unitary gates.
    \item Apply an independent $2$-design.
\end{itemize}
We will use this construction \emph{twice} to prove \cref{thm:ancilla-free-PRU}. First, we instantiate a high-depth version of the construction, where:
\begin{itemize}
    \item The building block unitaries act on $2\cdot n^\epsilon$ qubits. 
    \item The input register is divided into $n^{1-\epsilon}$ blocks of $n^\epsilon$ qubits. 
    \item Given an ancilla-independent implementation of the building block unitary using $n^{\alpha\cdot \epsilon}$ ancilla qubits (for some constant $\alpha > 0$), we construct a circuit from the glued unitary in which each building block unitary uses adjacent registers as its ancilla space.
\end{itemize}
This construction results in a quantum circuit implementation of a distribution of unitaries that is pseudorandom on its entire domain by the gluing lemma. However, its depth is large, especially because all gates in the ``brickwork layers'' must now be concatenated sequentially due to the circuit implementations having overlapping registers. Nevertheless, this yields the high-depth case of \cref{thm:ancilla-free-PRU}. Moreover, if the initial ancilla-independent PRU family is sub-exponentially secure, then so is the ancilla-free PRU family.

Finally, under this sub-exponential security assumption, we can plug our ancilla-free PRU back into \cref{thm:two-layer-PRU}, using block size $\poly(\log n)$, yielding an ancilla-free PRU family of depth $\poly(\log n)$.

\section{Analysis of the mixed Haar twirl} \label{sec: details mixed}

In this Appendix, we provide further details on our analysis of the mixed Haar twirl. This completes our proof of the additive-to-relative error translation result (Lemma~\ref{lemma:translating-strong}) for strong unitary designs. We also provide several additional results on the mixed Haar twirl not used in this work.

\subsection{Reformulating the mixed Haar twirl} 

In this section, we provide the full details of the derivation of our reformulation of the mixed Haar twirl [Eq.~(\ref{eq: exact mixed twirl subspace})], which was used to prove Lemma~\ref{lemma:translating-strong} in Appendix~\ref{sec: proof approx mixed}. Our derivation uses only basic properties of the partially transposed permutations. Nonetheless, it requires several detailed steps to formally construct the partial isometries $\tilde{I}_\ell$ and prove their essential properties.

For the ease of the reader, we structure this section in a pedagogical format. We first motivate and introduce the relevant objects one-by-one. We then provide a series of short proofs of their key properties, which leads to our reformulated expression for the mixed Haar twirl. 
We hope that our analysis may be useful in future works on the mixed Haar twirl, even beyond the context of strong unitary designs.


\subsubsection{The partially transposed permutations (PTPs)}

As previously discussed, we can write the exact expression for the mixed Haar twirl in terms of the partially transposed permutations (PTPs),
\begin{equation} 
    \Phi^{(p,q)}_{H}(X) = \sum_{\sigma, \tau \in S^\Gamma_k} \widetilde{\Wg}_{\sigma, \tau} \cdot \text{tr}(X \sigma^{\dagger}) \cdot \tau,
\end{equation}
where $\sigma = \pi^\Gamma$ and  $\tau = (\tilde{\pi})^\Gamma$ correspond to the original permutation operators with the partial transpose $\Gamma$ applied on the right $q$ copies. 
The summation is over all $(p+q)!$ possible PTPs for both $\sigma$ and $\tau$.
We also let $\widetilde{\Wg}_{\sigma,\tau} \equiv \Wg_{\pi,\tilde{\pi}}$ denote the analog of the Weingarten matrix elements for the PTPs.

Let us begin by reviewing a few basic facts regarding the PTPs.
We denote a \emph{mixed tensor unitary} (MTU) acting on $\mathcal{H}^{\otimes p} \otimes \mathcal{H}^{\otimes q}$ as,
\begin{equation}
\mathcal{U}^p_q \equiv U^{\otimes p} \otimes  U^{*, \otimes q},
\end{equation}
for any $U \in U(D)$ with $D = 2^n$.
We can draw each PTP using tensor network notation,
\begin{align}
\figbox{0.4}{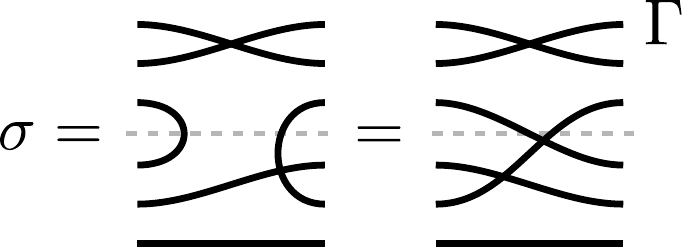}, \centering \label{eq: Brauer def}
\end{align}
where the ``left'' $p$ copies are depicted above the dashed line of the diagram, and the ``right'' $q$ copies are depicted below the dashed line.
One can easily verify that any PTP commutes with any MTU,
\begin{align}
\figbox{0.4}{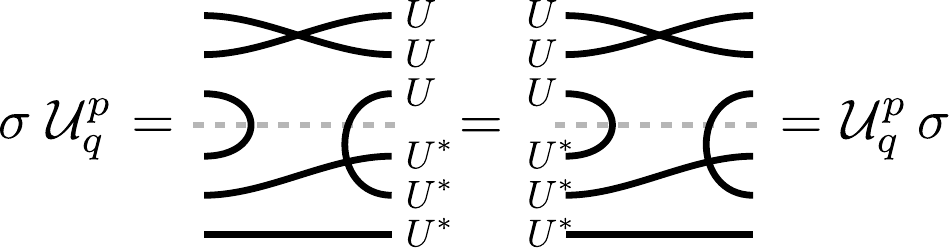}. \centering \label{eq: Brauer comm}
\end{align}
This is especially clear in the diagrammatic depiction of the PTP. Each $U$ in the MTU either slides from the right to left of the PTP, or cancels with a $U^*$ acting on a paired leg of the PTP.
The PTPs form a generating set for the commutant of the MTUs, i.e.~any operator that commutes with every MTU can be written as a sum of PTPs\footnote{An operator commutes with all MTUs if and only if its partial transpose commutes with all tensor power unitaries $U^{\otimes (p+q)}$. The statement that the PTPs form the commutant of the MTUs then follows from the well-known fact that the set of permutations $\pi \in S_k$ generates the commutant of the set of tensor power unitaries. This implies that the partial transpose of the aforementioned operator can be written as a sum of permutations, and hence the operator can be written as a sum of PTPs.}.

The PTPs form a representation of the so-called ``walled Brauer algebra'', $\mathcal{B}^D_{p,q}$. 
The multiplication rules of the algebra can be computed by connecting the legs of the associated PTPs.
For example,
\begin{align}
\figbox{0.4}{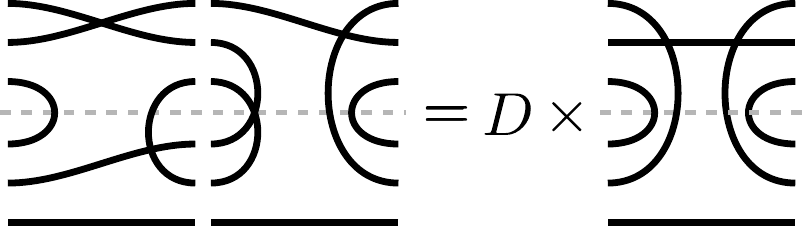} \centering \label{eq: Brauer multiply}
\end{align}
where each closed loop contributes a factor of the Hilbert space dimension $D$.
The representation is \emph{faithful} whenever $D \geq p+q$, i.e.~a linear combination of PTPs is equal to zero, $\sum_\sigma c_\sigma \sigma = 0$, if and only if each coefficient is zero\footnote{This follows from the well-known fact that the representation of the permutation group $S_k$ on $\mathcal{H}^{\otimes k}$ is faithful when $D \geq k$. A linear combination of PTPs is equal to zero if and only if its partial tranpose is equal to zero. Since the representation of the permutation group is faithful, this can be true only if every coefficient is equal to zero.}, $c_\sigma = 0$.

\subsubsection{Constructing complete orthogonal projectors from the PTPs}

Having defined the PTPs, we now begin our own analysis.
Let us first introduce some useful notation.
We can uniquely label each PTP by a set of four quantities, depicted as follows,
\begin{align}
\figbox{0.4}{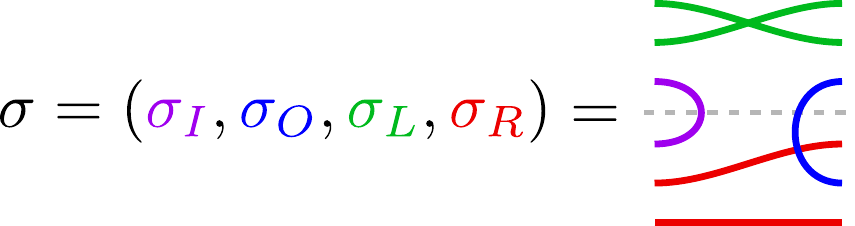} \centering \label{eq: Brauer notation}
\end{align}
The first quantity, $\sigma_I$, specifies the input legs that are paired under $\sigma$ (purple).
We term the number of pairs in $\sigma_I$ the \emph{size} of the PTP, and denote it as $\ell_\sigma = |\sigma_I|$, where $|\cdot|$ counts the number of pairs.
We have $0 \leq \ell_\sigma \leq \min(p,q)$.
The second quantity, $\sigma_O$, specifies the output legs that are paired under $\sigma$ (blue).
The number of output pairs is also $\ell_\sigma$, equal to the number of input pairs.
The third quantity, $\sigma_L \in S_{p-\ell_\sigma}$, specifies a permutation acting on the remaining $p-\ell_\sigma$  legs on the left side (green). 
Similarly, the fourth quantity, $\sigma_R \in S_{q-\ell_\sigma}$, specifies a permutation acting on the remaining $q-\ell_\sigma$  legs on the right side (red). 

The main aim of this section is to show that the PTPs naturally decompose the Hilbert space $\mathcal{H}^{\otimes p} \otimes \mathcal{H}^{\otimes q}$ into a  tensor sum of orthogonal subspaces.
Our derivation of this decomposition requires several steps.
%
To begin, we note that a subset of the PTPs with $\sigma_L = \sigma_R = \mathbbm{1}$ and $\sigma_I = \sigma_O \equiv \alpha$ (for any set of pairs $\alpha$), are proportional to projectors.
Namely, we have $\sigma^2 = D^{\ell_\sigma} \sigma$ for any such PTP where $\ell_\sigma = |\alpha|$.
This allows us to define the ``bare'' projectors,
\begin{equation}
	P_\alpha = \frac{1}{D^{|\alpha|}} \big( \alpha , \alpha , \mathbbm{1} , \mathbbm{1} \big).
\end{equation}
Each $P_\alpha$ projects onto the EPR state on each pair in $\alpha$, and acts as the identity on copies not in $\alpha$.

In general, the bare projectors are not orthogonal to one another. 
For example, whenever $\alpha \supset \beta$, the subspace defined by $P_\alpha$ is strictly contained within the subspace defined by $P_\beta$ (i.e.~$P_\alpha P_\beta = P_\alpha$).
This fact can make working with the bare projectors somewhat inconvenient.
To address this, we will need to orthogonalize the bare projectors.

We do so by introducing a new object: the \emph{no-EPR projector}.
For any subset $\beta$ of the $p+q$ copies, we define the no-EPR projector, $\noEPR_{\beta}$, as the projector onto the subspace of the Hilbert space of $\beta$ that is orthogonal to every EPR projector on $\beta$. That is, we let $\noEPR_{\beta}$ project onto the orthogonal complement of $\{ P_\alpha : \alpha \supseteq \beta \}$.
In the special case when $\beta$ contains every copy, we simply write $\noEPR$, which is the projector onto the orthogonal complement of $\{ P_\alpha  : \forall \alpha \neq \varnothing \}$.
%
We will write down an explicit expression for the no-EPR projector at the end of this subsection.
For now, we simply note that the no-EPR projector can be written as a sum of PTPs supported on $\beta$, since it commutes with every MTP unitary and acts as the identity on $\bar \beta$.

We can use the no-EPR projector to (partially) remedy the non-orthogonality of the bare projectors.
For each set of pairs $\alpha$, we define the ``nearly-orthogonal'' projectors,
\begin{align}
\figbox{0.4}{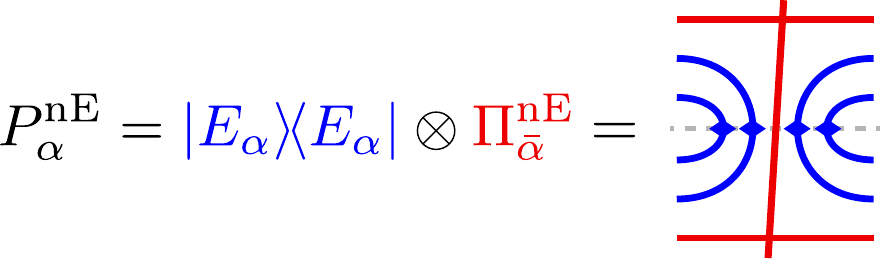} \centering \label{eq: PnE}
\end{align}
where $\ket{E_\alpha}$ denotes the EPR state on all pairs in $\alpha$.
Here, the slash represents the no-EPR projector on $\bar \alpha$, and each diamond represents a factor of the inverse square root Hilbert space dimension, $1/\sqrt{D}$, arising from the normalization of $\ket{E_\alpha}$. 
By definition, the nearly-orthogonal projector projects onto all states in $P_\alpha$ that are not contained in any $P_\beta$ for $\alpha \supset \beta$.
We can also write $\noP_\alpha = P_\alpha \noEPR_{\bar\alpha} = \noEPR_{\bar\alpha} P_\alpha$.

A simple argument shows that any two $\noP_\alpha$ and $\noP_\beta$ are orthogonal whenever the size of $\alpha$ and $\beta$ differ. We prove this in Appendix~\ref{app: nearly}.
\begin{proposition}[Nearly-orthogonal projectors] \label{prop: PnE orthogonality}
	Two nearly-orthogonal projectors $\noP_\alpha$ and $\noP_\beta$ are orthogonal, $\noP_\alpha \noP_\beta = 0$, if $|\alpha| \neq |\beta|$. 
\end{proposition}
\noindent However, the nearly-orthogonal projectors are \emph{not} orthogonal to one another when $|\alpha|=|\beta|$.

To address this latter fact, we can define a final set of ``orthogonal'' projectors, $\tilde{P}_\alpha$. 
We do so via the Gram-Schmidt process.
For each size $\ell = 0,\ldots,\min(p,q)$, we consider any ordering, $\{ \alpha_0, \alpha_1 , \ldots \}$, of the pairings $\alpha$ with size $|\alpha| = \ell$.
We then proceed $\alpha$-by-$\alpha$ through the ordering, and at each step define the orthogonal projector $\tilde{P}_\alpha$ as,
\begin{equation} \label{eq: def tilde P}
	\tilde{P}_\alpha \ket{\psi} = 
	\begin{cases}
		1, & \text{ if } \ket{\psi} \in \text{span} \left( \{ \noP_{\alpha'} : \alpha' \leq \alpha \} \right) \text{ and } \ket{\psi} \notin \text{span} \left( \{ \noP_{\alpha'} : \alpha' < \alpha \} \right) \\
		0, & \text{ else } \\
	\end{cases}.
\end{equation}
In the special case where $\alpha = \varnothing$ is the empty set, we have $\tilde{P}_{\varnothing} = \noEPR$.
The projectors $\tilde{P}_\alpha$ are mutually orthogonal by definition.
They are also complete, $\sum_\alpha \tilde{P}_\alpha = \mathbbm{1}$, by definition. 
Finally, it will be convenient to also define the projector onto the \emph{union} of all $\tilde{P}_\alpha$ of a given size $|\alpha| = \ell$. We term this the $\ell$\emph{-EPR projector},
\begin{equation}
	\tilde{P}_\ell \,\,\, = \sum_{\alpha : |\alpha| = \ell} \tilde{P}_\alpha.
\end{equation}
Unlike the individual orthogonal projectors $\tilde{P}_\alpha$, the projector $\tilde{P}_\ell$ is uniquely defined, independent of our ordering of the pairings $\alpha$ within each size $\ell$.

We can now demonstrate several useful properties of the orthogonal projectors.
We begin by proving (Appendix~\ref{app: orthogonal}) that the orthogonal projectors can be written as a sum of PTPs.
\begin{proposition}[Orthogonal projectors] \label{prop: proj diagrams}
	Each projector $\tilde{P}_\alpha$ can be written as a sum of PTPs with either (i) size greater than $|\alpha|$, or (ii) size equal to $|\alpha|$ and $\alpha_I, \alpha_O \leq \alpha$ with respect to the ordering in Eq.~(\ref{eq: def tilde P}).
\end{proposition}
\noindent Intuitively, this follows because the Gram-Schmidt process constructs an orthogonal vector by taking a linear combination of the current vector and all vectors previous to it.

We also have the following useful fact (Appendix~\ref{app: ell}).
\begin{proposition}[The $\ell$-EPR projector]  \label{ref: proj ell commutes}
	The projector $\tilde{P}_\ell$ commutes with every PTP.
\end{proposition}
\noindent This does not apply to the individual orthogonal projectors $\tilde{P}_\alpha$.

We can also characterize the rank of each orthogonal subspace.
To do so, for any $p$ and $q$, we let
\begin{equation}
	N_{\text{EPR}}^{(p,q)} = \text{rank } \left( \sum_{\alpha \neq \varnothing} P_{\alpha} \right),
\end{equation}
count the number of states on $\mathcal{H}^{\otimes p} \otimes \mathcal{H}^{\otimes q}$ that contain at least one EPR pair.
That is, the number of states in the span of any non-trivial bare projector $P_\alpha$.
We then show (Appendix~\ref{app: rank}),
\begin{proposition}[Rank of the orthogonal projectors] \label{prop: proj rank}
	For any $D \geq p + q$. Each projector $\noP_\alpha$ and $\tilde{P}_\alpha$ has rank $D^{p+q-2|\alpha|} - N_{\text{EPR}}^{(p-|\alpha|,q-|\alpha|)}$.
\end{proposition}
\noindent The positive term $D^{p+q-2|\alpha|}$ is the rank of the bare projector $P_\alpha$. The negative term $N_{\text{EPR}}^{(p-|\alpha|,q-|\alpha|)}$ counts the number of states in $P_{\alpha}$ that are contained in $P_{\alpha'}$ for any $\alpha' \supset \alpha$.
These states are removed from $\noP_\alpha$ and $\tilde{P}_\alpha$ and hence are subtracted from the rank.
The fact that the rank of $\noP_\alpha$ and $\tilde{P}_\alpha$ is the same implies that every state in $\noP_\alpha$ is outside the span of $\noP_\beta$ for all $\beta \neq \alpha$.
The proposition immediately implies that the rank of the $\ell$-EPR projector $\tilde{P}_\ell$ is equal to 
\begin{equation}
    D_\ell \equiv \text{rank } \tilde{P}_\ell = {p \choose \ell}{q \choose \ell} \ell! \cdot \bigg(D^{p+q-2\ell} - N_{\text{EPR}}^{(p-\ell,q-\ell)} \bigg),
\end{equation}
where ${p \choose \ell}{q \choose \ell} \ell!$ counts the number of $\alpha$ with size $\ell$.
The value of $N_{\text{EPR}}^{(p,q)}$ is tricky to compute in general.
However, one has an immediate upper bound, $N_{\text{EPR}}^{(p',q')} \leq \sum_{|\alpha|=1} \text{rank}(P_\alpha) = pq \cdot D^{p'+q'-2}$.

Before proceeding, we pause to provide the following explicit expression for the no-EPR projector $\noEPR$.
This expression is not needed for any of the results in the preceding sections; we mention it solely for completeness for the interested reader.
The expression is as follows (Appendix~\ref{app: expression}).
\begin{proposition}[Expression for the no-EPR projector] \label{prop: no EPR}
    For any $p+q \leq D$.
    The inner product of the no-EPR projector, $\noEPR$, with a permutation operator, $\pi_L \otimes \pi_R$, is given by,
    \begin{equation} \label{eq: trace noEPR}
        \tr( \noEPR (\pi_L \otimes \pi_R)^{-1} ) = [ \hat{\Wg}|_{\text{\emph {perm}}} ]^{-1}_{\mathbbm{1}, \pi_L \otimes \pi_R},
    \end{equation}
    where $\hat{\Wg}|_{\text{\emph {perm}}}$ is the $(p! q!) \times (p! q!)$ sub-matrix obtained by restricting  the $(p+q)! \times (p+q)!$ Weingarten matrix, $\hat{\Wg}$, to the permutation operators.
    As a consequence, the no-EPR projector can be written as a sum of PTPs with coefficients, 
    \begin{equation} \label{eq: noEPR projector expansion}
        \noEPR = \sum_\sigma \left( \sum_{\pi_L , \pi_R} [ \hat{\Wg}|_{\text{\emph {perm}}} ]^{-1}_{\mathbbm{1}, \pi_L \otimes \pi_R}
        \Wg_{\pi_L \otimes \pi_R,\sigma} \right) \sigma.
    \end{equation}.
\end{proposition}
\noindent The expression for the no-EPR projector is not particularly easy to work with, given that it involves the inverse of a sub-matrix of the Weingarten matrix. We provide further discussion and applications in Appendix~\ref{sec: additional}

\subsubsection{Partial isometries and reformulation of the mixed Haar twirl}

The mixed Haar twirl has a particularly simple action on the orthogonal subspaces constructed in the previous subsection.
To show this, let us first discuss the bare projectors, and then turn to the orthogonal projectors.
Each bare projector $P_\alpha$ naturally defines a \emph{partial isometry} from the $(p+q)$-copy Hilbert space to a smaller $(p+q-2|\alpha|)$-copy Hilbert space,
\begin{align}
\figbox{0.4}{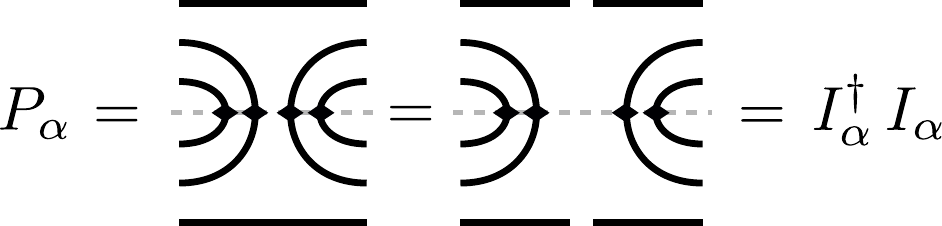}. \centering \label{eq: isometry def}
\end{align}
The $p+q-2|\alpha|$ copies correspond to the legs of the PTP that are not paired in $\alpha$.
The isometries are unitary-equivariant,
\begin{align}
\figbox{0.42}{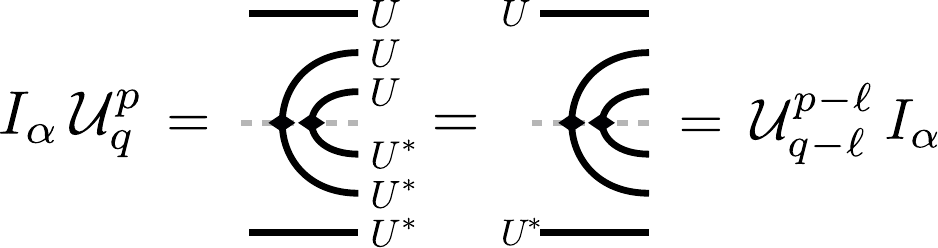}, \centering \label{eq: isometry comm}
\end{align}
meaning that the action of any MTU on the $(p+q)$-copy Hilbert space translates to a corresponding smaller MTU on the $(p+q-2|\alpha|)$-copy Hilbert space.

We will now show that an analogous set of isometries can be constructed for the orthogonal projectors $\tilde{P}_\alpha$.
We define these carefully in the following manner.
To begin, we expand each orthogonal projector as follows,
\begin{equation}
	\tilde{P}_\alpha = \tilde{P}_\alpha \tilde{P}_\alpha \tilde{P}_\alpha =  \tilde{P}_\alpha \left( \sum_{\pi_L,\pi_R} c_{\pi_L, \pi_R} \cdot \big( \alpha, \alpha, \pi_L, \pi_R \big) \right) \tilde{P}_\alpha,
\end{equation}
where, in the rightmost expression, we write the center $\tilde{P}_\alpha$ as a linear combination of PTPs and keep only the PTPs, $( \alpha, \alpha, \pi_L, \pi_R )$, that have both input and output pairs equal to $\alpha$ (since all other PTPs vanish upon conjugation by $\tilde{P}_\alpha$).
By definition, each PTP $( \alpha, \alpha, \pi_L, \pi_R )$ is proportional to the EPR projector, $\dyad{E_\alpha}$, on subsystem $\alpha$.
To proceed, we insert the following resolution of the identity on the complement of $\alpha$,
\begin{equation}
	\mathbbm{1}_{\bar \alpha} = \sum_{\gamma} \tilde{P}_\gamma^{(\bar\alpha)} = \noEPR_{\bar \alpha} + \sum_{\gamma \neq \varnothing} \tilde{P}_\gamma^{(\bar \alpha)},
\end{equation}
where $\gamma$ runs over sets of pairs in $\bar \alpha$, and the superscript denotes that the orthogonal projector is constructed from PTPs that act only on $\bar \alpha$.
From Proposition~\ref{prop: proj diagrams},
each $\tilde{P}_\gamma^{(\bar \alpha)}$ can be written as a sum of PTPs on $\bar \alpha$ with size at least $1$.
Thus, the tensor product, $\tilde{P}_\gamma^{(\bar \alpha)} \otimes \dyad{E_\alpha}$, can be written as a sum of PTPs with size at least $|\alpha|+1$.
This implies that $\tilde{P}_\gamma^{(\bar \alpha)} \otimes \dyad{E_\alpha}$ vanishes upon left or right multiplication with $\tilde{P}_\alpha$, which yields,
\begin{equation}
	\tilde{P}_\alpha 
	= \tilde{P}_\alpha \mathbbm{1}_{\bar \alpha} \left( \sum_{\pi_L,\pi_R} c_{\pi_L \pi_R} \cdot \big( \alpha, \alpha, \pi_L, \pi_R \big) \right) \mathbbm{1}_{\bar \alpha} \tilde{P}_\alpha 
	=  \tilde{P}_\alpha \noEPR_{\bar \alpha} \left( \sum_{\pi_L,\pi_R} c_{\pi_L \pi_R} \cdot \big( \alpha, \alpha, \pi_L, \pi_R \big) \right) \noEPR_{\bar \alpha}  \tilde{P}_\alpha.
\end{equation}
To proceed, let us take the square root of the middle operators,
\begin{equation} \label{eq: def M alpha}
	M_\alpha = \left(  \sum_{\pi_L,\pi_R} c_{\pi_L \pi_R} \cdot \noEPR_{\bar \alpha} \big( \alpha, \alpha, \pi_L, \pi_R \big) \noEPR_{\bar \alpha} \right)^{1/2} \!\!\!\! \equiv M'_{\bar \alpha} \otimes \dyad{E_\alpha},
\end{equation}
where $M'_{\bar \alpha}$ acts on $\bar \alpha$.
We prove that the square root is well-defined within the proof of Proposition~\ref{prop: isometry} (Appendix~\ref{app: partial}), by showing that the operator inside the parentheses is a positive operator.
From the right hand side, we can see that $M_\alpha = M_\alpha P_\alpha = P_\alpha M_\alpha$.
Thus, we can write
\begin{equation} \label{eq: tilde P decom tilde I}
	\tilde{P}_\alpha = \tilde{P}_\alpha \tilde{P}_\alpha \tilde{P}_\alpha =  \tilde{P}_\alpha M_\alpha M_\alpha \tilde{P}_\alpha =  \tilde{P}_\alpha M_\alpha P_\alpha M_\alpha \tilde{P}_\alpha =  ( \tilde{P}_\alpha M_\alpha I_\alpha^\dagger ) ( I_\alpha M_\alpha \tilde{P}_\alpha ).
\end{equation}
This decomposition immediately allows us to write $\tilde{P}_\alpha$ in terms of the ``orthogonal'' isometry,
\begin{equation} \label{eq: tilde I def}
	\tilde{I}_\alpha = I_\alpha M_\alpha \tilde{P}_\alpha,
\end{equation}
where we have $\tilde{P}_\alpha = \tilde{I}_\alpha^\dagger \tilde{I}_\alpha$ from Eq.~(\ref{eq: tilde P decom tilde I}).
This completes our definition of the partial isometries.

We establish a few key properties of the partial isometries (Appendix~\ref{app: partial}).
\begin{proposition}[Partial isometries] \label{prop: isometry}
	For any $D \geq p+q$, the map $\tilde{I}_\alpha$ is well-defined and is an isometry from the support of $\tilde{P}_\alpha$ to the no-EPR subspace of $\mathcal{H}^{\otimes p-|\alpha|} \otimes \mathcal{H}^{\otimes q-|\alpha|}$.
	The isometry is unitary-equivariant, $\tilde{I}_\alpha (\mathcal{U}^p_q)  = (\mathcal{U}^{p-\ell}_{q-\ell}) \, \tilde{I}_\alpha$.
\end{proposition}
\noindent The range of the isometry is restricted to the no-EPR subspace of $\mathcal{H}^{\otimes p-|\alpha|} \otimes \mathcal{H}^{\otimes q-|\alpha|}$, in contrast the bare isometry defined earlier which maps to the entire space.
This reflects the smaller size of $\tilde{P}_\alpha$ compared to $P_\alpha$.

We can also define a final set of partial isometries, $\tilde{I}_\ell$, associated to the $\ell$-EPR projectors $\tilde{P}_\ell$.
These will play a particularly important role since, as aforementioned, the definition of $\tilde{P}_\ell$ is unique and independent of our ordering of the $\alpha$, in contrast to $\tilde{P}_\alpha$.
The support of $\tilde{P}_\ell$ is equal to the tensor sum of the support of each $\tilde{P}_\alpha$ with size $|\alpha| = \ell$.
Hence, the partial isometry $\tilde{I}_\ell$ will map from this domain to the ${p \choose \ell}{q \choose \ell} \ell!$-fold tensor sum of the range of each $\tilde{I}_\alpha$.
Each $\tilde{I}_\alpha$ has range equal to the no-EPR subspace of $\mathcal{H}^{\otimes p-|\alpha|} \otimes \mathcal{H}^{\otimes q-|\alpha|}$.
Hence, the range of $\tilde{I}_\ell$ will be equal to,
\begin{equation}
    \left[ \mathcal{H}^{\otimes (p-\ell)} \otimes \mathcal{H}^{\otimes (q-\ell)} \right]_{\text{nE}} \otimes \mathcal{A}_\ell,
\end{equation}
where $|\mathcal{A}_\ell| = {p \choose \ell}{q \choose \ell} \ell!$ counts the number of $\alpha$ with size $|\alpha| = \ell$, i.e.~the number of terms in the tensor sum.
The first term denotes the no-EPR subspace of $\mathcal{H}^{\otimes (p-\ell)} \otimes \mathcal{H}^{\otimes (q-\ell)}$.
With the range defined, we write the partial isometry as,
\begin{equation} \label{eq: tilde I ell def}
	\tilde{I}_\ell = \sum_{\alpha : |\alpha| = \ell} \tilde{I}_\alpha,
\end{equation}
with the convention that each $\tilde{I}_\alpha$ maps from the support of $\tilde{P}_\alpha$ to $\left[ \mathcal{H}^{\otimes (p-\ell)} \otimes \mathcal{H}^{\otimes (q-\ell)} \right]_{\text{nE}} \otimes \dyad{\alpha}$ for $\ket{\alpha} \in \mathcal{A}_\ell$.
From Proposition~\ref{prop: isometry}, the orthogonal isometries $\tilde{I}_\ell$ are unitary-equivariant, in the sense that
$\tilde{I}_\ell \,  (\mathcal{U}^p_q)   = \left( \mathcal{U}^{p-\ell}_{q-\ell} \otimes \mathbbm{1}_{\mathcal{A}_\ell} \right) \, \tilde{I}_\ell$ for any MTU.

The partial isometries allow us to reformulate the mixed Haar twirl as follows.
Let us begin by inserting the resolution of the identity, $\mathbbm{1} = \sum_\ell \tilde{I}_\ell^\dagger \tilde{I}_\ell$, twice in the definition of the mixed Haar twirl,
\begin{equation}\nonumber
\begin{split}
    \Phi^{(p,q)}_H(X) & \equiv \E_{U \sim H} \bigg[ \bigg( \sum_{\ell} \tilde{I}_\ell^\dagger \tilde{I}^{}_\ell \bigg) \, (\mathcal{U}^p_q) X \, (\mathcal{U}^p_q)^\dagger  \bigg( \sum_{\ell'} \tilde{I}_{\ell'}^\dagger \tilde{I}^{}_{\ell'} \bigg) \bigg] 
     = \sum_{\ell,\ell'} \E_{U \sim H} \bigg[ \tilde{I}_\ell^\dagger 
    ((\mathcal{U}^{p-\ell}_{q-\ell}))
    \tilde{I}^{}_\ell \, X \, \tilde{I}_{\ell'}^\dagger (\mathcal{U}_{p-\ell',q-\ell'}^\dagger) \tilde{I}^{}_{\ell'} \bigg],
\end{split}
\end{equation}
where on the right side, we move the action of the MTU to the inside of the partial isometry.
We suppress the $\mathcal{A}_\ell$, $\mathcal{A}_{\ell'}$ registers for brevity, since the MTU acts trivially on these registers.
We can now apply the formula for the mixed Haar twirl in terms of PTPs to each term $\ell,\ell'$ above.
This yields,
\begin{equation} \label{eq: exact mixed twirl subspace 2}
\begin{split}
    \Phi^{(p,q)}_H(\rho) & = \sum_{\ell} \tilde{I}_{\ell}^\dagger \left[ \sum_{\pi_L  \pi_R} \sum_{\tilde{\pi}_L  \tilde{\pi}_R}
    \tr( \tilde{I}^{}_\ell \, \rho \, \tilde{I}_{\ell}^\dagger (\pi_L \otimes \pi_R)^{-1} ) \cdot \Wg^{(p+q-2\ell)}_{\pi_L \otimes \pi_R, \tilde{\pi}_L \otimes \tilde{\pi}_L} \cdot  (\tilde{\pi}_L \otimes \tilde{\pi}_R) \right] \tilde{I}_{\ell},
\end{split}
\end{equation}
where $\Wg^{(p+q-2\ell)}_{\pi_L \otimes \pi_R, \tilde{\pi}_L \otimes \tilde{\pi}_L}$ denotes the Weingarten matrix on $p+q-2\ell$ copies.
To compute this expression, we note that any PTP with size greater than zero vanishes upon conjugation by $\tilde{I}_\ell$, since the range of $\tilde{I}_\ell$ is the no-EPR subspace (Proposition~\ref{prop: isometry}) and PTPs with size greater than zero involve at least one EPR projector.
This implies that the terms with $\ell \neq \ell'$ vanish, and that the remaining $\ell = \ell'$ terms involve only tensor products permutation operators between the left and right side, $\pi_L,\tilde{\pi}_L \in S_{p-\ell}$ and $\pi_R,\tilde{\pi}_R \in S_{q-\ell}$.

This completes our derivation of the reformulation of the mixed Haar twirl.
In the remaining subsections, we provide the detailed proofs of each proposition above.

\subsubsection{Proof of Proposition~\ref{prop: PnE orthogonality}: Nearly-orthogonal projectors} \label{app: nearly}

Let us first establish a simple fact regarding the PTPs.
Namely, the size of a PTP can only increase under multiplication.
\begin{fact} \label{fact: PPT mult}
	The product of a size $\ell$ PPT and a size $\ell'$ PPT has size at least $\text{\emph{max}}(\ell,\ell')$.
\end{fact}
\begin{proof}
	Call the two PPTs $\sigma$ and $\sigma'$, respectively, and let $\sigma_I$ denote the input pairs of $\sigma$, and $\sigma_O'$ denote the output pairs of $\sigma'$.
	We have $| \sigma_I | = \ell$ and $| \sigma_O' | = \ell'$.
	One can easily see that the product $\sigma \sigma'$ is a PPT $\beta$ with input pairs $\beta_I \supseteq \sigma_I$ and output pairs $\beta_O \supseteq \sigma_O'$.
	The size of $\sigma \sigma'$ is given by $| \beta_I | = | \beta_O | \geq \text{max}(| \sigma_I |, | \sigma_O' |) = \text{max}(\ell,\ell')$.
\end{proof}

We can now prove Proposition~\ref{prop: PnE orthogonality}.
Assume $|\beta|| > |\alpha|$ without loss of generality.
We know that the product of the bare projectors, $P_\alpha P_\beta$, is proportional to a PTP $\gamma$ with $\gamma_I \supseteq \alpha$.
From Fact~\ref{fact: PPT mult}, the PTP has size at least $|\gamma_I| = |\gamma_O|=|\beta|$.
Since $|\beta| > |\alpha|$, this implies $\gamma_I \supsetneq \alpha$.
This implies that $\gamma_I$ contains at least one pair on $\bar \alpha$.
Thus, we have $\noEPR_{\bar \alpha} (P_\alpha P_\beta) = 0$ and hence, $\noP_\alpha \noP_\beta = \noEPR_{\bar \alpha} P_\alpha P_\beta \noEPR_{\bar \beta} = 0$. \qed

\subsubsection{Proof of Proposition~\ref{prop: proj diagrams}: Orthogonal projectors} \label{app: orthogonal}

We first note that each $\tilde{P}_\alpha$ commutes with any MTU, $\mathcal{U}^p_q = U^{\otimes p} \otimes  U^{*, \otimes q}$.
	To see this, recall that each bare projector $P_{\alpha}$ commutes with $\mathcal{U}^p_q$.
	Hence, if $\ket{\psi}$ is in the span of $P_\alpha$, then $\mathcal{U}^p_q \ket{\psi}$ is also in the span of $P_\alpha$.
	Similarly, it follows that if $\ket{\psi}$ is in $\text{span} \left( \{ P_{\alpha'} : \alpha' \leq \alpha \} \right)$, then $\mathcal{U}^p_q \ket{\psi}$ is also in the same span.
	Observing the definition of $\tilde{P}_\alpha$ [Eq.~(\ref{eq: def tilde P})], we see that $\tilde{P}_\alpha \mathcal{U}^p_q \ket{\psi} = \mathcal{U}^p_q \ket{\psi}$ if and only if $\tilde{P}_\alpha \ket{\psi} =  \ket{\psi}$.
	Hence, $\tilde{P}_\alpha$ commutes with $\mathcal{U}^p_q$.
	
	To complete our proof, let $\{ \ket{ \phi } \}$ denote an orthonormal basis for the unit eigenspace of $\tilde{P}_\alpha$.
	From the definition of $\tilde{P}_\alpha$, each vector $\ket{ \phi }$ can be written as a sum, $\ket{\phi} = \sum_{\alpha' \leq \alpha} \ket{ E_{\alpha'} } \otimes \ket{ \phi_{\alpha'} }$, of vectors in $P_{\alpha'}$ for $\alpha' \leq \alpha$.
	Since $\tilde{P}_\alpha$ commutes with $\mathcal{U}^p_q$, we have $\tilde{P}_\alpha = \E_{U \sim H} [ \mathcal{U}^p_q \tilde{P}_\alpha (\mathcal{U}^p_q)^\dagger ]$.
	Expressed in terms of the orthonormal basis $\{ \ket{ \phi } \}$, this gives
	\begin{equation}
		\tilde{P}_\alpha = \sum_\phi \E_{U \sim H} \left[ (\mathcal{U}^p_q) \dyad{\phi} (\mathcal{U}^p_q)^\dagger \right] = \sum_\phi \sum_{\alpha', \alpha'' \leq \alpha} \E_{U \sim H} \left[ (\mathcal{U}^p_q) \dyad{\phi_{\alpha'}}{\phi_{\alpha''}} (\mathcal{U}^p_q)^\dagger \right].
	\end{equation}
	By definition, we can write $\ket{ \phi_\alpha } = \ket{ E_\alpha } \otimes \ket{ \phi^c_{\bar \alpha} }$, where $\ket{E_\alpha}$ is the EPR projector on $\alpha$, and $\ket{ \phi^c_{\bar \alpha}}$ is a state on the complement of $\alpha$.
	When we apply $\mathcal{U}^p_q$ to this state, the factors acting on $\alpha$ cancel, leaving
	\begin{equation}
		\mathcal{U}^p_q \ket{ \phi_\alpha } = \ket{ E_\alpha } \otimes \mathcal{U}^{p-|\alpha'|}_{q-|\alpha'|} \ket{ \phi^c_{\bar \alpha} }.
	\end{equation}
	Applying the formula for the mixed Haar twirl to $\mathcal{U}^{p-|\alpha'|}_{q-|\alpha'|}$, one finds that each state
    \begin{equation}
        \E_{U \sim H} \left[ \mathcal{U}^p_q \dyad{\phi_{\alpha'}}{\phi_{\alpha''}} (\mathcal{U}^p_q)^\dagger \right]
    \end{equation}
    can be expressed as a sum of PTPs $\sigma$ with $\sigma_I \supseteq \alpha'$ and $\sigma_O \supseteq \alpha''$.
	This implies that $\sigma_I \leq \alpha'$ and $\sigma_O \leq \alpha''$, since our ordering of pairs was strictly decreasing.
	Hence, we obtain an expression for $\tilde{P}_\alpha$ as a sum of PTPs with $\sigma_I, \sigma_O \leq \alpha$. \qed

    \subsubsection{Proof of Proposition~\ref{ref: proj ell commutes}: The $\ell$-EPR projector} \label{app: ell}

    One can generate any PTP from a product of permutations, $\pi_L \otimes \pi_R$, and a single EPR projector, $\Pi_{ij}$, onto a copy $i$ on the left side and a copy $j$ on the right side.
	Hence, to prove that $\tilde{P}_\ell$ commutes with any PTP, it suffices to prove that $\tilde{P}_\ell$ commutes with each permutation, as well as $\Pi_{ij}$.
	The former follows from the definition of $\tilde{P}_\ell$, since conjugation by a permutation simply amounts to a re-ordering of the $\alpha$ within each size, and the projector $\tilde{P}_\ell$ is manifestly independent of this ordering.
	Thus, it remains only to show that $\tilde{P}_\ell$ commutes with $\Pi_{ij}$.
	
	Since $\tilde{P}_{\ell}$ can be written as a sum of PTPs with size greater than or equal to $\ell$ (Proposition~\ref{prop: proj diagrams}),
	 the operator $\Pi_{ij} \tilde{P}_{\ell}$ can also be written as a sum of PTPs with such size (using Fact~\ref{fact: PPT mult}). 
	This implies that $\tilde{P}_{\ell'} \Pi_{ij} \tilde{P}_\ell = 0$ for any $\ell' > \ell$, since $\tilde{P}_{\ell'}$ is orthogonal to any PTP with size greater than $\ell$ by definition.
	Applying the same argument to the Hermitian conjugate, $(\tilde{P}_{\ell'} \Pi_{ij} \tilde{P}_\ell)^\dagger = \tilde{P}_{\ell} \Pi_{ij} \tilde{P}_{\ell'}$, we have that $\tilde{P}_{\ell'} \Pi_{ij} \tilde{P}_\ell = 0$ for any $\ell' < \ell$, and hence $\ell' \neq \ell$.
	The desired commutation follows immediately. We write
	\begin{equation}
		\tilde{P}_\ell \Pi_{ij} \tilde{P}_\ell = ( \mathbbm{1} - \sum_{\ell' \neq \ell} \tilde{P}_{\ell'} ) \Pi_{ij} \tilde{P}_\ell = \Pi_{ij} \tilde{P}_\ell,
	\end{equation}
	where in the first equality we use that the projectors form a complete basis, $\sum_{\ell'} \tilde{P}_{\ell'} = \mathbbm{1}$, and in the second equality we use that $\tilde{P}_{\ell'} \Pi_{ij} \tilde{P}_\ell = 0$ if $\ell' \neq \ell$.
	Taking the Hermitian conjugate of both sides above yields $\tilde{P}_\ell \Pi_{ij} \tilde{P}_\ell = \tilde{P}_\ell \Pi_{ij}$.
	Thus, we have $\Pi_{ij} \tilde{P}_\ell = \tilde{P}_\ell \Pi_{ij} \tilde{P}_\ell = \tilde{P}_\ell \Pi_{ij}$, i.e.~$\tilde{P}_\ell$ and $\Pi_{ij}$  commute.
	This completes the proof. \qed

\subsubsection{Proof of Proposition~\ref{prop: proj rank}: Rank of the orthogonal projectors} \label{app: rank}

Let  $\oneEPR_{\bar \alpha} = \mathbbm{1} - \noEPR_{\bar \alpha} = \sum_{\gamma\neq \varnothing} \tilde{P}^{(\bar\alpha)}_\gamma$ denote the projector onto states with at least one EPR pair on $\bar \alpha$.
We have 
\begin{equation}
	P_\alpha = \dyad{E_\alpha} \otimes (\noEPR_{\bar \alpha} + \oneEPR_{\bar \alpha}).
\end{equation}
Only the first term contains states that contribute to $\tilde{P}_\alpha$.
To see this, we apply Proposition~\ref{prop: proj diagrams} to subsystem $\bar \alpha$, which shows that $\oneEPR_{\bar \alpha}$ can be written as a sum of PTPs on $\bar \alpha$ with at least one EPR projector.
After taking the tensor product with $\dyad{E_\alpha}$, we obtain a sum of PTPs with at least $|\alpha| + 1$ EPR projectors.
Thus, we have $\text{span} ( \dyad{E_\alpha} \otimes \oneEPR_{\bar \alpha} ) \subseteq \text{span}(\{ P_{\alpha'} : \alpha' < \alpha \})$, which is orthogonal to $\text{span} ( \tilde{P}_\alpha )$ by definition.
Since $P_\alpha$ has rank $D^{p+q-2|\alpha|}$ and $\oneEPR_{\bar \alpha}$ has rank $N_{\text{EPR}}^{(p-|\alpha|,q-|\alpha|)}$, this implies that the rank of  $\tilde{P}_\alpha$  is upper bounded by $D^{p+q-2|\alpha|} -  N_{\text{EPR}}^{(p-|\alpha|,q-|\alpha|)}$.

To show that the rank is equal to this value, we must show that every non-zero vector in $\text{span}( \dyad{E_\alpha} \otimes \noEPR_{\bar \alpha} )$ is outside of $\text{span}( \{ \tilde{P}_{\alpha'} : \alpha' < \alpha \})$.
We provide a proof by contradiction.
Suppose that a non-zero vector $\ket{\psi}$ in $\noP_\alpha$ is inside of $\text{span}(\{ P_{\alpha'} : \alpha' \leq \alpha \})$.
This implies that we can write both $\ket{\psi} = \ket{E_\alpha} \otimes \ket{\psi_{\bar \alpha}}$ for some $\ket{\psi_{\bar \alpha}} \in \noEPR_{\bar \alpha}$, as well as $\ket{\psi} = \sum_{\alpha' < \alpha} \ket{E_\alpha'} \otimes \ket{\psi_{\bar \alpha'}}$ for some $\ket{\psi_{\bar \alpha'}}$.
Since both $P_\alpha$ and $\noP_\alpha$ commute with any mixed tensor unitary $\mathcal{U}^p_{q}$, we also have that $\mathcal{U}^p_{q} \ket{\psi}$ is in both $\text{span}(\noP_\alpha)$ and $\text{span}(\{ P_{\alpha'} : \alpha' \leq \alpha \})$.
Thus, we are free to average over $U$ to obtain
\begin{equation} \label{eq: psi twirl 1}
\begin{split}
	\E_{U \sim H} \left[ \mathcal{U}^p_{q} \dyad{\psi} (\mathcal{U}^p_{q})^\dagger \right] & = 
	\E_{U \sim H} \left[ \dyad{E_\alpha} \otimes \mathcal{U}^{p-\ell}_{q-\ell} \dyad{\psi_{\bar \alpha}} (\mathcal{U}^{p-\ell}_{q-\ell})^\dagger \right]\\
	& = \sum_{\pi_L,\pi_R}  c_{\pi_L \pi_R} \Pi^{\bar \alpha}_{\text{noEPR}}   \big( \alpha, \alpha, \pi_L, \pi_R \big) \noEPR_{\bar \alpha}
\end{split}
\end{equation}
using our first decomposition of $\ket{\psi}$.
Here, we have used that $\ket{\psi_{\bar \alpha}} \in \noEPR_{\bar \alpha}$, i.e.~$\ket{\psi_{\bar \alpha}} = \noEPR_{\bar \alpha} \ket{\psi_{\bar \alpha}}$, to insert the no-EPR projector on $\bar \alpha$.
The coefficients $c_{\pi_L \pi_R}$ are obtained from the Haar twirl over $\dyad{\psi_{\bar \alpha}}$.
In a similar manner, we can obtain
\begin{equation} \label{eq: psi twirl 2}
	\E_{U \sim H} \left[ (\mathcal{U}^p_q) \dyad{\psi} (\mathcal{U}^p_q)^\dagger \right] = 
	\sum_{ \substack{ \alpha_I < \alpha \\ \alpha_O < \alpha }} \sum_{\pi_L,\pi_R} d^{\alpha_I \alpha_O}_{\pi_L \pi_R} \cdot \big( \alpha_I, \alpha_O, \pi_L, \pi_R \big),
\end{equation}
using our second decomposition of $\ket{\psi}$.
The coefficients $d^{\alpha_I \alpha_O}_{\pi_L \pi_R}$ are obtained from the Haar twirl over the various $( \ket{E_{\alpha_I}} \otimes \ket{\psi_{\bar \alpha_I}} ) ( \bra{E_{\alpha_O}} \otimes \bra{\psi_{\bar \alpha_O}} )$.

Note that the latter expression, Eq.~(\ref{eq: psi twirl 2}), contains only PTPs with $\alpha_I,\alpha_O < \alpha$.
Since the representation of the walled Brauer algebra is faithful for $D \geq p+q$, this implies that the former expression, Eq.~(\ref{eq: psi twirl 1}), must also contain only such PTPs. 
This follows because the PTPs are linearly independent operators when the representation is faithful.
We will now show that this leads to a contradiction.
To begin, recall that we can write $\noEPR_{\bar \alpha} = \mathbbm{1} - \oneEPR_{\bar \alpha}$, and that $\dyad{E_\alpha} \otimes \Pi^{\bar \alpha}_{>1\text{EPR}}$ can be expressed as a sum of diagrams with $\alpha_I, \alpha_O < \alpha$.
Thus, Eq.~(\ref{eq: psi twirl 1}) can be written as
\begin{equation}
	\sum_{\pi_L,\pi_R}  c_{\pi_L \pi_R} \Pi^{\bar \alpha}_{\text{noEPR}}   \big( \alpha, \alpha, \pi_L, \pi_R \big) \noEPR_{\bar \alpha} = \sum_{\pi_L,\pi_R} c_{\pi_L \pi_R}   \big( \alpha, \alpha, \pi_L, \pi_R \big)  + \left( \text{PTPs with $\alpha_I < \alpha$ or $\alpha_O < \alpha$} \right).
\end{equation}
The first term contains solely PTPs with $\alpha_I = \alpha_O = \alpha$.
Thus, if Eq.~(\ref{eq: psi twirl 1}) can be written as a sum of PTPs with $\alpha_I,\alpha_O < \alpha$, we must have $c_{\pi_L \pi_R} = 0$ for all $\pi_L, \pi_R$.
This implies that $\langle \psi_{\bar\alpha} | \psi_{\bar\alpha} \rangle$ is zero, which implies that $\ket{\psi}$ is zero, which is a contradiction.
We conclude that every non-zero vector in $\text{span}(\noP_\alpha)$ is outside $\text{span}(\{P_{\alpha'} : \alpha' \leq \alpha \})$.
This implies that the rank of $\tilde{P}_\alpha$ is equal to the rank of $\noP_\alpha$, which proves the proposition. \qed

\subsubsection{Proof of Proposition~\ref{prop: no EPR}: Expression for the no-EPR projector} \label{app: expression}

Let us first prove Eq.~(\ref{eq: trace noEPR}).
    We know from Proposition~\ref{prop: proj diagrams} that the no-EPR projector can be written as a sum of PTPs.
    Hence, we can write,
    \begin{equation} \label{eq: noEPR sigma tau}
        \noEPR = \sum_{\sigma,\tau} \widetilde{\Wg}_{\sigma,\tau} \cdot \text{tr}( \noEPR \sigma^\dagger ) \cdot \tau,
    \end{equation}
    where $\sigma,\tau$ run over the $(p+q)!$ PTPs.
    We can simplify this expression in two ways.
    First, we use that $\tr( \noEPR \sigma ) = 0$ unless $\sigma$ is a permutation, $\sigma = \pi_L \otimes \pi_R$.
    Second, we can use that $(\noEPR)^2 = \noEPR$ to multiply the right hand side by $\noEPR$.
    This similarly restricts the sum over $\tau$, since $\tau \noEPR = 0$ unless  $\tau$ is a permutation, $\tau = \tilde{\pi}_L \otimes \tilde{\pi}_R$.
    Together, these give
    \begin{equation} \label{eq: noEPR pis}
        \noEPR 
        = \sum_{\pi_L \pi_R} \sum_{\tilde{\pi}_L \tilde{\pi}_R} \Wg_{\pi_L \otimes \pi_R,\tilde{\pi}_L\otimes \tilde{\pi}_R} \cdot \tr( \noEPR ( \pi_L \otimes \pi_R )^{-1}  ) \cdot (\tilde{\pi}_L \otimes \tilde{\pi}_R) \noEPR.
    \end{equation}

    To proceed, consider the operators, $ \{ (\tilde{\pi}_L \otimes \tilde{\pi}_R) \noEPR \}$, appearing on the right side of Eq.~(\ref{eq: noEPR pis}).
    We will prove that these operators are linearly independent.
    Suppose that there exists coefficients, $c_{\tilde{\pi}}$, for $\tilde{\pi} \equiv \tilde{\pi}_L \otimes \tilde{\pi}_R$, such that $\sum_{\tilde{\pi}} c_{\tilde{\pi}} \tilde{\pi} \noEPR = 0$.
    Now, the no-EPR projector can be written as $\noEPR = \mathbbm{1} - \sum_{\ell=1}^{\min(p,q)} \tilde{P}_\ell$, where each projector $\tilde{P}_\ell$ can be written as a sum of PTPs with size greater than or equal to $\ell \geq 1$.
    Hence, we can write $\sum_{\tilde{\pi}} c_{\tilde{\pi}} \tilde{\pi} \noEPR = \sum_{\tilde{\pi}} c_{\tilde{\pi}} \tilde{\pi} + \Delta$, where $\Delta$ is a sum of PTPs with size $\geq 1$.
    If the left side of this expression were to vanish, as supposed, than the first term on the right side must vanish as well, since the PTPs are linearly independent for $p+q \leq D$, and the first term has only PTPs of size zero and the second term has only PTPs of size $\geq 1$.
    However, this requires $c_{\tilde{\pi}} = 0$ for all $\tilde{\pi}$, since the permutation operators are linearly independent as well.
    This establishes that that the operators, $\{ (\tilde{\pi}_L \otimes \tilde{\pi}_R) \noEPR \}$, are linearly independent.

    We can use this linear independence to complete our proof.
    Observing Eq.~(\ref{eq: noEPR pis}), the only way in which the two sides of the equation can be equal is if,
    \begin{equation}
        \sum_{\pi_L \pi_R} \Wg_{\pi_L \otimes \pi_R,\tilde{\pi}_L\otimes \tilde{\pi}_R} \cdot \tr( \noEPR ( \pi_L \otimes \pi_R )^{-1}  ) 
        =
        \delta_{\mathbbm{1},\tilde{\pi}_L\otimes \tilde{\pi}_R}
    \end{equation}
    for every $\tilde{\pi}_L, \tilde{\pi}_R$.
    The second term on the left side of the equation are the matrix elements of $\hat{\Wg}|_{\text{perm}}$.
    Applying the matrix inverse of $\hat{\Wg}|_{\text{perm}}$ to the left and right side, we obtain Eq.~(\ref{eq: trace noEPR}), as desired.
    The second statement of the proposition, Eq.~(\ref{eq: noEPR projector expansion}), follows immediately from the first statement and Eq.~(\ref{eq: noEPR sigma tau}).
    As before, we note that $\tr( \noEPR \sigma) = 0$ unless $\sigma = \pi_L \otimes \pi_R$, and we substitute the first statement, Eq.~(\ref{eq: trace noEPR}), in for $\tr( \noEPR \sigma) = \tr( \noEPR (\pi_L \otimes \pi_R)^{-1})$. \qed

\subsubsection{Proof of Proposition~\ref{prop: isometry}: Partial isometries} \label{app: partial}

To show that the isometry is well-defined, we simply need to show that the quantity within the parentheses in Eq.~(\ref{eq: def M alpha}) is positive.
To do so, let us expand the operator as a sum over its eigenvectors, $\ket{\lambda} \otimes \ket{E_\alpha}$, and eigenvalues, $\lambda$,
\begin{equation}
	\noEPR_{\bar \alpha} \sum_{\pi_L,\pi_R} c_{\pi_L \pi_R} \cdot \big( \alpha, \alpha, \pi_L, \pi_R \big) \noEPR_{\bar \alpha} = \sum_\lambda \lambda \dyad{\lambda \otimes E_\alpha}.
\end{equation}
To show that the eigenvalues are positive, we first use the fact that when we conjugate the operator above by $\tilde{P}_\alpha$, we obtain $\tilde{P}_\alpha$.
The conjugation maps each rank-1 projector to a new rank-1 operator, $\mathcal{N}_\lambda  \dyad*{\tilde{\lambda}}$, in $\tilde{P}_\alpha$,
\begin{equation} \nonumber
	\tilde{P}_\alpha = \tilde{P}_\alpha \noEPR_{\bar \alpha} \sum_{\pi_L,\pi_R} c_{\pi_L \pi_R} \cdot \big( \alpha, \alpha, \pi_L, \pi_R \big) \noEPR_{\bar \alpha} \tilde{P}_\alpha  = \sum_\lambda \lambda \cdot \tilde{P}_\alpha  \dyad{\lambda \otimes E_\alpha} \tilde{P}_\alpha  = \sum_\lambda \lambda \cdot \mathcal{N}_\lambda  \dyad*{\tilde{\lambda}} 
\end{equation}
with normalization $\mathcal{N}_\lambda = \bra{\lambda \otimes E_\alpha} \tilde{P}_\alpha \ket{\lambda \otimes E_\alpha}$, $0 \leq \mathcal{N}_\lambda \leq 1$.
Now, we note that there are at most $D^{p+q-2|\alpha|} - N_{\text{EPR}}^{(p-|\alpha|,q-|\alpha|)}$ non-zero eigenvalues $\lambda$, corresponding to the dimension of the no-EPR subspace on $\bar \alpha$.
However, from Proposition~\ref{prop: proj rank}, this is also the rank of $\tilde{P}_\alpha$.
Hence, there must be precisely this number of non-zero eigenvalues, and the set of conjugated vectors, $\{ \ket*{\tilde{\lambda}} \}$, must be linearly independent.
To determine the eigenvalues $\lambda$, let us apply $\tilde{P}_\alpha$ to $\ket*{\tilde{\lambda}'}$ for any $\tilde{\lambda}'$.
We have $\ket*{\tilde{\lambda}'} = \tilde{P}_\alpha \ket*{\tilde{\lambda}'} = \sum_\lambda \lambda \cdot \mathcal{N}_\lambda \langle \tilde{\lambda} | \lambda' \rangle \ket*{\tilde{\lambda}}$.
Since the $\{ \ket*{\tilde{\lambda}} \}$ are linearly independent, we must have $\langle \tilde{\lambda} | \tilde{\lambda}' \rangle = \delta_{\tilde{\lambda} , \tilde{\lambda}'}$ and $\lambda \cdot \mathcal{N}_\lambda = 1$.
Hence, $\lambda = 1/\mathcal{N}_\lambda$, so every eigenvalue is positive and greater than one.

To show that $\tilde{I}_\alpha$ is an isometry as described, we must show that $\tilde{I}_\alpha^\dagger \tilde{I}_\alpha = \tilde{P}_\alpha$ and $\tilde{I}_\alpha \tilde{I}_\alpha^\dagger = \noEPR_{\bar \alpha}$.
The first equality follows by construction [see Eqs.~(\ref{eq: tilde P decom tilde I}) and~(\ref{eq: tilde I def})].
To establish the second equality, we compute
\begin{equation}
	( \tilde{I}_\alpha \tilde{I}_\alpha^\dagger )^2 = I_\alpha M_\alpha \tilde{P}_\alpha (\tilde{P}_\alpha M_\alpha I_\alpha^\dagger I_\alpha M_\alpha \tilde{P}_\alpha ) \tilde{P}_\alpha M_\alpha I_\alpha^\dagger = I_\alpha M_\alpha \tilde{P}_\alpha (\tilde{P}_\alpha ) \tilde{P}_\alpha M_\alpha I_\alpha^\dagger = \tilde{I}_\alpha \tilde{I}_\alpha^\dagger.
\end{equation}
This shows that $\tilde{I}_\alpha \tilde{I}_\alpha^\dagger$ is a projector, and hence $\tilde{I}_\alpha$ is a partial isometry.

To establish the range of $\tilde{I}_\alpha$ (i.e.~on what subspace of $\bar \alpha$ does $\tilde{I}_\alpha \tilde{I}_\alpha^\dagger$ project onto), we first note that the range is contained in the no-EPR subspace on $\bar \alpha$.
This follows because the insertion of any EPR projector inside the isometry yields zero,
\begin{equation}
	\tilde{I}_\alpha^\dagger P_\gamma \tilde{I}_\alpha = \tilde{P}_\alpha M_\alpha I_\alpha^\dagger P_\gamma I_\alpha M_\alpha \tilde{P}_\alpha = \tilde{P}_\alpha M_\alpha P_{\alpha \cup \gamma} M_\alpha \tilde{P}_\alpha = 0,
\end{equation}
for all non-empty $\gamma \subseteq \bar \alpha$. 
The final expression is zero because $P_{\alpha \cup \gamma}$ has size $|\alpha| + |\gamma|$, which implies that $M_\alpha P_{\alpha \cup \gamma} M_\alpha$ is a sum of diagrams with size $|\alpha| + |\gamma|$, all of which vanish after conjugation by $\tilde{P}_\alpha$.
Hence, the range of $\tilde{I}_\alpha$ is contained in the orthogonal complement of $\text{span} (\{ P_\gamma : \gamma  \subseteq \bar \alpha \})$ on $\bar \alpha$, which is the definition of the no-EPR subspace.
To show that the range is equal to the no-EPR subspace, we simply note that the ranks of $\noEPR_{\bar \alpha}$ (restricted to subspace $\bar \alpha$) and $\tilde{I}_\alpha \tilde{I}_\alpha^\dagger$ are equal via Proposition~\ref{prop: proj rank}.
Namely, we apply Proposition~\ref{prop: proj rank} for $p, q, \ell$ for $\tilde{P}_\alpha$, and for $p-\ell,q-\ell,0$ for $\noEPR_{\bar \alpha}$. 
The rank of $\tilde{I}_\alpha \tilde{I}_\alpha^\dagger$ is equal to the rank of $\tilde{P}_\alpha = \tilde{I}_\alpha^\dagger \tilde{I}_\alpha$ since $\tilde{I}_\alpha$ is an isometry.
Thus, we have $\tilde{I}_\alpha \tilde{I}_\alpha^\dagger = \noEPR_{\bar \alpha}$, as claimed.

Finally, the isometry is unitary-equivariant,
\begin{equation}
	\tilde{I}_\alpha \mathcal{U}^p_q  = I_\alpha M_\alpha \tilde{P}_\alpha \mathcal{U}^p_q = \mathcal{U}^{p-\ell}_{q-\ell} I_\alpha M_\alpha \tilde{P}_\alpha = \mathcal{U}^{p-\ell}_{q-\ell} \tilde{I}_\alpha,
\end{equation}
since $\mathcal{U}^p_q$ commutes with $\tilde{P}_\alpha$ and $M_\alpha$, and $I_\alpha$ is unitary-equivariant.
The fact that $\mathcal{U}^p_q$ commutes with $M_\alpha$ follows because the square of $M_\alpha$ can be written as a sum of diagrams, and if $A$ and $B$ commute, then $A$ and $\sqrt{B}$ also commute, for any $A$, $B$. \qed

\subsection{Additional results on the mixed Haar twirl} \label{sec: additional}

In this section, we present several additional results on the objects appearing in the mixed Haar twirl. These results are not used in any of our main results in either the main text or the appendices. 
They were derived during an early unsuccessful attempt to prove the strong gluing lemma using only properties of the partially transposed permutations.
This attempt was aborted and replaced with the current proof via the path-recording framework after the introduction of this framework by Ref.~\cite{ma2024construct}.
We include these results here in case they may be useful in future work on the mixed Haar twirl or partially transposed permutations.

\subsubsection{Approximate orthogonality of the EPR projectors} 

The main result of this section is a proof that the ``nearly-orthogonal'' projectors defined in Appendix~\ref{sec: details mixed} are indeed nearly orthogonal, whenever the Hilbert space dimension $D$ is large.
This implies that the nearly-orthogonal projectors are approximately equal to the orthogonal projectors, $\noP_\alpha \approx \tilde{P}_\alpha$.
This can enable much easier analyses owing to the simpler definition of each $\noP_\alpha$.


To quantify the orthogonality of the nearly-orthogonal projectors, for each $\ell$, we define the  ${p \choose \ell}{q \choose \ell} \ell! \times {p \choose \ell}{q \choose \ell} \ell!$ matrix with elements,
\begin{equation}
	G^{(\ell)}_{\alpha \beta} =
	\begin{cases}
		\lVert P^{\text{{nE}}}_\alpha P^{\text{{nE}}}_\beta \rVert_\infty, & \alpha \neq \beta, \\
		0, & \alpha = \beta. \\
	 \end{cases}
\end{equation}
The matrix $\hat G^{(\ell)}$ would be zero if the projectors were perfectly orthogonal. We will show that in the limit of large $D$, its spectral norm is very small.
For each $\ell' \leq \ell$, we also consider the ${p \choose \ell}{q \choose \ell} \ell! \times {p \choose \ell}{q \choose \ell} \ell!$ matrix with elements,
\begin{equation} \label{eq: def F}
	F^{(\ell,\ell')}_{\alpha \beta} = 
	\begin{cases}
	\frac{1}{{\ell \choose \ell'}}\sum_{\gamma : |\gamma| = \ell'} \lVert P^{\text{{nE}}}_\alpha P_\gamma P^{\text{{nE}}}_\beta \rVert_\infty, & \alpha \neq \beta, \\
	\frac{1}{{\ell \choose \ell'}}\sum_{\gamma : |\gamma| = \ell', \gamma \not\subseteq \alpha} \lVert P^{\text{{nE}}}_\alpha P_\gamma P^{\text{{nE}}}_\alpha \rVert_\infty, & \alpha = \beta. \\
	\end{cases}
\end{equation}
We will show that the spectral norm of this matrix is also small.

We can now formally state our result on the approximate orthogonality of the projectors.
\begin{theorem}[Approximate orthogonality of EPR projectors] \label{thm: approx orth}
	The matrices $\hat G^{(\ell)}$ and $\hat F^{(\ell,\ell')}$ have small spectral norm,
	\begin{equation} \label{eq: G F bound}
		\big\lVert \hat G^{(\ell)}  \big\rVert_\infty \leq e^{\frac{\ell(p+q)}{D}} - 1, 
		\,\,\,\,\,\,\,\,\,\,\,\, \text{ and } \,\,\,\,\,\,\,\,\,\,\,\,
		\big\lVert \hat F^{(\ell,\ell')}  \big\rVert_\infty \leq e^{\frac{(\ell+\ell')(p+q)}{D}} - 1.
	\end{equation} 
	for any $(p+q)^2 \leq D$.
\end{theorem}
\noindent From Theorem~\ref{thm: approx orth}, we prove the following approximations for the orthogonal subspace projectors.
\begin{corollary}[Approximate expressions for the orthogonal projectors] \label{cor: approx projectors}
	The following approximations hold for any $(p+q)^2 \leq D$.
	First, 
	\begin{align}
        		\tilde{P}_\alpha & = P^{\text{\emph{nE}}}_\alpha + E_\alpha, 
		\,\,\,\,\,\,\,\, & \text{ with } 
		\left\lVert E_\alpha \right\rVert_\infty & \leq 2 \bigg( \frac{\ell(p+q)}{D} \bigg) + 10.78 \bigg( \frac{\ell(p+q)}{D} \bigg)^2, 
		\label{eq: approx 1} \\
        		\intertext{where $E_\alpha = \tilde{P}_{\ell} E_\alpha \tilde{P}_{\ell}$ and $\ell = |\alpha|$.
	Second,}
        		\tilde{P}_\ell & = P^{\text{\emph{nE}}}_\ell + E_\ell, 
		\,\,\,\,\,\,\,\, & \text{ with } 
		\left\lVert E_\ell \right\rVert_\infty & \leq  \bigg( \frac{\ell(p+q)}{D} \bigg) + 10.18 \bigg( \frac{\ell(p+q)}{D} \bigg)^2,
		 \label{eq: approx 2} \\
        		\intertext{where $E_\ell = \tilde{P}_{\ell} E_\ell \tilde{P}_{\ell}$.
	Third, for each $\ell'$,}
        		\sum_{\ell \geq \ell'} {\ell \choose \ell'}  \! \tilde{P}_{\ell} & = \! \sum_{\gamma : | \gamma | = \ell'} P_\gamma + \! \sum_{\ell \geq \ell'} {\ell \choose \ell'}   E^{(\ell')}_{\ell} , 
		\!\!\! & \text{ with } 
		\lVert E^{(\ell')}_{\ell} \rVert_\infty & \leq \bigg( \frac{(\ell+\ell')(p+q)}{D} \bigg) + 7.06 \bigg( \frac{(\ell+\ell')(p+q)}{D} \bigg)^2
		 \label{eq: approx 3} 
   	\end{align}
	where $E^{(\ell')}_{\ell} = \tilde{P}_{\ell} E^{(\ell')}_{\ell} \tilde{P}_{\ell}$ for each $\ell$.
\end{corollary}
\noindent We remark that the error bounds in Eqs.~(\ref{eq: approx 2}) and~(\ref{eq: approx 3}) are much tighter than would be obtained from applying Eq.~(\ref{eq: approx 1}) term by term.

The final result in Corollary~\ref{cor: approx projectors}, Eq.~(\ref{eq: approx 3}), is especially notable in the case $\ell =1$.
For this value of $\ell'$, the equation simplifies to
\begin{equation} \label{eq: approx 3 count}
		 N_E \equiv \sum_{\ell \geq \ell'} \ell  \, \tilde{P}_{\ell} = \sum_{\gamma : | \gamma | = 1} P_\gamma + \sum_{\ell \geq \ell'} \ell \,  E_{\ell} ,
\end{equation}
where we define  $N_E$ to equal the left hand side.
The operator $N_E$ simply counts the number of EPR pairs $\ell$ in a state.
From Corollary~\ref{cor: approx projectors}, we see $N_E$ can be approximated as a sum over all two-wise EPR projectors $P_\gamma$.

\subsubsection{Bound on the inverse Weingarten sub-matrix}

To establish our bounds in the previous section, we utilize the following bound on the inverse sub-matrix of the Weingarten matrix, whenever the Hilbert space dimension $D$ is large. This bound will be useful since the inverse sub-matrix appears in the expression for the no-EPR projector (Proposition~\ref{prop: no EPR}).

\begin{lemma}[Bound on the inverse Weingarten sub-matrix] \label{lemma: bound weingarten sum}
    For any $(p+q)^2 \leq D/2$.
    The sum of the absolute values of the matrix elements of the inverse of $\hat{\Wg}_{\text{\emph{perm}}}$ are bounded as
    \begin{equation}
        \frac{1}{p!q!} \sum_{\pi,\tilde{\pi}} \bigg| \delta_{\pi,\tilde{\pi}} -  \frac{1}{D^{p+q}}[\hat{\Wg}_{\text{\emph{perm}}}^{-1}]_{\pi,\tilde{\pi}} \bigg| \leq 2 \frac{(p+q)^2}{D},
    \end{equation}
    where we abbreviate $\pi = \pi_L \otimes \pi_R, \tilde{\pi} = \tilde{\pi}_L \otimes \tilde{\pi}_R$.
\end{lemma}

    \begin{proof}
    Since $\hat{\Wg}|_{\text{perm}}$ is a sub-matrix of $\hat{\Wg}$, its maximum eigenvalue is upper bounded by the maximum eigenvalue of $\hat{\Wg}$, and its minimum eigenvalue is lower bounded by the minimum eigenvalue of $\hat{\Wg}$.
    This follows since $v^T \hat{\Wg} v = v|_{\text{perm}}^T \hat{\Wg}|_{\text{perm}} v|_{\text{perm}}$ for any vector $v$ with support only on the permutation operators.
    From Ref.~\cite{harrow2023approximate}, the maximum eigenvalue of $\hat{\Wg}$ is less than $1+(p+q)^2/D$ and the minimum eigenvalue of $\hat{\Wg}$ is greater than $1-(p+q)^2/D$.
    Hence, the eigenvalues of $\hat{\Wg}|_{\text{perm}}$ are bounded by these values.

    Since the eigenvalues are bounded away from zero, the matrix inverse of $\hat{\Wg}|_{\text{perm}}$ can be Taylor expanded,
    \begin{equation}
        \hat{\Wg}|_{\text{perm}}^{-1} = \sum_{m=0}^\infty ( \mathbbm{1} - \hat{\Wg}|_{\text{perm}})^m.
    \end{equation}
    Subtracting this expression from the identity matrix, the $m=0$ term is canceled, and so we have
    \begin{equation}
        \mathbbm{1} - \hat{\Wg}|_{\text{perm}}^{-1} = - \sum_{m=1}^\infty ( \mathbbm{1} - \hat{\Wg}|_{\text{perm}})^m.
    \end{equation}
    We can bound our quantity of interest, the sum of absolute value matrix elements on the left, by a series of similar sums on the right,
    \begin{equation}
        \frac{1}{p!q!} \sum_{\pi,\tilde{\pi}} \bigg| \delta_{\pi,\tilde{\pi}} -  \frac{1}{D^{p+q}}[\hat{\Wg}_{\text{{perm}}}^{-1}]_{\pi,\tilde{\pi}} \bigg|
        \leq 
        \sum_{m=1}^\infty \left( \frac{1}{p!q!} \sum_{\pi,\tilde{\pi}} \bigg| \big[ ( \mathbbm{1} - \hat{\Wg}|_{\text{perm}})^m \big]_{\pi,\tilde{\pi}} \bigg| \right),
    \end{equation}
    which follows from the triangle inequality.

    To evaluate the terms within parentheses on the right hand side, we recall that the matrix elements of $\hat{\Wg}|_{\text{perm}}$ have an alternating pattern of signs, $(-1)^{|\pi_L| + |\pi_R|}$.
    If we define the diagonal matrix $\hat P$ element-wise via $P_{\pi,\pi} = (-1)^{|\pi_L| + |\pi_R|}$, this implies that $\hat P \hat{\Wg}|_{\text{perm}} \hat P$ is has all positive entries~\cite{aharonov2021quantum}.
    Since the diagonal elements of $\hat{\Wg}|_{\text{perm}}$ are greater than one, we further have that $\hat P (\hat{\Wg}|_{\text{perm}} - \mathbbm{1} ) \hat P$ has all positive entries~\cite{collins2017weingarten,schuster2024random}.
    Taking the $m$-th power, we find that $( \hat P (\hat{\Wg}|_{\text{perm}} - \mathbbm{1} ) \hat P )^m = \hat P (\hat{\Wg}|_{\text{perm}} - \mathbbm{1} )^m \hat P$ has all positive entries as well.
    The elements of these matrices have the same absolute values as the elements of the terms $(\mathbbm{1}-\hat{\Wg})^m$.
    Hence, the sum over the absolute value of the matrix elements of the latter is equal to the same sum for the former.

    The sum over matrix elements of $\hat P (\hat{\Wg}|_{\text{perm}} - \mathbbm{1} )^m \hat P$ is easy to evaluate, since the matrix has all positive elements.
    By the Perron-Frobenius theorem and the fact that the matrix is invariant under permutations, the maximum eigenvector of $\hat P (\hat{\Wg}|_{\text{perm}} - \mathbbm{1} )^m \hat P$ is the constant vector, $v_\pi = 1/\sqrt{p!q!}$~\cite{schuster2024random}.
    Hence, the maximum eigenvalue is equal to 
    \begin{equation}
        v^T \left( \hat P (\hat{\Wg}|_{\text{perm}} - \mathbbm{1} )^m \hat P \right) v 
        = 
        \frac{1}{p!q!} \sum_{\pi,\tilde{\pi}} \bigg| \big[ ( \mathbbm{1} - \hat{\Wg}|_{\text{perm}})^m \big]_{\pi,\tilde{\pi}} \bigg|,
    \end{equation}
    which is precisely the sum we would like to bound.
    From the above expression, we immediately have 
    \begin{equation}\nonumber
        \frac{1}{p!q!} \sum_{\pi,\tilde{\pi}} \bigg| \big[ ( \mathbbm{1} - \hat{\Wg}|_{\text{perm}})^m \big]_{\pi,\tilde{\pi}} \bigg|
        \leq 
        \left\lVert \hat P (\hat{\Wg}|_{\text{perm}} - \mathbbm{1} )^m \hat P \right\rVert_\infty
        =
        \left\lVert ( \hat{\Wg}|_{\text{perm}} - \mathbbm{1} )^m  \right\rVert_\infty
        \leq
        \left\lVert  \hat{\Wg}|_{\text{perm}} - \mathbbm{1}  \right\rVert_\infty^m.
    \end{equation}
    We have $\lVert  \hat{\Wg}|_{\text{perm}} - \mathbbm{1} \rVert_\infty \leq (p+q)^2/D$ from our discussion of the eigenvalues of $\hat{\Wg}|_{\text{perm}}$.
    Hence, the right side is less than $((p+q)^{2}/D)^m$.

    We can complete our proof by inserting this bound into the Taylor series and performing the sum over $m$.
    This yields,
    \begin{equation}
        \frac{1}{p!q!} \sum_{\pi,\tilde{\pi}} \bigg| \delta_{\pi,\tilde{\pi}} -  \frac{1}{D^{p+q}}[\hat{\Wg}_{\text{{perm}}}^{-1}]_{\pi,\tilde{\pi}} \bigg|
        \leq 
        \sum_{m=1}^\infty \left( \frac{(p+q)^2}{D} \right)^m
        = \frac{(p+q)^2/D}{1-(p+q)^2/D}
        \leq
        2(p+q)^2/D,
    \end{equation}
    where the final inequality holds if $(p+q)^2 \leq D/2$.
    This completes our proof. 
    \end{proof}

\subsubsection{Proof of Theorem~\ref{thm: approx orth}: Approximate orthogonality of EPR projectors}


Both $ \hat G^{(\ell)}$ and $ \hat F^{(\ell,\ell')}$ have entirely positive matrix elements.
The Perron-Frobenius theorem then states that the maximum eigenvalue of each matrix is achieved by an eigenvector with entirely positive elements.
Moreover, both $ \hat G^{(\ell)}$ and $ \hat F^{(\ell,\ell')}$ possess a ``permutation symmetry'',
\begin{equation}
	 G^{(\ell)}_{\alpha \beta} = \lVert \noP_{\alpha} \noP_{\beta} \rVert_\infty 
	= \lVert \pi \noP_{\alpha} \pi^{-1} \pi \noP_{\beta} \pi^{-1} \rVert_\infty 
	= \lVert \noP_{\pi(\alpha)} \noP_{\pi(\beta)} \rVert_\infty
	=  G^{(\ell)}_{\pi(\alpha) \pi(\beta)},
\end{equation}
and similar for $\hat F^{(\ell,\ell')}$.
Here, $\pi = \pi_L \otimes \pi_R$ is any tensor product of permutations on the left and right side.
Without loss of generality, we can assume that the maximum eigenvector, $v_\alpha$, is invariant under the permutation symmetry, $v_\alpha = v_{\pi(\alpha)}$ for any $\pi$.
(If not, we simply average the maximum eigenvector over all of its possible permutations, which produces a symmetric vector with the same eigenvalue.)
Since for every $\alpha, \beta$, there exists a permutation $\pi$ such that $\pi(\alpha) = \beta$, we have $v_\alpha = v_\beta$ for all $\alpha, \beta$.
Hence, the maximum eigenvalues of $ \hat G^{(\ell)}$ and $ \hat F^{(\ell,\ell')}$ are achieved by the constant vector, $v_\alpha = 1 / \left( \sum_\alpha 1 \right)^{1/2}$.

\vspace{3mm}
\noindent \textbf{Proof of the first statement, Eq.~(\ref{eq: G F bound}) left.} Let us begin with $\hat G$.
From the above, the spectral norm is
\begin{equation} \label{eq: G norm}
	\lVert  \hat G \rVert_\infty = \frac{ \sum_{\alpha} \sum_{\beta \neq \alpha} v_\beta G_{\beta \alpha} v_\alpha }{\sum_\alpha v_\alpha v_\alpha} = \frac{ \sum_\alpha \sum_{\beta \neq \alpha} \lVert \noP_\alpha \noP_\beta \rVert_\infty }{ \sum_\alpha 1 } = \sum_{\beta \neq \alpha} \lVert \noP_\alpha \noP_\beta \rVert_\infty.
\end{equation}
In the final expression, $\alpha$ is fixed to an arbitrary value and $\beta$ is summed over.

To proceed, we compute each term in the sum.
Let $L(\beta, \alpha)$ denote the number of loops when the PTPs associated with  $P_\beta$ and $P_\alpha$ are multiplied, and let $\gamma_I \supseteq \beta$ and $\gamma_O \supseteq \alpha$ denote the input and output pairs of the PTP obtained from the multiplication.
%
%
Then we have,
\begin{equation} \label{eq: spectral norm alpha beta}
	\left\lVert \noP_\beta \noP_\alpha \right\rVert_\infty = 
	\begin{cases}
		D^{L(\beta, \alpha)-\ell}, & \text{if } \gamma_I = \beta \text{ and } \gamma_O = \alpha \\
		0, & \text{else}.\\
	\end{cases}
\end{equation}
We can illustrate this formula with three examples:
\begin{align}\nonumber
\figbox{0.4}{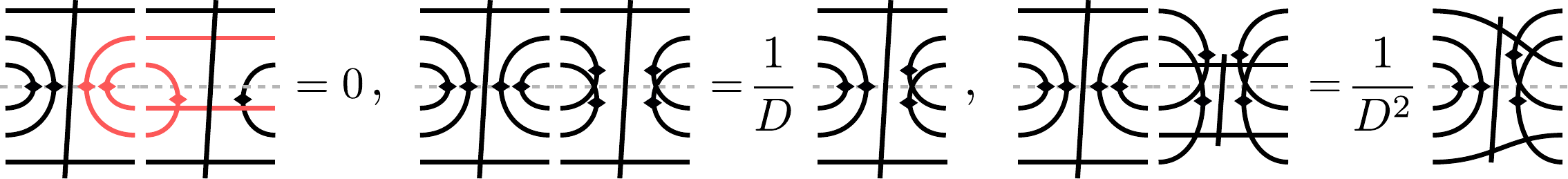} \centering \label{eq: PaPb}
\end{align}
In the first example, the projectors multiply to zero because the red leg forms an EPR projector on $\bar \beta$.
To derive Eq.~(\ref{eq: spectral norm alpha beta}), we write
\begin{equation}
\begin{split}
	 \noP_\beta \noP_\alpha 
	 = \noEPR_{\bar \beta} P_\beta P_\alpha \noEPR_{\bar \alpha} 
	 & = D^{-2\ell} \noEPR_{\bar \beta} \big( \beta, \beta, \mathbbm{1},\mathbbm{1} \big) \big( \alpha, \alpha, \mathbbm{1},\mathbbm{1} \big) \noEPR_{\bar \alpha} \\
     & = D^{L(\beta,\alpha)-2\ell}  \noEPR_{\bar \beta} \big( \gamma_I, \gamma_O, \pi_L, \pi_R \big) \noEPR_{\bar \alpha},
\end{split}
\end{equation}
where $\big( \gamma_I, \gamma_O, \pi_L, \pi_R \big)$ is the PTP obtained by multiplying $\big( \beta, \beta, \mathbbm{1},\mathbbm{1} \big)$ and $\big( \alpha, \alpha, \mathbbm{1},\mathbbm{1} \big)$.
The second clause in Eq.~(\ref{eq: spectral norm alpha beta}) follows because $\big( \gamma_I, \gamma_O, \pi_L, \pi_R \big)$ is annihilated by $\noEPR_{\bar \beta}$ if $\gamma_I$ contains a pair in $\bar \beta$, and similar for $\noEPR_{\bar \alpha}$ and $\gamma_O$.
The first clause follows because the spectral norm of $\noEPR_{\bar \beta} \big( \gamma_I, \gamma_O, \pi_L, \pi_R \big) \noEPR_{\bar \alpha}$ is one if $\gamma_I = \beta, \gamma_O = \alpha$.
We note that this condition implies $|\alpha| = |\beta|$, since $| \gamma_I | = | \gamma_O |$.

To bound the sum in Eq.~(\ref{eq: G norm}), we count the number of sets of pairs $\alpha$ that have spectral norm $\lVert \noP_\beta \noP_\alpha \rVert_\infty = D^{L-\ell}$ with a fixed set of pairs $\beta$, for each value of $L$.
Let us denote this number as $N(L,\ell)$.
Recall that $\alpha$ is a sequence of $\ell$ pairs of indices.
To determine $N(L,\ell)$, we enumerate the possible $\alpha$ pair-by-pair, as depicted below.
\vspace{1mm}
\begin{align}
\figbox{0.4}{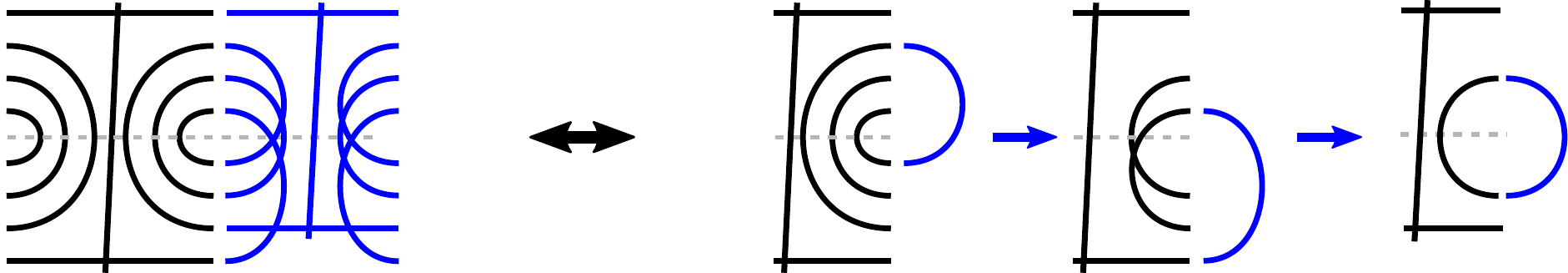} \centering \label{eq: PaPb sequence}
\end{align}
\vspace{2mm}
In more detail, let $\beta^{(0)} = \beta$ consider the ``right half'' of $\noP_\alpha$, the isometry $\noI_\alpha = \noEPR_{\bar \alpha} \otimes \bra{E_\alpha}$.
The first pair, $\alpha_1$, can be placed on any left and any right index, as long as at least one of the indices is contained in $\beta^{(0)}$.
(If neither index is contained in $\beta^{(0)}$, then the pair is annihilated by the no-EPR projector $\noEPR_{\bar \beta}$.)
There are at most $\ell (p+q)$ possible choices of the first pair, since: the index in $\beta^{(0)}$ can come from either the left or right side; on that side, the index corresponds to one of $\ell$ pairs; and on the other side, the index can be any of either $q$ or $p$ values.
Among these possible choices, there are $\ell$ possible choices of the first pair that produce a loop, since $\beta^{(0)}$ has $\ell$ pairs.

After the first pair is chosen, we can take the product of $\noP_\beta$ and $\ket{E_{\alpha_1}}$ to obtain a new projector $\noP_{\beta^{(1)}}$.
The new set $\noP_{\beta^{(1)}}$ acts on $p-1,q-1$ indices, and $\beta^{(1)}$ always contains one fewer pair than $\beta^{(0)}$. To verify the latter statement, note that there are three possible classes of pairs $\alpha_1$ that can be added. The first class connects one index of a pair in $\beta$ to one index in $\bar \beta$. This action annihilates the pair, and thus reduces the total number of pairs by one. The second class connects one index of a pair in $\beta$ to another index of another pair in $\beta$. This action joins the two pairs, and thus also reduces the number of pairs by one. The final class connects two indices of the same pair in $\beta$, producing a loop. Again, this reduces the total number of pairs by one.

We can iterate this process $\ell$ times to enumerate all possible sequences, $\alpha$, of $\ell$ pairs, such that $\noP_\beta \noP_\alpha$ is non-zero.
At the $j$-th step, for $j = 1,\ldots, \ell$, there are $(\ell-j+1)(p+q)$ total possible choices for the $j$-th pair, and, among these, $(\ell-j+1)$ possible choices that produce a loop.
At the end of the process, each set of pairs $\alpha$ is over-counted a total of $\ell!$ times, corresponding to the possible orderings of the pairs in $\alpha$.
If one wishes to consider only $\alpha$ that produce exactly $L$ loops, there are ${\ell \choose L} = {\ell \choose \ell-L}$ possible choices of $L$ steps at which to produce a loop.
Putting these three facts together, we have that there are at most,
\begin{equation}
	N(L,\ell) \leq \ell! \cdot (p+q)^{\ell-L} \cdot \frac{1}{\ell!} \cdot {\ell \choose \ell-L} \leq \frac{\ell^{\ell-L} (p+q)^{\ell-L}}{(\ell-L)!},
\end{equation}
possible $\alpha$ that produce $L$ loops.

This counting immediately enables us to bound our desired sum,
\begin{equation}\nonumber
	\lVert  \hat G \rVert_\infty \leq \sum_{\alpha \neq \beta} \lVert \noP_\beta \noP_\alpha \rVert_\infty = \sum_{L=1}^{\ell-1} N(L,\ell) D^{-(\ell-L)} \leq \sum_{L=0}^{\ell-1} \frac{1}{(\ell-L)!} \left( \frac{\ell(p+q)}{D} \right)^{\ell-L} \leq e^{\frac{\ell(p+q)}{D}} - 1.
\end{equation}
The upper bound of the second sum is $\ell-1$ and not $\ell$ because $\alpha = \beta$ is the sole choice of $\alpha$ that yields $\ell$ pairs, and this choice is excluded in the first sum.

\vspace{3mm}
\noindent \textbf{Proof of the second statement, Eq.~(\ref{eq: G F bound}) right.} We now turn to $ \hat F$.
The spectral norm is
\begin{equation} \label{eq: F norm}
	\lVert  \hat F \rVert_\infty = \frac{ \sum_{\alpha \beta \gamma} v_\beta F_{\beta \alpha} v_\alpha }{\sum_\alpha v_\alpha v_\alpha} = \sum_\gamma \sum_{\beta \neq \alpha} \lVert \noP_\alpha P_\gamma \noP_\beta \rVert_\infty + \sum_{\gamma \not\subseteq \alpha} \lVert \noP_\alpha P_\gamma \noP_\alpha \rVert_\infty .
\end{equation}
To bound the right hand side, we proceed similarly to our analysis for $ \hat G$, and enumerate all possible $\gamma, \beta$ that give a non-zero value of $\lVert \noP_\alpha P_\gamma \noP_\beta \rVert_\infty$.
We begin with $\gamma$ and proceed pair-by-pair as before.
Let $\alpha^{(0)} = \alpha$ and $\noI_{\alpha^{(0)}} = \noEPR_{\bar \alpha^{(0)}} \otimes \bra*{E^{}_{\alpha^{(0)}}}$.
We place the first pair, $\gamma_1$, on any two of the $p+q$ indices, so long as at least one of the indices is in $\alpha^{(0)}$.
(If it is not, then both indices of the pair are in $\bar{\alpha}^{(0)}$, and so are annihilated by the no-EPR projector in $\noI_{\alpha^{(0)}}$.)
This produces a new isometry, $\noI_{\alpha^{(1)}} = \noI_{\alpha^{(0)}} \ket{E_{\gamma_1}}$, acting on $p+q-2$ copies with $\ell-1$ EPR projectors, as depicted in Eq.~(\ref{eq: PaPb sequence}).
We iterate this procedure $j = 1,\ldots, \ell'$ times to generate all valid $\gamma$ with $\ell'$ pairs.
At the $j$-th step, there are $(\ell-j+1)(p+q-2j)$ possible locations at which to place the $j$-th pair, and, among these, $(\ell-j+1)$ possible locations that produce a loop.
In total, each $\gamma$ is over-counted a total of $\ell'!$ times, corresponding to the possible orderings of the $\ell'$ pairs in $\gamma$.

At the end of the process above, we obtain an isometry $\noI_{\alpha^{(\ell')}} = \noI_\alpha \ket{E_\gamma}$ acting on $p+q-2\ell'$ copies, with $\ell-\ell'$ EPR projectors.
The isometry is given by multiplying the ``right half'' of $\noP_\alpha$ with the ``left half'' of $P_\gamma$.
To proceed to enumerate the valid $\beta$, let us multiply this isometry by the ``right half'' of $P_\gamma$, to obtain $\noI_{\delta^{(0)}} \equiv \noI_\alpha \dyad{E_\gamma} = \noI_\alpha P_\gamma$.
The isometry $\noI_{\delta^{(0)}}$ acts on $p+q$ copies with $\ell$ EPR projectors.
We can now enumerate the possible $\beta$ pair-by-pair, exactly as we did in our analysis of $ \hat G$.
At the $j$-th step, for $j = 1,\ldots, \ell$, there are $(\ell-j+1)(p+q-2j)$ possible locations at which to place the $j$-th pair, and, among these, $(\ell-j+1)$ possible locations that produce a loop.
Each $\beta$ is over-counted a total of $\ell!$ times, corresponding to the possible orderings of the $\ell$ pairs in $\beta$.

We can now count the total number of $\gamma, \beta$ that produce $L$ loops in the multiplication $\noP_\alpha P_\gamma \noP_\beta$.
Since the enumeration of $\gamma, \beta$ contained a total of $\ell' + \ell$ steps, there are ${\ell' + \ell \choose L} = {\ell' + \ell \choose \ell' + \ell-L}$ possible choices of $L$ steps at which to produce a loop.
Thus, we have that there are at most
\begin{equation} \nonumber
	N(L,\ell,\ell') \leq \frac{\ell!}{(\ell-\ell')!} \cdot \ell! \cdot (p+q)^{\ell'+\ell-L} \cdot \frac{1}{\ell'!} \cdot \frac{1}{\ell!} \cdot {\ell' + \ell \choose \ell'+\ell-L}
	\leq {\ell \choose \ell'} \frac{ (\ell'+\ell)^{\ell'+\ell-L} (p+q)^{\ell'+\ell-L} }{(\ell'+\ell-L)!},
\end{equation}
possible choices of $\gamma, \beta$ that produce $L$ loops.

Turning to the sum in Eq.~(\ref{eq: F norm}), we have
\begin{equation}
\begin{split}
	\lVert  \hat F \rVert_\infty & \leq \frac{1}{{\ell \choose \ell'}} \bigg( \sum_\gamma \sum_{\beta \neq \alpha} \lVert \noP_\alpha P_\gamma \noP_\beta \rVert_\infty + \sum_{\gamma \not\subseteq \alpha} \lVert \noP_\alpha P_\gamma \noP_\alpha \rVert_\infty \bigg) \\
	 & = \frac{1}{{\ell \choose \ell'}} \sum_{L=1}^{\ell'+\ell-1} N(L,\ell,\ell') D^{-(\ell+\ell'-L)} \\
	 & \leq \sum_{L=0}^{\ell'+\ell-1} \frac{1}{(\ell'+\ell-L)!} \left( \frac{(\ell'+\ell)(p+q)}{D} \right)^{\ell'+\ell-L} \\
	 & \leq  \left( e^{\frac{(\ell'+\ell)(p+q)}{D}} - 1 \right).
\end{split}
\end{equation}
The upper bound of the second sum is $\ell'+\ell-1$ and not $\ell'+\ell$ because $\beta = \alpha$, $\gamma \subseteq \alpha$ is the sole choice of $\gamma, \beta$ that yields $\ell'+\ell$ pairs, and this choice is excluded in the first sum. \qed

\subsubsection{Proof of Corollary~\ref{cor: approx projectors}: Approximate expressions for subspace projectors}

As a starting point,
we consider a normalized vector $\ket{\psi} \in \tilde{P}_\ell$.
By construction, we can write $\ket{\psi}$ as a sum of vectors from each partly-orthogonal subspace $\alpha$ of size $\ell$,
\begin{equation} \label{eq: psi ell decomp}
	\ket{\psi} = \sum_\alpha c_\alpha \ket{\psi_{\alpha}} = \sum_\alpha c_\alpha \ket{E_\alpha} \otimes \ket{\phi_{\bar \alpha}},
\end{equation}
where each normalized state $\ket{\psi_{\alpha}} \equiv \ket{E_\alpha} \otimes \ket{\phi_{\bar \alpha}}$ is orthogonal to all EPR projectors on $\bar \alpha$.

Let us first understand how the normalization of $\ket{\psi}$ is related to the coefficients $c_\alpha$.
We have
\begin{equation} \label{eq: norm psi}
	1 = \langle \psi | \psi \rangle = \sum_\alpha | c_\alpha |^2 + \sum_{\alpha \neq \beta} c^*_\beta c_\alpha \langle \psi_\beta | \psi_\alpha \rangle.
\end{equation}
We can use Theorem~\ref{thm: approx orth} to show that the second sum is small,
\begin{equation} \label{eq: norm psi bound}
	\bigg| \sum_{\alpha \neq \beta} c^*_\beta c_\alpha \langle \psi_\beta | \psi_\alpha \rangle \bigg| 
	\leq 
	\sum_{\alpha \neq \beta} | c_\beta | | c_\alpha | \left\lVert \noP_\beta \noP_\alpha \right\rVert_\infty 
	\leq 
	\bigg( \sum_\alpha |c_\alpha|^2 \bigg) \cdot \lVert \hat G \rVert_\infty,
\end{equation}
where the first step follows from the triangle inequality as well as the expression, $| \langle \psi_\beta | \psi_\alpha \rangle | = | \langle \psi_\beta | \noP_\beta \noP_\alpha | \psi_\alpha \rangle | \leq \lVert \noP_\beta \noP_\alpha \rVert_\infty$.
Combining Eq.~(\ref{eq: norm psi}) and Eq.~(\ref{eq: norm psi bound}) gives
\begin{equation} \label{eq: coeffs bound}
	e^{-\ell(p+q)/D}\leq \frac{1}{1+\lVert  \hat G \rVert_\infty} \leq \sum_\alpha |c_\alpha|^2 \leq \frac{1}{1-\lVert  \hat G \rVert_\infty} \leq  \frac{1}{2-e^{\ell(p+q)/D} },
\end{equation}
where the outer inequalities follow from Theorem~\ref{thm: approx orth}.
We see that the coefficients are approximately normalized to one, as would be the case if the projectors were perfectly orthogonal.

\vspace{3mm}
\noindent \textbf{Proof of the first statement, Eq.~(\ref{eq: approx 1}).} By definition, both $\tilde{P}_\alpha$ and $\noP_\alpha$ are orthogonal to all subspaces $\tilde{P}_\beta$ with $| \beta | \neq \ell$, or with $|\beta| = \ell$ but $\beta > \alpha$ with respect to the ordering of the orthogonal projectors.
Hence, $E_\alpha = \tilde{P}_\ell E_\alpha \tilde{P}_\ell$, and we can restrict our attention to the action of $E_\alpha$ on states in $\text{span} \{ \tilde{P}_\beta : |\beta| = \ell, \beta \leq \alpha \}$.
Building upon the decomposition in Eq.~(\ref{eq: psi ell decomp}), we write
\begin{equation}
	\ket{\psi} = c_\alpha \ket{\psi_\alpha} + \sum_{\alpha' \neq \alpha} c_{\alpha'} \ket{\psi_{\alpha'}}  = c_\alpha 
	\left( b_\alpha \ket*{\tilde{\psi}_\alpha} + \sum_{\alpha' < \alpha} b_{\alpha'} \ket*{\tilde{\psi}_{\alpha'}} \right) + \sum_{\alpha' < \alpha} c_{\alpha'} \ket{\psi_{\alpha'}},
\end{equation}
where on the right hand side, we decompose $\ket{\psi_\alpha}$, which lies in $\noP_\alpha$, as a sum of vectors, $\ket*{\tilde{\psi}_\alpha}, \ket*{\tilde{\psi}_{\alpha'}}$, which lie in $\tilde{P}_\alpha, \tilde{P}_{\alpha'}$ for $\alpha' < \alpha$.
We have
\begin{equation}
	\noP_\alpha \ket{\psi} = c_\alpha \ket{\psi_\alpha} + \sum_{\alpha' < \alpha} c_{\alpha'} \noP_\alpha \ket{\psi_{\alpha'}},
\end{equation}
and 
\begin{equation}
	\tilde{P}_\alpha \ket{\psi} = c_\alpha b_\alpha \ket*{\tilde{\psi}_\alpha}.
\end{equation}
Taking the difference, we have
\begin{equation}
	E_\alpha \ket{\psi} = c_\alpha \sum_{\alpha' < \alpha} b_{\alpha'} \ket*{\tilde{\psi}_{\alpha'}} + \sum_{\alpha' < \alpha} c_{\alpha'} \noP_\alpha \ket{\psi_{\alpha'}}.
\end{equation}
We will now show that $E_\alpha \ket{\psi}$ has small norm.

The second term has norm at most,
\begin{equation}
	\bigg\lVert \sum_{\alpha' < \alpha} c_{\alpha'} \noP_\alpha \ket{\psi_{\alpha'}} \bigg\rVert \leq \sum_{\alpha' < \alpha} | c_{\alpha'} | \cdot \lVert \noP_\alpha \noP_{\alpha'} \rVert_\infty \leq \bigg( \sum_{\alpha' < \alpha} | c_{\alpha'} |^2 \bigg)^{1/2}  \lVert \hat G^{(\ell)} \rVert_\infty \leq \frac{\lVert \hat G^{(\ell)} \rVert_\infty}{1-\lVert \hat G^{(\ell)} \rVert_\infty},
\end{equation}
where $\lVert \ket{\psi} \rVert = \sqrt{ \langle \psi | \psi \rangle }$ denotes the vector norm, and the final inequality follows from Eq.~(\ref{eq: coeffs bound}).
To bound the first term, we note that 
\begin{equation}
	\bigg\lVert \sum_{\alpha' < \alpha} b_{\alpha'} \ket*{\tilde{\psi}_{\alpha'}} \bigg\rVert = \max_{\ket{\phi_{< \alpha}}} | \langle \phi_{< \alpha} | \psi_\alpha \rangle |,
\end{equation}
where the maximization is over all states $\ket{\phi_{< \alpha}}$ in the subspace $\text{span} \{ P_{\alpha'} : \alpha' < \alpha \}$.
The equation follows since the vector on the left hand side is the projection of $\ket{\psi_\alpha}$ onto the subspace.
Performing an analogous decomposition for $\ket{\phi_\alpha}$ as in Eq.~(\ref{eq: psi ell decomp}), with coefficients $d_{\alpha'}$, we have
\begin{equation}\nonumber
	| \langle \phi_{< \alpha} | \psi_\alpha \rangle | = \bigg| \sum_{\alpha' < \alpha} d_{\alpha'} \langle \phi_{\alpha'} | \psi_\alpha \rangle \bigg| \leq \sum_{\alpha' < \alpha} | d_{\alpha'} | \cdot \lVert \noP_{\alpha'} \noP_\alpha \rVert_\infty \leq \bigg( \sum_{\alpha' < \alpha} | d_{\alpha'} |^2 \bigg)^{1/2}  \lVert \hat G^{(\ell)} \rVert_\infty \leq \frac{\lVert \hat G^{(\ell)} \rVert_\infty}{1-\lVert \hat G^{(\ell)} \rVert_\infty},
\end{equation}
where the final inequality follows from Eq.~(\ref{eq: coeffs bound}).
Applying Theorem~\ref{thm: approx orth} yields Eq.~(\ref{eq: approx 1}),
\begin{equation}
	\lVert E_\alpha \rVert_\infty \leq \frac{2 \lVert \hat G^{(\ell)} \rVert_\infty}{1-\lVert \hat G^{(\ell)} \rVert_\infty} \leq \frac{e^{\ell (p+q)/D}-1}{1-\frac{1}{2}e^{\ell (p+q)/D}} \leq 2 \bigg( \frac{\ell(p+q)}{D} \bigg) + 10.78 \bigg( \frac{\ell(p+q)}{D} \bigg)^2 ,
\end{equation}
where in the final inequality we use $(e^x - 1)/(1-e^x/2) \leq 2x + 10.78 x^2$  for $0 \leq x \leq 1/2$, from Taylor's remainder theorem.

\vspace{3mm}
\noindent \textbf{Proof of the second statement, Eq.~(\ref{eq: approx 2}).} Similar to before, both $\tilde{P}_\ell$ and all $\noP_\alpha$ with $|\alpha|=\ell$ are orthogonal to all subspaces $\tilde{P}_\beta$ with $| \beta | \neq \ell$.
Hence, $E_\ell = \tilde{P}_\ell E_\ell \tilde{P}_\ell$, and we can restrict our attention to the action of $E_\ell$ on states in $\tilde{P}_\ell$.
To bound the magnitude of $E_\ell$, we recall the definition of the spectral norm, 
\begin{equation}
	\lVert E_\ell \rVert_\infty \equiv \max_{\ket{\psi},\ket{\phi}} \bra{\psi} E_\ell \ket{\phi},
\end{equation}
where we can assume  $\ket{\psi}, \ket{\phi} \in \tilde{P}_\ell$.
Now, we expand $\ket{\psi}$ and $\ket{\phi}$ as in Eq.~(\ref{eq: psi ell decomp}), $\ket{\psi} = \sum_{ \alpha : | \alpha | = \ell} c_\alpha \ket{\psi_{\alpha}}$ and $\ket{\phi} = \sum_{ \alpha : | \alpha | = \ell} d_\alpha \ket{\phi_{\alpha}}$.
We have
\begin{equation}
	\bra{\psi} \tilde{P}_\ell \ket{\phi} = \langle \psi | \phi \rangle = \sum_{\gamma,\alpha} c^*_\alpha d_\beta \langle \psi_\alpha | \phi_\beta \rangle,
\end{equation}
since $\tilde{P}_\ell \ket{\phi} = \ket{\phi}$ by assumption.
Meanwhile, we have
\begin{equation}
	\bra{\psi} \bigg( \sum_{\gamma : |\gamma| = \ell} \noP_\gamma \bigg) \ket{\phi} = \sum_{\alpha \beta \gamma} c^*_\alpha d_\beta \bra{\psi_\alpha} \noP_\gamma  \ket{\phi_\beta} = \langle \psi | \phi \rangle + \sum_{\alpha, \beta, \gamma \neq \beta} c^*_\alpha d_\beta \bra{\psi_\alpha} \noP_\gamma  \ket{\phi_\beta},
\end{equation}
where in the third expression we use that $\noP_\gamma \ket{\psi_\beta} = \ket{\psi_\beta}$ for $\gamma = \beta$.
Taking the difference of the two expressions, and applying the triangle inequality, we have
\begin{equation}
	| \bra{\psi} E_\ell \ket{\phi} | 
	=
	\sum_{\alpha, \beta, \gamma \neq \beta} c^*_\alpha d_\beta \bra{\psi_\alpha} \noP_\gamma  \ket{\phi_\beta}
	\leq 
	\sum_{\alpha, \beta, \gamma \neq \beta} | c_\alpha |  \cdot \lVert \noP_\alpha \noP_\gamma \rVert_\infty \cdot \lVert \noP_\gamma \noP_\beta \rVert_\infty \cdot | d_\beta |.
\end{equation}
We can view the second norm, $\lVert \noP_\gamma \noP_\beta \rVert_\infty$, as the elements of the matrix $\hat G^{(\ell)}$, since the diagonal elements, $\gamma = \beta$, are omitted.
We can view the first norm, $\lVert \noP_\alpha \noP_\gamma \rVert_\infty$, as the elements of the matrix, $\hat{\mathbbm{1}} + \hat G^{(\ell)}$, since it contains its diagonal elements, $\lVert \noP_\alpha \noP_\alpha \rVert_\infty = 1$.
Hence, we have
\begin{equation}
	| \bra{\psi} E_\ell \ket{\phi} | 
	\leq 
	\bigg( \sum_\alpha | c_\alpha |^2 \bigg)^{1/2}
	\bigg( \sum_\beta | d_\beta |^2 \bigg)^{1/2}
	\big( 1 + \lVert \hat G \rVert_\infty \big) 
	\lVert \hat G \rVert_\infty,
\end{equation}
Applying Eq.~(\ref{eq: coeffs bound}) and Theorem~\ref{thm: approx orth} yields Eq.~(\ref{eq: approx 2}),
\begin{equation} \nonumber
	\lVert E_\ell \rVert_\infty \leq \frac{(1+ \lVert \hat G^{(\ell)} \rVert_\infty) \lVert \hat G^{(\ell)} \rVert_\infty}{1-\lVert \hat G^{(\ell)} \rVert_\infty} \leq \frac{e^{\ell(p+q)/D}(e^{\ell(p+q)/D}-1)}{2-e^{\ell(p+q)/D}} \leq \bigg( \frac{\ell(p+q)}{D} \bigg) + 10.18 \bigg( \frac{\ell(p+q)}{D} \bigg)^2 ,
\end{equation}
where in the final inequality we use $e^x(e^x - 1)/(2-e^x) \leq x + 10.18 x^2$ for $0 \leq x \leq 1/2$, from Taylor's remainder theorem.

\vspace{3mm}
\noindent \textbf{Proof of the third statement, Eq.~(\ref{eq: approx 3}).} Our proof follows in a similar manner to the second statement.
Note that both the left hand side of Eq.~(\ref{eq: approx 3}), and the first term on the right hand side, commute $\tilde{P}_{\ell}$ for all $\ell$ (the latter follows from Proposition~\ref{ref: proj ell commutes}).
Thus, the difference of the two terms can be written as a sum of error terms, ${\ell \choose \ell'} E^{(\ell')}_\ell$, within each subspace, $\tilde{P}_\ell$.
To quantify each error, let us suppose $\ket{\psi},\ket{\phi} \in \tilde{P}_\ell$, and write,
\begin{equation}
	\bra{\psi} \bigg( \sum_{\ell'' \geq \ell'} {\ell'' \choose \ell'} \tilde{P}_{\ell''} \bigg) \ket{\phi} = {\ell \choose \ell'} \langle \psi | \phi \rangle.
\end{equation}
Meanwhile, expanding $\ket{\psi}, \ket{\phi}$ as in Eq.~(\ref{eq: psi ell decomp}), we have
\begin{equation}
	\bra{\psi} \bigg( \sum_{\gamma : |\gamma| = \ell'} P_\gamma \bigg) \ket{\phi} = \sum_{\alpha \beta \gamma} c^*_\alpha d_\beta \bra{\psi_\alpha} P_\gamma  \ket{\phi_\beta} = {\ell \choose \ell'} \langle \psi | \phi \rangle + \sum_{\alpha, \beta, \gamma \not\subseteq \beta} c^*_\alpha d_\beta \bra{\psi_\alpha} P_\gamma  \ket{\phi_\beta}.
\end{equation}
Taking the difference and dividing by ${\ell \choose \ell'}$, we have
\begin{equation}
	\lVert E^{(\ell')}_\ell \rVert_\infty = \bigg| \frac{1}{{\ell \choose \ell'}} \sum_{\alpha, \beta, \gamma \not\subseteq \beta} c^*_\alpha d_\beta \bra{\psi_\alpha} P_\gamma  \ket{\phi_\beta} \bigg| 
	\leq \frac{1}{{\ell \choose \ell'}} \sum_{\alpha, \beta, \gamma \not\subseteq \beta} | c_\alpha | \cdot | d_\beta | \cdot \lVert \noP_\alpha P_\gamma \noP_\beta \rVert_\infty.
\end{equation}
We are free to add terms to the sum, in order for the indices that are summed over to match those in the matrix $\hat F^{(\ell,\ell')}$ [Eq.~(\ref{eq: def F})].
Adding in terms where $\gamma \subseteq \beta$ for each $\beta \neq \alpha$, we find
\begin{equation}\nonumber
	\lVert E^{(\ell')}_\ell \rVert_\infty  \leq \frac{1}{{\ell \choose \ell'}} \bigg( \!\! \sum_{\alpha, \beta \neq \alpha, \gamma} \!+ \! \sum_{\alpha, \gamma \not\subseteq \alpha} \bigg) \, | c_\alpha |  | d_\beta |  \lVert \noP_\alpha P_\gamma \noP_\beta \rVert_\infty 
	\leq  \bigg( \sum_\alpha |c_\alpha|^2 \bigg)^{1/2} \!\! \bigg( \sum_\alpha |d_\alpha|^2 \bigg)^{1/2} \lVert \hat F^{(\ell,\ell')} \rVert_\infty.
\end{equation}
Applying Eq.~(\ref{eq: coeffs bound}) and Theorem~\ref{thm: approx orth} yields Eq.~(\ref{eq: approx 3}),
\begin{equation}\nonumber
	\lVert E^{(\ell')}_\ell \rVert_\infty  \leq \frac{ \lVert \hat F^{(\ell,\ell')} \rVert_\infty }{1-\lVert \hat G^{(\ell)} \rVert_\infty} \leq \frac{e^{(\ell+\ell')(p+q)/D}-1}{2-e^{\ell(p+q)/D}} \leq \bigg( \frac{(\ell+\ell')(p+q)}{D} \bigg) + 7.06 \bigg( \frac{(\ell+\ell')(p+q)}{D} \bigg)^2,
\end{equation}
where in the final inequality we use $(e^y-1)/(2-e^x) \leq (y+ 0.718 y^2)(1+ 3.693 x) \leq y + 7.06 y^2$ for $0 \leq y \leq 1$, $ 0 \leq x \leq \min(y,1/2)$, from Taylor's remainder theorem. \qed

\section{Fast scrambling} \label{app: scrambling}

In this Appendix, we provide full details on the applications of strong random unitaries to quantum information scrambling.
As mentioned in the main text, each of our results follows fairly immediately from the definition of strong unitary $k$-designs and strong PRUs.

\subsection{Out-of-time-order correlation functions}

Let $U$ be a random unitary and $\ket{\psi}$ a fixed quantum state.
A time-ordered $2k$-point correlation function takes the form,
\begin{equation}
    C_{\text{TO}}(P_1,\ldots,P_{2k}) = \bra{\psi} P_{2k} U^\dagger P_{2k-1} U^\dagger P_{2k-2} U^\dagger  \ldots P_{k+1} U^\dagger P_k U \ldots U P_2 U P_1 U \ket{\psi},
\end{equation}
where we assume that $P_i$ are Pauli operators for simplicity.
Any time-ordered correlation function can be measured in an experiment that applies the unitary $U$ $k$ times in sequence.
An out-of-time-order $2k$-point correlation function takes the form,
\begin{equation}
    C_{\text{OTO}}(P_1,\ldots,P_{2k}) = \bra{\psi} P_{2k}  U^\dagger P_{2k-1} U  P_{2k-2} \ U^\dagger P_{2k-3} U \ldots P_4  U^\dagger P_3 U  P_2  U^\dagger P_1 U \ket{\psi},
\end{equation}
where we again assume that $P_i$ are Pauli operators for simplicity.
Any out-of-time-order correlation function can be measured in an experiment that applies $U$ and $U^\dagger$ one after the other $k/2$ times in sequence.
Here, we assume $k$ is even.
The particular out-of-time-order correlation function shown in Fig.~\ref{fig:scrambling} of the main text sets all $P_i$ for even $i$ equal to one another and all $P_i$ for odd $i$ equal as well.

As discussed in the main text, the formation of strong unitary $k$-designs immediately implies the decay of all local $k$-point time-ordered and out-of-time-order correlation functions to zero.
\begin{proposition}
    For any $k = \mathcal{O}(1)$ and $\varepsilon = \Omega(1/2^n)$. Let $U$ be drawn from a strong $\frac{\varepsilon^2 \delta}{\poly n}$-approximate unitary $2k$-design and $\ket{\psi}$ be any quantum  state. Then with high probability $1-\delta$, every local $2k$-point time-ordered and out-of-time-order correlation function decays to within $\varepsilon$ of zero under $U$.
\end{proposition}
\noindent For example, if we one sets $\varepsilon$ and $\delta$ to be super-polynomially small in $n$, i.e.~$1/\varepsilon, 1/\delta = \omega(\poly n)$, then the proposition is satisfied whenever the error of the strong unitary, $\varepsilon' = \varepsilon^2 \delta/\poly n$, is also super-polynomially small in $n$.
From Theorem~\ref{thm:strong-design-depth}, this is achieved in $\mathcal{O}(\log n)$ circuit depth for structured quantum circuits and $\mathcal{O}(\log^3 n)$ circuit for  all-to-all connected random circuits.
\begin{proof}
    We consider the sum of squares of all local time-ordered and out-of-time-order correlation functions,
    \begin{equation}
        C(U) =  \sum_{P_1,\ldots,P_{2k}} C_{\text{TO}}(P_1,\ldots,P_{2k})^2 + C_{\text{OTO}}(P_1,\ldots,P_{2k})^2,
    \end{equation}
    where each $P_i$ in the sum is non-identity.
    If each $P_i$ is $r$-local with $r = \mathcal{O}(1)$, then there are at most $2(3n)^{2rk} = n^{\mathcal{O}(k)}$ correlation functions in the sum.
    This follows because each Pauli operator can take $3^r {n \choose r} \leq (3n)^r$ different values and there are $2k$ Pauli operators to choose.
    Here, we add in the label $U$ on the left side for specificity; the time-ordered and out-of-time-order correlation functions all implicitly depend on $U$ as well.

    The expected value, $\E_{U \sim \mathcal{E}} C(U)$, can be estimated to within $n^{\mathcal{O}(k)} \varepsilon'$ of its Haar-random value, since each individual term can be estimated to within $\mathcal{O}(\varepsilon')$.
    A straightforward calculation shows that each Haar-random correlation function is exponentially small (see e.g.~\cite{cotler2017chaos}), and hence $\E_{U \sim H} C(U) = \mathcal{O}(n^{\mathcal{O}(k)}/2^n)$.
    Thus, $\E_{U \sim \mathcal{E}} C(U) = \mathcal{O}(n^{\mathcal{O}(k)}(\varepsilon'+1/2^n)) = \mathcal{O}(n^{\mathcal{O}(k)} \varepsilon')$. 
    From Markov's inequality, we have 
    \begin{equation}
        \text{Pr}( C(U) \geq \varepsilon^2) \leq  \frac{n^{\mathcal{O}(k)} \varepsilon'}{\varepsilon^2}.
    \end{equation}
    The probability $\delta$ that any individual correlation function has absolute value greater than $\varepsilon$ is upper bounded by the probability above.
    Setting $\varepsilon' = \varepsilon^2 \delta/n^{\mathcal{O}(k)}$ completes the proof.
\end{proof}

\subsection{Operator size distributions} 

The size distribution of an operator $O$ evolved under a unitary $U$ is given by 
\begin{equation}
    P_U(w) = \frac{1}{2^n} \tr( O(t) \mathcal{P}_w [ O(t) ]),
\end{equation}
where $\mathcal{P}_w$ is a superoperator that projects onto Pauli strings of weight $w$.
We assume without loss of generality that $\frac{1}{2^n}\tr( O^\dagger O ) = 1$, which implies that the size distribution is normalized, $\sum_w P_U(w) = 1$.
If we consider the quantum state $(O \otimes \mathbbm{1})\ket{\Psi_{\text{EPR}}}$ on two copies of $n$ quits, then the size distribution corresponds to the expectation value,
\begin{equation}
    P_U(w) = \bra{\Psi_{\text{EPR}}}(O^\dagger \otimes \mathbbm{1})(U^\dagger \otimes U^T) \mathcal{P}_w (U \otimes U^*)(O \otimes \mathbbm{1})\ket{\Psi_{\text{EPR}}},
\end{equation}
where $\mathcal{P}_w$ is now an operator on the two-copy system that projects onto the span of states $(Q \otimes \mathbbm{1})\ket{\Psi_{\text{EPR}}}$ where $Q$ is any Pauli operator with weight $w$.

We can use strong approximate unitary 4-designs to bound the closeness of operator size distributions to their Haar-random values.
\begin{proposition}
    The expected total variation distance between the operator size distribution of a strong $\varepsilon$-approximate unitary 4-design and the Haar-random operator size distribution is less than $3n^2\varepsilon$.
\end{proposition}
\noindent From Theorem~\ref{thm:strong-design-depth}, strong $\varepsilon$-approximate unitary 4-designs with $\varepsilon = 1/\poly n$ can form in circuit depth $\mathcal{O}(\log n)$ in structured unitary ensembles and circuit depth $\mathcal{O}(\log^2 n)$ in random circuits.
This confirms empirical observations that operator size distributions can equilibrate to their Haar-random values in logarithmic depth~\cite{}.

\begin{proof}
We use the Cauchy-Schwarz inequality and the fact that the operator size distribution is the expectation value of a bounded operator $\mathcal{P}_w$ on two copies. The latter allows us to bound $| \E_U P_U(w) - P_H(w) | \leq \varepsilon$ and $| \E_U P_U(w)^2 - P_H(w)^2 | \leq \varepsilon$. This yields,
\begin{equation}
\begin{split}
    \E_{U \sim \mathcal{E}} \text{TVD}(P_U, P_H) & \equiv \E_{U \sim \mathcal{E}} \sum_{w=1}^n \left| P_U(w) - P_H(w) \right| \\
    & \leq \E_{U \sim \mathcal{E}} n \sum_{w=1}^n \left| P_U(w) - P_H(w) \right|^2 \\
    & = \E_{U \sim \mathcal{E}} n \sum_{w=1}^n \left( P_U(w)^2 - 2 P_U(w) P_H(w) + P_H(w)^2 \right) \\
    & = n \sum_{w=1}^n \left( \varepsilon + 2\varepsilon \right) \\
    & = 3 n^2 \varepsilon. \\
\end{split}
\end{equation}
This completes the proof.
\end{proof}

\subsection{Entanglement and operator entanglement entropy.}

Consider a state $\ket{\psi(t)} \equiv U \ket{\psi}$ and an operator $O(t) \equiv U O U^\dagger$.
Let $A$ denote a subsystem of $n$ qubits and $B$ its complement.
To define the entanglement entropy and operator entanglement entropy of $\ket{\psi(t)}$ and $O(t)$, respectively, we can first write the Schmidt decomposition of each object between $A$ and $B$,
\begin{equation}
    \ket{\psi(t)} = \sum_i \sqrt{\lambda^\psi_i} \cdot \ket{\psi_A^i} \otimes \ket{\psi_B^i},
\end{equation}
where $\langle \psi_A^i | \psi_A^j \rangle = \langle \psi_B^i | \psi_B^j \rangle = \delta_{ij}$, and 
\begin{equation}
    O(t) = \sum_i \sqrt{\lambda^O_i} \cdot O_A^i \otimes O_B^i,
\end{equation}
where $\frac{1}{2^n}\text{tr}( (O_A^i)^\dagger O_A^j) = \frac{1}{2^n}\text{tr}( (O_B^i)^\dagger O_B^j) = \delta_{ij}$.
We have $\sum_i \lambda_i^\psi = \langle \psi | \psi \rangle = 1$ and $\sum_i \lambda_i^O = \frac{1}{2^n} \text{tr}(O^\dagger O) = 1$ (assuming we normalize $O$ to one).
The von Neumann entanglement entropy of $\ket{\psi(t)}$ is equal to $\sum_i \lambda^\psi_i \ln \lambda^\psi_i$ and the von Neumann operator entanglement entropy of $O(t)$ is equal to $\sum_i \lambda^O_i \ln \lambda^O_i$.

The von Neumann entanglement entropy is difficult to analyze using unitary $k$-designs due to the logarithmic factor.
To this end, we consider a Renyi version of the entanglement and operator entanglement entropies.
The Renyi-2 entanglement entropy of $\ket{\psi(t)}$ between $A$ and $B$ is given by
\begin{equation}
    S^{(2)}_A(\ket{\psi(t)}) = -\ln\Big( \sum_i (\lambda^\psi_i)^2\Big) = -\ln\Big(\! \tr_A(\tr_B( \dyad{\psi(t)} )^2)\Big),
\end{equation}
while the Renyi-2 operator entanglement entropy of $O(t)$ is given by
\begin{equation}
    S^{(2)}_A(O(t)) = -\ln\Big( \sum_i (\lambda^O_i)^2\Big).
\end{equation}
The entanglement and operator entanglement entropies are difficult to tightly bound using standard unitary designs.
Fundamentally, this is because each quantity requires an exponential overhead to experimentally measure.
Here, we show that strong unitary designs with small \emph{relative error} can nonetheless be used to tightly bound both quantities near their Haar-random values.

We prove that the entanglement entropy and for any initial state $\ket{\psi}$ and the operator entanglement entropy for any initial operator $O$ saturate to their Haar-random values $U$ is drawn from a strong unitary design with relative error.
\begin{proposition}\label{prop: ent}
    Consider the state $\ket{\psi(t)} \equiv U \ket{\psi}$ formed by applying a strong unitary 2-design with relative error $\varepsilon$ to any state $\ket{\psi}$.
    The Renyi-2 entanglement entropy of any subsystem of $\ket{\psi(t)}$ is equal to its Haar value to within  error $\varepsilon$.
\end{proposition}
\begin{proposition}\label{prop: op ent}
    Consider the operator $O(t) \equiv U O U^\dagger$ formed by applying a strong unitary 4-design with relative error $\varepsilon$ to any operator $O$.
    The Renyi-2 operator entanglement entropy of any subsystem of $O(t)$ is equal to its Haar value to within  error $\varepsilon$.
\end{proposition}
\noindent We recall from Theorem~\ref{thm:strong-design-depth} that strong unitary 4-designs with relative error $\varepsilon$ can be formed in circuit depth $\mathcal{O}(\log n + \log \log 1/\varepsilon)$ in structured circuits.
Hence, the entanglement and operator entanglement entropies saturate to within $\varepsilon = 1/\exp n$ of their Haar-random values at $\mathcal{O}(\log n)$ depth.

\begin{proof}[Proof of Proposition~\ref{prop: ent}]
    The proposition follows by reformulating the purity as the expectation value of a positive operator on a larger system involving $U$ and $U^*$.
    Let us abbreviate $\rho \equiv \dyad{\psi(t)}$. We have
    \begin{equation}
        \tr_A( \tr_B( \rho )^2 ) = 2^{|A|} \tr( \dyad*{\Psi^A_{\text{EPR}}} \cdot (\rho \otimes \rho^*)),
    \end{equation}
    where $|A|$ denotes the number of qubits in subsystem $A$, and $\ket*{\Psi^A_{\text{EPR}}}$ denotes the EPR state between two copies of subsystem $A$.
    Note that this formula differs from the standard reformulation of the purity in terms of a swap operator, $\tr_A( \tr_B( \rho )^2 ) = \tr( \mathcal{S}_A \cdot (\rho \otimes \rho))$.
    The expression in terms of the EPR state can be obtained from the expression in terms of the swap operator by taking a partial transpose on the second copy of both terms inside the trace.

    The state $\rho \otimes \rho^* = (U \otimes U^*) (\dyad{\psi} \otimes \dyad{\psi^*}) (U^\dagger \otimes U^T)$ is obtained by evolving the state $\ket{\psi} \otimes \ket{\psi^*}$ under one application of $U$ and one application of $U^*$.
    Since $\dyad*{\Psi^A_{\text{EPR}}}$ is positive, the expectation value above is captured within multiplicative error $\varepsilon$ by any strong unitary 2-design with relative error~$\varepsilon$.
\end{proof}

\begin{proof}[Proof of Proposition~\ref{prop: op ent}]
    The proposition follows immediately from Proposition~\ref{prop: ent} by noting that the purity of the operator $O(t)$ is equal to the purity of the state $(O(t) \otimes \mathbbm{1})\ket{\Psi_{\text{EPR}}}$ formed by applying $O(t)$ to one side of the EPR state on a two-copy system.
    The latter state can be written as
    \begin{equation}
        (O(t) \otimes \mathbbm{1})\ket{\Psi_{\text{EPR}}} = (U \otimes U^*)(O \otimes \mathbbm{1}) \ket{\Psi_{\text{EPR}}}.
    \end{equation}
    To estimate the purity, from Proposition~\ref{prop: ent}, we use one copy of the state above and one copy of its conjugate.
    This requires 2 applications of $U$ and 2 applications of $U^*$.
    Hence, the expectation value of the purity is captured to within multiplicative error $\varepsilon$ by any strong unitary 4-design with relative error $\varepsilon$.
\end{proof}

\clearpage
\bibliography{refs}
\bibliographystyle{unsrt}

\end{document}